\documentclass[a4paper,11pt,oneside,final,openbib,pdftex]{scrbook}
\pdfoutput=1

\usepackage[utf8]{inputenc}
\usepackage{bbold}
\usepackage[mathscr]{euscript}
\usepackage{graphicx}
\usepackage{enumerate}
\usepackage{multirow}
\usepackage{subfigure}
\usepackage{dsfont}
\usepackage{slashed}
\usepackage{ytableau}
\usepackage{textcomp}
\usepackage{todonotes}
\usepackage{lipsum}
\usepackage{url}
\usepackage{amsmath}
\usepackage{amssymb}
\usepackage{amsthm}
\usepackage{mathtools}
\usepackage{romannum}
\usepackage{tensor}
\usepackage{leftidx}
\usepackage[compat=1.0.0]{tikz-feynman}
\usepackage{xcolor}
\usepackage{braket}
\usepackage{marvosym}
\usepackage{esint}
\usepackage{wasysym}
\usepackage{gensymb}
\usepackage{footnote}
\usepackage{caption}
\usepackage{float}
\usepackage{tabularx}
\usepackage{pdflscape}
\usepackage{subfigure}
\usepackage{blindtext}
\usepackage{todonotes}
\usepackage{siunitx}
\usepackage{tikz}
\usepackage{pgf}
\usepackage{pgfplots}
\usetikzlibrary{shapes,arrows}
\usepackage{standalone}
\usepackage{hyperref}

\newtheorem*{Lem}{Lemma}

\newcommand*\xbar[1]{
   \hbox{%
     \vbox{%
       \hrule height 0.5pt 
       \kern0.5ex
       \hbox{%
         \kern-0.1em
         \ensuremath{#1}%
         \kern-0.1em
       }%
     }%
   }%
}

\makeatletter
\DeclareRobustCommand*{\bfseries}{%
   \not@math@alphabet\bfseries\mathbf
   \fontseries\bfdefault\selectfont
   \boldmath
}
\makeatother 

\font\afont=cmssbx10 scaled \magstep5     

\font\dfont=cmssbx10 scaled \magstep2     

\hypersetup{
	unicode = true,
	pdfborder = {0 0 0},
	linktoc = page,
	colorlinks,
	linkcolor = blue,
	citecolor = black,
	filecolor = black,
	urlcolor = blue
}

\let\oldtheequation\theequation
\makeatletter
\def\tagform@#1{\maketag@@@{\ignorespaces#1\unskip\@@italiccorr}}
\renewcommand{\theequation}{(\oldtheequation)}
\makeatother

\makeatletter
\renewcommand\paragraph{\@startsection{paragraph}{4}{\z@}%
  {-3.25ex\@plus -1ex \@minus -.2ex}%
  {1.5ex \@plus .2ex}%
  {\normalfont\normalsize\bfseries}
}
\makeatother

\begin{document}

\begin{titlepage}
  \vspace*{6mm}
  \begin{center}
     {\afont Revisiting Hadronic Mass Relations\\}\vspace{0.5cm}
     {\afont from First Order Flavor Symmetry\\}\vspace{0.5cm}
     {\afont Breaking}
     \\[3.5cm]
     {\large by}
     \\[3.5cm]
     {\dfont Luiz Frederic Wagner}
     \\[2cm]
     {\large Master's Thesis in Theoretical Physics\/\\
        Handed in to the Johannes Gutenberg-Universit\"at Mainz on March 18, 2020\/\\
        Revised for publication on March 28, 2021\/\\}
   \end{center}
\end{titlepage}


\newpage
\pagenumbering{roman}
\tableofcontents
\newpage
\pagenumbering{arabic}

\chapter*{Introduction}
\addcontentsline{toc}{chapter}{Introduction}
\markboth{}{Introduction}

The Gell-Mann--Okubo (GMO) mass relations following from $\text{SU}(3)$-flavor symmetry breaking in the strong sector have been proven to be of great success for describing hadron masses and classifying them into multiplets. The GMO mass relations of the decuplet, for instance, allowed Gell-Mann to predict the mass of the $\Omega^-$-particle, before it was discovered (cf. \cite{Zee2016} and \cite{Langacker2017}). Nowadays, it is widely believed that baryons being fermions have to obey GMO relations which are linear in the baryon masses, whilst mesons as bosons follow quadratic GMO relations (cf. \cite{Langacker2017}). This distinction -- allegedly first introduced by R.P. Feynman (cf. \cite{DeSwart1963}) -- is commonly justified with the observation that the mass enters linearly in a fermionic Lagrangian, but quadratically in a bosonic Lagrangian. In a supersymmetrical world, however, one expects a symmetry between fermionic and bosonic multiplets in a supermultiplet. This symmetry implies that every mass relation satisfied by a fermionic multiplet has to be satisfied by its bosonic supersymmetrical counterpart and vice versa. As the argument for distinguishing fermionic and bosonic mass relations should apply to all quantum field theories (QFTs) independent of the question whether the theory is realized in Nature, it should also apply to a supersymmetrical theory. 
Clearly, the symmetry between fermion and boson masses conflicts with Feynman's distinction\footnote{``Feynman's distinction'' is used throughout this thesis as a phrase for the distinction of baryons and mesons into linear and quadratic GMO relations (or linear and quadratic hadronic mass relations, in general), respectively.} in this case. Since this reasoning is independent of the question whether Nature is supersymmetrical, the problem with Feynman's distinction must already arise on a theoretical level rendering the distinction itself questionable. Therefore, I challenge Feynman's distinction in this master's thesis. In the course of the thesis, I will further support this claim by presenting an in-depth analysis of the GMO relations on a theoretical and experimental level.\par
The theoretical side of this thesis concerns the derivation of the mass relations. In the course of this work, we will develop two approaches to the GMO relations: The first is the description of hadron masses within the framework of a hadronic effective field theory (EFT). An example for this EFT approach is chiral perturbation theory (cf. \cite{Scherer2011}). The second approach which we will call state formalism describes the hadrons and their masses as eigenstates and -values of the Hamilton operator, respectively. In both the EFT approach and the state formalism, we will obtain the GMO mass relations as a first order result of the $\text{SU}(3)$-flavor symmetry breaking to the isospin symmetry group $\text{SU}(2)\times\text{U}(1)$.\par
One difference between the EFT approach and the state formalism is central for our investigations: While Feynman's distinction seems to arise naturally in the EFT approach, the state formalism does not exhibit this distinction. At first glance, that seems to indicate that the EFT approach and the state formalism are incompatible, nevertheless, we will see that the discrepancy can easily be explained: The state formalism is only applicable, if the flavor symmetry breaking is small and can be treated as a perturbation. In that case, both linear and quadratic GMO mass relations are actually equivalent and Feynman's distinction is artificial.\par
Additional to the $\text{SU}(3)\rightarrow\text{SU}(2)\times\text{U}(1)$ symmetry breaking, we will also incorporate isospin symmetry breaking, electromagnetic contributions, and heavy quark symmetry into the model of hadron masses. These effects give rise to additional mass relations in the hadronic sector which we will use for further investigations of Feynman's distinction.\par
The part of this work related to experimental data services as a ``proof of concept'' for the theoretical considerations: I will support my claims about Feynman's distinction by analyzing experimental data on hadron masses. In particular, we will see that the experimental hadron masses indicate that both linear and quadratic mass relations are applicable in most cases rendering Feynman's distinction artificial. Even if we need to distinguish between linear and quadratic mass relations, the distinction does not originate from the difference between baryons and mesons.\\
The thesis is organized as follows:\par
\autoref{chap:hadron_masses} provides an instructive introduction to the basic concepts of this work. For the sake of clarity, the mathematical details are postponed to \autoref{sec:GMO_formula}. At the beginning of the chapter (\autoref{sec:mass_matrix}), we will explore an easily understandable access to the GMO relations by considering the \text{SU}(3)-flavor transformation properties of mass matrices corresponding to singlets, triplets, sextets, octets, and decuplets. As a next step (\autoref{sec:Trafo_QCD}), we will consider the transformation properties of a Lagrangian involving three light quark flavors with a flavor symmetric interaction of the quarks exemplified by quantum chromodynamics (QCD). In particular, we will discuss which terms in the Lagrangian may spoil the \text{SU}(3)-flavor symmetry between the quarks. At the end of the chapter (\autoref{sec:EFT+H_Pert}), we will see how the transformation behavior of the Lagrangian from \autoref{sec:Trafo_QCD} dictates the symmetry structure of the hadron masses. We will find that the transition from the quark Lagrangian to the hadron masses can be achieved in two ways, namely the EFT approach and the state formalism. In the context of these two approaches, Feynman's distinction will be addressed.\par
In \autoref{chap:GMO_formula}, the description of hadron masses will be elevated to a higher mathematical level (\autoref{sec:GMO_formula}) and expanded to isospin symmetry breaking, electromagnetic interaction (\autoref{sec:add_con}), and heavy quark symmetry (\autoref{sec:heavy_quark}).\par
\autoref{chap:mass_relations} gives an overview over the relations that follow from the mass formulae of \autoref{chap:GMO_formula} for all multiplets realized in Nature. Additionally, the order of magnitude of the dominant correction(s) is calculated for each mass relation.\par
In \autoref{chap:data}, we check both linear and quadratic mass relations against experimentally determined hadron masses to verify our interpretation of the mass relations. Moreover, we use current scientific results and phenomenological observations to classify unassigned resonances into multiplets and utilize the data of yet incomplete multiplets to predict the masses of the missing particles.

\newpage
\chapter{Hadron Masses under \text{SU}(3)-Flavor Transformations}
\label{chap:hadron_masses}

The derivation of the GMO mass relations features many complex and noteworthy steps with regard to both mathematics and physics. A bottom-up derivation, starting with a Lagrangian and presenting every point in detail, may easily be convoluted and fails to be instructive. Therefore, in this chapter, we only point out the main features that lead to the GMO mass relations. In particular, we cover how the \text{SU}(3)-flavor transformations of hadrons in a multiplet, namely singlet, triplet, sextet, octet, and decuplet, define the transformation of the hadron mass matrices and how \text{SU}(3)-flavor symmetry breaking leads to mass relations in those multiplets. Furthermore, we explain how \text{SU}(3)-flavor symmetry breaking enters a Lagrangian with three quark flavors and how this symmetry breaking can be linked to the hadron masses. The last point is demonstrated in two different ways to show how Feynman's distinction may arise and how it can be understood. Throughout this chapter, many assumptions and mathematical claims are made without deeper reasoning for the sake of clarity. A detailed discussion of those points is postponed to \autoref{sec:GMO_formula}.

\section{Decomposition of Mass Matrices and GMO Mass Relations}\label{sec:mass_matrix}

In the early 1960s, when people began to realize that hadrons show a symmetry structure related to the Lie group \text{SU}(3) (cf. \cite{Gell-Mann1961}), it was observed that hadrons form multiplets of \text{SU}(3). The mathematical definition of the term ``multiplet'' is connected to the term ``representation''. A real or complex \textit{representation} $D^{(\rho)}$ -- sometimes simply written as $\rho$ or denoted by its vector space $V$ -- of a topological group $G$, for instance \text{SU}(3), on a real or complex Hilbert space $V$ is a group homomorphism
\begin{gather*}
D^{(\rho)}:G\rightarrow\text{GL}(V)
\end{gather*}
with $\text{GL}(V)\coloneqq \{A:V\rightarrow V\text{ linear and bounded}\mid A^{-1}\text{ exists and is bounded}\}$ such that the map $\Phi^{(\rho)}:G\times V\rightarrow V, (g,v)\mapsto D^{(\rho)}(g)(v)$ is continuous (cf. \cite{Knapp2001}). A subspace $W\subset V$ is called \textit{invariant}, if $D^{(\rho)}(g)(W)\subset W\ \forall g\in G$, and the representation $D^{(\rho)}$ is called \textit{irreducible}, if the only closed invariant subspaces of $V$ are $\{0\}$ and $V$ (cf. \cite{Knapp2001}). With these definitions in mind, we define a \textit{multiplet} of a group (mostly \text{SU}(3) in this work) to be an irreducible representation of that group.\par
The expression ``hadrons form multiplets of \text{SU}(3)" can now be understood in multiple ways: In the sense of a QFT, we may assume that each hadron which we want to group into one multiplet is described by a field and that the fields transform as the components of a vector in the vector space $V$ of the multiplet $D^{(\rho)}$. For example, if $V$ is the complex vector space $\mathbb{C}^n$ and the hadrons at hand are spin-$\frac{1}{2}$ fermions described by the fields $\psi_a$, then the fields $\psi_a$ transform under \text{SU}(3) as:
\begin{align*}
\begin{pmatrix}\psi_1\\ \vdots\\ \psi_n\end{pmatrix}&\xrightarrow{A\in\,\text{SU}(3)}\begin{pmatrix}\psi^\prime_1\\ \vdots\\ \psi^\prime_n\end{pmatrix} =  D^{(\rho)}(A)\begin{pmatrix}\psi_1\\ \vdots\\ \psi_n\end{pmatrix}\\
\Leftrightarrow \psi_a &\xrightarrow{A\in\,\text{SU}(3)}\psi^\prime_a = \sum\limits^n_{b=1}\left(D^{(\rho)}(A)\right)_{ab}\psi_b.
\end{align*}
Note that if the fields $\psi_a$ can be decomposed in terms of creation and annihilation operators, the transformed fields $\psi^\prime_a$ cannot, in general, be decomposed in the same way. The reason for this is that we cannot, in general, pass down the transformation of the fields to the creation and annihilation operators (cf. \autoref{app:stateform}).\par
Alternatively, we may think of the hadrons in one multiplet as quantum mechanical states $\Ket{a}$ and assume that these states transform under {${A\in\text{SU}(3)}$} as:
\begin{gather*}
\Ket{a}\xrightarrow{A\in\,\text{SU}(3)}\Ket{a^\prime} = \sum\limits_{b} \left(D^{(\rho)}(A)\right)^\ast_{ab}\Ket{b}.
\end{gather*}
Note that the transformation coefficients of $\Ket{b}$ are complex conjugated in regard to the transformation of $\psi_a$. This is due to the fact that the states $\Ket{a}$ are related to $\overline{\psi}_a$.\par
In both cases, we may now identify mass matrices, i.e., self-adjoint matrices whose eigenvalues are the hadron masses. The definition of the mass corresponding to a particle or resonance and its relation to operators and parameters of the Lagrangian are by no means trivial and require an in-depth discussion (cf. \autoref{sec:polemass} and \autoref{app:stateform}). However, we skip this discussion for now and use very simplistic pictures for both cases. In the case of fields, we identify the mass matrix (neglecting loop corrections) with the coefficients of the quadratic fields terms in a(n) (effective) Lagrangian:
\begin{gather*}
\mathcal{L}_M = -\sum\limits_{a,b}\overline{\psi}_a M_{ab}\psi_b.
\end{gather*}
Note that the mass matrix defined by the equation above has to be replaced by the squared mass matrix in the case of bosonic fields. We will discuss this point in more detail in \autoref{sec:EFT+H_Pert}.\par
In the case of states, the mass, as a physical observable, is given by a self-adjoint operator $\hat M$ and the mass matrix can be defined via the matrix elements of $\hat M$:
\begin{gather*}
M_{ab}\coloneqq \Bra{a}\hat M\Ket{b}.
\end{gather*}
If we now assume that $M_{ab} = 0$ if $a$ is a hadron in the multiplet and $b$ is not, only the mass matrix elements of the multiplet contribute to the mass of a hadron in that multiplet. Thus, it is sufficient to only consider the mass matrix elements of the multiplet for the transformation of the hadron masses in that multiplet and we are able to define the transformation of the mass matrix under \text{SU}(3) for both cases:
\begin{align*}
\sum\limits_{a,b}\overline{\psi}_a M_{ab}\psi_b&\xrightarrow{A\in\, \text{SU}(3)}\sum\limits_{a,b}\overline{\psi^\prime}_a M^\prime_{ab}\psi^\prime_b\coloneqq \sum\limits_{a,b}\overline{\psi}_a M_{ab}\psi_b;\\
M_{ab}&\xrightarrow{A\in\, \text{SU}(3)}M^\prime_{ab}\coloneqq \Bra{a^\prime}\hat M\Ket{b^\prime}.
\end{align*}
A quick calculation shows:
\begin{align}
\sum\limits_{a,b}\overline{\psi^\prime}_a M^\prime_{ab}\psi^\prime_b &= \sum\limits_{a,b}\left(\sum\limits_{c}\overline{\psi}_c \left(D^{(\rho)}(A)\right)_{ac}^\ast \right) M^\prime_{ab} \left(\sum\limits_{d}\left(D^{(\rho)}(A)\right)_{bd}\psi_d\right)\nonumber\\
& = \sum\limits_{c,d}\overline{\psi}_c\left(\sum\limits_{a,b}\left(D^{(\rho)}(A)\right)_{ac}^\ast M^\prime_{ab} \left(D^{(\rho)}(A)\right)_{bd}\right)\psi_d\nonumber\\
& = \sum\limits_{c,d}\overline{\psi}_c M_{cd}\psi_d\nonumber\\
\Rightarrow M^\prime_{ab} &= \sum\limits_{c,d}\left(D^{(\rho)}(A^{-1})\right)_{ca}^\ast M_{cd} \left(D^{(\rho)}(A^{-1})\right)_{db};\label{eq:M_compon_field}\\
M^\prime_{ab} &= \Bra{a^\prime}\hat M\Ket{b^\prime}\nonumber\\
&= \left(\sum\limits_{c}\Bra{c}\left(D^{(\rho)}(A)\right)_{ac}\right)\hat M\left(\sum\limits_{d}\left(D^{(\rho)}(A)\right)^\ast_{bd}\Ket{d}\right)\nonumber\\
&= \sum\limits_{c,d}\left(D^{(\rho)}(A)\right)_{ac} M_{cd} \left(D^{(\rho)}(A)\right)^\ast_{bd}\nonumber\\
&= \sum\limits_{c,d}\left(D^{(\rho)}(A)^\dagger\right)^\ast_{ca} M_{cd} \left(D^{(\rho)}(A)^\dagger\right)_{db}\label{eq:M_compon_state}.
\end{align}
In \autoref{eq:M_compon_field}, we have used that $D^{(\rho)}$ is a group homomorphism:\linebreak {${D^{(\rho)}(A)^{-1} = D^{(\rho)}\left(A^{-1}\right)}$}. Without loss of generality, we can assume that the matrix $D^{(\rho)}(A)$ is always unitary, as for every representation of the compact Lie group \text{SU}(3), there exists an Hermitian scalar product for which the representation is unitary (cf. \cite{Hamilton2017}). With this property, we can see that the mass matrix in both cases transforms in the same way, even though $\psi_a$ and $\Ket{a}$ transform differently. If we rewrite the transformation of $M_{ab}$ in matrix notation, we find:
\begin{gather*}
M^\prime = D^{(\rho)}(A)\cdot M\cdot D^{(\rho)}(A)^\dagger.
\end{gather*}
Note that $M^\prime$ is self-adjoint, as $M$ is a mass matrix and, therefore, self-adjoint.\par 
Until now, the discussion of the transformation behavior of $M$ was rather abstract. To make the flavor transformation properties of hadronic mass matrices more concrete, it is instructive to consider examples of multiplets. In addition to clarity, examples of multiplets also provide an easy access to group theoretical mass parameterizations and relations. The assumptions and considerations we need to make in order to find the GMO mass formula and relations in \autoref{chap:GMO_formula} and \autoref{chap:mass_relations} are easily applied to singlets, triplets, sextets, octets, and decuplets. The structures these multiplets exhibit guide our expectation of the flavor transformation behavior of general multiplets and provide us with a list of (mathematical) observations we will prove for all \text{SU}(3)-multiplets in \autoref{sec:GMO_formula}.

\subsection*{Singlet}

First, take the multiplet $D^{(\rho)}$ to be the trivial representation or, equivalently, a singlet of \text{SU}(3), denoted by $\rho = 1$. The singlet operates on an one-dimensional vector space and is simply given by:
\begin{gather*}
D^{(1)}(A) = \mathbb{1}\quad\forall\, A\in\text{SU}(3).
\end{gather*}
This implies that $M$ does not transform at all under \text{SU}(3), meaning that $M$ itself transforms as $1$ under \text{SU}(3), and is just given by one component. Since $M$ is self-adjoint, this component is real and corresponds to the mass of the hadron in the singlet. In this case, the transformation behavior of $M$ does not reveal any information about the mass of the hadron.

\subsection*{Triplet}

A more interesting transformation behavior arises when we take $D^{(\rho)}$ to be a triplet or, equivalently, the fundamental representation of \text{SU}(3), denoted by $\rho = 3$. It operates on the three-dimensional vector space $\mathbb{C}^3$ and is given by:
\begin{gather*}
D^{(3)}(A) = A\quad\forall\, A\in\text{SU}(3).
\end{gather*}
The triplet describes three hadrons which we can identify with the vectors
\begin{gather*}
\begin{pmatrix}1\\0\\0\end{pmatrix},\ \begin{pmatrix}0\\1\\0\end{pmatrix},\ \begin{pmatrix}0\\0\\1\end{pmatrix}.
\end{gather*}
Hence, the mass matrix $M$ is, in this case, a complex self-adjoint $(3\times 3)$-matrix that transforms under $A\in\text{SU}(3)$ as:
\begin{gather*}
M^\prime = A\cdot M\cdot A^\dagger.
\end{gather*}
Note that, while the triplet $3$ is a complex representation acting on $\mathbb{C}^3$, the transformation of $M$ is a real representation of \text{SU}(3), as the space of complex self-adjoint $(3\times 3)$-matrices is a real nine-dimensional vector space.\par
First, let us consider the case where $M = m\mathbb{1}$ is a real factor times the identity. Then:
\begin{gather*}
M^\prime = A\cdot m\mathbb{1}\cdot A^\dagger = mA A^\dagger = m\mathbb{1} = M\quad\text{for } A\in\text{SU}(3).
\end{gather*}
As we can see, $M = m\mathbb{1}$ transforms trivially under \text{SU}(3). Therefore, the multiples of the identity are an invariant subspace under the transformation of $M$ and, since this subspace is one-dimensional, they even furnish an irreducible representation.\par
Next, consider $M$ to be a traceless self-adjoint matrix, i.e., $M^\dagger = M$ and\linebreak $\text{Tr}(M) = 0$. As stated earlier, $M^\prime$ is self-adjoint, if $M$ is self-adjoint. Furthermore, we find:
\begin{gather*}
\text{Tr}(M^\prime) = \text{Tr}(A\cdot M\cdot A^\dagger) = \text{Tr}(M\cdot A^\dagger\cdot A) = \text{Tr}(M) = 0\quad\text{for } A\in\text{SU}(3).
\end{gather*}
This shows that the real eight-dimensional vector space of traceless self-adjoint $(3\times 3)$-matrices is also an invariant subspace under the transformation of $M$. One can show that this subspace even furnishes an irreducible representation under the transformation of $M$. This irreducible representation is called octet or, equivalently, the adjoint representation of \text{SU}(3), denoted by $8$. A basis for the octet, i.e., a basis for the traceless self-adjoint $(3\times 3)$-matrices is given by the Gell-Mann matrices (cf. \cite{Langacker2017}):
\begin{align*}
\lambda_1 &= \begin{pmatrix}0&1&0\\1&0&0\\0&0&0\end{pmatrix},\ \lambda_2 = \begin{pmatrix}0&-i&0\\i&0&0\\0&0&0\end{pmatrix},\ \lambda_3 = \begin{pmatrix}1&0&0\\0&-1&0\\0&0&0\end{pmatrix}\\
\lambda_4 &= \begin{pmatrix}0&0&1\\0&0&0\\1&0&0\end{pmatrix},\ \lambda_5 = \begin{pmatrix}0&0&-i\\0&0&0\\i&0&0\end{pmatrix},\ \lambda_6 = \begin{pmatrix}0&0&0\\0&0&1\\0&1&0\end{pmatrix}\\
\lambda_7 &= \begin{pmatrix}0&0&0\\0&0&-i\\0&i&0\end{pmatrix},\ \lambda_8 = \frac{1}{\sqrt{3}}\begin{pmatrix}1&0&0\\0&1&0\\0&0&-2\end{pmatrix}.
\end{align*}
With these considerations in mind, we can now decompose a general self-adjoint $(3\times 3)$-mass matrix $M$ as:
\begin{gather*}
M = \frac{\text{Tr}(M)}{3}\mathbb{1} + \tilde{M}
\end{gather*}
with $\tilde{M}\coloneqq M - \frac{\text{Tr}(M)}{3}\mathbb{1}$. As we can see, the first term in the decomposition is just a real factor times the identity, while the second term is a traceless self-adjoint matrix. Therefore, we have found a decomposition of $M$ into terms that transform as irreducible representations under the transformation of $M$.\par
There is also a more general, but abstract way that leads to this decomposition. If we take a look at the transformation of $M$ in components (\autoref{eq:M_compon_field} and \autoref{eq:M_compon_state}), we can see that the components of the matrix $D^{(M)}(A)$ that transforms $M$ is a product of the components of the (complex) conjugate representation $\bar{\rho}$ of $\rho$ and of the components of the representation $\rho$ itself:
\begin{gather*}
\left(D^{(M)}(A)\right)_{(ab)(cd)} = \left(D^{(\rho)}(A)\right)_{ac}\cdot \left(D^{(\rho)}(A)\right)^\ast_{bd}\eqqcolon \left(D^{(\rho)}(A)\right)_{ac}\cdot \left(D^{(\bar{\rho})}(A)\right)_{bd}.
\end{gather*}
This equation implies that the representation $D^{(M)}$ is (equivalent to) the tensor product $\rho\otimes\bar{\rho}$. Since every finite-dimensional representation of a compact Lie group decomposes into a (direct) sum of irreducible representations (cf. \cite{Hamilton2017}), $\rho\otimes \bar{\rho}$ as a representation of \text{SU}(3) decomposes into irreducible representations. This decomposition of tensor products of representations into irreducible representations is called \textit{Clebsch-Gordan series} (cf. \cite{DeSwart1963}).\par
In the case of the triplet, this means that the transformation of $M$ is given by $3\otimes \bar{3}$. The Clebsch-Gordan series of $3\otimes \bar{3}$ can be calculated in multiple ways, for instance by using Young tableaux (cf. \cite{Lichtenberg}). One obtains:
\begin{gather*}
3\otimes \bar{3} = 1 \oplus 8.\label{eq:3ast3}
\end{gather*}
The Clebsch-Gordan series of $3\otimes \bar{3}$ implies that the transformation of $M$ contains one singlet and one octet. This result coincides with our prior analysis of the mass matrix $M$, where we identified the singlet with the multiples of the identity and the octet with the traceless self-adjoint $(3\times 3)$-matrices.\par
So far, we have only rewritten a nine-dimensional vector space as a sum of an one-dimensional and an eight-dimensional vector space with interesting transformation properties, but not gained any information about hadron masses or the entries of $M$. In order to do this, we have to make additional assumptions about the transformation behavior of $M$. For instance, we may assume that \text{SU}(3) is an exact symmetry, i.e., that $M$ is invariant under any \text{SU}(3)-transformation: $M^\prime = M\ \text{for } A\in\text{SU}(3)$. Then, we can directly deduce that $M$ has to be an element of the singlet and, therefore, a multiple of the identity, because the octet as an irreducible representation cannot contain an element that transforms trivially under \text{SU}(3). If the octet contained such an element, the span of this element would be a singlet and, hence, an invariant subspace of the octet which contradicts the fact that the octet is irreducible. In this regard, exact \text{SU}(3)-symmetry implies that all hadrons in the triplet have to have the same mass.\par
However, we often find that symmetries are broken in Nature. Therefore, let us consider the case where \text{SU}(3) is not taken to be an exact, but only an approximate symmetry. This means that the term which explicitly breaks the \text{SU}(3)-symmetry, i.e., the octet is small in comparison to the \text{SU}(3)-invariant term, i.e., the singlet:
\begin{gather*}
M = m_0\mathbb{1} + \sum\limits^8_{i=1}m_i\frac{\lambda_i}{2}
\end{gather*}
with $m_i\ll m_0$. Nonetheless, we have to have some residual symmetry, since without it we cannot make any further statement about the masses of the hadrons in the triplet. Hence, we require the subgroup of \text{SU}(3)-transformations that at most mix two specific hadrons, say the hadrons described by $(1,0,0)^\text{T}$ and $(0,1,0)^\text{T}$, to leave the mass matrix $M$ invariant. This subgroup contains \text{SU}(2)-transformations mixing $(1,0,0)^\text{T}$ and $(0,1,0)^\text{T}$:
\begin{gather*}
A = \begin{pmatrix}\tilde{A}&0\\0&1\end{pmatrix}\quad\text{for }\tilde{A}\in\text{SU}(2)
\end{gather*}
and phase modulations not mixing any hadrons:
\begin{gather*}
A = \begin{pmatrix}e^{i\alpha}&0&0\\0&e^{i\alpha}&0\\0&0&e^{-2i\alpha}\end{pmatrix}\quad\text{for }\alpha\in\mathbb{R}.
\end{gather*}
Therefore, all transformations that shall leave $M$ invariant are given by the subgroup
\begin{gather*}
E\coloneqq \left\{\begin{pmatrix}e^{i\alpha}\tilde{A}&0\\0&e^{-2i\alpha}\end{pmatrix}\mid\tilde{A}\in\text{SU}(2),\ \alpha\in\mathbb{R}\right\}.
\end{gather*}
Considering the map
\begin{gather*}
f:\text{SU}(2)\times\text{U}(1)\rightarrow E,\ (\tilde{A},\, e^{i\alpha})\mapsto\begin{pmatrix}e^{i\alpha}\tilde{A}&0\\0&e^{-2i\alpha}\end{pmatrix},
\end{gather*}
we can see that $\text{SU}(2)\times\text{U}(1)$ is the covering group of $E$ covering $E$ twice, as $f(\tilde{A},\, e^{i\alpha})$ is equal to $f(-\tilde{A},\, -e^{i\alpha})$. However, we will ignore this in the future and simply identify $\text{SU}(2)\times\text{U}(1)$ with $E$, because they have the same or similar properties regarding our purposes\footnote{We are mostly interested in the irreducible representations of $E$. Since the map $f$ connecting $\text{SU}(2)\times\text{U}(1)$ and $E$ is a surjective Lie group homomorphism, we can promote every irreducible representation of $E$ via $f$ to an irreducible representation of $\text{SU}(2)\times\text{U}(1)$. This allows us to identify the irreducible representations of $E$ with irreducible representations of $\text{SU}(2)\times\text{U}(1)$. The identification of the irreducible representations is unique, i.e., there are no two distinct irreducible representations of $E$ which are identified with the same irreducible representation of {${\text{SU}(2)\times\text{U}(1)}$}. Note, however, that not every irreducible representation of $\text{SU}(2)\times\text{U}(1)$ corresponds to a representation of $E$. Also, note that $E\cong \text{U}(2)$.}.\par
In this sense, we consider a \text{SU}(3)-symmetry breaking to its conserved subgroup $\text{SU}(2)\times\text{U}(1)$ here. To see the impact of the symmetry breaking on the hadron masses, we have to parametrize all mass matrices $M$ that are invariant under $\text{SU}(2)\times\text{U}(1)$. Clearly, the singlet is invariant under $\text{SU}(2)\times\text{U}(1)$, as it is already invariant under \text{SU}(3). Furthermore, the Gell-Mann matrix $\lambda_8$ is invariant under $\text{SU}(2)\times\text{U}(1)$, since it leaves the first two components unchanged, aside from a factor. Indeed, no other Gell-Mann matrix satisfies this condition and, hence, can be invariant under $\text{SU}(2)\times\text{U}(1)$. This means that $M$ is only parametrized by $\mathbb{1}$ and $\lambda_8$:
\begin{gather*}
M = m_0\mathbb{1} + m_8\frac{\lambda_8}{2} = \begin{pmatrix}m_0 + m_8/2\sqrt{3}&0&0\\0&m_0 + m_8/2\sqrt{3}&0\\0&0&m_0 - m_8/\sqrt{3}\end{pmatrix}.
\end{gather*}
As $M$ is already diagonal, we can directly read off the hadron masses: The masses of $(1,0,0)^\text{T}$ and $(0,1,0)^\text{T}$ -- the two hadrons that mix under $\text{SU}(2)\times\text{U}(1)$ -- are degenerate, while the third hadron has a different mass. However, the difference $\sqrt{3}m_8/2$ between the two non-degenerate masses is small in comparison to the average mass $m_0$, as we assume that the breaking is small.\par
The triplet as a model for hadron masses is realized in both the mesonic and baryonic sector. The $B$- and $D$-mesons show a mass structure very similar to the one described above and likewise the $\Lambda_c/\Xi_c$- and $\Lambda_b/\Xi_b$-baryons. The only difference for both mesons and baryons is that the hadrons corresponding to $(1,0,0)^\text{T}$ and $(0,1,0)^\text{T}$ are not exactly degenerate and show a very small deviation. This difference is related to isospin symmetry breaking.

\subsection*{Sextet}

Although the triplet exhibits a non-trivial mass structure, the triplet as well as the singlet do not lead to any mass relation. The lowest dimensional multiplet $D^{(\rho)}$ which does lead to mass relations is the sextet, denoted by $\rho=6$. It operates on the six-dimensional vector space $\mathbb{C}^6$ and, therefore, describes six hadrons. Like all representations of \text{SU}(3), the sextet can be chosen to be unitary, however, its explicit dependence on $A\in\text{SU}(3)$ is not as simple as in the case of the singlet or triplet. The transformation of the mass matrix $M$, a self-adjoint $(6\times 6)$-matrix now, is obviously given by:
\begin{gather*}
M^\prime = D^{(6)}(A)\cdot M\cdot D^{(6)}(A)^\dagger.
\end{gather*}
Similar to the triplet, the multiples of the identity and the traceless self-adjoint matrices are invariant subspaces of the transformation of $M$. Again, the multiples of the identity furnish an irreducible representation equivalent to the singlet. However, the traceless self-adjoint matrices only furnish a reducible representation, i.e., they contain non-trivial invariant subspaces. We can see this by calculating the Clebsch-Gordan series of $6\otimes \bar{6}$, as explained for the triplet:
\begin{gather*}
6\otimes \bar{6} = 1\oplus 8\oplus 27.
\end{gather*}
The 35 dimensional space of traceless self-adjoint $(6\times 6)$-matrices decomposes into an octet and a 27-dimensional irreducible representation of \text{SU}(3), denoted by $27$. By choosing a basis of the singlet, the octet, and $27$, we can parametrize the mass matrix $M$ in terms of irreducible representations.\par
As for the triplet, we can now consider exact and approximate \text{SU}(3)-symmetry. If the \text{SU}(3)-symmetry is exact, the mass matrix $M$ has to transform trivially and, hence, be an element of the singlet. This means that $M$ is a multiple of the identity and all hadrons in the sextet have the same mass.\par
If we take \text{SU}(3) to be an approximate symmetry, the \text{SU}(3)-breaking contributions of the octet and $27$ to the mass matrix are small in comparison to the contribution of the \text{SU}(3)-invariant singlet. Even though \text{SU}(3)-symmetry is broken, we want to have some residual symmetry. Again, we consider the case where the subgroup $\text{SU}(2)\times \text{U}(1)$ of \text{SU}(3) is an exact symmetry. Group theoretical considerations\footnote{Every representation of a group is also a representation of any subgroup of that group. We can use this here for \text{SU}(3) and $\text{SU}(2)\times\text{U}(1)\subset\text{SU}(3)$: 1, 8, and 27 are representations of \text{SU}(3) and, thus, furnish representations of $\text{SU}(2)\times\text{U}(1)$. These representations decompose into irreducible representations of $\text{SU}(2)\times\text{U}(1)$. Each decomposition of 1, 8, or 27 contains exactly one trivial representation of $\text{SU}(2)\times\text{U}(1)$, each corresponding to one $\text{SU}(2)\times\text{U}(1)$-invariant matrix. One can show this by noting that the sextet 6 decomposes into three different irreducible representations of $\text{SU}(2)\times\text{U}(1)$ and by using a lemma we prove in \autoref{sec:GMO_formula}.} then imply that the mass matrix $M$ is parametrized by three linearly independent matrices, one from each representation 1, 8, and 27.\par
At this point, there is no hierarchy between the contributions of 8 and 27, both could be of the same order of magnitude. However, it was observed in the 1960s that, for hadrons, the contribution from octets is dominant in comparison to contributions from other non-trivial representations. This phenomenological observation, known as \textit{octet enhancement} (cf. \cite{Lichtenberg}), was not understood at that time (cf. \cite{Langacker2017}), but could be explained years later in the framework of QCD. Since we want to apply our considerations to hadron masses, we assume octet enhancement from now on, meaning that we take all contributions to the mass matrix aside from singlets and octets to be negligible. The origin of octet enhancement will be the concern of the last two sections (\autoref{sec:Trafo_QCD} and \autoref{sec:EFT+H_Pert}) in this chapter. Note that the singlet and triplet are trivially subject to octet enhancement, as there are no other contributions aside from singlets and octets in these cases.\par
The mass matrix $M$ can now be parametrized by only two matrices, one corresponding to the singlet and one corresponding to the octet. If we denote the octet matrix with $F^{(6\otimes\bar{6})}_8$, we can write:
\begin{gather*}
M = m_0\mathbb{1} + m_8F^{(6\otimes\bar{6})}_8
\end{gather*}
with $m_8\ll m_0$. There are a lot of unitary representations acting on $\mathbb{C}^6$ that are equivalent to the sextet, all linked by similarity transformations. Every one of these representations is a possible candidate for $D^{(6)}$. In general, $F^{(6\otimes\bar{6})}_8$ is not diagonal for an arbitrary choice of $D^{(6)}$, but we will see in \autoref{sec:GMO_formula} (and by using the weight diagrams in \autoref{chap:mass_relations}) that one can choose $D^{(6)}$ such that $F^{(6\otimes\bar{6})}_8$ is given by:
\begin{gather*}
F^{(6\otimes\bar{6})}_8 = \text{diag}\left(\frac{1}{\sqrt{3}},\frac{1}{\sqrt{3}},\frac{1}{\sqrt{3}},\frac{-1}{2\sqrt{3}},\frac{-1}{2\sqrt{3}},\frac{-2}{\sqrt{3}}\right).
\end{gather*}
For this choice, $M$ is diagonal and we can read off the hadron masses. There are three different masses: $m_1 = m_0 + \frac{m_8}{\sqrt{3}}$ is the mass of three hadrons in the sextet, $m_2 = m_0 - \frac{m_8}{2\sqrt{3}}$ is the mass of two hadrons, and $m_3 = m_0 - \frac{2m_8}{\sqrt{3}}$ is the mass of the remaining hadron. As for the triplet, the difference between two hadron masses is small in comparison to the average mass $m_0$ of the hadrons in the sextet.\par
On top of that, we see that there are three distinct hadron masses which are given by only two parameters. This means that we can find a relation between the hadron masses. The existence of a mass relation follows from assuming $\text{SU}(3)\rightarrow \text{SU}(2)\times \text{U}(1)$ symmetry breaking and octet enhancement. In this work, we call mass relations that follow from these two assumptions \textit{Gell-Mann--Okubo mass relations}. The relation for the sextet is given by:
\begin{gather*}
m_1 - m_2 = \frac{\sqrt{3}}{2}m_8 = m_2 - m_3.
\end{gather*}
This relation states that the difference between two neighboring hadron masses is the same for the entire sextet. Mass relations of this kind are known as \textit{equal spacing rules} and common for totally symmetric multiplets like sextets, decuplets etc.\par
The sextet as a model for hadrons is, to our knowledge, only realized in the baryonic sector. Baryons like $\Sigma_c,\ \Xi^\prime_c,\ \Omega_c$ and $\Sigma_b,\ \Xi^\prime_b,\ \Omega_b$, for instance, exhibit a structure very similar to the sextet structure described above. However, the equal spacing rule is not exactly satisfied, but only holds true to good approximation. As for triplets, the mass degeneracy is lifted in Nature due to isospin symmetry breaking.

\subsection*{Octet}

The octet $D^{(8)}$ operates on an eight-dimensional vector space. As shown in the discussion of the triplet, we can take this vector space to be the space of traceless self-adjoint $(3\times 3)$-matrices. In this case, $D^{(8)}$ is a real representation of \text{SU}(3) as this vector space is a real eight-dimensional space. If we let \text{SU}(3) act on the complexification of this space, i.e., the space of traceless complex $(3\times 3)$-matrices, in the same way, we obtain the octet $D^{(8)}$ as a complex representation of \text{SU}(3). We now identify the space of traceless complex $(3\times 3)$-matrices with the space $\mathbb{C}^8$ to get the octet $D^{(8)}$ as a unitary representation acting on $\mathbb{C}^8$. With this identification, the mass matrix $M$ is a self-adjoint $(8\times 8)$-matrix with the transformation:
\begin{gather*}
M^\prime = D^{(8)}(A)\cdot M\cdot D^{(8)}(A)^\dagger.
\end{gather*}
As for the other multiplets, the decomposition of $M$ into irreducible representations is given by the Clebsch-Gordan series of $8\otimes \bar{8}$. Note, however, that the octet is equivalent to a purely real representation, like explained above. This means that $\bar{8}$ and $8$ are equivalent and it is sufficient to consider $8\otimes 8$ instead of $8\otimes \bar{8}$:
\begin{gather*}
8\otimes 8 = 1\oplus 8\oplus 8\oplus 10\oplus \overline{10}\oplus 27,
\end{gather*}
where $10$ is a ten-dimensional irreducible representation of \text{SU}(3) called decuplet and $\overline{10}$ is its conjugate representation. Again, the singlet $1$ is given by the multiples of the identity which transform trivially and choosing a basis of each irreducible representation gives a parametrization of $M$ in terms of irreducible representations.\par
Let us consider exact \text{SU}(3)-symmetry now. In this case, $M$ has to be an element of the singlet and, hence, a multiple of the identity which means that all hadrons in the octet have the same mass.\par
If we only have approximate \text{SU}(3)-symmetry where the subgroup $\text{SU}(2)\times \text{U}(1)$ is an exact symmetry, group theory dictates that the mass matrix $M$ is parametrized by four linearly independent matrices, each from one of the representations 1, 8, 8, and 27. Imposing octet enhancement, we are left with only three contributions, one from the singlet and two from the two octets. In contrast to the triplet and the sextet, $M$ is given by three terms now instead of two after imposing octet enhancement. Actually, the mass matrix $M$ is only given by either two or three terms for arbitrary multiplets (except for the singlet), if octet enhancement is assumed. Depending on this, the multiplet satisfies different mass relations. We will show this in \autoref{sec:GMO_formula}.\par
If we denote the contribution from the first octet with $F^{(8\otimes\bar{8})}_8$ and the contribution from the second with $D^{(8\otimes\bar{8})}_8$, we can write:
\begin{gather*}
M = m_0\mathbb{1} + m^F_8F^{(8\otimes\bar{8})}_8 + m^D_8D^{(8\otimes\bar{8})}_8,
\end{gather*}
where $m^{F/D}_8\ll m_0$. One can choose $D^{(8)}$ such that $F^{(8\otimes\bar{8})}_8$ and $D^{(8\otimes\bar{8})}_8$ are diagonal, as we will see in \autoref{sec:GMO_formula}. If we make such a choice for $D^{(8)}$, we find:
\begin{align*}
F^{(8\otimes\bar{8})}_8 &= \text{diag}\left(\frac{\sqrt{3}}{2},\frac{\sqrt{3}}{2},0,0,0,0,\frac{-\sqrt{3}}{2},\frac{-\sqrt{3}}{2}\right),\\
D^{(8\otimes\bar{8})}_8 &= \text{diag}\left(\frac{-1}{2\sqrt{3}},\frac{-1}{2\sqrt{3}},\frac{1}{\sqrt{3}},\frac{1}{\sqrt{3}},\frac{1}{\sqrt{3}},\frac{-1}{\sqrt{3}},\frac{-1}{2\sqrt{3}},\frac{-1}{2\sqrt{3}}\right).
\end{align*}
For this choice, $M$ is diagonal and we can read off the hadron masses:
\begin{align*}
m_1 &= m_0 + \frac{\sqrt{3}}{2}m^F_8 - \frac{1}{2\sqrt{3}}m^D_8,\\
m_2 &= m_0 + \frac{1}{\sqrt{3}}m^D_8,\\
m_3 &= m_0 - \frac{1}{\sqrt{3}}m^D_8,\\
m_4 &= m_0 - \frac{\sqrt{3}}{2}m^F_8 - \frac{1}{2\sqrt{3}}m^D_8,
\end{align*}
where $m_1$ is the mass of the first two hadrons, $m_2$ is the mass of the following three hadrons, $m_3$ is the mass of the only non-degenerate hadron, and $m_4$ is the mass of the last two hadrons. Again, the average mass $m_0$ of the hadrons is way bigger than the difference of any two hadron masses in the octet and, like for the sextet, there is one relation between the masses, as four masses are parametrized by only three quantities. The relation can be expressed as:
\begin{gather*}
2m_1 + 2m_4 = 4m_0 - \frac{2}{\sqrt{3}}m^D_8 = m_2 + 3m_3.
\end{gather*}
This equation is the mass relation Gell-Mann wrote down when examining \text{SU}(3)-symmetries in the strong sector (cf. \cite{Gell-Mann1961}) and the original GMO mass relation.\par
The octet as a model for hadrons is very much present in both the mesonic and baryonic sector. Examples include the baryon octet with $J^P = 1/2^+$ containing proton and neutron and the pseudoscalar meson octet containing the pions. However, both of these octets deviate from the described model. While, for the baryon octet, the deviations are small and the GMO mass relation applies with a precision of few percents (cf. \autoref{sec:mass_testing}), the pseudoscalar meson octet shows a large discrepancy from the octet model described above: The average mass of the mesons in the pseudoscalar meson octet -- roughly \SI{400}{MeV} -- is in the same order of magnitude as the difference between meson masses and the GMO mass relation is rather strongly violated (broken by about 15\%; cf. \autoref{sec:mass_testing}). It is also noteworthy that the quadratic GMO relation, i.e., the mass relation where the masses in the GMO relation above are replaced by squared masses, is much better satisfied by the pseudoscalar meson octet than the linear GMO relation, even though it is still slightly broken (cf. \autoref{sec:mass_testing}). Both relations are broken due to $\eta$-$\eta^\prime$-mixing. The difference between linear and quadratic relations regarding the pseudoscalar meson octet will be the concern of discussions in \autoref{sec:EFT+H_Pert} and \autoref{sec:mass_testing}.

\subsection*{Decuplet}

The decuplet $D^{(10)}$ operates, as already stated, on a ten-dimensional vector space. The mass matrix $M$, a self-adjoint $(10\times 10)$-matrix now, is subject to the transformation
\begin{gather*}
M^\prime = D^{(10)}(A)\cdot M\cdot D^{(10)}(A)^\dagger
\end{gather*}
under $A\in \text{SU}(3)$. Its decomposition into irreducible representations is given by the Clebsch-Gordan series of $10\otimes \overline{10}$:
\begin{gather*}
10\otimes \overline{10} = 1\oplus 8\oplus 27\oplus 64,
\end{gather*}
where $64$ is a 64-dimensional irreducible representation of \text{SU}(3). Once again, the singlet $1$ is the space spanned by the identity. If \text{SU}(3) is an exact symmetry, $M$ can only be an element of this singlet which means that, yet again, all hadrons in the decuplet have the same mass.\par
If we consider approximate \text{SU}(3)-symmetry with exact $\text{SU}(2)\times\text{U}(1)$-symmetry, group theoretical considerations show that $M$ is a linear combination of four linearly independent matrices, one from each representation in the Clebsch-Gordan series. Octet enhancement limits the number of free parameters to two, since only the singlet and octet contribution remain. As for the triplet and sextet, the mass matrix is given by two contributions. Denoting the octet contribution with $F^{(10\otimes\overline{10})}_8$, we find:
\begin{gather*}
M = m_0\mathbb{1} + m_8F^{(10\otimes\overline{10})}_8
\end{gather*}
with $m_8\ll m_0$. If $D^{(10)}$ is chosen such that $F^{(10\otimes\overline{10})}_8$ is diagonal, one obtains:
\begin{gather*}
F^{(10\otimes\overline{10})}_8 = \text{diag}\left(\frac{\sqrt{3}}{2},\frac{\sqrt{3}}{2},\frac{\sqrt{3}}{2},\frac{\sqrt{3}}{2},0,0,0,\frac{-\sqrt{3}}{2},\frac{-\sqrt{3}}{2},-\sqrt{3}\right).
\end{gather*}
We can directly read off the hadron masses:
\begin{align*}
m_1 &= m_0 + \frac{\sqrt{3}}{2}m_8,\\
m_2 &= m_0,\\
m_3 &= m_0 - \frac{\sqrt{3}}{2}m_8,\\
m_4 &= m_0 - \sqrt{3}m_8,
\end{align*}
where $m_1$ is the mass of four hadrons, $m_2$ is the mass of three hadrons, $m_3$ is the mass of two hadrons, and $m_4$ is the mass of the only non-degenerate hadron. As two parameters describe four masses, there are two mass relations this time:
\begin{align*}
m_1 - m_2 &= \frac{\sqrt{3}}{2}m_8 = m_2 - m_3,\\
m_2 - m_3 &= \frac{\sqrt{3}}{2}m_8 = m_3 - m_4.
\end{align*}
Like in the case of the sextet, these two relations are equivalent to an equal spacing rule:
\begin{gather*}
m_1 - m_2 = m_2 - m_3 = m_3 - m_4.
\end{gather*}
To our knowledge, decuplets as a model for hadrons are only realized in the baryonic sector. The baryon decuplet with $J^P = 3/2^+$ made up out of the $\Delta$-, $\Sigma^\ast$-, $\Xi^\ast$-, and $\Omega$-baryons is the most prominent example for this. In this decuplet, the equal spacing rule holds true to good approximation. Like for the other multiplets, the degeneracy is lifted due to isospin symmetry breaking.\\\par
The discussion of the hadronic multiplets and the GMO mass relations revealed some reoccurring features. In particular, the previous segments indicate that:
\begin{itemize}
\item The mass matrix of every multiplet contains exactly one singlet $1$. This singlet is given by the multiples of the identity. For exact \text{SU}(3)-symmetry, this implies that every hadron in the multiplet has to have the same mass.
\item The mass matrix of every non-trivial multiplet contains at least one, but at most two octets.
\item The multiplets whose mass matrix contains only one octet are exactly the totally symmetric multiplets like triplets, sextets, decuplets, and so on.
\item Aside from the triplet, all totally symmetric multiplets satisfy equal spacing rules.
\end{itemize}
We will prove these properties in \autoref{sec:GMO_formula}.

\newpage
\section{Transformation Behavior of a 3-Flavors Lagrangian}\label{sec:Trafo_QCD}

As we have seen in the previous section, assuming $\text{SU}(3)\rightarrow\text{SU}(2)\times\text{U}(1)$ symmetry breaking and octet enhancement is crucial for the derivation of the GMO mass relations. The origin of these two assumptions lies in the Lagrangian governing the dynamics of the hadrons and their constituents. In order to understand how \text{SU}(3)-symmetry breaking and octet enhancement arise from a Lagrangian, we have to specify the Lagrangian first. In principle, the Lagrangian we need to investigate would be the Lagrangian of the Standard Model (SM), if we assume that the SM Lagrangian describes hadrons and quarks as their constituents to good approximation. For the sake of simplicity, let us drop everything in the SM Lagrangian which is not of immediate interest for the consideration at hand, meaning that we only keep the three light quarks u, d, and s and the strong interaction. By doing this, we are left with a QCD Lagrangian describing three flavors:\footnote{In the following calculations, all indices like spinor and color indices that are unchanged under flavor transformations are suppressed for the sake of clarity. Only the flavor is given explicitly.}
\begin{gather*}
\mathcal{L}_{\text{QCD}}(\bar{q},q) = \sum\limits_{q\in\{\text{u,d,s}\}}\bar{q}\left(i\slashed{D} -  m_q\right)q + \mathcal{L}_\text{YM},
\end{gather*}
where $q\in\{\text{u, d, s}\}$ ($q\coloneqq q_\text{L} + q_\text{R}$) is the field of the u, d, or s quark, $D_\mu$ is the covariant derivative of QCD containing $\text{SU}(3)_\text{c}$-gauge fields\footnote{$\text{SU}(3)_\text{c}$ is the gauge group of QCD.}, $m_q$ is -- at this stage, i.e., without formulating a renormalization scheme -- the bare quark mass of u, d, or s generated by spontaneous symmetry breaking of the Higgs field, and $\mathcal{L}_\text{YM}$ is the Yang-Mills Lagrangian of QCD containing the kinetic terms and self-interaction of the gauge fields. For physical applications of the theories we present, we need finite values of quark masses, thus, we need to choose a renormalization scheme. The quark mass values we use and reference throughout this work are $\overline{\text{MS}}$ masses at some scale $\mu$ ($\mu = \SI{2}{GeV}$ for u, d, and s quark and $\mu = \overline{m}_Q$ for heavy quarks $Q$) and taken from review \textit{66. Quark Masses} in \cite{PDG}.\par
We immediately see that all terms in the Lagrangian besides the mass term of the quarks are flavor symmetric\footnote{We will see later on in this section how we have to understand this statement.}. For our investigations in \autoref{chap:hadron_masses} and \autoref{chap:GMO_formula}, we only use that the interaction governing the dynamics of the quarks is flavor symmetric. Indeed, the considerations we make in the course of this work (mostly) apply to all theories where the interaction is flavor symmetric. This is the reason why we suppressed all indices but the flavor indices. If we want to restore the color indices and write the gluon fields explicitly, we have to make the following replacements:
\begin{gather*}
\bar{q}\rightarrow\bar{q}_i,\quad q\rightarrow q_j,\\
\slashed{D}\rightarrow \slashed{D}_{ij} = \gamma^\mu\partial_\mu\delta_{ij} - ig_s\gamma^\mu A^a_\mu T^a_{ij},\\
m_q\rightarrow m_q\delta_{ij},\\
\mathcal{L}_\text{YM} = -\frac{1}{4}G^{\mu\nu;\, a}G^a_{\mu\nu}\quad\text{with}\\
G^{a}_{\mu\nu} \coloneqq \partial_\mu A^a_\nu - \partial_\nu A^a_\mu + g_sf^{abc}A^b_\mu A^c_\nu,
\end{gather*}
where we sum over doubly occurring indices, $i$ and $j$ are color indices of the quarks, i.e., ($i$) $j$ is an index transforming under the (anti)fundamental representation of $\text{SU}(3)_\text{c}$, $a$, $b$, and $c$ are color indices of the gluon, i.e., $a$, $b$, and $c$ are indices transforming under the adjoint representation of $\text{SU}(3)_\text{c}$, $g_s$ is the coupling constant of QCD, the fields $A^a_\mu$ are the gauge fields of QCD, the matrices $T^a_{ij}$ are the generators of the adjoint representation of $\text{SU}(3)_\text{c}$, and the constants $f^{abc}$ are the structure constants of $\text{SU}(3)_\text{c}$. Even though the particular kind of interaction is not important for the following considerations as long as the interaction is flavor symmetric, we still label the Lagrangian $\mathcal{L}_\text{QCD}$ with ``QCD'' to indicate that, for our purposes, this Lagrangian already takes the strong interaction into account.\par
Let us now turn to the transformation behavior of $\mathcal{L}_\text{QCD}$ under flavor transformations. By rewriting the Lagrangian as
\begin{gather*}
\mathcal{L}_{\text{QCD}} = \sum\limits_{p,q\in\{\text{u,d,s}\}}\bar{p}\left(i\slashed{D}\delta_{pq} - \mathscr{M}_{pq}\right)q + \mathcal{L}_\text{YM},
\end{gather*}
we see that the quark mass matrix $\mathscr{M}$ is given by:
\begin{gather*}
\mathscr{M} = \begin{pmatrix} m_\text{u} & 0 & 0 \\ 0 & m_\text{d} & 0 \\ 0 & 0 & m_\text{s} \end{pmatrix}.
\end{gather*}
As $\mathscr{M}$ is a diagonal, self-adjoint $(3\times 3)$-matrix, we can express it via $\mathbb{1},\ \lambda_3$ and $\lambda_8$:
\begin{gather}\label{eq:quark_mass_matrix}
\mathscr{M} = \frac{m_\text{u} + m_\text{d} + m_\text{s}}{3}\cdot\mathbb{1} + (m_\text{u} - m_\text{d})\cdot\frac{\lambda_3}{2} + \frac{m_\text{u} + m_\text{d} - 2m_\text{s}}{\sqrt{3}}\cdot\frac{\lambda_8}{2}.
\end{gather}
With this in mind, we can define a flavor transformation of the fields $q$ and the matrix $\mathscr{M}$ under $A\in\text{SU}(3)$:
\begin{align*}
q^\prime &\coloneqq \sum\limits_{\tilde{q}\in\{\text{u,d,s}\}} A_{q\tilde{q}}\, \tilde{q},\\
\sum\limits_{p,q\in\{\text{u,d,s}\}}\bar{p}^{\, \prime}\,\mathscr{M}^\prime_{pq}\, q^\prime &\coloneqq \sum\limits_{p,q\in\{\text{u,d,s}\}}\bar{p}\,\mathscr{M}_{pq}\, q.
\end{align*}
The transformation of $\mathscr{M}$ is then given by:
\begin{gather*}
\mathscr{M}^\prime = A\cdot\mathscr{M}\cdot A^\dagger.
\end{gather*}
The quark mass matrix transformation coincides with the transformation of the triplet mass matrix from \autoref{sec:mass_matrix}. This means that the quark mass matrix decomposes under \text{SU}(3)-flavor transformations into a singlet and an octet, similar to what we have seen in the previous section. The other terms in $\mathcal{L}_{\text{QCD}}$, i.e., $\sum_{q}\bar{q}i\slashed{D}q$ and $\mathcal{L}_\text{YM}$, do not change at all under flavor transformations of the fields $q$ and, hence, transform as a singlet of $\text{SU}(3)$. In this sense, we can say that the Lagrangian decomposes under global flavor transformations, given by
\begin{gather*}
\mathcal{L}^\prime_{\text{QCD}}(\bar{q}^{\, \prime}, q^\prime) \coloneqq \mathcal{L}_{\text{QCD}}(\bar{q}, q),
\end{gather*}
into a singlet and an octet. Under appropriate assumptions, this flavor transformation behavior of the Lagrangian (or of any Lagrangian describing hadrons that decomposes similarly) can be linked to octet enhancement. We will explore this in greater detail in \autoref{sec:EFT+H_Pert}.\par
But before we move on, we take a moment to consider different quark mass configurations and the symmetries/symmetry breakings that arise from these configurations. Let us start by considering the case where all quark masses are equal, i.e., $m_\text{u} = m_\text{d} = m_\text{s}$. In this case, the matrix $\mathscr{M}$ is just the quark mass times the identity and, therefore, a singlet under \text{SU}(3)-flavor transformations. This means that the \text{SU}(3)-flavor transformations are an exact, global symmetry of the Lagrangian.\par
Next, let us consider the case where only two quarks have the same mass. Typically, these quarks are chosen to be u and d meaning $m_\text{u} = m_\text{d}$. Then, according to \autoref{eq:quark_mass_matrix}, $\mathscr{M}$ is a linear combination of $\mathbb{1}$ and $\lambda_8$. Similar to the case of the triplet, the \text{SU}(3)-flavor symmetry is broken now. However, there is still a residual flavor symmetry, namely $\text{SU}(2)\times\text{U}(1)$. Hence, the case where two quark masses are equal is described by $\text{SU}(3)\rightarrow\text{SU}(2)\times\text{U}(1)$ symmetry breaking and leads to the GMO mass relations, as we will see in \autoref{sec:EFT+H_Pert}.\par
Lastly, let us consider the case where all quark masses are taken to be different. Even though $\mathscr{M}$ is a linear combination of $\mathbb{1}$, $\lambda_3$ and $\lambda_8$ and the \text{SU}(3)-flavor symmetry is clearly broken now, a residual symmetry remains. The transformations of $\mathscr{M}$ where $A\in\text{SU}(3)$ is chosen to describe a change of phases along the diagonal,
\begin{gather*}
A = \begin{pmatrix} e^{i\alpha} & 0 & 0 \\ 0 & e^{i\beta} & 0 \\ 0 & 0 & e^{-i(\alpha + \beta)} \end{pmatrix},
\end{gather*}
leave the quark mass matrix invariant. This group is equivalent to $\text{U}(1)\times \text{U}(1)$. In this sense, the case of the most general parametrization of $\mathscr{M}$ deals with\linebreak $\text{SU}(3)\rightarrow \text{U}(1)\times \text{U}(1)$ symmetry breaking. This case will be important in \autoref{sec:add_con}.

\newpage
\section{\text{SU}(3)-Flavor Symmetry Breaking of Hadron Masses}
\label{sec:EFT+H_Pert}

In this section, we want to obtain the hadron mass patterns and relations discussed in \autoref{sec:mass_matrix} from the Lagrangian
of \autoref{sec:Trafo_QCD}. To do so, we have to link the Lagrangian and its transformation behavior under \text{SU}(3)-flavor transformations to hadron masses. We will investigate two approaches to this problem. For the first approach, one assumes that the hadrons we want to consider are described by fields of an EFT and that the transformation behavior of $\mathcal{L}_\text{QCD}$ ``carries over'' to the EFT Lagrangian. This EFT approach is especially interesting as Feynman's distinction between baryons and mesons arises naturally in this approach. The second approach, which we will call \textit{state formalism} for the remainder of this thesis, is based on the assumption that there are states describing the hadrons. We will see that the state formalism predicts GMO relations without distinguishing between baryons and mesons. We conclude this section with a discussion of Feynman's distinction and problems each approach faces.

\subsection*{EFT Approach}

In the EFT approach, the hadrons are described by fields in an EFT Lagrangian $\mathcal{L}$. For the sake of simplicity, we restrict ourselves to a finite number of scalar mesons represented by scalar fields $\phi_i$ and spin-$\frac{1}{2}$ baryons represented by Dirac fields $\psi_i$. $\mathcal{L}$ contains kinetic terms for the fields and an interaction part $\mathcal{L}_\text{Int}$:
\begin{gather*}
\mathcal{L} = \sum^{n_\text{M}}_{i,j=1}\left(\delta_{ij}\left(\partial_\mu\phi_i\right)^\dagger\left(\partial^\mu\phi_j\right) - \left(M^2_\text{M}\right)_{ij}\phi_i^\dagger\phi_j\right) + \sum^{n_\text{B}}_{i,j=1}\xbar\psi_i\left(\delta_{ij}i\slashed\partial - \left(M_\text{B}\right)_{ij}\right)\psi_j + \mathcal{L}_\text{Int},
\end{gather*}
where $M^2_{\text{M}}$ and $M_{\text{B}}$ are the mass matrices of the mesons and baryons and $n_M$ and $n_B$ are the numbers of mesons and baryons described by the EFT, respectively. We assume now that the fields $\phi_i$ and $\psi_i$ transform via unitary representations\linebreak $D^{(\rho_\text{M})}:\text{SU}(3)\rightarrow \text{GL}(\mathbb{C}^{n_\text{M}})$ and $D^{(\rho_\text{B})}:\text{SU}(3)\rightarrow \text{GL}(\mathbb{C}^{n_\text{B}})$ of \text{SU}(3):
\begin{gather*}
\phi_i \xrightarrow{A\,\in\,\text{SU}(3)} \phi^\prime_i = \sum^{n_\text{M}}_{j=1} \left(D^{(\rho_\text{M})}(A)\right)_{ij}\phi_j\quad\text{and}\quad\psi_i \xrightarrow{A\,\in\,\text{SU}(3)} \psi^\prime_i = \sum^{n_\text{B}}_{j=1} \left(D^{(\rho_\text{B})}(A)\right)_{ij}\psi_j.
\end{gather*}
This allows us to define a transformation of $\mathcal{L}$:
\begin{gather*}
\mathcal{L}^\prime(\phi_i^\prime, \psi_i^\prime) \coloneqq \mathcal{L}(\phi_i, \psi_i).
\end{gather*}
Let us now assume that the transformation behavior of $\mathcal{L}_\text{QCD}$ under $A\in\text{SU}(3)$ ``carries over'' to $\mathcal{L}$. As explained in \autoref{sec:Trafo_QCD}, $\mathcal{L}_\text{QCD}$ decomposes under \text{SU}(3)-flavor transformations into a singlet and an octet:
\begin{gather}\label{eq:L_QCD_decomp}
\mathcal{L}_\text{QCD} = \mathcal{L}^{0}_\text{QCD} + \varepsilon_3\cdot\mathcal{L}^{8}_{\text{QCD};\, 3} + \varepsilon_8\cdot\mathcal{L}^{8}_{\text{QCD};\, 8},
\end{gather}
where $\varepsilon_3 \coloneqq m_\text{u} - m_\text{d}$ and $\varepsilon_8 \coloneqq \frac{m_\text{u} + m_\text{d} - 2m_\text{s}}{\sqrt{3}}$ and
\begin{align*}
\mathcal{L}^{0}_{\text{QCD}} &\coloneqq \sum\limits_{q\in\{\text{u,d,s}\}}\bar{q}\left(i\slashed{D} -  \frac{m_\text{u} + m_\text{d} + m_\text{s}}{3}\right)q + \mathcal{L}_\text{YM},\\
\mathcal{L}^{8}_{\text{QCD};\, k} &\coloneqq -\sum\limits_{p,q\in\{\text{u,d,s}\}}\frac{\bar{p}\left(\lambda_k\right)_{pq}q}{2}\quad\forall k\in\{1,\ldots, 8\}.
\end{align*}
The singlet term of $\mathcal{L}_{\text{QCD}}$ is given by $\mathcal{L}^{0}_{\text{QCD}}$, while the octet term is a linear combination of $\mathcal{L}^{8}_{\text{QCD};\, 3}$ and $\mathcal{L}^{8}_{\text{QCD};\, 8}$. We assume now that $\mathcal{L}$ transforms to first order in $\text{SU}(3)$-symmetry breaking in a similar way under \text{SU}(3):
\begin{gather}\label{eq:Sing+Oct}
\mathcal{L} = \mathcal{L}^{0} + \tilde{\varepsilon}_3\cdot\mathcal{L}^{8}_{3} + \tilde{\varepsilon}_8\cdot\mathcal{L}^{8}_{8} + \mathcal{O}(\tilde{\varepsilon}_i\tilde{\varepsilon}_j),
\end{gather}
where $\mathcal{L}^{0}$ is a singlet under \text{SU}(3), $\mathcal{L}^{8}_{3}$ and $\mathcal{L}^{8}_{8}$ are the 3rd and 8th component of an octet, and $\tilde{\varepsilon}_3\propto \varepsilon_3$ and $\tilde{\varepsilon}_8\propto \varepsilon_8$ are parameters governing how strongly \text{SU}(3) is broken. This equation with $\tilde{\varepsilon}_3 = 0$ corresponds to $\text{SU}(3)\rightarrow\text{SU}(2)\times \text{U}(1)$ symmetry breaking to first order in $\tilde{\varepsilon}_8$. Since the kinetic terms in $\mathcal{L}$ for both $\phi_i$ and $\psi_i$ and the interaction term $\mathcal{L}_\text{Int}$ do not mix into each other under \text{SU}(3)-transformations of the fields, both kinetic terms and $\mathcal{L}_\text{Int}$ have to transform analogously to \autoref{eq:Sing+Oct} under \text{SU}(3)-transformations. The part containing derivatives of fields in each kinetic term is a singlet under \text{SU}(3), hence, the mass terms have the same decomposition as in \autoref{eq:Sing+Oct}:
\begin{align*}
\sum^{n_\text{M}}_{i,j=1} \left(M^2_\text{M}\right)_{ij}\phi_i^\dagger\phi_j &= \mathcal{L}^{0}_{\text{M}} + \tilde{\varepsilon}_3\cdot\mathcal{L}^{8}_{\text{M};\, 3} + \tilde{\varepsilon}_8\cdot\mathcal{L}^{8}_{\text{M};\, 8} + \mathcal{O}(\tilde{\varepsilon}_i\tilde{\varepsilon}_j),\\
\sum^{n_\text{B}}_{i,j=1}\bar{\psi}_i \left(M_\text{B}\right)_{ij}\psi_j &= \mathcal{L}^{0}_{\text{B}} + \tilde{\varepsilon}_3\cdot\mathcal{L}^{8}_{\text{B};\, 3} + \tilde{\varepsilon}_8\cdot\mathcal{L}^{8}_{\text{B};\, 8} + \mathcal{O}(\tilde{\varepsilon}_i\tilde{\varepsilon}_j),
\end{align*}
where $\mathcal{L}^{0}_{\text{M}/\text{B}}$, $\mathcal{L}^{8}_{\text{M}/\text{B};\, 3}$, and $\mathcal{L}^{8}_{\text{M}/\text{B};\, 8}$ behave similar to \autoref{eq:Sing+Oct}. The transformation of the mass terms under $A\in \text{SU}(3)$ can be expressed as a transformation of the mass matrices:
\begin{align*}
\sum^{n_\text{M}}_{i,j=1} \left(M^{2\,\prime}_\text{M}\right)_{ij}\phi_i^{\prime\, \dagger}\phi^\prime_j = \sum^{n_\text{M}}_{i,j=1} \left(M^{2}_\text{M}\right)_{ij}\phi_i^{\dagger}\phi_j &\Leftrightarrow M^{2\,\prime}_\text{M} = D^{(\rho_\text{M})}(A)\cdot M^2_\text{M}\cdot D^{(\rho_\text{M})}(A)^\dagger,\\
\sum^{n_\text{B}}_{i,j=1}\bar{\psi}^\prime_i \left(M^\prime_\text{B}\right)_{ij}\psi^\prime_j = \sum^{n_\text{B}}_{i,j=1}\bar{\psi}_i \left(M_\text{B}\right)_{ij}\psi_j &\Leftrightarrow M^{\prime}_\text{B} = D^{(\rho_\text{B})}(A)\cdot M_\text{B}\cdot D^{(\rho_\text{B})}(A)^\dagger.
\end{align*}
It is now easy to see that the mass matrices have to satisfy the following structure:
\begin{align*}
M^2_\text{M} &= M^{2;\, 0}_\text{M} + \tilde{\varepsilon}_3\cdot M^{2;\, 8}_{\text{M};\, 3} + \tilde{\varepsilon}_8\cdot M^{2;\, 8}_{\text{M};\, 8} + \mathcal{O}(\tilde{\varepsilon}_i\tilde{\varepsilon}_j),\\
M_\text{B} &= M^{0}_\text{B} + \tilde{\varepsilon}_3\cdot M^{8}_{\text{B};\, 3} + \tilde{\varepsilon}_8\cdot M^{8}_{\text{B};\, 8} + \mathcal{O}(\tilde{\varepsilon}_i\tilde{\varepsilon}_j),\\
\end{align*}
with $M^{2;\, 0}_\text{M}$/$M^{0}_\text{B}$ being singlets of \text{SU}(3) and $M^{2;\, 8}_{\text{M};\, 3}$/$M^{8}_{\text{B};\, 3}$ and $M^{2;\, 8}_{\text{M};\, 8}$/$M^{8}_{\text{B};\, 8}$ being the 3rd and 8th component of an octet of \text{SU}(3), respectively.\par
To proceed, we need to require that \text{SU}(3) is an approximate symmetry, i.e., the \text{SU}(3)-invariant singlet contribution is much larger than the \text{SU}(3)-breaking contribution, in particular the octet contribution. For most hadrons like baryons or heavy mesons, this holds true. Usually, the octet contribution in \text{SU}(3)-multiplets, roughly given by their mass splitting, is of the order of \SI{100}{MeV}, while the singlet contribution given by the average of the masses in the multiplet is about \SI{1}{GeV} or higher. Furthermore, we set $\tilde{\varepsilon}_3$ equal to 0, as the 3rd component of the octet leads to isospin symmetry breaking which we want to discuss in \autoref{sec:add_con}.\par
With this, we can calculate the eigenvalues of the mass matrices, i.e., the hadron masses, by treating the \text{SU}(3)-breaking term as a small perturbation. In order to do so, we need to know the eigenstates and eigenvalues of the unperturbed matrices first. The unperturbed terms are just the singlets. Since $D^{(\rho_{\text{M}/\text{B}})}$ is unitary, we find:
\begin{align*}
D^{(\rho_\text{M})}(A)\cdot M^{2;\, 0}_\text{M}\cdot D^{(\rho_\text{M})}(A)^{-1} &= M^{2;\, 0}_\text{M}\quad\forall A\in \text{SU}(3),\\
D^{(\rho_\text{B})}(A)\cdot M^0_\text{B}\cdot D^{(\rho_\text{B})}(A)^{-1} &= M^0_\text{B}\quad\forall A\in \text{SU}(3)
\end{align*}
or equivalently:
\begin{gather*}
\left[M^{2;\, 0}_\text{M}, D^{(\rho_\text{M})}(A)\right] = 0\quad\text{and}\quad\left[M^0_\text{B}, D^{(\rho_\text{B})}(A)\right] = 0\quad\forall A\in \text{SU}(3).
\end{gather*}
As $D^{(\rho_\text{M})}$ and $D^{(\rho_\text{B})}$ are finite-dimensional, unitary representations, they decompose completely into a direct sum of irreducible representations or multiplets of \text{SU}(3). We will show in \autoref{sec:GMO_formula} that the commutation relations imply the existence of complete orthonormal eigenbases {${\{v^0_{\text{M};\, i}\mid i = 1,\ldots, n_M\}}$} of $M^{2;\, 0}_\text{M}$ and {${\{v^0_{\text{B};\, i}\mid i = 1,\ldots, n_B\}}$} of $M^0_\text{B}$ such that each eigenbasis consists of bases for the multiplets from the decomposition of $D^{(\rho_\text{M})}$ and $D^{(\rho_\text{B})}$ and that the eigenvalues of eigenvectors in the same multiplet are equal. Hence, $M^{2;\, 0}_\text{M}$ and $M^0_\text{B}$ are constant on those multiplets. This implies that the hadron masses are degenerate in each \text{SU}(3)-flavor multiplet, if there is no perturbation, i.e., if \text{SU}(3) is an exact symmetry.\par
With this knowledge, we can calculate the hadron masses to first order in perturbation theory:
\begin{align*}
m^2_{\text{M};\, i} &= v^{0\, \dagger}_{\text{M};\, i}M^{2;\, 0}_\text{M} v^{0}_{\text{M};\, i} + \tilde{\varepsilon}_8\cdot v^{0\, \dagger}_{\text{M};\, i} M^{2;\, 8}_{\text{M};\, 8} v^{0}_{\text{M};\, i} + \mathcal{O}(\tilde{\varepsilon}^{\, 2}_8),\\
m_{\text{B};\, i} &= v^{0\, \dagger}_{\text{B};\, i}M^{0}_\text{B} v^{0}_{\text{B};\, i} + \tilde{\varepsilon}_8\cdot v^{0\, \dagger}_{\text{B};\, i} M^{8}_{\text{B};\, 8} v^{0}_{\text{B};\, i} + \mathcal{O}(\tilde{\varepsilon}^{\, 2}_8),
\end{align*}
where $m^2_{\text{M};\, i}$ is the mass squared of the meson $i$ and $m_{\text{B};\, i}$ is the mass of the baryon $i$. To obtain these formulae, we assumed that it is sufficient to consider the \text{SU}(3)-multiplets separately. We will see in \autoref{sec:GMO_formula} how we can understand this statement. At this point, it should be mentioned that these two formulae are only correct if the eigenbases $\{v^0_{\text{M};\, i}\mid i = 1,\ldots, n_M\}$ and $\{v^0_{\text{B};\, i}\mid i = 1,\ldots, n_B\}$ are chosen such that
\begin{align*}
v^{0\, \dagger}_{\text{M};\, i} M^{2;\, 8}_{\text{M};\, 8} v^{0}_{\text{M};\, j} &= \delta_{ij}\cdot v^{0\, \dagger}_{\text{M};\, i} M^{2;\, 8}_{\text{M};\, 8} v^{0}_{\text{M};\, i}\quad\text{and}\\
v^{0\, \dagger}_{\text{B};\, i} M^{8}_{\text{B};\, 8} v^{0}_{\text{B};\, j} &= \delta_{ij}\cdot v^{0\, \dagger}_{\text{B};\, i} M^{8}_{\text{B};\, 8} v^{0}_{\text{B};\, i},
\end{align*}
if $v^{0}_{\text{M/B};\, i}$ and $v^{0}_{\text{M/B};\, j}$ are elements of the same multiplet. These two conditions are a consequent of degenerate perturbation theory (cf. \autoref{sec:GMO_formula}).\par
Now we need to convince ourselves that the given expressions for the masses actually reflect the patterns explained in \autoref{sec:mass_matrix}. We only consider the baryonic case for this, the mesonic case works analogously. To do so, we consider a multiplet $D^{(\sigma)}$ from the decomposition of $D^{(\rho_\text{B})}$. Without loss of generality, we can take the first $d = \text{dim}(V^{(\sigma)})$ vectors of $\{v^0_{\text{B};\, i}\mid i = 1,\ldots, n_B\}$ to be a basis of the vector space $V^{(\sigma)}$ on which $D^{(\sigma)}$ acts. Let us define the matrix $m^{(\sigma)}_\text{B}$ by:
\begin{gather*}
m^{(\sigma)}_{\text{B};\, kl} \coloneqq v^{0\, \dagger}_{\text{B};\, k}M^{0}_\text{B} v^{0}_{\text{B};\, l} + \tilde{\varepsilon}_8\cdot v^{0\, \dagger}_{\text{B};\, k} M^{8}_{\text{B};\, 8} v^{0}_{\text{B};\, l}
\end{gather*}
with $k,l\in\{1,\ldots,d\}$. The diagonal elements of the already diagonalized matrix $m^{(\sigma)}_\text{B}$ coincide with the baryon masses of the multiplet $D^{(\sigma)}$ to first order in $\tilde{\varepsilon}_8$. We can define a transformation of $m^{(\sigma)}_\text{B}$ under $A\in\text{SU}(3)$:
\begin{gather*}
m^{(\sigma)\,\prime}_\text{B} \coloneqq D^{(\sigma)}(A)\cdot m^{(\sigma)}_\text{B}\cdot D^{(\sigma)}(A)^\dagger\\
\text{with } \left(D^{(\sigma)}(A)\right)_{kl} \coloneqq v^{0\,\dagger}_{\text{B};\, k}D^{(\rho_\text{B})}(A)v^{0}_{\text{B};\, l}.
\end{gather*}
This transformation corresponds to the transformation of mass matrices defined in \autoref{sec:mass_matrix}. A quick calculation now shows that this transformation is equivalent to a transformation of $M^0_\text{B} + \tilde{\varepsilon}_8\cdot M^{8}_{\text{B};\, 8}$ under $A\in\text{SU}(3)$:
\begin{align*}
m^{(\sigma)\,\prime}_{\text{B};\, kl} &= \sum\limits^d_{r,s = 1} \left(D^{(\sigma)}(A)\right)_{kr}\cdot m^{(\sigma)}_{\text{B};\, rs}\cdot \left(D^{(\sigma)}(A)\right)^\ast_{ls}\\
&= \left[\sum\limits^d_{r=1}\left(D^{(\sigma)}(A)\right)^\ast_{kr} v^{0}_{\text{B};\, r}\right]^\dagger M^0_\text{B} + \tilde{\varepsilon}_8\cdot M^{8}_{\text{B};\, 8}\left[\sum\limits^d_{s=1}\left(D^{(\sigma)}(A)\right)^\ast_{ls} v^{0}_{\text{B};\, s}\right]\\
&= v^{0\,\dagger}_{\text{B};\, k}D^{(\rho_\text{B})}(A)\left(M^0_\text{B} + \tilde{\varepsilon}_8\cdot M^{8}_{\text{B};\, 8}\right)D^{(\rho_\text{B})}(A)^\dagger v^{0}_{\text{B};\, l}\\
&= v^{0\,\dagger}_{\text{B};\, k}\left(M^{0\,\prime}_\text{B} + \tilde{\varepsilon}_8\cdot M^{8\,\prime}_{\text{B};\, 8}\right)v^{0}_{\text{B};\, l},
\end{align*}
where we used
\begin{align*}
\sum\limits^d_{r=1}\left(D^{(\sigma)}(A)\right)^\ast_{kr} v^{0}_{\text{B};\, r} &= \sum\limits^d_{r=1}\left(v^{0\,\dagger}_{\text{B};\, k}D^{(\rho_\text{B})}(A)v^{0}_{\text{B};\, r}\right)^\ast v^{0}_{\text{B};\, r}\\
&= \sum\limits^{n_\text{B}}_{r=1}\left(v^{0\,\dagger}_{\text{B};\, k}D^{(\rho_\text{B})}(A)v^{0}_{\text{B};\, r}\right)^\dagger v^{0}_{\text{B};\, r}\\
&= \sum\limits^{n_\text{B}}_{r=1}\left(v^{0\,\dagger}_{\text{B};\, r}D^{(\rho_\text{B})}(A)^\dagger v^{0}_{\text{B};\, k}\right) v^{0}_{\text{B};\, r}\\
&= \left(\sum\limits^{n_\text{B}}_{r=1} v^{0}_{\text{B};\, r}\cdot v^{0\,\dagger}_{\text{B};\, r}\right)D^{(\rho_\text{B})}(A)^\dagger v^{0}_{\text{B};\, k}\\
&= D^{(\rho_\text{B})}(A)^\dagger v^{0}_{\text{B};\, k}.
\end{align*}
In the second line of the second calculation, we used the fact that the first $d$ eigenvectors $v^0_{\text{B};\, k}$ live in the same invariant subspace and that the eigenbasis is orthogonal, meaning $v^{0\,\dagger}_{\text{B};\, k}D^{(\rho_\text{B})}(A)v^{0}_{\text{B};\, r} = 0$ for $k\leq d$ and $r> d$. In the last line, we used the completeness of the eigenbasis.\par
As $M^0_\text{B}$ transforms as a singlet and $M^8_{\text{B};\, 8}$ transforms as the 8th component of an octet, $m^{(\sigma)}_\text{B}$ transforms as a singlet plus an octet under \text{SU}(3). This means that octet enhancement applies to the transformation of $m^{(\sigma)}_\text{B}$. Furthermore, the 8th component of an octet is invariant under $\text{SU}(2)\times \text{U}(1)$, as we have seen in \autoref{sec:mass_matrix}. Therefore, $m^{(\sigma)}_\text{B}$ is invariant under $\text{SU}(2)\times \text{U}(1)$ as well, since $M^0_\text{B}$ is invariant under any \text{SU}(3)-transformation. Lastly, the \text{SU}(3)-breaking contribution to $m^{(\sigma)}_\text{B}$ is small in comparison to the \text{SU}(3)-invariant term, hence, \text{SU}(3) is an approximate symmetry. In total, the mass matrix $m^{(\sigma)}_\text{B}$ is subject to $\text{SU}(3)\rightarrow \text{SU}(2)\times \text{U}(1)$ symmetry breaking with octet enhancement and approximate \text{SU}(3)-symmetry. As seen in \autoref{sec:mass_matrix}, these properties lead to the GMO mass relations for sextets, octets, and decuplets.\par
Let us summarize what we have found: By describing the hadrons in an EFT and requiring that the transformation behavior of the EFT Lagrangian coincides with the transformation behavior of the fundamental Lagrangian $\mathcal{L}_\text{QCD}$ to first order in flavor symmetry breaking, we found that the baryon masses obey linear GMO relations, whilst meson masses obey quadratic GMO relations to first order in \text{SU}(3)-symmetry breaking. The difference between baryons and mesons can be traced back to their treatment in the EFT Lagrangian: Since baryons are fermions, their mass enters linearly into the Lagrangian. In contrast to that, mesons are bosons and, hence, their mass enters quadratically into the Lagrangian. This means that Feynman's distinction arises naturally in the EFT approach. However, this does not imply that, in Nature, baryons are only subject to linear relations, while mesons only follow quadratic relations. We will discuss this point in detail at the end of this section.

\subsection*{State Formalism}

Alternative to the EFT approach, we can try to link the masses of the hadrons to $\mathcal{L}_\text{QCD}$ via the corresponding Hamilton operator $H_\text{QCD}$. In this approach, the state formalism, we describe the hadrons by eigenstates of $H_\text{QCD}$ where the eigenvalues of these states correspond to the masses of the hadrons. We then calculate these eigenvalues in a perturbative treatment of $H_\text{QCD}$. $H_\text{QCD}$ is given by:
\begin{gather*}
H_\text{QCD} = \int d^3x\, \mathcal{H}_\text{QCD};\quad \mathcal{H}_\text{QCD} = \sum_{q\in\{\text{u,d,s}\}} \frac{\partial \mathcal{L}_\text{QCD}}{\partial \dot{q}}\dot{q} - (\mathcal{L}_\text{QCD} - \mathcal{L}_\text{YM}) + \mathcal{H}_\text{YM},
\end{gather*}
where $\mathcal{H}_\text{YM}$ is the Hamiltonian of $\mathcal{L}_\text{YM}$. With $\frac{\partial \mathcal{L}_\text{QCD}}{\partial \dot{q}} = iq^\dagger$, we obtain:
\begin{gather*}
\mathcal{H}_\text{QCD} = \sum\limits_{q\in\{\text{u,d,s}\}} iq^\dagger\dot{q} + \mathcal{L}_\text{YM} + \mathcal{H}_\text{YM} - \mathcal{L}_\text{QCD}.
\end{gather*}
Since $\sum\limits_{q\in\{\text{u,d,s}\}} iq^\dagger\dot{q} + \mathcal{L}_\text{YM} + \mathcal{H}_\text{YM}$ transforms as a singlet under \text{SU}(3)-transformations and $\mathcal{L}_\text{QCD}$ transforms as a singlet plus an octet, $\mathcal{H}_\text{QCD}$ transforms as a singlet plus an octet as well. Flavor transformations are global transformations, therefore, $H_\text{QCD}$ exhibits the same transformation behavior as $\mathcal{H}_\text{QCD}$ and, in turn, $\mathcal{L}_\text{QCD}$:
\begin{gather*}
H_\text{QCD} = H^{0}_\text{QCD} + \varepsilon_3\cdot H^{8}_{\text{QCD};\, 3} + \varepsilon_8\cdot H^{8}_{\text{QCD};\, 8},
\end{gather*}
where $H^{0}_\text{QCD}$, $H^{8}_{\text{QCD};\, 3}$, and $H^{8}_{\text{QCD};\, 8}$ are defined analogously to \autoref{eq:L_QCD_decomp}.\par
Now, in order to make statements about hadron masses by applying a perturbative treatment to the Hamilton operator, we need to make three assumptions:\footnote{The assumptions listed here do not apply to QFTs in a rigorous mathematical sense. However, they may hold true to good approximation. One can find a discussion of this point among others in \autoref{app:stateform}.}
\begin{itemize}
 \item[1)] For every hadron $a$, there exists an eigenstate $\Ket{a}$ with $\Braket{a|a} = 1$ of the Hamilton operator $H_\text{QCD}$ from which the vacuum energy is already subtracted such that the mass $m_a$ of the hadron $a$ is given by
 \begin{gather*}
  m_a = \Bra{a}H_\text{QCD} \Ket{a}.
 \end{gather*}
 \item[2)] The subspace $V$ of the physical states which is spanned by the states $\Ket{a}$ from 1), i.e., {${V:= \overline{\text{Span}\left\{\Ket{a}\mid a\text{ hadron}\right\}}}$}, is a Hilbert space.
 \item[3)] There is a unitary representation $D^{(\rho)}:V\rightarrow V$ of \text{SU}(3) on $V$ such that the following equation holds for every $A\in\text{SU}(3)$:
 \begin{gather*}
  \Bra{a} D^{(\rho)}(A)^\dagger\circ H_\text{QCD}\left(\bar{q},\, q\right)\circ D^{(\rho)}(A) \Ket{b} = \Bra{a}H_\text{QCD}\left(\bar{q}^{\, \prime},\, q^\prime\right) \Ket{b}\ \ \forall\Ket{a},\Ket{b}\in V
 \end{gather*}
 where $q^\prime\coloneqq \sum\limits_{\tilde{q}\in\{\text{u,d,s}\}}A_{q \tilde{q}}\cdot \tilde{q}$.
\end{itemize}
I cannot prove these assumptions, but they are explored and motivated in \autoref{app:stateform}. It should be noted at this point that we have yet to give a proper definition of the hadron mass. We postpone this discussion to \autoref{sec:polemass}.\par
Given these assumptions and assuming that the subtraction of vacuum energy does not spoil the transformation behavior of $H_\text{QCD}$, we can calculate the hadron masses. The following steps are very similar to the calculation of hadron masses in the EFT approach: For the computation of the masses, we make use of perturbation theory again. Like for the EFT approach, we take the \text{SU}(3)-symmetry breaking to be small such that \text{SU}(3) is an approximate symmetry. This means that the eigenvalues of $H^0_\text{QCD}$ are much larger than the eigenvalues of $\varepsilon_3\cdot H^{8}_{\text{QCD};\, 3} + \varepsilon_8\cdot H^{8}_{\text{QCD};\, 8}$. Furthermore, we take $\varepsilon_3$ to be zero, as, for now, we do not consider isospin symmetry breaking.\par
For a perturbative treatment, we first have to know the eigenvalues and -vectors of the unperturbed operator, i.e., $H^0_\text{QCD}$. However, we are not interested in the entire spectrum of $H_\text{QCD}$, but only in its eigenvalues on $V$. Therefore, we can restrict $H_\text{QCD}$ and likewise $H^0_\text{QCD}$ and $H^8_{\text{QCD};8}$ to $V$. Hereby, we understand the restriction\footnote{The steps performed here may not be rigorous in a mathematical sense. Their purpose is rather to guide our intuition.} $H_\text{QCD}\vert_V$ of $H_\text{QCD}$ to $V$ to be
\begin{gather*}
H_\text{QCD}\vert_V\ket{b} \coloneqq\sum_a \Braket{a|H_\text{QCD}|b}\cdot\ket{a}\quad\forall\ket{b}\in V,
\end{gather*}
where $\{\ket{a}\}$ is any complete orthonormal basis of $V$. We define $H^0_\text{QCD}\vert_V$ and $H^8_{\text{QCD};8}\vert_V$ in a similar way. For the remainder of this section, we will only consider the restricted operators, therefore, we are going to drop ``$\vert_V$'' from now on. Note that we already had to understand $H_\text{QCD}$ in assumption 3) as $H_\text{QCD}\vert_V$, but suppressed ``$\vert_V$''.\par
Using assumption 3) and the fact that $H^0_\text{QCD}$ is a singlet under \text{SU}(3), we find for the unperturbed operator:
\begin{gather*}
D^{(\rho)}(A)\, H^0_\text{QCD}\, D^{(\rho)}(A)^\dagger = H^0_\text{QCD}\  \Leftrightarrow\  \left[D^{(\rho)}(A),\, H^0_\text{QCD}\right] = 0\quad\forall A\in\text{SU}(3).
\end{gather*}
As we will see in \autoref{sec:GMO_formula}, this implies that (the closure of) each eigenspace of $H^0_\text{QCD}$ is an invariant subspace of $D^{(\rho)}$. Similar to the EFT approach, $D^{(\rho)}$ decomposes into multiplets of \text{SU}(3) on this space. This means that $H^0_\text{QCD}$ is constant on multiplets of \text{SU}(3) which again implies that the hadron masses are degenerate in each \text{SU}(3)-flavor multiplet, if there is no perturbation, i.e., if \text{SU}(3) is an exact symmetry.\par
This allows us to compute the hadron masses to first order in perturbation theory. If we consider each multiplet separately (cf. \autoref{sec:GMO_formula}), we find for the hadron masses in a multiplet $D^{(\sigma)}$:
\begin{gather*}
m_{a} = m^{(0)} + \varepsilon_8\cdot\Braket{a^{(\sigma)}|H^8_{\text{QCD};8}|a^{(\sigma)}} + \mathcal{O}\left(\varepsilon_8^2\right),
\end{gather*}
where $m_a$ is the mass of the hadron $a$ in the multiplet $D^{(\sigma)}$, $m^{(0)}$ is the eigenvalue of $H^0_\text{QCD}$ corresponding to the multiplet $D^{(\sigma)}$ and {${\left\{\Ket{a^{(\sigma)}}\mid a\text{ is a hadron in }D^{(\sigma)}\right\}}$} is an orthonormal basis of the multiplet $D^{(\sigma)}$. As for the EFT approach, we made a special choice for the basis $\{\ket{a^{(\sigma)}}\}$ such that $H^8_{\text{QCD};8}$ is diagonal on the multiplet $D^{(\sigma)}$ in this basis.\par
Lastly, it remains to show that the given mass expressions exhibit the very same symmetry structures that led us to the GMO relations in \autoref{sec:mass_matrix}. We do this in the same way, as we did it in the EFT approach.
Let us define the mass matrix $m^{(\sigma)}$ by:
\begin{gather*}
m^{(\sigma)}_{ab} \coloneqq m^{(0)}\cdot\delta_{ab} + \varepsilon_8\cdot\Braket{a^{(\sigma)}|H^8_{\text{QCD};8}|b^{(\sigma)}}.
\end{gather*}
Earlier, we have chosen the basis $\{\ket{a^{(\sigma)}}\}$ such that $m^{(\sigma)}$ is diagonal. To first order in perturbation theory, the eigenvalues of $m^{(\sigma)}$, i.e., its diagonal entries, coincide with the hadron masses. We can define a transformation of $m^{(\sigma)}$ under $A\in\text{SU}(3)$:
\begin{gather*}
m^{(\sigma)\,\prime} \coloneqq D^{(\sigma)}(A)\cdot m^{(\sigma)}\cdot D^{(\sigma)}(A)^\dagger\\
\text{with } \left(D^{(\sigma)}(A)\right)_{ab} \coloneqq \Braket{a^{(\sigma)}|D^{(\rho)}(A)|b^{(\sigma)}}.
\end{gather*}
This transformation corresponds to the transformation of mass matrices defined in \autoref{sec:mass_matrix}. It is fairly easy to see that the transformation of $m^{(\sigma)}$ is given by the transformation of $m^{(0)}\mathbb{1} + \varepsilon_8\cdot H^{8}_{\text{QCD};\, 8}$ (cf. \autoref{sec:GMO_formula}):
\begin{gather*}
m^{(\sigma)\,\prime}_{ab} \coloneqq \Braket{a^{(\sigma)}|m^{(0)}\mathbb{1} + \varepsilon_8\cdot H^{8\,\prime}_{\text{QCD};\, 8}|b^{(\sigma)}}.
\end{gather*}
This means that the mass matrix $m^{(\sigma)}$ transforms as a singlet plus the 8th component of an octet. As the 8th component of an octet is $\text{SU}(2)\times\text{U}(1)$-invariant, $m^{(\sigma)}$ is subject to $\text{SU}(3)\rightarrow\text{SU}(2)\times\text{U}(1)$ symmetry breaking with octet enhancement and approximate \text{SU}(3)-symmetry. In \autoref{sec:mass_matrix}, we saw that these properties are sufficient to obtain the GMO mass relations. Note that the state formalism makes no difference between fermions and bosons, meaning between baryons and mesons. All hadrons satisfy linear GMO relations in this approach.

\subsection*{Is Feynman's Distinction Physical?}

We have seen that Feynman's distinction arises naturally in the EFT approach. However, this alone is no compelling reason to distinguish the GMO relations for baryons and mesons, that is to say that baryons only satisfy linear relations, while mesons are only described by quadratic relations. In particular, such a statement faces three problems:\par
First of all, it is questionable whether the EFT approach represents a valid model for hadron masses and, if it does, how precise this model is. Hadrons are composite, often unstable particles which are mostly just experimentally accessible as resonances. It is not clear at all that these resonances are ``good'' degrees of freedom, i.e., that it is possible to describe the hadrons as fields. Additionally, it is unclear whether the mass parameters appearing in the EFT Lagrangian are connected to the mass-like quantities of these resonances that are determined in experiments and, if so, how they are connected. Lastly, it is unknown how to obtain the EFT Lagrangian from the fundamental SM Lagrangian or from $\mathcal{L}_\text{QCD}$. It is especially not obvious whether this process spoils the symmetry structure and transformation behavior of the fundamental Lagrangian. If it spoils the symmetry structure, the transformation behavior of the fundamental Lagrangian does not ``carry over'' to the EFT Lagrangian and the EFT approach is invalid. For some multiplets like the pseudoscalar meson octet, one can find in the framework of theories like chiral perturbation theory (cf. \cite{Scherer2011}) that the transformation behavior is actually preserved. Nonetheless, it is not clear whether this holds true for all multiplets. In total, the validity of the EFT approach is at least debatable.\par
Secondly, we have already seen that there is an alternative way of deriving the GMO mass relations which does not lead to Feynman's distinction, namely the state formalism. Truly, this model also faces several problems. There are technical difficulties like the subtraction of vacuum energy that may spoil the symmetry structure of the Hamilton operator and unproven assumptions like the assumptions 1)-3) which enter the state formalism. However, we can at least motivate these assumptions (cf. \autoref{app:stateform}).\par
Lastly, we have to note that both quadratic and linear GMO mass relations are equivalent approximations, if the \text{SU}(3)-symmetry breaking contribution is small enough. This is illustrated by the following consideration: In both the EFT approach and the state formalism, we were able to boil the calculation of the hadron masses down to the computation of eigenvalues of some matrix $N$ that transforms with
\begin{gather*}
N^\prime = D^{(\rho)}(A)\cdot N\cdot D^{(\rho)}(A)^\dagger\quad\text{for } A\in\text{SU}(3),
\end{gather*}
where $D^{(\rho)}$ is a multiplet of \text{SU}(3). This matrix $N$ was in both cases subject to $\text{SU}(3)\rightarrow\text{SU}(2)\times\text{U}(1)$ symmetry breaking and octet enhancement to first order in some small parameter $\epsilon$ governing the symmetry breaking:
\begin{gather*}
N = N_0 + \epsilon\cdot N^8_8 + \mathcal{O}(\epsilon^2),
\end{gather*}
where $N_0$ is a singlet and $N^8_8$ is the 8th component of an octet. Depending on the approach and hadron multiplet, the matrix $N$ was either the hadron mass matrix or its square. If $N$ is the hadron mass matrix, its square is given by:
\begin{gather*}
N^2 = N^2_0 + \epsilon\left(N_0\cdot N^8_8 + N^8_8\cdot N_0\right) + \mathcal{O}(\epsilon^2).
\end{gather*}
Now, the transformation of $N$ implies a transformation of $N^2$:
\begin{gather*}
N^{2\,\prime} \coloneqq N^\prime\cdot N^\prime = D^{(\rho)}(A)\cdot N^2\cdot D^{(\rho)}(A)^\dagger\quad\text{for } A\in\text{SU}(3)
\end{gather*}
using the unitary of $D^{(\rho)}$. Under this transformation, $N^2_0$ transforms as a singlet and $N_0\cdot N^8_8 + N^8_8\cdot N_0$ transforms as the 8th component of an octet. This means that the eigenvalues of $N^2$, in our example the hadron masses squared, also satisfy GMO relations to first order in symmetry breaking. Hence, the hadrons in a multiplet satisfy quadratic GMO relations to first order in symmetry breaking, if they satisfy linear GMO relations to first order in symmetry breaking and the symmetry breaking in this multiplet is small. Likewise, if $N$ is the hadron mass matrix squared, we can take the square root of the eigenvalues of $N$ and perform a Taylor expansion of the square root about the \text{SU}(3)-invariant term. Note that the symmetry breaking term has to be small for the Taylor expansion to converge. We then find that the square roots of the eigenvalues of $N$, i.e., the hadron masses, coincide with the eigenvalues of the matrix $\sqrt{N}$:
\begin{gather*}
\sqrt{N}\coloneqq \sqrt{n_0}\mathbb{1} + \frac{\epsilon}{2\sqrt{n_0}}\cdot N^8_8 + \mathcal{O}(\epsilon^2),
\end{gather*}
where $N_0 \eqqcolon n_0\mathbb{1}$. Since $\sqrt{N}$ is a singlet plus the 8th component of an octet to first order in $\epsilon$, its eigenvalues and, therefore, the hadron masses satisfy GMO relations to first order in symmetry breaking. Again, this means that the hadrons in a multiplet satisfy linear GMO relations to first order in symmetry breaking, if they satisfy quadratic GMO relations to first order in symmetry breaking and the symmetry breaking in this multiplet is small. With this, we arrive at the statement that both quadratic and linear GMO relations are equivalent approximations in a hadronic multiplet, if the \text{SU}(3)-symmetry breaking contribution in this multiplet is small enough. We will see in \autoref{sec:mass_testing} that this is the case for most known hadronic multiplets.\par
With these three problems regarding Feynman's distinction in mind, we arrive at a rather surprising result: Most of the time, it does not appear to be sensible to ask which kind of mass relation, linear or quadratic, applies to which kind of particle, baryons or mesons, as most hadrons satisfy both relations within the same range of validity. For those hadrons, both kinds of GMO relations are satisfied to first order in flavor symmetry breaking. In this sense, we can say that the EFT approach and the state formalism are not contradicting considerations, but rather complementary and equivalent to some extent. Still, we have to note that one multiplet poses an exception to this statement: For the pseudoscalar meson octet, the quadratic GMO relation is clearly favored over the linear GMO relation (cf. \autoref{sec:mass_testing}). The linear and quadratic GMO mass relation are not equivalent in this case, as the symmetry breaking contribution in this multiplet is roughly of the same size as the \text{SU}(3)-invariant contribution (cf. \autoref{sec:mass_testing}). Nevertheless, the pseudoscalar meson octet cannot be seen as a confirmation of Feynman's distinction, since the linear and quadratic GMO mass relation are not inequivalent because of the spin, but because of the size of the flavor symmetry breaking: For heavier meson octets like the vector meson octet, the equivalence of the linear and quadratic GMO mass relation is restored (cf. \autoref{sec:mass_testing}). To this end, one might say that Feynman's distinction is artificial as a distinction of baryons and mesons into linear and quadratic GMO mass relations is not observable\footnote{We also have to exclude the mass relations following from heavy quark symmetry from this statement (cf. \autoref{sec:heavy_quark} and \autoref{sec:mass_testing}). However, these mass relations are not GMO mass relations as defined in \autoref{sec:mass_matrix}.} for the hadronic multiplets aside from the pseudoscalar meson octet.\par
Even though a global distinction of baryons and mesons into linear and quadratic relations, respectively, cannot be observed, the question remains whether there are reasons to prefer one relation over the other for certain multiplets or for a certain class of multiplets that is not only characterized by spin. We will investigate this question, among other questions, in the following chapters, but in particular in \autoref{sec:mass_testing}.

\newpage
\chapter{Mathematical Derivation of Hadronic Mass Formulae}
\label{chap:GMO_formula}

The main aspects of the derivation of the GMO mass relations were outlined in \autoref{chap:hadron_masses}, but a lot of mostly mathematical details were only glanced over in favor of clarity and simplicity. In this chapter, we want to take a closer look at these aspects on a deeper mathematical level. To do this, we repeat the derivation of the GMO mass relations from \autoref{chap:hadron_masses} in \autoref{sec:GMO_formula}. This time, we start with $\mathcal{L}_\text{QCD}$ and then proceed by linking this Lagrangian to the hadron masses and calculating them in a perturbative treatment. The perturbative description of hadron masses leads us to the notion of multiplets. We find that the hadron masses in a multiplet $\sigma$ transform as $\sigma\otimes\bar{\sigma}$ under $\text{SU}(3)$-flavor transformations and are subject to $\text{SU}(3)\rightarrow\text{SU}(2)\times\text{U}(1)$ symmetry breaking and octet enhancement. These properties allow us to parametrize the hadron masses in a multiplet.
The mass parametrization gives rise to the GMO mass formula which implies the GMO mass relations. However, we only repeat the derivation of the GMO mass relations for the state formalism and not for the EFT approach. The mathematical statements we derive apply almost in the same manner to the EFT approach as they do to the state formalism, so an additional derivation for the EFT approach would be superficial. The differences between the EFT approach and the state formalism were already presented and discussed in \autoref{sec:EFT+H_Pert}.\par
In \autoref{sec:add_con}, we want to consider additional contributions to the GMO mass formula. In particular, we want to consider isospin symmetry breaking and electromagnetic contributions. In \autoref{sec:heavy_quark}, we will see how we can use heavy quark symmetry to link hadronic mass formulae of different multiplets that only differ by the exchange of a charm with a bottom quark.

\section{The Gell-Mann--Okubo Mass Formula}
\label{sec:GMO_formula}

As already stated, we want to repeat the derivation of the GMO mass relations by starting with $\mathcal{L}_\text{QCD}$:
\begin{gather*}
\mathcal{L}_{\text{QCD}}(\bar{q},q) = \sum\limits_{q\in\{\text{u,d,s}\}}\bar{q}\left(i\slashed{D} -  m_q\right)q + \mathcal{L}_\text{YM},
\end{gather*}
where we used the same notations as in \autoref{sec:Trafo_QCD}. Defining \text{SU}(3)-flavor transformations of the fields $q$ implies \text{SU}(3)-transformations of $\mathcal{L}_\text{QCD}$:
\begin{align*}
q&\xrightarrow{A\in\text{SU}(3)}q^\prime \coloneqq \sum\limits_{p\in\{\text{u,d,s}\}} A_{qp}p,\\
\bar{q}&\xrightarrow{A\in\text{SU}(3)}\bar{q}^\prime \coloneqq \sum\limits_{p\in\{\text{u,d,s}\}} A^\ast_{qp}\bar{p}\\
\Rightarrow\mathcal{L}_\text{QCD}&\xrightarrow{A\in\text{SU}(3)}\mathcal{L}^\prime_\text{QCD}\text{ with } \mathcal{L}^\prime_\text{QCD}(\bar{q}^\prime, q^\prime) \coloneqq \mathcal{L}_\text{QCD}(\bar{q},q).
\end{align*}
We have seen in \autoref{sec:Trafo_QCD} and \autoref{sec:EFT+H_Pert} that $\mathcal{L}_\text{QCD}$ transforms under \text{SU}(3)-flavor transformations as a singlet plus an octet:
\begin{gather*}
\mathcal{L}_\text{QCD} = \mathcal{L}^{0}_\text{QCD} + \varepsilon_3\cdot\mathcal{L}^{8}_{\text{QCD};\, 3} + \varepsilon_8\cdot\mathcal{L}^{8}_{\text{QCD};\, 8},
\end{gather*}
where $\varepsilon_3 \coloneqq m_\text{u} - m_\text{d}$ and $\varepsilon_8 \coloneqq \frac{m_\text{u} + m_\text{d} - 2m_\text{s}}{\sqrt{3}}$ and
\begin{align*}
\mathcal{L}^{0}_{\text{QCD}} &\coloneqq \sum\limits_{q\in\{\text{u,d,s}\}}\bar{q}\left(i\slashed{D} -  \frac{m_\text{u} + m_\text{d} + m_\text{s}}{3}\right)q + \mathcal{L}_\text{YM},\\
\mathcal{L}^{8}_{\text{QCD};\, k} &\coloneqq -\sum\limits_{p,q\in\{\text{u,d,s}\}}\frac{\bar{p}\left(\lambda_k\right)_{pq}q}{2}\quad\forall k\in\{1,\ldots, 8\}.
\end{align*}
In \autoref{sec:EFT+H_Pert}, we have seen that this transformation behavior of $\mathcal{L}_\text{QCD}$ also applies to the corresponding Hamilton operator $H_\text{QCD}$ of $\mathcal{L}_\text{QCD}$:
\begin{gather*}
H_\text{QCD} = H^{0}_\text{QCD} + \varepsilon_3\cdot H^{8}_{\text{QCD};\, 3} + \varepsilon_8\cdot H^{8}_{\text{QCD};\, 8},
\end{gather*}
where $H^{0}_\text{QCD}$ is a singlet of \text{SU}(3) and $H^{8}_{\text{QCD};\, 3}$ and $H^{8}_{\text{QCD};\, 8}$ are the 3rd and 8th component of an octet, respectively. For now, we set $m_\text{u} = m_\text{d}$. This implies that $\varepsilon_3 = 0$. In this case, $\text{SU}(2)\times\text{U}(1)\subset\text{SU}(3)$ is an exact global symmetry of $H_\text{QCD}$. We are going to investigate the case of $m_\text{u}\neq m_\text{d}$, i.e., the case of isospin symmetry breaking, in \autoref{sec:add_con}.\par
Like explained, we mainly want to consider the state formalism in this section. The state formalism makes the following three assumptions:
\begin{itemize}
 \item[1)] For every hadron $a$, there exists an eigenstate $\Ket{a}$ with $\Braket{a|a} = 1$ of the Hamilton operator $H_\text{QCD}$ from which the vacuum energy is already subtracted such that the mass $m_a$ of the hadron $a$ is given by
 \begin{gather*}
  m_a = \Bra{a}H_\text{QCD} \Ket{a}.
 \end{gather*}
 \item[2)] The subspace $V$ of the physical states which is spanned by the states $\Ket{a}$ from 1), i.e., {${V:= \overline{\text{Span}\left\{\Ket{a}\mid a\text{ hadron}\right\}}}$}, is a Hilbert space.
 \item[3)] There is a unitary representation $D^{(\rho)}:V\rightarrow V$ of \text{SU}(3) on $V$ such that the following equation holds for every $A\in\text{SU}(3)$:
 \begin{gather*}
  \Bra{a} D^{(\rho)}(A)^\dagger\circ H_\text{QCD}\left(\bar{q},\, q\right)\circ D^{(\rho)}(A) \Ket{b} = \Bra{a}H_\text{QCD}\left(\bar{q}^{\, \prime},\, q^\prime\right) \Ket{b}\ \ \forall\Ket{a},\Ket{b}\in V
 \end{gather*}
 where $q^\prime\coloneqq \sum\limits_{\tilde{q}\in\{\text{u,d,s}\}}A_{q \tilde{q}}\cdot \tilde{q}$.
\end{itemize}
Again, we understand $H_\text{QCD}$, $H^0_\text{QCD}$, $H^8_{\text{QCD};3}$, and $H^8_{\text{QCD};8}$ in assumption 3) and from now on as operators restricted and projected down to $V$, like explained in \autoref{sec:EFT+H_Pert}. A motivation of these assumptions is given in \autoref{app:stateform}. Before we can apply a perturbative treatment to $H_\text{QCD}$ to calculate the hadron masses, we need two more statements. The first one is a statement about the theory at hand and depends on the free parameters in $\mathcal{L}_\text{QCD}$. It is the assumption that the octet term $\varepsilon_8\cdot H^8_{\text{QCD};8}$ is actually just a small correction to the Hamilton operator $H_\text{QCD}$ such that we can treat it as a small perturbation. In Nature, we find this to be a sensible assumption for most hadron masses. The second statement is the assumption that the subtraction of vacuum energy does not spoil the transformation behavior of $H_\text{QCD}$. Like the assumptions 1)-3), I cannot prove this, as, in general, the subtraction of vacuum energy for an arbitrary QFT is not known. However, the subtraction of vacuum energy for a free QFT is accomplished by taking the normal ordered product of the Hamilton operator. The normal ordered product does not spoil the transformation behavior of the Hamilton operator, hence, the second statement is at least true for free QFTs. This makes it seem likely that the second statement is also true for other QFTs.\par
Now, we want to calculate the hadron masses in a perturbative treatment of $H_\text{QCD}$ with $\varepsilon_8\cdot H^8_{\text{QCD};8}$ as a small perturbation. Obviously, we need to know (degenerate) perturbation theory for this. At this point, it is instructive to quickly recapitulate perturbation theory.

\subsection*{Degenerate Perturbation Theory}

Let us consider linear operators $H(\varepsilon):V\rightarrow V$ on a Hilbert space $V$. Suppose that $H(\varepsilon)$ is given by:
\begin{gather*}
H(\varepsilon) = H_0 + \varepsilon\cdot\Delta H,
\end{gather*}
where $H_0$ is a self-adjoint linear operator on $V$, $\Delta H$ is a linear operator on $V$, and $\varepsilon$ is a ``small'' parameter, i.e., the contribution of $\varepsilon\cdot\Delta H$ to the eigenvectors and -values of $H(\varepsilon)$ is small and can be written as a Taylor series in $\varepsilon$. Now suppose that $\Ket{a(\varepsilon)}$ is such an eigenvector of $H(\varepsilon)$ with eigenvalue $m(\varepsilon)$ which can be expanded in a Taylor series:
\begin{align*}
\Ket{a(\varepsilon)} &= \ket{a^{(0)}} + \varepsilon\cdot \ket{a^{(1)}} + \mathcal{O}(\varepsilon^2),\\
m(\varepsilon) &= m^{(0)} + \varepsilon\cdot m^{(1)} + \mathcal{O}(\varepsilon^2).
\end{align*}
The eigenvector equation for $\Ket{a(\varepsilon)}$ then reads:
\begin{align*}
&H(\varepsilon)\Ket{a(\varepsilon)} = m(\varepsilon) \Ket{a(\varepsilon)}\\
\Rightarrow\ &H_0\ket{a^{(0)}} + \varepsilon\cdot (\Delta H\ket{a^{(0)}} + H_0\ket{a^{(1)}}) + \mathcal{O}(\varepsilon^2)\\
&= m^{(0)}\ket{a^{(0)}} + \varepsilon\cdot (m^{(1)}\ket{a^{(0)}} + m^{(0)}\ket{a^{(1)}}) + \mathcal{O}(\varepsilon^2).
\end{align*}
Note that the last equation holds true if we can commute $H_0$ and $\Delta H$ with the infinite sum of the Taylor expansion. This is, for instance, possible, if $H_0$ and $\Delta H$ are continuous or, equivalently, bounded which is always the case, if $V$ is finite-dimensional. As the coefficients of a Taylor series are unique, the equation above has to be satisfied to every order in $\varepsilon$. The equation for the zeroth order reads:
\begin{gather*}
H_0\ket{a^{(0)}} = m^{(0)}\ket{a^{(0)}}.
\end{gather*}
This equation states that $\ket{a^{(0)}}$ is an eigenvector of $H_0$ with eigenvalue $m^{(0)}$. The first order equation is given by:
\begin{gather}\label{eq:pert_e1}
\Delta H\ket{a^{(0)}} + H_0\ket{a^{(1)}} = m^{(1)}\ket{a^{(0)}} + m^{(0)}\ket{a^{(1)}}.
\end{gather}
We want to solve \autoref{eq:pert_e1} for $m^{(1)}$ now. This task is rather easy, if $\ket{a^{(0)}}$ is non-degenerate, i.e., if the eigenspace $V_{m^{(0)}}$ of $H_0$ with eigenvalue $m^{(0)}$ is just given by $\text{Span}\{\ket{a^{(0)}}\}$. In this case, we simply act with $\bra{a^{(0)}}$ on \autoref{eq:pert_e1}:
\begin{align*}
\braket{a^{(0)}|\Delta H|a^{(0)}} + \braket{a^{(0)}|H_0|a^{(1)}} &= \braket{a^{(0)}|\Delta H|a^{(0)}} + m^{(0)}\braket{a^{(0)}|a^{(1)}}\\
&= m^{(1)}\braket{a^{(0)}|a^{(0)}} + m^{(0)}\braket{a^{(0)}|a^{(1)}}\\
\Rightarrow m^{(1)} &= \frac{\braket{a^{(0)}|\Delta H|a^{(0)}}}{\braket{a^{(0)}|a^{(0)}}},
\end{align*}
where we used the hermicity of $H_0$. If we insert $\ket{\tilde{a}^{(0)}} = c\cdot \ket{a^{(0)}}$ ($c\in\mathbb{C}\backslash\{0\}$) instead of $\ket{a^{(0)}}$ into the expression for $m^{(1)}$, the equation does not change. This means that we can use any non-zero vector from the eigenspace $V_{m^{(0)}} = \text{Span}\{\ket{a^{(0)}}\}$ to calculate $m^{(1)}$.\par
However, if $\ket{a^{(0)}}$ is degenerate, which, in general, it is, the eigenspace $V_{m^{(0)}}$ is not just one-dimensional and it is not immediately obvious how to choose $\ket{a^{(0)}}$. In this case, let us suppose that there is an orthonormal basis $\{\ket{\alpha}\mid \alpha\in I\}$ of $V_{m^{(0)}}$ for some index set $I$.
We then have:
\begin{gather*}
\ket{a^{(0)}} = \sum\limits_{\alpha\in I} \braket{\alpha|a^{(0)}}\cdot \ket{\alpha}.
\end{gather*}
Let us now act with $\bra{\alpha}$ on \autoref{eq:pert_e1}:
\begin{gather*}
\braket{\alpha|\Delta H|a^{(0)}} + \braket{\alpha|H_0|a^{(1)}} = \braket{\alpha|\Delta H|a^{(0)}} + m^{(0)}\braket{\alpha|a^{(1)}} = m^{(1)}\braket{\alpha|a^{(0)}} + m^{(0)}\braket{\alpha|a^{(1)}}\\
\Rightarrow \braket{\alpha|\Delta H|a^{(0)}} = m^{(1)}\braket{\alpha|a^{(0)}}\\
\Rightarrow \sum\limits_{\alpha\in I} \braket{\alpha|\Delta H|a^{(0)}}\cdot \ket{\alpha} = m^{(1)}\sum\limits_{\alpha\in I} \braket{\alpha|a^{(0)}}\cdot \ket{\alpha} = m^{(1)}\ket{a^{(0)}}.
\end{gather*}
If we define the linear operator\footnote{Of course, this only makes sense, if the operator $\Delta H\vert_{\overline{V_{m^{(0)}}}}$ exists, which is not clear at this stage.} $\Delta H\vert_{\overline{V_{m^{(0)}}}}:\overline{V_{m^{(0)}}}\rightarrow\overline{V_{m^{(0)}}}, \ket{b}\mapsto \sum\limits_{\alpha\in I} \braket{\alpha|\Delta H|b}\cdot \ket{\alpha}$ on the closure $\overline{V_{m^{(0)}}}$ of $V_{m^{(0)}}$, the equation above states that $\ket{a^{(0)}}$ is an eigenvector of $\Delta H\vert_{\overline{V_{m^{(0)}}}}$ with eigenvalue $m^{(1)}$.\par
These considerations give us a method on how to calculate the eigenvalue $m(\varepsilon)$ to first order in $\varepsilon$: First, we have to calculate the eigenvalues and -spaces of $H_0$ to obtain the possible values for the zeroth order contribution $m^{(0)}$ of $m(\varepsilon)$. Next, we have to diagonalize the restriction $\Delta H\vert_{\overline{V_{m^{(0)}}}}$ of $\Delta H$ to (the closure of) each eigenspace $\overline{V_{m^{(0)}}}$ to determine the possible values for the first order contribution $m^{(1)}$ of $m(\varepsilon)$. Again, we are going to drop the restriction ``$\vert_{\overline{V_{m^{(0)}}}}$'' in further discussions. We should note at this point that we have to be cautious, if there are eigenstates of $H_0$ outside of $V_{m^{(0)}}$ that are ``almost degenerate'' to $V_{m^{(0)}}$, i.e, if the difference between the eigenvalues of these eigenstates and the eigenvalue $m^{(0)}$ of $V_{m^{(0)}}$ is in the order of or small in comparison to $m^{(1)}$ originating from the perturbation $\varepsilon\cdot\Delta H$. In this case, we can split up $H_0$ into $H^\prime_0$ and $\Delta H^\prime$ such that all ``nearly degenerate'' states of $H_0$ are now exactly degenerate for $H^\prime_0$ and $\Delta H^\prime$ is a small perturbation in the order of the difference between the eigenvalues. We then treat $H^\prime_0$ as the large contribution and $\Delta H^\prime + \varepsilon\cdot\Delta H$ as the perturbation.\\\par
We want to apply this now to $H_\text{QCD} = H^0_\text{QCD} + \varepsilon_8
\cdot H^8_{\text{QCD};8}$, where we treat $\varepsilon_8\cdot H^8_{\text{QCD};8}$ as a small perturbation. For this, we have to investigate the eigenvalues and -spaces of $H^0_\text{QCD}$.

\subsection*{Eigenspaces of $H^0_\text{\normalfont{QCD}}$}

$H^0_\text{QCD}$ is a singlet under \text{SU}(3). Together with assumption 3), we can rewrite this property as:
\begin{gather*}
D^{(\rho)}(A)\, H^0_\text{QCD}\, D^{(\rho)}(A)^\dagger = H^0_\text{QCD}\  \Leftrightarrow\  \left[D^{(\rho)}(A),\, H^0_\text{QCD}\right] = 0\quad\forall A\in\text{SU}(3).
\end{gather*}
Now let $V_m$ be the eigenspace of $H^0_\text{QCD}$ to the eigenvalue $m$ and $\ket{\alpha}$ an element of $V_m$, then:
\begin{gather*}
H^0_\text{QCD}D^{(\rho)}(A)\ket{\alpha} = D^{(\rho)}(A)H^0_\text{QCD}\ket{\alpha} = D^{(\rho)}(A)m\ket{\alpha} = m D^{(\rho)}(A)\ket{\alpha}\quad\forall A\in\text{SU}(3).
\end{gather*}
This equation states that $D^{(\rho)}(A)\ket{\alpha}$ is also an eigenvector of $H^0_\text{QCD}$ with eigenvalue $m$, if $\ket{\alpha}$ is an eigenvector with eigenvalue $m$, hence, $D^{(\rho)}(A)(V_m) = V_m\quad\forall A\in\text{SU}(3)$. Now consider the closure $\overline{V_m}$ of $V_m$. Let $\ket{\alpha}$ be an element of $\overline{V_m}$, then we can write $\ket{\alpha} = \lim_{n\to\infty} \ket{\alpha_n}$ with $(\ket{\alpha_n})_{n\in\mathbb{N}}$ being a Cauchy sequence in $V_m$. For every $A\in\text{SU}(3)$, $D^{(\rho)}(A)$ is unitary, hence, $D^{(\rho)}(A)$ is bounded and, therefore, continuous. With this, we obtain:
\begin{gather*}
D^{(\rho)}(A)\ket{\alpha} = D^{(\rho)}(A)(\lim_{n\to\infty} \ket{\alpha_n}) = \lim_{n\to\infty} D^{(\rho)}(A)\ket{\alpha_n}\quad\forall A\in\text{SU}(3).
\end{gather*}
This means that $D^{(\rho)}(A)\ket{\alpha}$ is an element of $\overline{V_m}$ for every $A\in\text{SU}(3)$ which implies that $D^{(\rho)}(A)(\overline{V_m}) = \overline{V_m}\ \forall A\in\text{SU}(3)$. As a closed subspace of the Hilbert space $V$, $\overline{V_m}$ itself is a Hilbert space. Therefore, the restriction $D^{(\rho)}\vert_{\overline{V_m}}:\text{SU}(3)\rightarrow\text{GL}(\overline{V_m})$, $A\mapsto D^{(\rho)}(A)\vert_{\overline{V_m}}$ of $D^{(\rho)}$ is a unitary representation of the compact Lie group \text{SU}(3) on the Hilbert space $\overline{V_m}$. Applying the Peter-Weyl theorem (cf. \cite{Knapp2001}), $\overline{V_m}$ can be written as the closure of the direct sum of irreducible, finite-dimensional, and orthogonal spaces $W_i$:
\begin{gather*}
\overline{V_m} = \overline{\bigoplus_i W_i}.
\end{gather*}
This means that if we choose an orthonormal basis for every $W_i$ (which is possible as they all are finite-dimensional) and combine all these bases into one set, we obtain a complete orthonormal basis for $\overline{V_m}$ only consisting of complete multiplets of \text{SU}(3).\par
Note that all multiplets in the complete orthonormal basis are complex representations, as $D^{(\rho)}$ is a complex representation. Also note that either $W_i\cap V_m = \{0\}$ or $W_i\subset V_m$ for every i, because if there is an element $\ket{\alpha}$ of $W_i\backslash\{0\}$ that is an eigenvector with eigenvalue $m$, then $W\coloneqq \text{Span}\{D^{(\rho)}(A)\ket{\alpha}\mid A\in\text{SU}(3)\}$ is a subset of $V_m$ and a closed invariant subspace of $W_i$. However, the only closed invariant subspaces of the irreducible representation $W_i$ are $\{0\}$ and $W_i$, hence, $W_i = W$ because of $0\neq \ket{\alpha}\in W$. If $V_m$ is closed, all $W_i$ are trivially contained in the eigenspace $V_m$ as, in this case, $\overline{V_m} = V_m$. $V_m$ is closed, if $H^0_\text{QCD}\vert_{\overline{V_m}}:\overline{V_m}\rightarrow H^0_\text{QCD}(\overline{V_m})$ is bounded (or, equivalently, continuous) or if $V_m$ is finite-dimensional.\par
Now suppose that there is a complete orthonormal eigenbasis $\{\ket{a}\mid a\in I\}$ of $H^0_\text{QCD}$ for some index set $I$, then every $\ket{a}$ is contained in an eigenspace $V_{m_a}$ for some eigenvalue $m_a$. Now choose an orthonormal basis for every eigenspace $\overline{V_{m_a}}$ such that each basis only consists of complete multiplets of \text{SU}(3). Combine all these bases into one set and denote it by $\{\ket{\alpha}\mid \alpha\in\tilde{I}\}$ for some index set $\tilde{I}$. All vectors of {${\{\ket{\alpha}\mid \alpha\in\tilde{I}\}}$} are pairwise orthogonal, as the spaces $\overline{V_{m_a}}$ are pairwise orthogonal because of the hermicity of $H^0_\text{QCD}$. Furthermore, $\ket{a}\in\overline{\text{Span}\{\ket{\alpha}\mid \alpha\in\tilde{I}\}}$ for every {${a\in I}$}, as {${\{\ket{\alpha}\mid \alpha\in\tilde{I}\}}$} contains an orthonormal basis for every $\overline{V_{m_a}}$ and {${\ket{a}\in V_{m_a}\subset\overline{V_{m_a}}}$}. This implies:
\begin{gather*}
V = \overline{\text{Span}\{\ket{a}\mid a\in I\}} \subset \overline{\text{Span}\left(\overline{\text{Span}\{\ket{\alpha}\mid \alpha\in\tilde{I}\}}\right)} = \overline{\text{Span}\{\ket{\alpha}\mid \alpha\in\tilde{I}\}}\subset V\\
\Rightarrow V = \overline{\text{Span}\{\ket{\alpha}\mid \alpha\in\tilde{I}\}}.
\end{gather*}
This means that $\{\ket{\alpha}\mid \alpha\in\tilde{I}\}$ is a complete orthonormal basis of $V$ that only consists of complete multiplets of \text{SU}(3) where every such multiplet is completely contained in one eigenspace $\overline{V_{m_a}}$.\\\par
Let us summarize what we found: There is a complete orthonormal basis for every closure $\overline{V_m}$ of an eigenspace $V_m$ of $H^0_\text{QCD}$ where this basis only consists of complete finite-dimensional multiplets of \text{SU}(3). Furthermore, if $H^0_\text{QCD}$ is diagonalizable, i.e., if there is a complete orthonormal eigenbasis of $H^0_\text{QCD}$, there is a complete orthonormal basis of $V$ only consisting of complete finite-dimensional multiplets of \text{SU}(3) where every such multiplet is completely contained in the closure of some eigenspace of $H^0_\text{QCD}$. Before we move on, let us make some remarks:
\begin{enumerate}
\item We only need the hermicity and the \text{SU}(3)-invariance of $H^0_\text{QCD}$ for the eigenspace analysis given above. Therefore, the very same statements apply to $M^{2;0}_\text{M}$ and $M^0_\text{B}$ from the EFT approach in \autoref{sec:EFT+H_Pert}. For both matrices $M^{2;0}_\text{M}$ and $M^0_\text{B}$, the underlying Hilbert spaces were assumed to be finite-dimensional, so all eigenspaces are closed and $M^{2;0}_\text{M}$ and $M^0_\text{B}$ are diagonalizable. This means that there is a complete orthonormal eigenbasis for each matrix $M^{2;0}_\text{M}$ and $M^0_\text{B}$ such that each basis consists of complete multiplets of \text{SU}(3) which each are completely contained in some eigenspace.
\item If there are multiple ``nearly degenerate'' eigenspaces of $H^0_\text{QCD}$ with respect to the perturbation $\varepsilon_8\cdot H^8_{\text{QCD};8}$, we need to diagonalize the perturbation on the closure of the direct sum of the ``nearly degenerate'' eigenspaces. One can show in similar fashion to the considerations above that there is a complete orthonormal basis for the closure of any direct sum of eigenspaces only consisting of complete finite-dimensional multiplets of \text{SU}(3).
\end{enumerate}
Now, we want to calculate the contribution of the perturbation $\varepsilon_8\cdot H^8_{\text{QCD};8}$ to the hadron masses in first order. In order to do so, we have to diagonalize $H^8_{\text{QCD};8}$ on (the closure of the direct sum of) some (``nearly degenerate'') eigenspace(s) of $H^0_\text{QCD}$. Let us call this space $W$. As we have just seen, there exists a complete orthonormal basis of $W$ only consisting of complete finite-dimensional multiplets of \text{SU}(3). In general, however, $H^8_{\text{QCD};8}$ is not diagonal in such a basis. It is possible that multiple multiplets in $W$ contribute to the same hadron mass. Nevertheless, we only want to consider the case where the diagonalization of $H^8_{\text{QCD};8}$ on $W$ is compatible with the multiplet structure of $W$, i.e, $H^8_{\text{QCD};8}$ on $W$ is diagonal in multiplets of $\text{SU}(3)$. In this case, there exists a complete orthonormal basis of $W$ only consisting of complete finite-dimensional multiplets of \text{SU}(3) such that this basis is also an eigenbasis of $H^8_{\text{QCD};8}$ on $W$. This assumption is a good approximation in multiple scenarios: Firstly, it is possible that the space $W$ at hand only consists of one multiplet. Then, $H^8_{\text{QCD};8}$ on $W$ is trivially diagonal in multiplets of \text{SU}(3). In Nature, this scenario applies when the mass difference between the average mass of the hadron multiplet at hand and any other hadron multiplet is larger than the mass contribution of the perturbation $\varepsilon_8\cdot H^8_{\text{QCD};8}$. Secondly, even if the space $W$ consists of several multiplets, the mixing of these multiplets in $H^8_{\text{QCD};8}$ might be negligibly small or suppressed such that we can diagonalize $H^8_{\text{QCD};8}$ on each multiplet separately. This suppression may originate from quantum numbers that are preserved under $H^8_{\text{QCD};8}$ like baryon number, total angular momentum, parity, and so on. For instance, the contribution of mesons to baryon masses should be zero, as mesons and baryons have different total angular momenta and, therefore, cannot mix.\par
The assumption that $H^8_{\text{QCD};8}$ on $W$ is diagonal in \text{SU}(3)-multiplets allows us to compute the hadron masses in each multiplet individually. Thus, we only have to parametrize the hadron masses in each of these multiplets in order to find the GMO mass formula. However, this task is still quite difficult and it is not clear yet how we can accomplish this. Therefore, we wish to rephrase the problem to clearly see what we have to do. For this, we consider the transformation behavior of the hadron masses in \text{SU}(3)-multiplets.

\subsection*{Flavor Transformation Behavior of Hadronic Mass Multiplets}

Applying the assumption that $H^8_{\text{QCD};8}$ on $W$ is diagonal in multiplets of \text{SU}(3), we can diagonalize $H^8_{\text{QCD};8}$ on each of these multiplets separately. Let us now pick out one of these finite-dimensional multiplets and call it $D^{(\sigma)}$. Then, the mass of the hadron $a$ in the multiplet $D^{(\sigma)}$ is given by:
\begin{gather*}
m_{a} = m^{(0)} + \varepsilon_8\cdot\Braket{a^{(\sigma)}|H^8_{\text{QCD};8}|a^{(\sigma)}} + \mathcal{O}\left(\varepsilon_8^2\right),
\end{gather*}
where $m^{(0)}$ is the eigenvalue of $H^0_\text{QCD}$ corresponding to the multiplet $D^{(\sigma)}$ and\linebreak $\left\{\Ket{a^{(\sigma)}}\mid a\text{ is a hadron in }D^{(\sigma)}\right\}$ is a basis of the multiplet $D^{(\sigma)}$ chosen such that it is an orthonormal eigenbasis of $H^8_{\text{QCD};8}$ restricted to the multiplet $D^{(\sigma)}$.\par
Although we have found this formula, we still need to parametrize\linebreak $\Braket{a^{(\sigma)}|H^8_{\text{QCD};8}|a^{(\sigma)}}$ to obtain the GMO mass formula. It is easier to solve this problem, if we rephrase it. Consider the matrix $m^{(\sigma)}$:
\begin{gather*}
m^{(\sigma)}_{ab} \coloneqq m^{(0)}\cdot\delta_{ab} + \varepsilon_8\cdot\Braket{a^{(\sigma)}|H^8_{\text{QCD};8}|b^{(\sigma)}}.
\end{gather*}
Since $\left\{\Ket{a^{(\sigma)}}\mid a\text{ is a hadron in }D^{(\sigma)}\right\}$ is an eigenbasis of $H^8_{\text{QCD};8}$ on $D^{(\sigma)}$, $m^{(\sigma)}$ is diagonal. This means that the eigenvalues of $m^{(\sigma)}$, i.e., its diagonal entries coincide with the hadron masses of the hadrons in $D^{(\sigma)}$ to first order. Furthermore, we can define the following transformation of $m^{(\sigma)}$ under $A\in\text{SU}(3)$:
\begin{gather*}
m^{(\sigma)\,\prime} \coloneqq D^{(\sigma)}(A)\cdot m^{(\sigma)}\cdot D^{(\sigma)}(A)^\dagger\\
\text{with } \left(D^{(\sigma)}(A)\right)_{ab} \coloneqq \Braket{a^{(\sigma)}|D^{(\rho)}(A)|b^{(\sigma)}}.
\end{gather*}
This makes $m^{(\sigma)}$ a hadronic mass matrix transforming under $\sigma\otimes\bar{\sigma}$. We investigated such mass matrices in \autoref{sec:mass_matrix}. $\left\{\Ket{a^{(\sigma)}}\mid a\text{ is a hadron in }D^{(\sigma)}\right\}$ is a complete orthonormal basis of the multiplet $D^{(\sigma)}$. We can extend such a basis into a complete orthonormal basis $\{\Ket{\alpha}\}$ of the entire Hilbert space $V$ $D^{(\rho)}$ is acting on. With this, we see that $m^{(\sigma)}$ transforms as a singlet plus the 8th component of an octet under $A\in\text{SU}(3)$:
\begin{align*}
m^{(\sigma)\,\prime}_{ab} &= \sum\limits_{c,d} \left(D^{(\sigma)}(A)\right)_{ac}\cdot m^{(\sigma)}_{cd}\cdot \left(D^{(\sigma)}(A)\right)^\ast_{bd}\\
&= \left[\sum\limits_{c}\left(D^{(\sigma)}(A)\right)^\ast_{ac} \Ket{c^{(\sigma)}}\right]^\dagger m^{(0)}\mathbb{1} + \varepsilon_8\cdot H^8_{\text{QCD};8}\left[\sum\limits_{d}\left(D^{(\sigma)}(A)\right)^\ast_{bd} \Ket{d^{(\sigma)}}\right]\\
&= \Braket{a^{(\sigma)}|D^{(\rho)}(A)\left(m^{(0)}\mathbb{1} + \varepsilon_8\cdot H^8_{\text{QCD};8}\right)D^{(\rho)}(A)^\dagger|b^{(\sigma)}}\\
&= \Braket{a^{(\sigma)}|m^{(0)}\mathbb{1} + \varepsilon_8\cdot H^{8\,\prime}_{\text{QCD};8}|b^{(\sigma)}},
\end{align*}
where we used
\begin{align*}
\sum\limits_{c}\left(D^{(\sigma)}(A)\right)^\ast_{ac} \Ket{c^{(\sigma)}} &= \sum\limits_{c}\Braket{a^{(\sigma)}|D^{(\rho)}(A)|c^{(\sigma)}}^\ast \Ket{c^{(\sigma)}}\\
&= \sum\limits_{c}\Braket{a^{(\sigma)}|D^{(\rho)}(A)|c^{(\sigma)}}^\dagger \Ket{c^{(\sigma)}}\\
&= \sum\limits_{c}\Ket{c^{(\sigma)}}\Braket{c^{(\sigma)}|D^{(\rho)}(A)^\dagger|a^{(\sigma)}}\\
&= \sum\limits_{\alpha}\left(\Ket{\alpha^{\color{white}(}}\Bra{\alpha^{\color{white}(}}\right)D^{(\rho)}(A)^\dagger\Ket{a^{(\sigma)}}\\
&= D^{(\rho)}(A)^\dagger\Ket{a^{(\sigma)}}.
\end{align*}
In the fourth line of the calculation, we used the unitarity of $D^{(\rho)}$ and the fact that $D^{(\sigma)}$ as a multiplet in $D^{(\rho)}$ is an invariant subspace of $D^{(\rho)}$. In the last line, we used the completeness of $\{\Ket{\alpha}\}$.\par
With this consideration, we find that the masses of the hadrons in the multiplet $D^{(\sigma)}$ are given to first order in flavor symmetry breaking by the eigenvalues of a mass matrix transforming under $\sigma\otimes\bar{\sigma}$ and decomposing into a singlet plus the 8th component of an octet under this transformation. As we have seen in \autoref{chap:hadron_masses}, this implies octet enhancement and $\text{SU}(3)\rightarrow\text{SU}(2)\times\text{U}(1)$ symmetry breaking. Therefore, we now want to parametrize and diagonalize mass matrices transforming under $\sigma\otimes\bar{\sigma}$ for arbitrary complex finite-dimensional multiplets $\sigma$ which are subject to octet enhancement and $\text{SU}(3)\rightarrow\text{SU}(2)\times\text{U}(1)$ symmetry breaking.

\subsection*{Multiplet Classification}

In \autoref{sec:mass_matrix}, we already studied mass matrices transforming under $\sigma\otimes\bar{\sigma}$ for a small selection of multiplets $\sigma$, namely for singlets, triplets, sextets, octets, and decuplets. In the course of this investigation, we observed some reoccurring features which we listed at the end of \autoref{sec:mass_matrix}. This list provides a guideline for the remainder of this section and tells us which aspects we need to investigate. We start by examining how many singlets and octets occur in the Clebsch-Gordan series of $\sigma\otimes\bar{\sigma}$ for arbitrary complex finite-dimensional multiplets $\sigma$. To answer this question, we derive the following lemma:
\begin{Lem}
Let $G$ be a compact Lie group, $\sigma$ a complex, finite-dimensional, and irreducible representation of $G$, and $\rho$ a complex finite-dimensional representation of $G$. Then $n_\sigma(\rho) = n_1(\bar{\sigma}\otimes\rho),$ where $1$ is the trivial representation of $G$ and $n_\mu(\nu)$ denotes the multiplicity of an irreducible representation $\mu$ of $G$ in the decomposition of a finite-dimensional representation $\nu$ of $G$ into irreducible representations, i.e., $n_\mu(\nu)$ denotes how often $\mu$ occurs in a decomposition of $\nu$ into irreducible representations.
\end{Lem}
\begin{proof}
Let $G$, $\sigma$, and $\rho$ be like above. As $\rho$ is a finite-dimensional representation of a compact Lie group, it decomposes completely into a direct sum of (finite-dimensional) irreducible representations of $G$ (cf. \cite{Knapp2001}):
\begin{gather*}
\rho = \bigoplus^n_{i=1}\rho_i,
\end{gather*}
where $\rho_i$ is an irreducible representation of $G$ for every $i\in\{1,\ldots,n\}$ and $n\in\mathbb{N}$. Then, the decomposition of $\bar{\sigma}\otimes\rho$ into irreducible representations can be obtained by decomposing $\bar{\sigma}\otimes\rho_i$ into irreducible representations for every $i$:
\begin{gather*}
\bar{\sigma}\otimes\rho = \bigoplus^n_{i=1}\bar{\sigma}\otimes\rho_i.
\end{gather*}
Denote the vector spaces $\sigma$ and $\rho_i$ are acting on with $V^{(\sigma)}$ and $V^{(\rho_i)}$, respectively. Then, $\bar{\sigma}\otimes\rho_i$ can be understood as:
\begin{gather*}
D^{(\bar{\sigma}\otimes\rho_i)}:G\rightarrow\text{GL}\left(\text{Hom}\left(V^{(\sigma)},V^{(\rho_i)}\right)\right),\\ g\mapsto\left(L\mapsto D^{(\bar{\sigma}\otimes\rho_i)}(g)(L) = D^{(\rho_i)}(g)\circ L\circ D^{(\sigma)}(g)^\dagger\right)
\end{gather*}
with Hom$\left(V^{(\sigma)}, V^{(\rho_i)}\right)\coloneqq\left\{L\mid L:V^{(\sigma)}\rightarrow V^{(\rho_i)}\text{ linear}\right\}$ being the vector space of linear maps from $V^{(\sigma)}$ to $V^{(\rho_i)}$. Suppose there is a $L\in\text{Hom}\left(V^{(\sigma)}, V^{(\rho_i)}\right)$ such that $D^{(\bar{\sigma}\otimes\rho_i)}(g)(L) = L\ \forall g\in G$, then:
\begin{gather*}
D^{(\rho_i)}(g)\circ L = L\circ D^{(\sigma)}(g)\ \forall g\in G,
\end{gather*}
since $\sigma$ can be chosen to be unitary, as it is finite-dimensional (cf. \cite{Knapp2001}). By Schur's lemma (cf. \cite{Knapp2001}), $\sigma$ and $\rho_i$ are either equivalent or $L = 0$. If $\sigma$ and $\rho_i$ are equivalent, we can assume $\sigma = \rho_i$ without loss of generality. Then:
\begin{gather*}
D^{(\sigma)}(g)\circ L = L\circ D^{(\sigma)}(g)\ \forall g\in G
\end{gather*}
Again, by Schur's lemma (cf. \cite{Knapp2001}), all maps $L$ that satisfy this equation are the multiples of the identity. This means that $\bar{\sigma}\otimes\rho_i$ contains exactly one trivial representation, if $\sigma$ and $\rho_i$ are equivalent, and no trivial representation, if they are not. Since $\bar{\sigma}\otimes\rho$ is the direct sum of all $\bar{\sigma}\otimes\rho_i$, the trivial representation $1$ occurs as often in the decomposition of $\bar{\sigma}\otimes\rho$ as $\sigma$ in the decomposition of $\rho$:
\begin{gather*}
n_\sigma(\rho) = n_1(\bar{\sigma}\otimes\rho).
\end{gather*}
\end{proof}
Using this lemma, we can easily calculate the number of singlets in the Clebsch-Gordan series of $\sigma\otimes\bar{\sigma}$ for arbitrary finite-dimensional multiplets $\sigma$ of \text{SU}(3):
\begin{gather*}
n_1(\sigma\otimes\bar{\sigma}) = n_1(\bar{\sigma}\otimes\sigma) = n_\sigma(\sigma) = 1.
\end{gather*}
This means that there is exactly one singlet in $\sigma\otimes\bar{\sigma}$. If we choose $\sigma$ to be unitary, which we always can and will do, the multiples of the identity form trivially a singlet and, therefore, are the only singlet in $\sigma\otimes\bar{\sigma}$.\par
Now, we want to determine the number of octets in $\sigma\otimes\bar{\sigma}$. Although this task is more laborious than calculating the number of singlets, we can again use the previous lemma:
\begin{gather*}
n_8(\sigma\otimes\bar{\sigma}) = n_1(\bar{8}\otimes(\sigma\otimes\bar{\sigma})) = n_1(\bar{\sigma}\otimes(\sigma\otimes\bar{8})) = n_\sigma(\sigma\otimes\bar{8}) = n_\sigma(\sigma\otimes 8),
\end{gather*}
where we used the fact that $8$ and $\bar{8}$ are equivalent in the last step. This equation allows us to calculate the number of multiplets $\sigma$ in $\sigma\otimes 8$ to obtain the number of octets in $\sigma\otimes\bar{\sigma}$. We will accomplish this by determining the Clebsch-Gordan series of $\sigma\otimes 8$ for arbitrary complex finite-dimensional multiplets $\sigma$. However, in order to do this, we need to classify the complex finite-dimensional multiplets of \text{SU}(3) first.\par
All complex finite-dimensional multiplets of \text{SU}(3) can be uniquely characterized by two non-negative integers $p$ and $q$ (cf. \cite{DeSwart1963}, \cite{Lichtenberg}, and review \textit{46. $\text{SU}(n)$ Multiplets and Young Diagrams} in \cite{PDG}). Given two numbers $p$ and $q$, we denote the corresponding multiplet by $D(p,q)$ following \cite{DeSwart1963}. In this notation, the complex conjugate representation of $D(p,q)$ is equivalent to $D(q,p)$, $\overline{D(p,q)} = D(q,p)$, and the dimension of $D(p,q)$ is given by $(p+1)(q+1)\frac{p+q+2}{2}$ (cf. {\cite{DeSwart1963}}). There is also a graphical representation of this classification. We can assign each multiplet $D(p,q)$ a Young tableau (cf. \cite{Lichtenberg}, \cite{sternberg1995}, and review \textit{46. $\text{SU}(n)$ Multiplets and Young Diagrams} in \cite{PDG}). A Young tableau is a diagram consisting of $n\in\mathbb{N}_0$ boxes which are ordered in rows and columns such that all rows are left-aligned. The number of boxes in one row cannot increase from top to bottom. For \text{SU}(3), a Young tableau has no more than three rows, but is not restricted in the number of columns. However, any column containing three boxes can be dropped (cf. \cite{Lichtenberg}). The Young tableau corresponding to $D(p,q)$ consists of $p+2q$ boxes arranged in the following way:
\begin{gather*}
D(p,q) = \ytableausetup{boxsize = 2.5em}\begin{ytableau} \, & \none[\dots] & \, & \, & \none[\dots] & \, \\ \, & \none[\dots] & \,\\ \none[1] & \none[\dots] & \none[q] & \none[q+1] & \none[\dots] & \none[q+p]\end{ytableau}
\end{gather*}
Of course, we can express the multiplets of \text{SU}(3) that already appeared in the course of this thesis, mainly the singlet, triplet, sextet, octet, and decuplet, in terms of $p$ and $q$ and Young tableaux:
\begin{gather*}
1 = D(0,0) = \emptyset\\
3 = D(1,0) = \ytableausetup{boxsize = 1em, centertableaux}\ydiagram{1},\quad \bar{3} = D(0,1) = \ydiagram{1,1}\\
6 = D(2,0) = \ydiagram{2},\quad \bar{6} = D(0,2) = \ydiagram{2,2}\\
8 = D(1,1) = \ydiagram{2,1}\\
10 = D(3,0) = \ydiagram{3},\quad \overline{10} = D(0,3) = \ydiagram{3,3}\\
27 = D(2,2) = \ydiagram{4,2}\\
64 = D(3,3) = \ydiagram{6,3},
\end{gather*}
where we denoted the Young tableau consisting of no boxes with $\emptyset$. Before we move on, we want to single out one important class of multiplets. We will see that the mass matrix decomposition for non-trivial multiplets with $p = 0$ or $q = 0$ is simpler than for the other non-trivial multiplets and all follow the same structure. Therefore, we want to denote them separately. We call a multiplet of the type $D(p,0)$ with $p\in\mathbb{N}$ or $D(0,q)$ with $q\in\mathbb{N}$ \textit{totally symmetric}.\par
One interesting aspect of Young tableaux is that they are very well suited for the calculation of the Clebsch-Gordan series of tensor product representations. One can derive a collection of rules for the Young tableaux -- readily available in literature (cf. \cite{Lichtenberg} and review \textit{46. $\text{SU}(n)$ Multiplets and Young Diagrams} in \cite{PDG}, for instance) -- on how to perform this calculation. We use these rules for the Young tableaux to determine the Clebsch-Gordan series of $\sigma\otimes 8$ for arbitrary complex finite-dimensional multiplets $\sigma$. For this, we represent $8$ by $D(1,1)$ and $\sigma$ by $D(p,q)$ with $p$ and $q$ being non-negative integers. As the actual calculations are rather lengthy, we only present the results here and display the full computation in \autoref{app:Young}. We find:\\\\
\textbf{Singlet}
\begin{flalign*}
D(0,0)\otimes D(1,1) &= D(1,1)&&
\end{flalign*}
\textbf{Totally symmetric multiplets}
\begin{flalign*}
D(1,0)\otimes D(1,1) &= D(1,0)\oplus D(0,2)\oplus D(2,1)&&\\
D(p,0)\otimes D(1,1) &= D(p-2,1)\oplus D(p,0)\oplus D(p-1,2)\oplus D(p+1,1),\quad p\geq 2&&\\
D(0,1)\otimes D(1,1) &= D(0,1)\oplus D(2,0)\oplus D(1,2)&&\\
D(0,q)\otimes D(1,1) &= D(1,q-2)\oplus D(0,q)\oplus D(2,q-1)\oplus D(1,q+1),\quad q\geq 2&&
\end{flalign*}
\textbf{Remaining multiplets}
\begin{flalign*}
D(1,1)\otimes D(1,1) &= D(0,0)\oplus D(1,1)\oplus D(1,1)\oplus D(3,0)\oplus D(0,3)\oplus D(2,2)&&\\
D(p,1)\otimes D(1,1) &= D(p-1,0)\oplus D(p-2,2)\oplus D(p,1)\oplus D(p,1)\oplus D(p+2,0)&&\\
&\ \ \ \oplus D(p-1,3)\oplus D(p+1,2),\quad p\geq 2&&\\
D(1,q)\otimes D(1,1) &= D(0,q-1)\oplus D(2,q-2)\oplus D(1,q)\oplus D(1,q)\oplus D(0,q+2)&&\\
&\ \ \ \oplus D(3,q-1)\oplus D(2,q+1),\quad q\geq 2&&\\
D(p,q)\otimes D(1,1) &= D(p-1,q-1)\oplus D(p+1,q-2)\oplus D(p-2,q+1)&&\\
&\ \ \ \oplus D(p,q)\oplus D(p,q)\oplus D(p+2,q-1)\oplus D(p-1,q+2)&&\\
&\ \ \ \oplus D(p+1,q+1),\quad p,q\geq 2&&
\end{flalign*}
Looking at the Clebsch-Gordan series, we see that the singlet $1$ does not occur in the decomposition of $1\otimes 8$, that every totally symmetric multiplet $\sigma$ appears exactly once in the decomposition of $\sigma\otimes 8$, and that every remaining multiplet $\sigma$ appears exactly twice in the decomposition of $\sigma\otimes 8$. Therefore, the Clebsch-Gordan series of $\sigma\otimes\bar{\sigma}$ contains no octet for $\sigma$ being the singlet, exactly one octet for totally symmetric multiplets $\sigma$, and exactly two octets for the remaining multiplets $\sigma$.\par
Although we have found the number of octets in every hadronic mass matrix, we still need to identify the $\text{SU}(2)\times\text{U}(1)$-symmetric elements of these octets and parametrize them. When discussing the structure of the octet in \autoref{sec:rel_within_multiplets}, we will observe that the octet contains exactly one $\text{SU}(2)\times\text{U}(1)$-singlet. This means that all $\text{SU}(2)\times\text{U}(1)$-symmetric elements in the octet are given by one non-zero element, aside from multiplication with scalars. Hence, every octet in the decomposition of the mass matrix gives rise to one contribution to the hadron masses. We now describe these contributions by parametrizing every octet appearing in the hadronic mass matrices transforming under $\sigma\otimes\bar{\sigma}$.

\subsection*{Mass matrix Parametrization \Romannum{1}: Constructing Octets}

For the parametrization of the octets in $\sigma\otimes\bar{\sigma}$, we need some mathematical background knowledge of Lie group theory. Consider a Lie group $G$, for instance \text{SU}(3). As $G$ is a differentiable manifold, we can assign every point $g\in G$ a tangent space $T_gG$, a real finite-dimensional vector space, and define the tangent bundle {${TG\coloneqq \bigcup_{g\in G}T_gG}$}. With this, we can define vector fields on $G$ as smooth sections of the tangent bundle $TG$. There exists a Lie bracket on the space of vector fields on $G$, called the commutator. Since $G$ is also a Lie group, one can define an important class of vector fields on $G$, the left-invariant vector fields (cf. \cite{Hamilton2017}). They span a real finite-dimensional subspace of the vector fields on $G$ and are closed under the commutator, i.e., the commutator of two left-invariant vector fields is also left-invariant. This turns the restriction of the commutator to the left-invariant vector fields into a Lie bracket on the left-invariant vector fields. Furthermore, the subspace of left-invariant vector fields is isomorphic to the tangent space $T_eG$ of $G$ at the neutral element $e\in G$. With this, we can define a Lie bracket on $T_eG$ by identifying the Lie bracket/commutator on $T_eG$ with the commutator on the left-invariant vector fields via the canonical isomorphism. We call the real finite-dimensional vector space $T_eG$ equipped with this Lie bracket/commutator the Lie algebra of $G$ (cf. \cite{Hamilton2017}). Usually, we denote the Lie algebra of a Lie group with lower case letters, often in Fraktur, for short. For instance, we denote the Lie algebra of $G$ with $\mathfrak{g}$ and the Lie algebra of \text{SU}(3) with $\mathfrak{su}(3)$.\par
If $G$ is a matrix group, i.e., if $G$ is a(n) (embedded) Lie subgroup of the Lie group $\text{GL}(V)$ for a finite-dimensional, real or complex vector space $V$, we can make an additional identification of the Lie algebra $\mathfrak{g}$: For matrix groups $G\subset\text{GL}(V)$, the tangent space $T_eG$ is isomorphic to a real subspace of {${\text{End}(V)\coloneqq\{A:V\rightarrow V\mid \text{A linear}\}}$}. Using the canonical isomorphism between this subspace and $T_eG$, we can equip this subspace with a Lie bracket. This Lie bracket then coincides with the standard commutator for matrices (cf. \cite{Hamilton2017}):
\begin{gather*}
[A,B]\equiv A\circ B - B\circ A\quad\forall A,B\in\text{End}(V).
\end{gather*}
Therefore, we also call this real subspace of $\text{End}(V)$ equipped with the commutator the Lie algebra $\mathfrak{g}$ of the Lie group $G$, if $G\subset\text{GL}(V)$ is a matrix group. In this sense, the Lie algebra $\mathfrak{gl}(V)$ of the Lie group $\text{GL}(V)$ and $\mathfrak{su}(3)$ of \text{SU}(3) are given by (cf. \cite{Hamilton2017}):
\begin{align*}
\mathfrak{gl}(V) &= \text{End}(V),\\
\mathfrak{su}(3) &= \{A\in\text{Mat}(3\times 3,\mathbb{C})\mid A^\dagger = -A\text{ and }\text{Tr}(A) = 0\}.
\end{align*}
We will use both identifications of the Lie algebra -- the Lie algebra as the tangent space at the neutral element and as a subspace of $\text{End}(V)$ -- in this work.\par
Let us now consider a multiplet $D^{(\sigma)}:\text{SU}(3)\rightarrow\text{GL}(V)$ of \text{SU}(3) on a finite-dimensional, complex vector space $V$.
By the definition of a representation, the map $\Phi^{(\sigma)}:\text{SU}(3)\times V\rightarrow V, (A,v)\mapsto D^{(\sigma)}(A)(v)$ is continuous. As $V$ is a finite-dimensional vector space, the continuity of $\Phi^{(\sigma)}$ implies the continuity of $D^{(\sigma)}$. $D^{(\sigma)}$ is a continuous group homomorphism between Lie groups, hence, $D^{(\sigma)}$ is also smooth and a Lie group homomorphism (cf. \cite{Hamilton2017}). \text{SU}(3) and $\text{GL}(V)$ are differentiable manifolds and $D^{(\sigma)}$ is smooth, therefore, we can define a differential $DD^{(\sigma)}\vert_A$ of $D^{(\sigma)}$ for every $A\in\text{SU}(3)$ that maps vectors from the tangent space $T_A\text{SU}(3)$ linearly into the tangent space $T_{D^{(\sigma)}(A)}\text{GL}(V)$. Hence, the differential $DD^{(\sigma)}\vert_\mathbb{1}$ at $\mathbb{1}\in\text{SU}(3)$ is a linear map between Lie algebras:
\begin{gather*}
DD^{(\sigma)}\vert_\mathbb{1}:\mathfrak{su}(3)\rightarrow\mathfrak{gl}(V).
\end{gather*}
As $D^{(\sigma)}$ is a Lie group homomorphism, $DD^{(\sigma)}\vert_\mathbb{1}$ is a Lie algebra homomorphism (cf. \cite{Hamilton2017}), i.e.:
\begin{gather*}
\left[DD^{(\sigma)}\vert_\mathbb{1}(X),DD^{(\sigma)}\vert_\mathbb{1}(Y)\right] = DD^{(\sigma)}\vert_\mathbb{1}\left(\left[X,Y\right]\right)\quad\forall X,Y\in\mathfrak{su}(3).
\end{gather*}
Furthermore, the image {${\text{Im}(D^{(\sigma)})\coloneqq\{D^{(\sigma)}(A)\mid A\in\text{SU}(3)\}}$} of $D^{(\sigma)}$ is a compact subgroup of $\text{GL}(V)$, since \text{SU}(3) is compact and $D^{(\sigma)}$ is a continuous group homomorphism. By Cartan's theorem (cf. \cite{Hamilton2017}), this makes $\text{Im}(D^{(\sigma)})$ an (embedded) Lie subgroup of $\text{GL}(V)$. The Lie algebra $\mathfrak{im}(D^{(\sigma)})$ of $\text{Im}(D^{(\sigma)})$ is, therefore, a real subspace of the Lie algebra $\mathfrak{gl}(V)$ of $\text{GL}(V)$. The differential $DD^{(\sigma)}\vert_\mathbb{1}$ only maps into $\mathfrak{im}(D^{(\sigma)})$, hence, $DD^{(\sigma)}\vert_\mathbb{1}$ is a Lie algebra homomorphism between $\mathfrak{su}(3)$ and $\mathfrak{im}(D^{(\sigma)})$:
\begin{gather*}
DD^{(\sigma)}\vert_\mathbb{1}:\mathfrak{su}(3)\rightarrow\mathfrak{im}(D^{(\sigma)}).
\end{gather*}
One can further show that the map $DD^{(\sigma)}\vert_\mathbb{1}$ is surjective on $\mathfrak{im}(D^{(\sigma)})$\footnote{One can see this in the following way: If $\sigma$ is the trivial representation, the statement is obvious. If $\sigma$ is not trivial, one can apply Sard's theorem to show that $DD^{(\sigma)}\vert_A$ has to be surjective on $T_{D^{(\sigma)}(A)}\text{Im}(D^{(\sigma)})$ for at least one $A\in\text{SU}(3)$. As $D^{(\sigma)}$ is a Lie group homomorphism, this has to be true for every $A\in\text{SU}(3)$, so in particular for $A = \mathbb{1}$.}\linebreak (cf. Ch. \Romannum{3}, § 3, no. 8, Proposition 28 in \cite{Bourbaki}). Now consider the kernel\linebreak {${K\coloneqq\left\{X\in\mathfrak{su}(3)\mid DD^{(\sigma)}\vert_\mathbb{1}(X)=0\right\}}$} of $DD^{(\sigma)}\vert_\mathbb{1}$. The kernel $K$ is a subspace of $\mathfrak{su}(3)$ satisfying:
\begin{gather*}
DD^{(\sigma)}\vert_\mathbb{1}\left(\left[X,Y\right]\right) = \left[DD^{(\sigma)}\vert_\mathbb{1}(X),0\right] = 0\quad\forall X\in\mathfrak{su}(3)\,\forall Y\in K\\
\Rightarrow \left[X,Y\right]\in K\quad\forall X\in\mathfrak{su}(3)\,\forall Y\in K.
\end{gather*}
This implies that $K$ is an ideal of $\mathfrak{su}(3)$. However, \text{SU}(3) is a simple Lie group, hence, its Lie algebra $\mathfrak{su}(3)$ is also simple. A simple Lie algebra has no proper ideals, i.e., every ideal is either $\{0\}$ or the entire Lie algebra. Therefore, we find $K=\{0\}$ or $K=\mathfrak{su}(3)$. If $K$ equates to $\mathfrak{su}(3)$, $DD^{(\sigma)}\vert_\mathbb{1}(X)$ is zero for every $X\in\mathfrak{su}(3)$. It is easy to see that this implies $DD^{(\sigma)}\vert_A = 0\forall A\in\text{SU}(3)$. Using the connectedness of \text{SU}(3) and the fact that $D^{(\sigma)}$ is a Lie group homomorphism, we thus obtain $D^{(\sigma)}(A)=\mathbb{1}\forall A\in\text{SU}(3)$. Therefore, every finite-dimensional complex multiplet $D^{(\sigma)}$ of \text{SU}(3) for which $K$ is equal to $\mathfrak{su}(3)$ is equivalent to the trivial representation of \text{SU}(3).\par
If $D^{(\sigma)}$ is not trivial or, equivalently, if the kernel $K$ is $\{0\}$, the Lie algebra homomorphism $DD^{(\sigma)}\vert_\mathbb{1}$ is bijective, thus, the Lie algebras $\mathfrak{su}(3)$ and $\mathfrak{im}(D^{(\sigma)})$ are isomorphic. This implies that $D^{(\sigma)}$ is a local diffeo- and isomorphism on $\text{Im}(D^{(\sigma)})$ and that the Lie groups \text{SU}(3) and $\text{Im}(D^{(\sigma)})$ are locally diffeo- and isomorphic.\par
From now on, we want to choose $D^{(\sigma)}$ such that it is unitary. As \text{SU}(3) is compact, we can always achieve this by a similarity transformation of $D^{(\sigma)}$ (cf. \cite{Knapp2001}). If $D^{(\sigma)}$ is not unitary, one has to transform the octets in $\sigma\otimes\bar{\sigma}$ we are going to identify with the help of such a similarity transformation. Define the representation $D^{(\sigma\otimes\bar{\sigma})}$:
\begin{gather*}
D^{(\sigma\otimes\bar{\sigma})}:\text{SU}(3)\rightarrow\text{GL}\left(\text{End}(V)\right),\\
D^{(\sigma\otimes\bar{\sigma})}(A)(X)\coloneqq D^{(\sigma)}(A)\circ X\circ D^{(\sigma)}(A)^\dagger = D^{(\sigma)}(A)\circ X\circ D^{(\sigma)}(A)^{-1}.
\end{gather*}
As $\mathfrak{gl}(V)=\text{End}(V)$ is the Lie algebra of $\text{GL}(V)$, $\mathfrak{im}(D^{(\sigma)})$ is a real subspace of $\text{End}(V)$. Let $\tilde{X}$ be an element of $\mathfrak{im}(D^{(\sigma)})$, then we can express it as $\tilde{X} = DD^{(\sigma)}\vert_\mathbb{1}(X)$ with $X\in\mathfrak{su}(3)$. If we insert $\tilde{X}$ in $D^{(\sigma\otimes\bar{\sigma})}$, we find for $A\in\text{SU}(3)$:
\begin{align*}
D^{(\sigma\otimes\bar{\sigma})}(A)(\tilde{X}) &= \left(D^{(\sigma\otimes\bar{\sigma})}(A)\circ DD^{(\sigma)}\vert_\mathbb{1}\right)(X)\\
&= D^{(\sigma)}(A)\circ DD^{(\sigma)}\vert_\mathbb{1}(X)\circ D^{(\sigma)}(A)^{-1}\\
&= D^{(\sigma)}(A)\circ \left(\left.\frac{d}{dt}\right\vert_{t=0}D^{(\sigma)}\left(\exp(tX)\right)\right)\circ D^{(\sigma)}(A^{-1})\\
&= \left.\frac{d}{dt}\right\vert_{t=0} D^{(\sigma)}(A)\circ D^{(\sigma)}\left(\exp(tX)\right)\circ D^{(\sigma)}(A^{-1})\\
&= \left.\frac{d}{dt}\right\vert_{t=0}D^{(\sigma)}\left(A\exp(tX)A^{-1}\right)\\
&= DD^{(\sigma)}\vert_\mathbb{1}\left(AXA^{-1}\right)\\
&= DD^{(\sigma)}\vert_\mathbb{1}\left(D^{(3\otimes\bar{3})}(A)(X)\right)\\
&= DD^{(\sigma)}\vert_\mathbb{1}\left(D^{(8)}(A)(X)\right)\\
&= \left(DD^{(\sigma)}\vert_\mathbb{1}\circ D^{(8)}(A)\right)(X),
\end{align*}
where $D^{(3\otimes\bar{3})}(A)(X)\coloneqq AXA^\dagger = AXA^{-1}$ is the tensor product representation of the fundamental representation $3$ of \text{SU}(3) with its conjugate representation $\bar{3}$. Before we discuss the last two lines of this calculation, we have to pay attention to the difference between real and complex vector spaces: The calculation shows that every element of $\mathfrak{im}(D^{(\sigma)})$ is mapped into $\mathfrak{im}(D^{(\sigma)})$ under $D^{(\sigma\otimes\bar{\sigma})}$. However, $\mathfrak{im}(D^{(\sigma)})$ is not an invariant subspace of $\sigma\otimes\bar{\sigma}$, as $\mathfrak{im}(D^{(\sigma)})$ is a real vector space and $\sigma\otimes\bar{\sigma}$ is a complex representation, i.e., operates on a complex vector space. Nevertheless, we can extend $\mathfrak{im}(D^{(\sigma)})$ to an invariant subspace of $\sigma\otimes\bar{\sigma}$ via complexification. If we take the complexification of $\mathfrak{im}(D^{(\sigma)})$ to be {${\mathfrak{im}(D^{(\sigma)}) + i\cdot\mathfrak{im}(D^{(\sigma)})}$}, the calculation above and the linearity of $D^{(\sigma\otimes\bar{\sigma})}(A)$ are sufficient to show that {${\mathfrak{im}(D^{(\sigma)}) + i\cdot\mathfrak{im}(D^{(\sigma)})}$} is an invariant subspace of $\sigma\otimes\bar{\sigma}$.\par
Similar remarks are important for the understanding of the last two lines: We have seen in \autoref{sec:mass_matrix} that $3\otimes\bar{3}$ as a real representation, i.e., $3\otimes\bar{3}$ only acting on Hermitian $(3\times 3)$-matrices, admits a non-trivial invariant subspace, the space of traceless, Hermitian $(3\times 3)$-matrices. $3\otimes\bar{3}$ restricted to this subspace is a real irreducible representation which we identified as the (real) octet $8$. If we complexify all the spaces, similar statements apply. $3\otimes\bar{3}$ as a complex representation, i.e., $3\otimes\bar{3}$ acting on all complex $(3\times 3)$-matrices, admits the space of traceless, complex $(3\times 3)$-matrices as an invariant subspace and $3\otimes\bar{3}$ restricted to this subspace is a complex multiplet, the (complex) octet $8$. We can easily see that the complexification of $\mathfrak{su}(3)$ is just the space of traceless, complex $(3\times 3)$-matrices. Thus, $3\otimes\bar{3}$ on $\mathfrak{su}(3)$ is just (equivalent to) the real octet and $3\otimes\bar{3}$ on the complexification of $\mathfrak{su}(3)$ is just the complex octet. This fact was used in the last two lines.\par
If we complexify the map $DD^{(\sigma)}\vert_\mathbb{1}$:
\begin{gather*}
DD^{(\sigma)}\vert_\mathbb{1}:\mathfrak{su}(3)+i\cdot\mathfrak{su}(3)\rightarrow\mathfrak{im}(D^{(\sigma)}) + i\cdot\mathfrak{im}(D^{(\sigma)}),\\
DD^{(\sigma)}\vert_\mathbb{1}(X+iY)\coloneqq DD^{(\sigma)}\vert_\mathbb{1}(X) + iDD^{(\sigma)}\vert_\mathbb{1}(Y)\quad\forall X,Y\in\mathfrak{su}(3),
\end{gather*}
we can now see that the prior calculation also holds true, if we take $\tilde{X}$ and $X$ to be elements of the complexified Lie algebras. With this, we find an interesting result: If $\sigma$ is trivial, the invariant subspace $\mathfrak{im}(D^{(\sigma)}) + i\cdot\mathfrak{im}(D^{(\sigma)})$ of $\sigma\otimes\bar{\sigma}$ is trivial as well. However, if $\sigma$ is non-trivial, the complexified map $DD^{(\sigma)}\vert_\mathbb{1}$ is an isomorphism, as the real map $DD^{(\sigma)}\vert_\mathbb{1}$ is an isomorphism. If we now compare the first and last line of the prior calculation, we see that $\sigma\otimes\bar{\sigma}$ restricted to $\mathfrak{im}(D^{(\sigma)}) + i\cdot\mathfrak{im}(D^{(\sigma)})$ is equivalent to the octet $8$, i.e, $D^{(8)}(A) = DD^{(\sigma)}\vert^{-1}_\mathbb{1}\circ D^{(\sigma\otimes\bar{\sigma})}(A)\circ DD^{(\sigma)}\vert_\mathbb{1}$ for every $A\in\text{SU}(3)$. This means that we can parametrize one octet in $\sigma\otimes\bar{\sigma}$ with a (complex) basis of $\mathfrak{im}(D^{(\sigma)}) + i\cdot\mathfrak{im}(D^{(\sigma)})$ for every non-trivial, complex, and finite-dimensional multiplet $\sigma$.\par
Let us now choose a (complex) basis for $\mathfrak{su}(3) + i\cdot\mathfrak{su}(3)$. As explained in \autoref{sec:mass_matrix}, the Gell-Mann matrices $\lambda_k$, $k\in\{1,\ldots, 8\}$, constitute a (real) basis of the space of traceless, Hermitian $(3\times 3)$-matrices, thus, they constitute a (complex) basis of its complexification which coincides with $\mathfrak{su}(3) + i\cdot\mathfrak{su}(3)$. We modify the Gell-Mann matrices to obtain the basis $\{F_k\mid k\in\{1,\ldots, 8\}\}$ of $\mathfrak{su}(3) + i\cdot\mathfrak{su}(3)$ following \cite{Lichtenberg}:
\begin{gather*}
F_k\coloneqq \frac{\lambda_k}{2}\quad\forall k\in\{1,\ldots,8\}.
\end{gather*}
This gives us a (complex) basis $\{F^{(\sigma\otimes\bar{\sigma})}_k\mid k\in\{1,\ldots,8\}\}$ of $\mathfrak{im}(D^{(\sigma)}) + i\cdot\mathfrak{im}(D^{(\sigma)})$ for non-trivial multiplets $\sigma$:
\begin{gather*}
F^{(\sigma\otimes\bar{\sigma})}_k\coloneqq DD^{(\sigma)}\vert_\mathbb{1}\left(F_k\right)\quad\forall k\in\{1,\ldots,8\}.
\end{gather*}
Note that $F^{(3\otimes\bar{3})}_k \equiv F_k$. The matrices $F^{(\sigma\otimes\bar{\sigma})}_k$ parametrize one octet in $\sigma\otimes\bar{\sigma}$ for non-trivial multiplets $\sigma$.\par
However, as we have seen prior in this section, $\sigma\otimes\bar{\sigma}$ contains, in general, two octets. The remaining octet can be constructed out of products of $F^{(\sigma\otimes\bar{\sigma})}_k$. Nevertheless, not any product of $F^{(\sigma\otimes\bar{\sigma})}_k$ is an element of an octet. In this regard, it is instructive to first consider products of $F_k$ to guide our intuition on how to form products of $F^{(\sigma\otimes\bar{\sigma})}_k$. Consider the product of two matrices $F_k$.
We can write any product of operators in terms of commutators and anticommutators:
\begin{gather*}
F_kF_l = \frac{1}{2}[F_k,F_l] + \frac{1}{2}\{F_k,F_l\}\quad\forall k,l\in\{1,\ldots,8\}
\end{gather*}
with $\{A,B\}\coloneqq AB + BA$ being the anticommutator. The commutator is closed in $\mathfrak{su}(3)$, therefore, we can write the commutator $[F_k,F_l]$ as a linear combination of $F_m$ with purely imaginary coefficients:
\begin{gather*}
[F_k,F_l] = \sum^{8}_{m=1} if_{klm}F_m\quad\forall k,l\in\{1,\ldots,8\},
\end{gather*}
where $f_{klm}$ are real constants defining the commutation relations of the $F_k$-matrices. $f_{klm}$ are called structure constants. Using the orthogonality of the $F_k$-matrices under the Hermitian form $\text{Tr}\left(A^\dagger B\right)$ (cf. \cite{Lichtenberg}) and the cyclicity of the trace, one can easily see that the structure constants $f_{klm}$ are totally antisymmetric in $k,\ l,\text{ and }m$:
\begin{gather*}
\text{Tr}(F_kF_l) = \frac{\delta_{kl}}{2}\quad\Rightarrow f_{klm} = -2i\text{Tr}\left([F_k,F_l]F_m\right)\quad\forall k,l,m\in\{1,\ldots,8\}.
\end{gather*}
The anticommutator $\{F_k,F_l\}$ is just an Hermitian $(3\times 3)$-matrix, hence, can be written as a real linear combination of $\mathbb{1}$ and $F_m$:
\begin{gather*}
\{F_k,F_l\} = a_{kl}\mathbb{1} + \sum^{8}_{m=1} d_{klm}F_m\quad\forall k,l\in\{1,\ldots, 8\},
\end{gather*}
where $a_{kl}$ and $d_{klm}$ are real constants. Again, using the orthogonality of the $F_k$-matrices and the fact that they are traceless, we can determine $a_{kl}$ and $d_{klm}$:
\begin{gather*}
3a_{kl} = \text{Tr}\left(\{F_k,F_l\}\right) = \delta_{kl}\Rightarrow a_{kl} = \frac{\delta_{kl}}{3}\quad\forall k,l\in\{1,\ldots, 8\},\\
d_{klm} = 2\text{Tr}\left(\{F_k,F_l\}F_m\right)\quad\forall k,l\in\{1,\ldots, 8\}.
\end{gather*}
Similar to $f_{klm}$, we can directly follow from this that $d_{klm}$ is totally symmetric in $k,\ l,\text{ and }m$ by using the cyclicity of the trace.\par
The analysis of products of $F_k$ has provided us with three tensors: $\delta_{kl} = 3a_{kl}$, $f_{klm}$, and $d_{klm}$. We can now try to contract these tensors with products of $F^{(\sigma\otimes\bar{\sigma})}_k$ to see what kind of objects we can form. First, consider the contraction with $\delta_{kl}$:
\begin{gather*}
C^{(\sigma\otimes\bar{\sigma})}\coloneqq \sum^{8}_{k,l = 1}\delta_{kl}F^{(\sigma\otimes\bar{\sigma})}_kF^{(\sigma\otimes\bar{\sigma})}_l.
\end{gather*}
$C^{(\sigma\otimes\bar{\sigma})}$ commutes with $F^{(\sigma\otimes\bar{\sigma})}_k$ for every $k\in\{1,\ldots,8\}$ (cf. \autoref{app:F-and_D-symbols}). Therefore, $C^{(\sigma\otimes\bar{\sigma})}$ is a Casimir operator. It is known as the quadratic Casimir operator. Casimir operators are constant on multiplets, i.e., $C^{(\sigma\otimes\bar{\sigma})} = c^{(\sigma\otimes\bar{\sigma})}\cdot\mathbb{1}$ for multiplets $\sigma$ (cf. Schur's lemma for Lie algebra representations in \cite{Hall2015}). This means that $C^{(\sigma\otimes\bar{\sigma})}$ spans the singlet in $\sigma\otimes\bar{\sigma}$ for non-trivial multiplets $\sigma$ (cf. \autoref{app:F-and_D-symbols}).\par
Now consider the contraction with $f_{klm}$:
\begin{gather*}
\tilde{F}^{(\sigma\otimes\bar{\sigma})}_k\coloneqq \sum^{8}_{l,m=1}f_{klm}F^{(\sigma\otimes\bar{\sigma})}_lF^{(\sigma\otimes\bar{\sigma})}_m\quad\forall k\in\{1,\ldots,8\}.
\end{gather*}
As the tensor $f_{klm}$ is totally antisymmetric, we have:
\begin{gather*}
\tilde{F}^{(\sigma\otimes\bar{\sigma})}_k = \sum^{8}_{l,m=1}\frac{f_{klm}}{2}\left[F^{(\sigma\otimes\bar{\sigma})}_l, F^{(\sigma\otimes\bar{\sigma})}_m\right]\quad\forall k\in\{1,\ldots,8\}.
\end{gather*}
The map $DD^{(\sigma)}\vert_\mathbb{1}$ that connects $F_k$ and $F^{(\sigma\otimes\bar{\sigma})}_k$ is (the complexification of) a Lie algebra homomorphism, hence:
\begin{gather*}
\left[F^{(\sigma\otimes\bar{\sigma})}_l, F^{(\sigma\otimes\bar{\sigma})}_m\right] = \sum^8_{n=1}if_{lmn}F^{(\sigma\otimes\bar{\sigma})}_n\quad\forall l,m\in\{1,\ldots,8\}.
\end{gather*}
We obtain:
\begin{gather*}
\tilde{F}^{(\sigma\otimes\bar{\sigma})}_k = \sum^{8}_{l,m,n=1}\frac{if_{klm}f_{lmn}}{2}F^{(\sigma\otimes\bar{\sigma})}_n\quad\forall k\in\{1,\ldots,8\}.
\end{gather*}
We can see that $\tilde{F}^{(\sigma\otimes\bar{\sigma})}_k$ is just a linear combination of $F^{(\sigma\otimes\bar{\sigma})}_k$. It is easy to show that the matrices $\tilde{F}^{(\sigma\otimes\bar{\sigma})}_k$ are all linearly independent, if $\sigma$ is non-trivial (cf. \autoref{app:F-and_D-symbols}). As the matrices $F^{(\sigma\otimes\bar{\sigma})}_k$ span an octet in $\sigma\otimes\bar{\sigma}$ for non-trivial multiplets $\sigma$, the matrices $\tilde{F}^{(\sigma\otimes\bar{\sigma})}_k$ form the very same octet.\par
This leaves us with the contraction with the last tensor $d_{klm}$:
\begin{gather*}
D^{(\sigma\otimes\bar{\sigma})}_k\coloneqq \frac{2}{3}\sum^{8}_{l,m=1}d_{klm}F^{(\sigma\otimes\bar{\sigma})}_lF^{(\sigma\otimes\bar{\sigma})}_m\quad\forall k\in\{1,\ldots,8\}.
\end{gather*}
As the products of $F^{(\sigma\otimes\bar{\sigma})}_k$ contracted with the tensors $\delta_{kl}$ and $f_{klm}$ give rise to multiplets in $\sigma\otimes\bar{\sigma}$, it seems reasonable to conjecture that the same applies to $d_{klm}$, i.e, the matrices $D^{(\sigma\otimes\bar{\sigma})}_k$ form a multiplet in $\sigma\otimes\bar{\sigma}$. As we will see, this is the case for non-trivial multiplets $\sigma$. However, the full proof of this requires some formulae whose derivations involve lengthy calculations. We will skip a part of these calculations in this section, but display them in \autoref{app:F-and_D-symbols}. Partially, these computations can be found in \cite{Lichtenberg}.\par
First, we want to consider the commutator of $F^{(\sigma\otimes\bar{\sigma})}_k$ and $D^{(\sigma\otimes\bar{\sigma})}_l$. To compute the commutator of these matrices, we need the following formula (cf. \autoref{app:F-and_D-symbols} and \cite{Lichtenberg}):
\begin{gather*}
\sum^8_{m=1}f_{klm}d_{nrm} = -\sum^8_{m=1}\left(f_{nlm}d_{krm} +f_{rlm}d_{nkm}\right)\quad\forall k,l,n,r\in\{1,\ldots,8\}.
\end{gather*}
Using this formula, we can calculate (cf. \autoref{app:F-and_D-symbols} and \cite{Lichtenberg}):
\begin{gather}
\left[F^{(\sigma\otimes\bar{\sigma})}_k,D^{(\sigma\otimes\bar{\sigma})}_l\right] = \sum^8_{m=1}if_{klm}D^{(\sigma\otimes\bar{\sigma})}_m\quad\forall k,l\in\{1,\ldots,8\}.\label{eq:D-F_com}
\end{gather}
Now define the real vector space $V_D\coloneqq\left\{\sum^8_{k=1}ia_kD^{(\sigma\otimes\bar{\sigma})}_k\mid a_i\in\mathbb{R}\, \forall i\in\{1,\ldots,8\}\right\}$ and consider the (unique) $\mathbb{R}$-linear map $f$ defined by:
\begin{gather*}
f:\mathfrak{su}(3)\rightarrow V_D,\\
f(-iF_k)\coloneqq -iD^{(\sigma\otimes\bar{\sigma})}_k\quad\forall k\in\{1,\ldots,8\}.
\end{gather*}
Consider the kernel $K_D\coloneqq\text{ker}(f)$ of $f$ and let $X\in\mathfrak{su}(3)$ and $Y\in K_D$. We then find:
\begin{align*}
f\left(\left[-iF_k, -iF_l\right]\right) &= f\left(\sum^8_{m=1}f_{klm}(-iF_m)\right) = -\sum^8_{m=1}if_{klm}D^{(\sigma\otimes\bar{\sigma})}_m\\
&= -\left[F^{(\sigma\otimes\bar{\sigma})}_k,D^{(\sigma\otimes\bar{\sigma})}_l\right]\\
&= \left[DD^{(\sigma)}\vert_\mathbb{1}\left(-iF_k\right),f\left(-iF_l\right)\right]\quad\forall k,l\in\{1,\ldots,8\}\\
\Rightarrow f([X,Y]) &= [DD^{(\sigma)}\vert_\mathbb{1}(X),f(Y)] = [DD^{(\sigma)}\vert_\mathbb{1}(X),0] = 0\\
&\Rightarrow [X,Y]\in K_D.
\end{align*}
This means that the kernel $K_D$ is an ideal of $\mathfrak{su}(3)$. However, $\mathfrak{su}(3)$ is simple, therefore, $K_D$ is either $\{0\}$ or $\mathfrak{su}(3)$. Hence, $V_D$ is either isomorphic to $\mathfrak{su}(3)$ or $\{0\}$. Thus, all matrices $D^{(\sigma\otimes\bar{\sigma})}_k$ are either linearly independent 
in the sense of a real vector space or zero. Moreover, we can show for unitary multiplets $\sigma$ that all matrices $D^{(\sigma\otimes\bar{\sigma})}_k$ are also linearly independent in the sense of a complex vector space, if they all are linearly independent in the sense of a real vector space: First, note that the Lie group $\text{Im}(D^{(\sigma)})$ is a(n) (embedded) Lie subgroup of the Lie group $\text{U}(V)\coloneqq{\{A:V\rightarrow V\mid A\text{ linear and unitary}\}}$, if $D^{(\sigma)}$ is a unitary representation. Therefore:
\begin{gather*}
\mathfrak{im}(D^{(\sigma)})\subset\mathfrak{u}(V) = \{A:V\rightarrow V\mid A^\dagger = -A\}.
\end{gather*}
Hence, the matrices $iF^{(\sigma\otimes\bar{\sigma})}_k$ are antihermitian:
\begin{gather*}
iF_k\in\mathfrak{su}(3)\Rightarrow iF^{(\sigma\otimes\bar{\sigma})}_k = DD^{(\sigma)}\vert_\mathbb{1}(iF_k)\in\mathfrak{im}(D^{(\sigma)})\subset\mathfrak{u}(V)\quad\forall k\in\{1,\ldots,8\}.
\end{gather*}
Thus, the matrices $F^{(\sigma\otimes\bar{\sigma})}_k$ and $D^{(\sigma\otimes\bar{\sigma})}_k$ are Hermitian. Now let $a_k$ and $b_k$ be real numbers and define $c_k\coloneqq a_k +ib_k$ for every $k\in\{1,\ldots,8\}$. Choose $a_k$ and $b_k$ such that:
\begin{gather*}
\sum^8_{k=1}c_kD^{(\sigma\otimes\bar{\sigma})}_k = \sum^8_{k=1}a_kD^{(\sigma\otimes\bar{\sigma})}_k + i\sum^8_{k=1}b_kD^{(\sigma\otimes\bar{\sigma})}_k = 0\\
\Rightarrow\left(\sum^8_{k=1}c_kD^{(\sigma\otimes\bar{\sigma})}_k\right)^\dagger = \sum^8_{k=1}a_kD^{(\sigma\otimes\bar{\sigma})}_k - i\sum^8_{k=1}b_kD^{(\sigma\otimes\bar{\sigma})}_k = 0\\
\Rightarrow \sum^8_{k=1}a_kD^{(\sigma\otimes\bar{\sigma})}_k = 0\ \wedge\ \sum^8_{k=1}b_kD^{(\sigma\otimes\bar{\sigma})}_k = 0.
\end{gather*}
Since the matrices $D^{(\sigma\otimes\bar{\sigma})}_k$ are linearly independent in the sense of a real vector space, the coefficients $a_k$ and $b_k$ have to be all zero. This makes the matrices $D^{(\sigma\otimes\bar{\sigma})}_k$ linearly independent in the sense of a complex vector space.\par
Prior in this section, we have seen that the matrices $F_k$ transform as an octet under $A\in\text{SU}(3)$:
\begin{gather*}
\sum^{8}_{l=1}\left(D^{(8)}(A)\right)_{kl}F_l \coloneqq D^{(8)}(A)(F_k) = AF_kA^\dagger\quad\forall k\in\{1,\ldots,8\},
\end{gather*}
where $\left(D^{(8)}(A)\right)_{kl}$ are the coefficients mediating the octet transformation of $F_k$. We show in \autoref{app:F-and_D-symbols} that the same coefficients mediate the transformation of $D^{(\sigma\otimes\bar{\sigma})}_k$ under $A\in\text{SU}(3)$:
\begin{gather*}
D^{(\sigma\otimes\bar{\sigma})}(A)\left(D^{(\sigma\otimes\bar{\sigma})}_k\right) = \sum^{8}_{l=1}\left(D^{(8)}(A)\right)_{kl}D^{(\sigma\otimes\bar{\sigma})}_l\quad\forall k\in\{1,\ldots,8\}.
\end{gather*}
Thus, the matrices $D^{(\sigma\otimes\bar{\sigma})}_k$ form an invariant subspace in $\sigma\otimes\bar{\sigma}$. Furthermore, we can see that they either span an octet or are all zero, as the matrices $D^{(\sigma\otimes\bar{\sigma})}_k$ are either linearly independent or zero and as their transformation coefficients correspond to the coefficients of an octet.

\subsection*{Mass matrix Parametrization \Romannum{2}: Multiplet Classification using Octets}

To summarize our results so far, we have found that both the $F^{(\sigma\otimes\bar{\sigma})}_k$- and $D^{(\sigma\otimes\bar{\sigma})}_k$-matrices span an invariant subspace of $\sigma\otimes\bar{\sigma}$ for \text{SU}(3)-multiplets $\sigma$ that either is $\{0\}$ or an octet. If $\sigma$ is a non-trivial multiplet, the matrices $F^{(\sigma\otimes\bar{\sigma})}_k$ span an octet. What we do not know so far is for which cases the matrices $D^{(\sigma\otimes\bar{\sigma})}_k$ span an octet and, if they do, when this octet is different from the octet of the matrices $F^{(\sigma\otimes\bar{\sigma})}_k$. Therefore, in the next step, we want to proof the following classification for \text{SU}(3)-multiplets $\sigma$:
\begin{align*}
\text{1. }&F^{(\sigma\otimes\bar{\sigma})}_k=0\,\forall k\in\{1,\ldots,8\}\Leftrightarrow D^{(\sigma\otimes\bar{\sigma})}_k=0\,\forall k\in\{1,\ldots,8\}\Leftrightarrow \sigma\text{ trivial}\\
\text{2. }&F^{(\sigma\otimes\bar{\sigma})}_k = c\cdot D^{(\sigma\otimes\bar{\sigma})}_k\neq0\,\forall k\in\{1,\ldots,8\}\text{ for }c\in\mathbb{R}\backslash\{0\}\Leftrightarrow \sigma\text{ totally symmetric}\\
\text{3. }&F^{(\sigma\otimes\bar{\sigma})}_k\text{ and }D^{(\sigma\otimes\bar{\sigma})}_k\text{ are all linearly independent}\Leftrightarrow \sigma\text{ neither trivial nor tot. sym.}
\end{align*}
To show this classification, we need to know how the matrices $F^{(\sigma\otimes\bar{\sigma})}_k$ and $D^{(\sigma\otimes\bar{\sigma})}_k$ act on $V$ which is the vector space the multiplet $D^{(\sigma)}$ acts on. For this, it is helpful to choose an orthonormal basis of $V$ such that the greatest possible number of matrices $F^{(\sigma\otimes\bar{\sigma})}_k$ is diagonal in this basis. To find this basis, we need to introduce the mathematical concept of weights. We will only provide a short overview over the most important results for weights here. These results are mostly taken from \cite{Lichtenberg} where a more in-depth discussion of weights can be found. Let us start by determining the greatest possible number of simultaneously diagonalizable matrices $F^{(\sigma\otimes\bar{\sigma})}_k$. Clearly, the condition that two Hermitian matrices can be diagonalized in the same basis is equivalent to requirement that their commutator has to be zero. Thus, we are looking for the maximal number of mutually commuting matrices $F^{(\sigma\otimes\bar{\sigma})}_k$. As the map $DD^{(\sigma)}\vert_\mathbb{1}$ connecting the matrices $F_k$ and $F^{(\sigma\otimes\bar{\sigma})}_k$ is a Lie algebra isomorphism for non-trivial multiplets $\sigma$, two matrices $F^{(\sigma\otimes\bar{\sigma})}_k$ and $F^{(\sigma\otimes\bar{\sigma})}_l$ commute for non-trivial multiplets $\sigma$ if and only if the matrices $F_k$ and $F_l$ commute. Therefore, the maximal number of commuting matrices $F^{(\sigma\otimes\bar{\sigma})}_k$ is the same as the maximal number of commuting matrices $F_k$. The maximal number of mutually commuting basis elements in a Lie algebra of a Lie group is called the \textit{rank} of that group. The rank of \text{SU}(3) is 2 meaning that at most two matrices $F_k$ and, in return, $F^{(\sigma\otimes\bar{\sigma})}_k$ can mutually commute. As we already saw in \autoref{sec:mass_matrix}, the Gell-Mann matrices $\lambda_3$ and $\lambda_8$ are diagonal and, therefore, commute. This implies that also the matrices $F_3$ and $F_8$ and likewise the matrices $F^{(\sigma\otimes\bar{\sigma})}_3$ and $F^{(\sigma\otimes\bar{\sigma})}_8$ commute. With this, we have found the set of matrices we need to diagonalize simultaneously. Now consider a non-zero vector $v$ in $V$ such that $v$ is both an eigenvector of $F^{(\sigma\otimes\bar{\sigma})}_3$ and of $F^{(\sigma\otimes\bar{\sigma})}_8$, i.e.:
\begin{gather*}
F^{(\sigma\otimes\bar{\sigma})}_3(v) = w^{(\sigma\otimes\bar{\sigma})}_3\cdot v\quad\text{and}\quad F^{(\sigma\otimes\bar{\sigma})}_8(v) = w^{(\sigma\otimes\bar{\sigma})}_8\cdot v,
\end{gather*}
where $w^{(\sigma\otimes\bar{\sigma})}_3$ and $w^{(\sigma\otimes\bar{\sigma})}_8$ are the corresponding eigenvalues of $F^{(\sigma\otimes\bar{\sigma})}_3$ and $F^{(\sigma\otimes\bar{\sigma})}_8$, respectively. The two-dimensional vector $w^{(\sigma\otimes\bar{\sigma})}\coloneqq (w^{(\sigma\otimes\bar{\sigma})}_3,w^{(\sigma\otimes\bar{\sigma})}_8)^\text{T}$ and, sometimes, its components are called \textit{weights} (cf. \cite{Lichtenberg}). The set of all weights for a \text{SU}(3)-multiplet $\sigma$ completely determines the multiplet $\sigma$. Therefore, multiplets are often just characterized by a graphical representation of their weights. This graphical representation, called \textit{weight diagram}, is usually just a coordinate system where each weight is indicated by a dot. Often, the multiplicity of each weight $w^{(\sigma\otimes\bar{\sigma})}$, i.e., the dimension of the set of vectors with weight $w^{(\sigma\otimes\bar{\sigma})}$, is encoded in a weight diagram, too. We will consider the weight diagrams for the triplet, sextet, octet, and decuplet in \autoref{chap:mass_relations}.\par
Physically, weights correspond to quantum numbers of a state. To explore this statement, consider the matrices $F_1,\ F_2\text{, and }F_3$. Under the commutator, they form a subalgebra of $\mathfrak{su}(3) + i\cdot\mathfrak{su}(3)$. This subalgebra is isomorphic to the (complexified) Lie algebra $\mathfrak{su}(2)+i\cdot\mathfrak{su}(2)$. As the Lie algebra $\mathfrak{su}(2)$ of the group \text{SU}(2) represents the spin algebra, we say that $F_1,\ F_2\text{, and }F_3$ also form a spin algebra. This spin algebra follows from flavor transformations and not from a spacetime symmetry, therefore, the matrices $F_1,\ F_2\text{, and }F_3$ are said to constitute the \textit{isospin} instead of the spin. Isospin related properties are usually denoted by a capital ``$I$'', for instance, the isospin generators $I_k$ and the total isospin $I^2$:
\begin{gather*}
I_k \coloneqq F_k\quad\forall k\in\{1,2,3\}\quad\text{and}\quad I^2 \coloneqq \sum^3_{k=1} I^2_k.
\end{gather*}
As the map $DD^{(\sigma)}\vert_\mathbb{1}$ mediating the transition from $F_k$ to $F^{(\sigma\otimes\bar{\sigma})}_k$ is a Lie algebra isomorphism for non-trivial multiplets $\sigma$, the isospin algebra for a non-trivial multiplet $\sigma$ is described by:
\begin{gather*}
I^{(\sigma\otimes\bar{\sigma})}_k \coloneqq F^{(\sigma\otimes\bar{\sigma})}_k\quad\forall k\in\{1,2,3\}\quad\text{and}\quad I^{2;\,(\sigma\otimes\bar{\sigma})} \coloneqq \sum^3_{k=1} \left(I^{(\sigma\otimes\bar{\sigma})}_k\right)^2.
\end{gather*}
In this sense, the weight $w^{(\sigma\otimes\bar{\sigma})}_3$ is the third component of the isospin and, hence, a spin-like quantum number.\par
To interpret the weight $w^{(\sigma\otimes\bar{\sigma})}_8$ in a similar way, we need to consider what kind of algebra $F_1,\ F_2,\ F_3\text{, and }F_8$ form. By looking at the Gell-Mann matrices from \autoref{sec:mass_matrix}, we can easily see that the matrix $F_8$ commutes with every matrix $F_k$ for $k\in\{1,2,3,\}$. Of course, this also applies to the matrices $F^{(\sigma\otimes\bar{\sigma})}_1,$ $F^{(\sigma\otimes\bar{\sigma})}_2,$ $F^{(\sigma\otimes\bar{\sigma})}_3\text{, and }F^{(\sigma\otimes\bar{\sigma})}_8$.
This means that $F^{(\sigma\otimes\bar{\sigma})}_8$ has to be constant on the isospin multiplets making up the multiplet $\sigma$. The operator $F^{(\sigma\otimes\bar{\sigma})}_8$ behaves like a \text{U}(1)-charge in the sense that its eigenvalues are an integer times a constant that is independent of the multiplet $\sigma$. Therefore, we can enforce the eigenvalues to be 1/3 times an integer by changing the normalization of $F^{(\sigma\otimes\bar{\sigma})}_8$ (cf. \cite{Lichtenberg}):
\begin{gather*}
Y\coloneqq \frac{2}{\sqrt{3}}F_8\quad\text{and}\quad Y^{(\sigma\otimes\bar{\sigma})}\coloneqq \frac{2}{\sqrt{3}}F^{(\sigma\otimes\bar{\sigma})}_8.
\end{gather*}
The third-integer-valued quantum number corresponding to $Y$ and $Y^{(\sigma\otimes\bar{\sigma})}$ is called \textit{hypercharge}. Note that the hypercharge follows from flavor transformations and not from a gauge symmetry like a lot of other charges. Now we can see that the weight $w^{(\sigma\otimes\bar{\sigma})}_8$ is related to the hypercharge, a charge-like quantum number.\par
In total, this means that a non-zero vector $v$ in $V$ with weight $w^{(\sigma\otimes\bar{\sigma})}$ has the following isospin and hypercharge quantum numbers:
\begin{gather*}
I_3 = w^{(\sigma\otimes\bar{\sigma})}_3\quad\text{and}\quad Y = \frac{2}{\sqrt{3}}w^{(\sigma\otimes\bar{\sigma})}_8.
\end{gather*}
Now consider an orthonormal basis of $V$ such that $F^{(\sigma\otimes\bar{\sigma})}_3$ and $F^{(\sigma\otimes\bar{\sigma})}_8$ are diagonal in this basis. Then, every vector in this basis is an eigenvector of both $F^{(\sigma\otimes\bar{\sigma})}_3$ and $F^{(\sigma\otimes\bar{\sigma})}_8$ and, thus, has a weight $w^{(\sigma\otimes\bar{\sigma})}$. For the trivial multiplet and for totally symmetric multiplets, such a basis is unique up to phases. For the remaining multiplets, this is not the case. To obtain a unique basis (up to phases), one additionally requires the basis of $V$ to be an eigenbasis of $I^{2;\,(\sigma\otimes\bar{\sigma})}$. This is always possible as the matrices $F^{(\sigma\otimes\bar{\sigma})}_3$ and $F^{(\sigma\otimes\bar{\sigma})}_8$ commute with $I^{2;\,(\sigma\otimes\bar{\sigma})}$. In the following, we are going to label the elements $v$ of such a basis of $V$ by three quantum numbers:
\begin{gather*}
v\equiv\Ket{Y,I,I_3}
\end{gather*}
where $Y$ is the hypercharge of $v$, $I$ the isospin quantum number of $v$, and $I_3$ the third component of the isospin of $v$. In summary, the three quantum numbers $Y$, $I$, and $I_3$ give rise to an orthonormal basis of $V$ that is unique up to phases.\par
One can define a special order on the weights of the multiplet $\sigma$. The first weight in this order is called the highest weight (cf. \cite{Lichtenberg}). The highest weight completely characterizes a multiplet in the sense that two multiplets are equivalent, if they have the same highest weight. If the multiplet $\sigma$ is equivalent to the multiplet $D(p,q)$ for two non-negative integers $p$ and $q$, its highest weight $w^{(\sigma\otimes\bar{\sigma})}_h$ is given by (cf. \cite{Lichtenberg}):
\begin{gather*}
w^{(\sigma\otimes\bar{\sigma})}_h = p\cdot\begin{pmatrix}\frac{1}{2}\\\frac{1}{2\sqrt{3}}\end{pmatrix} + q\cdot\begin{pmatrix}0\\\frac{1}{\sqrt{3}}\end{pmatrix}.
\end{gather*}
For every multiplet, a normalized vector $v_1\in V$ whose weight is the highest weight is unique up to a phase. For the multiplet $D(p,q)$, it is given by:
\begin{gather*}
v_1 = \Ket{Y = \frac{p+2q}{3}, I=\frac{p}{2}, I_3=\frac{p}{2}}.
\end{gather*}
As every weight diagram of \text{SU}(3) is invariant under rotations by $\frac{2\pi}{3}$ (cf. \cite{Lichtenberg}), the highest weight in a multiplet generates two more weights in the multiplet by rotation. For the multiplet $D(p,q)$, the corresponding vectors $v_2$ and $v_3$ of these rotated weights are:
\begin{align*}
v_2 &= \Ket{Y = \frac{p-q}{3}, I=\frac{p+q}{2}, I_3=-\frac{p+q}{2}},\\
v_3 &= \Ket{Y = -\frac{2p+q}{3}, I=\frac{q}{2}, I_3=\frac{q}{2}}.
\end{align*}
Let us now consider the matrix $D^{(\sigma\otimes\bar{\sigma})}_8$. By using the definition of the matrices $D^{(\sigma\otimes\bar{\sigma})}_k$ and a list of the constants $d_{klm}$ (cf. \cite{Lichtenberg}), we find:
\begin{align*}
D^{(\sigma\otimes\bar{\sigma})}_8 &= \frac{2}{3}\left(\frac{1}{\sqrt{3}}\sum^3_{k=1}\left(F^{(\sigma\otimes\bar{\sigma})}_k\right)^2 - \frac{1}{2\sqrt{3}}\sum^7_{k=4}\left(F^{(\sigma\otimes\bar{\sigma})}_k\right)^2 - \frac{1}{\sqrt{3}}\left(F^{(\sigma\otimes\bar{\sigma})}_8\right)^2\right)\\
&= \frac{2}{3}\left(\frac{\sqrt{3}}{2}\sum^3_{k=1}\left(F^{(\sigma\otimes\bar{\sigma})}_k\right)^2 - \frac{1}{2\sqrt{3}}\left(F^{(\sigma\otimes\bar{\sigma})}_8\right)^2 - \frac{1}{2\sqrt{3}}\sum^8_{k=1}\left(F^{(\sigma\otimes\bar{\sigma})}_k\right)^2\right)\\
&= \frac{1}{\sqrt{3}}\left(I^{2;\, (\sigma\otimes\bar{\sigma})} - \frac{1}{4}\left(Y^{(\sigma\otimes\bar{\sigma})}\right)^2 - \frac{1}{3}C^{(\sigma\otimes\bar{\sigma})}\right),
\end{align*}
where $C^{(\sigma\otimes\bar{\sigma})}$ is the quadratic Casimir operator. As already mentioned, the quadratic Casimir operator is constant on multiplets. For $\sigma$ being the multiplet $D(p,q)$, one finds (cf. \cite{Pais1966})\footnote{In this article, $p$ is used differently. The integer $p$ in \cite{Pais1966} has to be understood as $p+q$ in our notation.}:
\begin{gather*}
C^{(\sigma\otimes\bar{\sigma})} = \frac{1}{3}\left(p^2 + q^2 + pq + 3p + 3q\right)\mathbb{1}.
\end{gather*}
With this, we find for multiplets $\sigma = D(p,q)$:
\begin{align}
v^\dagger_1D^{(\sigma\otimes\bar{\sigma})}_8v_1 &= \frac{1}{\sqrt{3}}\left(\frac{p}{2}\left(\frac{p}{2}+1\right) - \frac{1}{4}\cdot\frac{(p+2q)^2}{9} - \frac{1}{3}\cdot\frac{p^2 + q^2 + pq + 3p + 3q}{3}\right)\nonumber\\
\Rightarrow v^\dagger_1D^{(\sigma\otimes\bar{\sigma})}_8v_1 &= \frac{1}{3\sqrt{3}}\left(\frac{p^2 -2pq -2q^2}{3} + \frac{p}{2} - q\right),\label{eq:v_1}\\
v^\dagger_2D^{(\sigma\otimes\bar{\sigma})}_8v_2 &= \frac{1}{3\sqrt{3}}\left(\frac{p^2 -4pq +q^2}{3} + \frac{p+q}{2}\right),\label{eq:v_2}\\
v^\dagger_3D^{(\sigma\otimes\bar{\sigma})}_8v_3 &=  \frac{1}{3\sqrt{3}}\left(\frac{q^2 -2pq -2p^2}{3} + \frac{q}{2} - p\right)\label{eq:v_3}.
\end{align}
\autoref{eq:v_1}, \autoref{eq:v_2}, and \autoref{eq:v_3} are three equations that have to be satisfied for every multiplet $\sigma = D(p,q)$. These equations allow us now to prove the classification for \text{SU}(3)-multiplets mentioned above. Let us start with the first statement of the classification:
\begin{gather*}
\text{1. }F^{(\sigma\otimes\bar{\sigma})}_k=0\,\forall k\in\{1,\ldots,8\}\Leftrightarrow D^{(\sigma\otimes\bar{\sigma})}_k=0\,\forall k\in\{1,\ldots,8\}\Leftrightarrow \sigma\text{ trivial.}
\end{gather*}
First, suppose $\sigma$ is trivial. Then, all matrices $F^{(\sigma\otimes\bar{\sigma})}_k$ and, thus, $D^{(\sigma\otimes\bar{\sigma})}_k$ are trivially zero. Next, suppose all matrices $F^{(\sigma\otimes\bar{\sigma})}_k$ are zero. Then, $DD^{(\sigma)}\vert_\mathbb{1}$ has to be zero. We have already seen that $\sigma$ needs to be trivial in this case. To conclude the proof of the first statement, we need to show that the multiplet $\sigma$ is trivial, if all matrices $D^{(\sigma\otimes\bar{\sigma})}_k$ are zero. Suppose that $\sigma$ is a multiplet such that all matrices $D^{(\sigma\otimes\bar{\sigma})}_k$ are zero. As $\sigma$ is a multiplet of \text{SU}(3), there are non-negative integers $p$ and $q$ such that $\sigma$ is equivalent to $D(p,q)$. All matrices $D^{(\sigma\otimes\bar{\sigma})}_k$ are zero, so, in particular, $D^{(\sigma\otimes\bar{\sigma})}_8$ is zero. Thus, the left-hand sides of \autoref{eq:v_1}, \autoref{eq:v_2}, and \autoref{eq:v_3} are equal to zero:
\begin{align*}
0 &= \frac{1}{3\sqrt{3}}\left(\frac{p^2 -2pq -2q^2}{3} + \frac{p}{2} - q\right),\\
0 &= \frac{1}{3\sqrt{3}}\left(\frac{p^2 -4pq +q^2}{3} + \frac{p+q}{2}\right),\\
0 &=  \frac{1}{3\sqrt{3}}\left(\frac{q^2 -2pq -2p^2}{3} + \frac{q}{2} - p\right).
\end{align*}
Subtracting the third from the first equation and multiplying the result with $3\sqrt{3}$ yields:
\begin{gather*}
0 = p^2 - q^2 + \frac{3}{2}\left(p-q\right) = (p-q)\left(p+q+\frac{3}{2}\right).
\end{gather*}
As $p$ and $q$ are non-negative numbers, we can divide by $p+q+3/2$ to find that $p$ is equal to $q$. Inserting this into the second equation gives:
\begin{gather*}
0 = \frac{1}{3\sqrt{3}}\left(\frac{-2}{3}p^2 + p\right).
\end{gather*}
The only integer-valued solution of this equation for $p$ is zero. Hence, $p$ and $q$ are zero implying that $\sigma$ is trivial. This concludes the proof for the first statement. Let us now turn to the proof of the second statement:
\begin{gather*}
\text{2. }F^{(\sigma\otimes\bar{\sigma})}_k = c\cdot D^{(\sigma\otimes\bar{\sigma})}_k\neq0\,\forall k\in\{1,\ldots,8\}\text{ for }c\in\mathbb{R}\backslash\{0\}\Leftrightarrow \sigma\text{ totally symmetric.}
\end{gather*}
First, suppose that $\sigma$ is a totally symmetric multiplet. We have seen prior in this section that the Clebsch-Gordan series of $\sigma\otimes\bar{\sigma}$ contains exactly one octet in this case. Furthermore, we know that the matrices $F^{(\sigma\otimes\bar{\sigma})}_k$ are either all linearly independent or zero. Likewise, the same statement applies to the matrices $D^{(\sigma\otimes\bar{\sigma})}_k$. However, no matrix $F^{(\sigma\otimes\bar{\sigma})}_k$ or $D^{(\sigma\otimes\bar{\sigma})}_k$ can be zero, as then all matrices $F^{(\sigma\otimes\bar{\sigma})}_k$ or $D^{(\sigma\otimes\bar{\sigma})}_k$ would have to be zero. By the first statement of the classification, this would imply that $\sigma$ is trivial which would be a contradiction. This means that the matrices $F^{(\sigma\otimes\bar{\sigma})}_k$ (and, likewise, $D^{(\sigma\otimes\bar{\sigma})}_k$) are linearly independent and span an octet in $\sigma\otimes\bar{\sigma}$. As $\sigma\otimes\bar{\sigma}$ only contains one octet, $F^{(\sigma\otimes\bar{\sigma})}_k$ and $D^{(\sigma\otimes\bar{\sigma})}_k$ have to span the same octet. We identified this octet with $\mathfrak{im}(D^{(\sigma)}) + i\cdot\mathfrak{im}(D^{(\sigma)})$ prior in this section. Now consider the linear map $f$ that is uniquely defined by:
\begin{gather*}
f:\mathfrak{im}(D^{(\sigma)}) + i\cdot\mathfrak{im}(D^{(\sigma)})\rightarrow\mathfrak{im}(D^{(\sigma)}) + i\cdot\mathfrak{im}(D^{(\sigma)}),\\
f(D^{(\sigma\otimes\bar{\sigma})}_k)\coloneqq F^{(\sigma\otimes\bar{\sigma})}_k\quad\forall k\in\{1,\ldots,8\}.
\end{gather*}
As both matrices $F^{(\sigma\otimes\bar{\sigma})}_k$ and $D^{(\sigma\otimes\bar{\sigma})}_k$ transform with the same transformation coefficients under $A\in\text{SU}(3)$ (cf. \autoref{app:F-and_D-symbols}):
\begin{align*}
&D^{(\sigma\otimes\bar{\sigma})}(A)\left(F^{(\sigma\otimes\bar{\sigma})}_k\right) = \sum^8_{l=1} \left(D^{(8)}(A)\right)_{kl} F^{(\sigma\otimes\bar{\sigma})}_l\quad\forall k\in\{1,\ldots,8\}\quad\text{and}\\
&D^{(\sigma\otimes\bar{\sigma})}(A)\left(D^{(\sigma\otimes\bar{\sigma})}_k\right) = \sum^8_{l=1} \left(D^{(8)}(A)\right)_{kl} D^{(\sigma\otimes\bar{\sigma})}_l\quad\forall k\in\{1,\ldots,8\},
\end{align*}
the map $f$ commutes with the octet transformation $D^{(8)}(A)$:
\begin{align*}
f\left(D^{(\sigma\otimes\bar{\sigma})}(A)\left(D^{(\sigma\otimes\bar{\sigma})}_k\right)\right) &= \sum^8_{l=1} \left(D^{(8)}(A)\right)_{kl} f\left(D^{(\sigma\otimes\bar{\sigma})}_l\right) = \sum^8_{l=1} \left(D^{(8)}(A)\right)_{kl} F^{(\sigma\otimes\bar{\sigma})}_l\\
&= D^{(\sigma\otimes\bar{\sigma})}(A)\left(F^{(\sigma\otimes\bar{\sigma})}_k\right)\\
&= D^{(\sigma\otimes\bar{\sigma})}(A)\left(f\left(D^{(\sigma\otimes\bar{\sigma})}_k\right)\right)\quad\forall k\in\{1,\ldots,8\}\\
\Rightarrow f\circ D^{(8)}(A) &= D^{(8)}(A)\circ f,
\end{align*}
where $D^{(8)}(A)\coloneqq D^{(\sigma\otimes\bar{\sigma})}(A)\vert_{\mathfrak{im}(D^{(\sigma)}) + i\cdot\mathfrak{im}(D^{(\sigma)})}$. By Schur's lemma (cf. \cite{Knapp2001}), every map that commutes with all transformations of a finite-dimensional, unitary, and irreducible representation of a compact Lie group is a multiple of the identity. Thus, $f=c\cdot\mathbb{1}$ for some constant $c\in\mathbb{C}$. However, $c$ has to be real, as all matrices $F^{(\sigma\otimes\bar{\sigma})}_k$ and $D^{(\sigma\otimes\bar{\sigma})}_k$ are Hermitian. Moreover, $c$ cannot be zero, as the matrices $F^{(\sigma\otimes\bar{\sigma})}_k$ cannot be zero. With this, we find:
\begin{gather*}
F^{(\sigma\otimes\bar{\sigma})}_k = f\left(D^{(\sigma\otimes\bar{\sigma})}_k\right) = c\cdot D^{(\sigma\otimes\bar{\sigma})}_k\neq 0\ \forall k\in\{1,\ldots,8\}\text{ for some constant }c\in\mathbb{R}\backslash\{0\}.
\end{gather*}
Let us now suppose that $\sigma$ is a \text{SU}(3)-multiplet such that $F^{(\sigma\otimes\bar{\sigma})}_k = c\cdot D^{(\sigma\otimes\bar{\sigma})}_k\neq 0$ for every $k\in\{1,\ldots,8\}$ and some constant $c\in\mathbb{R}\backslash\{0\}$. As $\sigma$ is a multiplet of \text{SU}(3), there are non-negative integers $p$ and $q$ such that $\sigma$ is equivalent to $D(p,q)$. Consider \autoref{eq:v_1}, \autoref{eq:v_2}, and \autoref{eq:v_3}. We can calculate the left-hand sides by using $D^{(\sigma\otimes\bar{\sigma})}_8 = c^{-1}F^{(\sigma\otimes\bar{\sigma})}_8 =2/\sqrt{3}c\cdot Y^{(\sigma\otimes\bar{\sigma})}$:
\begin{align*}
\frac{2}{3\sqrt{3}c}(p+2q) &= \frac{1}{3\sqrt{3}}\left(\frac{p^2 -2pq -2q^2}{3} + \frac{p}{2} - q\right),\\
\frac{2}{3\sqrt{3}c}(p-q) &= \frac{1}{3\sqrt{3}}\left(\frac{p^2 -4pq +q^2}{3} + \frac{p+q}{2}\right),\\
\frac{-2}{3\sqrt{3}c}(2p+q) &=  \frac{1}{3\sqrt{3}}\left(\frac{q^2 -2pq -2p^2}{3} + \frac{q}{2} - p\right).
\end{align*}
Subtracting the third from the first equation and multiplying the result with $3\sqrt{3}$ yields:
\begin{gather*}
\frac{6}{c}(p+q) = (p-q)\left(p+q+\frac{3}{2}\right).
\end{gather*}
Likewise, subtracting the third from the second equation and multiplying the result with $3\sqrt{3}$ yields:
\begin{gather*}
\frac{6}{c}p = p\left(p-\frac{2}{3}q+\frac{3}{2}\right).
\end{gather*}
Using the first statement of the classification, we know that $\sigma$ cannot be trivial, as the matrices $F^{(\sigma\otimes\bar{\sigma})}_k$ are not zero. Thus, $p$ or $q$ has to be greater than zero. First, suppose that $p$ is greater than zero. Then, $p>0$ and $p+q>0$. This allows us to solve the last two equations for $6/c$ and set the results equal to each other:
\begin{align*}
\frac{p-q}{p+q}\left(p+q+\frac{3}{2}\right) &= p -\frac{2}{3}q + \frac{3}{2}\\
\Rightarrow (p-q)\left(p+q+\frac{3}{2}\right) &= (p+q)\left(p -\frac{2}{3}q + \frac{3}{2}\right)\\
\Rightarrow q\left(\frac{p+q}{3} + 3\right) = 0.
\end{align*}
As $p$ and $q$ are non-negative, $(p+q)/3 + 3$ is greater than zero. This implies that $q$ has to be zero to satisfy the last equation. In this case, $\sigma$ is just $D(p,0)$ with $p>0$, i.e., $\sigma$ is totally symmetric. Now, consider the remaining case of $p=0$. Then, $q$ has to be greater than zero, as $\sigma$ is not trivial. Therefore, $\sigma$ is just $D(0,q)$ with $q>0$, i.e., $\sigma$ is totally symmetric. This concludes the proof of the second statement of the classification. Let us now turn to the proof of the last statement:
\begin{gather*}
\text{3. }F^{(\sigma\otimes\bar{\sigma})}_k\text{ and }D^{(\sigma\otimes\bar{\sigma})}_k\text{ are all linearly independent}\Leftrightarrow \sigma\text{ neither trivial nor tot. sym.}
\end{gather*}
First, suppose that $\sigma$ is a \text{SU}(3)-multiplet that is neither trivial nor totally symmetric. Then, by the first statement of the classification, none of the matrices $F^{(\sigma\otimes\bar{\sigma})}_k$ or $D^{(\sigma\otimes\bar{\sigma})}_k$ can be zero. Therefore, both the matrices $F^{(\sigma\otimes\bar{\sigma})}_k$ and the matrices $D^{(\sigma\otimes\bar{\sigma})}_k$ have to span an octet in $\sigma\otimes\bar{\sigma}$. Let us call the octet spanned by $F^{(\sigma\otimes\bar{\sigma})}_k$ $V_F$ and the octet spanned by $D^{(\sigma\otimes\bar{\sigma})}_k$ $V_D$. Then, we either have $V_F = V_D$ or $V_F\cap V_D = \{0\}$, since the vector space $V_F\cap V_D$ is an invariant subspace of the irreducible representation $V_{F/D}$, so it has to be $V_{F/D}$ or $\{0\}$. However, $V_F$ cannot be equal to $V_D$, as, in this case, we would be able to repeat the first part of the proof of the second statement to show that $F^{(\sigma\otimes\bar{\sigma})}_k$ and $D^{(\sigma\otimes\bar{\sigma})}_k$ are proportional to each other. But if $F^{(\sigma\otimes\bar{\sigma})}_k$ and $D^{(\sigma\otimes\bar{\sigma})}_k$ were proportional to each other, we would find with the second statement that $\sigma$ is totally symmetric contradicting our assumption. Hence, $V_F\cap V_D$ has to be $\{0\}$. This implies that combining a basis of $V_F$ and a basis of $V_D$ gives a set of linearly independent vectors. With this, we obtain that the matrices $F^{(\sigma\otimes\bar{\sigma})}_k$ and $D^{(\sigma\otimes\bar{\sigma})}_k$ are all linearly independent. To conclude this proof, we need to show the converse. Therefore, suppose that $\sigma$ is a \text{SU}(3)-multiplet such that the matrices $F^{(\sigma\otimes\bar{\sigma})}_k$ and $D^{(\sigma\otimes\bar{\sigma})}_k$ are all linearly independent. This implies that both the matrices $F^{(\sigma\otimes\bar{\sigma})}_k$ and the matrices $D^{(\sigma\otimes\bar{\sigma})}_k$ span an octet in $\sigma\otimes\bar{\sigma}$. These octets cannot be equal, as the matrices $F^{(\sigma\otimes\bar{\sigma})}_k$ and $D^{(\sigma\otimes\bar{\sigma})}_k$ are all linearly independent. Hence, the Clebsch-Gordan series of $\sigma\otimes\bar{\sigma}$ contains at least two different octets. However, neither the trivial multiplet nor any totally symmetric multiplet contains more than one octet. Therefore, $\sigma$ is neither trivial nor totally symmetric. This concludes the proof of the classification.\par
Our initial motivation for classifying and parametrizing \text{SU}(3)-multiplets was to describe mass matrices transforming under $\sigma\otimes\bar{\sigma}$ for arbitrary complex finite-dimensional \text{SU}(3)-multiplets $\sigma$, where these mass matrices are, among others, subject to octet enhancement, i.e., are merely a sum of singlets and octets. In the course of our investigations, we have found that $\sigma\otimes\bar{\sigma}$ contains exactly one singlet for every \text{SU}(3)-multiplet $\sigma$, namely the multiples of the identity. Furthermore, we have seen that the Clebsch-Gordan series of $\sigma\otimes\bar{\sigma}$ contains no octet for $\sigma$ being the trivial multiplet, exactly one octet for $\sigma$ being totally symmetric, and exactly two octets for the remaining multiplets $\sigma$. After this classification, we aimed to parametrize the octets in $\sigma\otimes\bar{\sigma}$. We accomplished this by introducing matrices $F^{(\sigma\otimes\bar{\sigma})}_k$ and $D^{(\sigma\otimes\bar{\sigma})}_k$. We have found that each span of the matrices $F^{(\sigma\otimes\bar{\sigma})}_k$ and $D^{(\sigma\otimes\bar{\sigma})}_k$ is a candidate for an octet in $\sigma\otimes\bar{\sigma}$, i.e., is an octet or $\{0\}$. Lastly, we have proven a classification of \text{SU}(3)-multiplets $\sigma$ in terms of the matrices $F^{(\sigma\otimes\bar{\sigma})}_k$ and $D^{(\sigma\otimes\bar{\sigma})}_k$:
\begin{align*}
\text{1. }&F^{(\sigma\otimes\bar{\sigma})}_k=0\,\forall k\in\{1,\ldots,8\}\Leftrightarrow D^{(\sigma\otimes\bar{\sigma})}_k=0\,\forall k\in\{1,\ldots,8\}\Leftrightarrow \sigma\text{ trivial}\\
\text{2. }&F^{(\sigma\otimes\bar{\sigma})}_k = c\cdot D^{(\sigma\otimes\bar{\sigma})}_k\neq0\,\forall k\in\{1,\ldots,8\}\text{ for }c\in\mathbb{R}\backslash\{0\}\Leftrightarrow \sigma\text{ totally symmetric}\\
\text{3. }&F^{(\sigma\otimes\bar{\sigma})}_k\text{ and }D^{(\sigma\otimes\bar{\sigma})}_k\text{ are all linearly independent}\Leftrightarrow \sigma\text{ neither trivial nor tot. sym.}
\end{align*}
We can now sum up this classification by saying that each set of matrices $F^{(\sigma\otimes\bar{\sigma})}_k$ and $D^{(\sigma\otimes\bar{\sigma})}_k$ spans no octet for $\sigma$ being trivial, spans one (and the same) octet for $\sigma$ being totally symmetric, and each spans a different octet for the remaining multiplets $\sigma$. Together with the number of octets in $\sigma\otimes\bar{\sigma}$, we arrive at the conclusion that every matrix in $\sigma\otimes\bar{\sigma}$ only consisting of a sum of octets is a linear combination of the matrices $F^{(\sigma\otimes\bar{\sigma})}_k$ and $D^{(\sigma\otimes\bar{\sigma})}_k$ for complex finite-dimensional \text{SU}(3)-multiplets $\sigma$. In total, this means that every mass matrix transforming under $\sigma\otimes\bar{\sigma}$ for arbitrary complex finite-dimensional \text{SU}(3)-multiplets $\sigma$ and subject to octet enhancement is a linear combination of the 17 matrices $\mathbb{1}$, $F^{(\sigma\otimes\bar{\sigma})}_k$ ($k\in\{1,\ldots,8\}$), and $D^{(\sigma\otimes\bar{\sigma})}_l$ ({${l\in\{1,\ldots,8\}}$}).

\subsection*{GMO Mass Formula}

Now, we have all tools at hand to write down the Gell-Mann--Okubo mass formula. Prior in \autoref{sec:GMO_formula}, we have seen that we can group hadrons into complex finite-dimensional \text{SU}(3)-multiplets $\sigma$ such that the masses of these hadrons are given to first order in perturbation theory (neglecting isospin symmetry breaking) by the eigenvalues of a mass matrix transforming under $\sigma\otimes\bar{\sigma}$ which is subject to octet enhancement and $\text{SU}(3)\rightarrow\text{SU}(2)\times\text{U}(1)$ symmetry breaking. We have just shown that every matrix transforming under $\sigma\otimes\bar{\sigma}$ and subject to octet enhancement is a linear combination of $\mathbb{1}$, $F^{(\sigma\otimes\bar{\sigma})}_k$, and $D^{(\sigma\otimes\bar{\sigma})}_l$. However, we are only interested in the $\text{SU}(2)\times\text{U}(1)$-invariant part of this parametrization. When discussing the weight diagram of the octet in \autoref{sec:rel_within_multiplets}, we will see that the octet only contains one $\text{SU}(2)\times\text{U}(1)$-invariant element, aside from multiplication with scalars. Therefore, only one matrix from each set of matrices $F^{(\sigma\otimes\bar{\sigma})}_k$ and $D^{(\sigma\otimes\bar{\sigma})}_k$ contributes to the hadronic mass matrix. The matrices $F^{(\sigma\otimes\bar{\sigma})}_k$ and $D^{(\sigma\otimes\bar{\sigma})}_k$ transform in the same way as $F_k$ under \text{SU}(3). As $F_8$ is invariant under $\text{SU}(2)\times\text{U}(1)$ (cf. \autoref{sec:mass_matrix}), $F^{(\sigma\otimes\bar{\sigma})}_8$ and $D^{(\sigma\otimes\bar{\sigma})}_8$ are invariant under $\text{SU}(2)\times\text{U}(1)$ as well. This means that the masses of hadrons in a \text{SU}(3)-multiplet $\sigma$ are given to first order in perturbation theory (neglecting isospin symmetry breaking) by the eigenvalues of the matrix $m^{(\sigma\otimes\bar{\sigma})}$:
\begin{gather*}
m^{(\sigma\otimes\bar{\sigma})} = m_0\mathbb{1} + m^F_8F^{(\sigma\otimes\bar{\sigma})}_8 + m^D_8D^{(\sigma\otimes\bar{\sigma})}_8,
\end{gather*}
where $m_0$, $m^F_8$, and $m^D_8$ are parameters which can, in general, be different for different multiplets $\sigma$. We can rewrite this in terms of isospin and hypercharge operators:
\begin{align*}
m^{(\sigma\otimes\bar{\sigma})} &= \left(m_0 - \frac{c^{(\sigma\otimes\bar{\sigma})}}{3\sqrt{3}}m^D_8\right)\mathbb{1} + \frac{\sqrt{3}m^F_8}{2}Y^{(\sigma\otimes\bar{\sigma})} + \frac{m^D_8}{\sqrt{3}}\left(I^{2;\, (\sigma\otimes\bar{\sigma})} - \frac{1}{4}\left(Y^{(\sigma\otimes\bar{\sigma})}\right)^2\right)\\
&\eqqcolon \tilde{m}_0\mathbb{1} + \tilde{m}^F_8Y^{(\sigma\otimes\bar{\sigma})} + \tilde{m}^D_8\left(I^{2;\, (\sigma\otimes\bar{\sigma})} - \frac{1}{4}\left(Y^{(\sigma\otimes\bar{\sigma})}\right)^2\right),
\end{align*}
where $c^{(\sigma\otimes\bar{\sigma})}$ is the constant describing the Casimir operator $C^{(\sigma\otimes\bar{\sigma})}\eqqcolon c^{(\sigma\otimes\bar{\sigma})}\mathbb{1}$. All that is left to do is to diagonalize the matrix $m^{(\sigma\otimes\bar{\sigma})}$ to obtain its eigenvalues. However, this is rather easy as we have already introduced the orthonormal basis $\{\Ket{Y,I,I_3}\}$ of $\sigma$. In this basis, $Y^{(\sigma\otimes\bar{\sigma})}$ and $I^{2;\, (\sigma\otimes\bar{\sigma})}$ are diagonal, hence, $m^{(\sigma\otimes\bar{\sigma})}$ is also diagonal in this basis. Let us denote the mass of the hadron in the multiplet $\sigma$ with hypercharge $Y$, total isospin $I$, third component $I_3$ of the isospin by $m(Y,I,I_3)$. Neglecting isospin symmetry breaking, $m(Y,I,I_3)$ is given by:
\begin{gather}
m(Y,I,I_3) = \tilde{m}_0 + \tilde{m}^F_8\cdot Y + \tilde{m}^D_8\left(I(I+1) - \frac{Y^2}{4}\right) + \mathcal{O}(\varepsilon^2_8).\label{eq:GMO_mass_formula}
\end{gather}
This is the famous \textit{Gell-Mann--Okubo mass formula} (cf. \cite{Okubo}). It implies all hadronic mass formulae and relations from \autoref{sec:mass_matrix}. In particular, the GMO mass formula predicts equal spacing rules for totally symmetric multiplets $\sigma$. As the Clebsch-Gordan series of $\sigma\otimes\bar{\sigma}$ for totally symmetric multiplets $\sigma$ contains exactly one octet, the parametrization of $m^{(\sigma\otimes\bar{\sigma})}$ with both $F^{(\sigma\otimes\bar{\sigma})}_8$ and $D^{(\sigma\otimes\bar{\sigma})}_8$ is redundant and we can freely choose one of the parameters $m^F_8$ and $m^D_8$. For instance, we can choose $m^D_8 = 0$ for totally symmetric multiplets $\sigma$. We then find for totally symmetric multiplets $\sigma$:
\begin{gather*}
m(Y,I,I_3) = \tilde{m}_0 + \tilde{m}^F_8\cdot Y + \mathcal{O}(\varepsilon^2_8).
\end{gather*}
In this case, the mass of a hadron in $\sigma$ only depends its hypercharge. If $Y_\text{min}$ and $Y_\text{max}$ are the minimal and maximal hypercharge in a totally symmetric multiplet $\sigma$, the other hypercharges occurring in $\sigma$ are element of $\{Y_\text{min}, Y_\text{min}+1,\ldots, Y_\text{max}-1, Y_\text{max}\}$ (cf. \cite{Lichtenberg}). Hence, all hadron masses in a totally symmetric multiplet $\sigma$ are equidistant and, thus, subject to an equal spacing rule.

\newpage
\section{Additional Mass Contributions}
\label{sec:add_con}

In \autoref{sec:GMO_formula}, we made a lot of assumptions and used a lot of approximations to arrive at the GMO mass formula. In particular, we considered a Lagrangian only consisting of kinetic and mass terms of three quark flavors (up, down, and strange) and of flavor symmetric interaction terms as a starting point for our investigations. However, the only (quark) flavor symmetric interaction in the SM is the strong interaction. Following such an approach, we clearly neglected the non-flavor symmetric electroweak interaction and the remaining particle content of the SM. Furthermore, we took the masses of the up and down quark to be equal, but we observe in Nature that the masses of the up and down quark are most likely to be different. Obviously, these effects give rise to corrections to the GMO mass formula. The goal of this section is to describe additional contributions to the GMO mass formula originating from these effects. In particular, we aim to derive expressions for the isospin symmetry breaking induced by the difference of the up and down quark mass and for electromagnetic corrections in first order perturbation theory.

\subsection*{Isospin Symmetry Breaking induced by the Up-Down-Mass Difference}

In \autoref{chap:hadron_masses}, we considered $\mathcal{L}_\text{QCD}$, a Lagrangian describing three quark flavors and a flavor symmetric interaction between them, and showed that its corresponding Hamilton operator $H_\text{QCD}$ decomposes as follows under \text{SU}(3)-flavor transformations:
\begin{gather*}
H_\text{QCD} = H^{0}_\text{QCD} + \varepsilon_3\cdot H^{8}_{\text{QCD};\, 3} + \varepsilon_8\cdot H^{8}_{\text{QCD};\, 8},
\end{gather*}
where $H^{0}_\text{QCD}$ is a singlet of \text{SU}(3), $H^{8}_{\text{QCD};\, 3}$ and $H^{8}_{\text{QCD};\, 8}$ are the 3rd and 8th component of an octet, respectively, $\varepsilon_3\coloneqq m_u-m_d$, and $\varepsilon_8\coloneqq (m_u+m_d-2m_s)/\sqrt{3}$. In \autoref{sec:GMO_formula}, we only investigated the case of $\varepsilon_3 = 0\Leftrightarrow m_u=m_d$ to ensure that the Hamilton operator $H_\text{QCD}$ is invariant under $\text{SU}(2)\times\text{U}(1)$-flavor transformations. The generators of this $\text{SU}(2)\times\text{U}(1)$-symmetry are the three isospin operators $I_1$, $I_2$, $I_3$, and the hypercharge operator $Y$. Therefore, the $\text{SU}(2)\times\text{U}(1)$-flavor symmetry of $H_\text{QCD}$ (or its subsymmetry $\text{SU}(2)$ to be precise) is often referred to as \textit{isospin symmetry}. In the most general case of all quark masses being different, the isospin symmetry is broken by $H^{8}_{\text{QCD};\, 3}$, as the third component of an octet is not isospin symmetric. Nevertheless, there is still a residual symmetry left in this case. As we have seen in \autoref{sec:Trafo_QCD}, the quark mass matrix $\mathcal{M}\coloneqq\text{diag}(m_u,m_d,m_s)$ is still invariant under diagonal \text{SU}(3)-phase modulations, i.e.:
\begin{gather*}
A\mathcal{M}A^\dagger = \mathcal{M}\text{ with }A=\text{diag}(e^{i\alpha},e^{i\beta},e^{-i(\alpha+\beta)})\in\text{SU}(3)\ \forall\alpha,\beta\in\mathbb{R}.
\end{gather*}
We can identify this residual symmetry with the Lie group $\text{U}(1)\times\text{U}(1)$. The generators of this group are the third component $I_3$ of the isospin and the hypercharge operator $Y$.\par
Now, we want to examine how the introduction of isospin symmetry breaking in terms of $H^{8}_{\text{QCD};\, 3}$ affects the derivation of the GMO mass formula from \autoref{sec:GMO_formula}. Even though one might think that this is rather laborious task given the length of \autoref{sec:GMO_formula}, it is actually surprisingly simple. All statements from \autoref{sec:GMO_formula} apply and can be derived in a very similar fashion for a Hamilton operator $H_\text{QCD}$ including a non-vanishing $\varepsilon_3\cdot H^{8}_{\text{QCD};\, 3}$ term. We just have to be cautious to include the third component of an octet when needed and be aware that the flavor symmetry of $H_\text{QCD}$ is now just $\text{U}(1)\times\text{U}(1)$ instead of $\text{SU}(2)\times\text{U}(1)$. In sloppy terms, one might say we just have to replace the phrases ``singlet plus the 8th component of an octet'', ``$\text{SU}(2)\times\text{U}(1)$'', and ``$\mathcal{O}\left(\varepsilon^2_8\right)$'' in \autoref{sec:GMO_formula} by ``singlet plus the 3rd and 8th component of an octet'', ``$\text{U}(1)\times\text{U}(1)$'', and ``$\mathcal{O}\left(\varepsilon^2_8\right)+\mathcal{O}\left(\varepsilon^2_3\right)+\mathcal{O}\left(\varepsilon_3\varepsilon_8\right)$'', respectively. Repeating \autoref{sec:GMO_formula} with these comments in mind, we find that we can group hadrons into complex finite-dimensional \text{SU}(3)-multiplets $\sigma$ such that the masses of the hadrons in $\sigma$ are given to first order in perturbation theory (meaning to first order in both $\varepsilon_3$ and $\varepsilon_8$ this time) by the eigenvalues of a mass matrix transforming under $\sigma\otimes\bar{\sigma}$ which is subject to octet enhancement and $\text{SU}(3)\rightarrow\text{U}(1)\times\text{U}(1)$ symmetry breaking. We already know that every mass matrix transforming under $\sigma\otimes\bar{\sigma}$ which is subject to octet enhancement is a linear combination of the matrices $\mathbb{1}$, $F^{(\sigma\otimes\bar{\sigma})}_k$, and $D^{(\sigma\otimes\bar{\sigma})}_k$. Hence, all that is left to do is to find the $\text{U}(1)\times\text{U}(1)$-invariant elements of this decomposition. When discussing the weight diagram of the octet in \autoref{sec:rel_within_multiplets}, we will see that every octet contains exactly two linearly independent elements that are $\text{U}(1)\times\text{U}(1)$-invariant. These elements correspond to the Gell-Mann matrices $\lambda_3$ and $\lambda_8$, i.e., $F_3$ and $F_8$. This means the only $F^{(\sigma\otimes\bar{\sigma})}_k$- and $D^{(\sigma\otimes\bar{\sigma})}_k$-matrices that are $\text{U}(1)\times\text{U}(1)$-invariant are $F^{(\sigma\otimes\bar{\sigma})}_3$, $F^{(\sigma\otimes\bar{\sigma})}_8$, $D^{(\sigma\otimes\bar{\sigma})}_3$, and $D^{(\sigma\otimes\bar{\sigma})}_8$. Thus, the masses of hadrons in a \text{SU}(3)-multiplet $\sigma$ are given to first order in perturbation theory by the eigenvalues of the matrix $m^{(\sigma\otimes\bar{\sigma})}$:
\begin{gather*}
m^{(\sigma\otimes\bar{\sigma})} = m_0\mathbb{1} + m^F_3F^{(\sigma\otimes\bar{\sigma})}_3 + m^F_8F^{(\sigma\otimes\bar{\sigma})}_8 + m^D_3D^{(\sigma\otimes\bar{\sigma})}_3 + m^D_8D^{(\sigma\otimes\bar{\sigma})}_8,
\end{gather*}
where $m_0$, $m^F_3$, $m^F_8$, $m^D_3$, and $m^D_8$ are parameters which can, in general, be different for different multiplets $\sigma$. Again, we only need the $F^{(\sigma\otimes\bar{\sigma})}_k$- or the $D^{(\sigma\otimes\bar{\sigma})}_k$-matrices for the description of totally symmetric multiplets $\sigma$. For instance, we can set $m^D_3 = m^D_8 = 0$:
\begin{align*}
m^{(\sigma\otimes\bar{\sigma})} &= m_0\mathbb{1} + m^F_3F^{(\sigma\otimes\bar{\sigma})}_3 + m^F_8F^{(\sigma\otimes\bar{\sigma})}_8\\
&= m_0\mathbb{1} + m^F_3I^{(\sigma\otimes\bar{\sigma})}_3 + \tilde{m}^F_8Y^{(\sigma\otimes\bar{\sigma})}
\end{align*}
with $\tilde{m}^F_8\coloneqq\sqrt{3}m^F_8/2$.\par
To find the eigenvalues of $m^{(\sigma\otimes\bar{\sigma})}$, we need to diagonalize it. For totally symmetric multiplets $\sigma$, this poses no problem, as $I^{(\sigma\otimes\bar{\sigma})}_3$ and $Y^{(\sigma\otimes\bar{\sigma})}$ are diagonal in the basis $\{\Ket{Y,I,I_3}\}$. If we denote by $m(Y,I,I_3)$ the mass of the hadron in the multiplet $\sigma$ whose hypercharge is $Y$, whose total isospin is $I$, and whose third component of the isospin is $I_3$, we obtain for totally symmetric multiplets $\sigma$:
\begin{gather}
m(Y,I,I_3) = m_0 + m^F_3\cdot I_3 + \tilde{m}^F_8\cdot Y + \mathcal{O}\left(\varepsilon^2_3\right) + \mathcal{O}\left(\varepsilon_3\varepsilon_8\right) + \mathcal{O}\left(\varepsilon^2_8\right).\label{eq:iso_tot_sym}
\end{gather}
If $\sigma$ is not totally symmetric, we need to take $D^{(\sigma\otimes\bar{\sigma})}_3$ and $D^{(\sigma\otimes\bar{\sigma})}_8$ into account as well. We have already seen that $D^{(\sigma\otimes\bar{\sigma})}_8$ is diagonal in the basis $\{\Ket{Y,I,I_3}\}$. However, $D^{(\sigma\otimes\bar{\sigma})}_3$ does not need to be diagonal in this basis, as $I^{2;\,(\sigma\otimes\bar{\sigma})}$ and $D^{(\sigma\otimes\bar{\sigma})}_3$ do not, in general, commute. Nevertheless, $D^{(\sigma\otimes\bar{\sigma})}_3$ does commute with $I^{(\sigma\otimes\bar{\sigma})}_3$ and $Y^{(\sigma\otimes\bar{\sigma})}$, since $F^{(\sigma\otimes\bar{\sigma})}_3$ commutes with $I^{(\sigma\otimes\bar{\sigma})}_3$ and $Y^{(\sigma\otimes\bar{\sigma})}$ (cf. \autoref{eq:D-F_com}). This means that $D^{(\sigma\otimes\bar{\sigma})}_3$ links hadrons with same hypercharge $Y$ and third isospin component $I_3$, but different total isospin $I$. In this regard, $D^{(\sigma\otimes\bar{\sigma})}_3$ is a source for $\Lambda^0$-$\Sigma^0$-mixing in the spin-$\frac{1}{2}$ baryon octet\footnote{In fact, one can compute the mixing angle $\alpha$ of the $\Lambda^0$-$\Sigma^0$-mixing by calculating the entries of $D^{(8\otimes \bar{8})}_3$. One obtains:
\[\frac{1}{2}\tan (2\alpha) = \frac{m_{\Sigma} - m_{\Sigma^+} + m_p - m_n}{\sqrt{3}(m_{\Sigma} - m_{\Lambda})},\]
where $m_{\Sigma^+}$, $m_p$, and $m_n$ are the masses of the baryons given in the index and $m_{\Sigma}$ and $m_{\Lambda}$ are the diagonal entries of $m^{(8\otimes \bar{8})}$ corresponding to $\ket{0,1,0}$ (roughly $\Sigma^0$) and $\ket{0,0,0}$ (roughly $\Lambda^0$) respectively. As the mixing and, thus, $\alpha$ is small, we have $m_{\Sigma} \approx m_{\Sigma^0}$ and $m_{\Lambda} \approx m_{\Lambda^0}$ to good approximation. Furthermore, one can use $\tan (2\alpha)/2 \approx \tan (\alpha)$ for small $\alpha$ to recover the formula of R. H. Dalitz and F. von Hippel for the mixing angle $\alpha$ (cf. \cite{Dalitz1964}; note the different sign convention for $\alpha$).}.\par
To diagonalize $m^{(\sigma\otimes\bar{\sigma})}$ for multiplets $\sigma$ that are not totally symmetric, we make use of the following trick: In physical application, we are only interested in the eigenvalues of $m^{(\sigma\otimes\bar{\sigma})}$ for isospin breaking terms that are way smaller than their $\text{SU}(2)\times\text{U}(1)$-invariant counterparts, i.e., for $\varepsilon_3\ll\varepsilon_8$. This allows us to treat the $D^{(\sigma\otimes\bar{\sigma})}_3$-term as a perturbation of $m^{(\sigma\otimes\bar{\sigma})}$. However, the eigenvectors of the unperturbed part of $m^{(\sigma\otimes\bar{\sigma})}$ are just $\Ket{Y,I,I_3}$. Hence, we obtain for the eigenvalue $m^{(\sigma\otimes\bar{\sigma})}_{Y,I,I_3}$ of $m^{(\sigma\otimes\bar{\sigma})}$ corresponding to $\Ket{Y,I,I_3}$:
\begin{align*}
m^{(\sigma\otimes\bar{\sigma})}_{Y,I,I_3} &= \tilde{m}_0 + m^F_3\cdot I_3 + \tilde{m}^F_8\cdot Y + \tilde{m}^D_8\left(I(I+1) - \frac{Y^2}{4}\right)\\
&\ \ \ + m^D_3\Braket{Y,I,I_3|D^{(\sigma\otimes\bar{\sigma})}_3|Y,I,I_3} + \mathcal{O}\left(\varepsilon_8\cdot\left(\frac{\varepsilon_3}{\varepsilon_8}\right)^2\right),
\end{align*}
where $\tilde{m}_0$, $\tilde{m}^F_8$, and $\tilde{m}^D_8$ are defined like in \autoref{sec:GMO_formula}. If $\sigma$ is an octet $8$, one can calculate:
\begin{gather*}
\Braket{Y,I,I_3|D^{(8\otimes\bar{8})}_3|Y,I,I_3} = I_3Y.
\end{gather*}
Denoting the mass of a hadron in an octet whose hypercharge is $Y$, whose total isospin is $I$, and whose third component of the isospin is $I_3$ by $m(Y,I,I_3)$, we find for the octet:
\begin{align}
m(Y,I,I_3) &= \tilde{m}_0 + m^F_3\cdot I_3 + \tilde{m}^F_8\cdot Y + m^D_3\cdot I_3Y + \tilde{m}^D_8\left(I(I+1) - \frac{Y^2}{4}\right)\nonumber\\
&\ \ \ + \mathcal{O}\left(\varepsilon^2_3\right) + \mathcal{O}\left(\varepsilon_3\varepsilon_8\right) + \mathcal{O}\left(\varepsilon^2_8\right) + \mathcal{O}\left(\varepsilon_8\cdot\left(\frac{\varepsilon_3}{\varepsilon_8}\right)^2\right).\label{eq:iso_octet}
\end{align}
\indent The mass difference between up and down quark is not the only source of isospin symmetry breaking. The electroweak interaction also breaks the isospin symmetry, as, for instance, the up and down quark carry a different electric charge. However, the electroweak interaction also breaks the residual $\text{U}(1)\times\text{U}(1)$-symmetry originating from the quark mass matrix $\mathcal{M}$, since the weak interaction couples quarks of different flavor to each other. Nevertheless, the contribution of the weak interaction to the hadron masses is parametrically suppressed by at least $(M/M_W)^2$, where $M\sim\tilde{m}_0$ is the mass scale of the hadronic multiplet at hand and $M_W$ is the mass of the $W$-bosons. Thus, we can safely assume that for most hadronic multiplets the contribution of the weak interaction to the hadron masses is very small in comparison to the other isospin symmetry breaking effects like up-down-mass difference and electromagnetism. This allows us to neglect all flavor changing currents in the SM and restore the residual $\text{U}(1)\times\text{U}(1)$-flavor symmetry. In principal, the restoration of this symmetry lets us describe electromagnetic contributions to the GMO mass formula in the same way as the contribution arising from the mass difference of the up and down quark. For this, we simply need to incorporate the expansion parameter of the electromagnetic interaction into $\varepsilon_3$. As the electromagnetic coupling constant is directly related to $\sqrt{\alpha}$, we just need to incorporate $\sqrt{\alpha}$ into $\varepsilon_3$. However, the GMO mass formula including isospin symmetry breaking only describes contributions to first order in $\varepsilon_3$. This way, we can only describe electromagnetic contributions at order $\sqrt{\alpha}$. Since we do not expect any electromagnetic contribution at order $\sqrt{\alpha}$, but first at order $\alpha$, we have to think of something else to describe electromagnetic contributions at order $\alpha$.

\subsection*{Electromagnetic Contributions}

Up to this point, we have aimed to be very precise and rigorous with every statement and tried to motivate every assumption we had to make. However, we now want to employ a more heuristic and phenomenological model to describe the corrections to the GMO mass formula arising from the electromagnetic interaction. We will see in \autoref{chap:mass_relations} that the model we present in this section leads to mass relations that are well known in the literature like the Coleman-Glashow mass relation (cf. \cite{coleman-glashow}). For our model, we make the assumption that the electromagnetic interaction inside hadrons can be described in an effective approach: We imagine that we are able to ``integrate out'' the photon field inside a hadron\footnote{We refrain from defining the meaning of the phrase ``integrate out'', as it is not crucial for our reasoning.}. This leaves us with an effective hadronic Hamilton operator $H$. Let us now suppose that $H$ takes the following form:
\begin{align*}
H &= H_\text{QCD} + \alpha\Delta H + \mathcal{O}\left(\alpha^2\right),
\end{align*}
where $H_\text{QCD}$ is the Hamilton operator from the previous section that led us to the GMO mass formula (including isospin symmetry breaking originating from the mass difference between up and down quark) and $\Delta H$ describes the electromagnetic interaction inside hadrons at order $\alpha$.
Again, we identify the mass of a hadron with an eigenvalue of $H$ and calculate the eigenvalues of $H$ in a perturbative treatment. The perturbative treatment of the electromagnetic interaction for the computation of hadron masses is justified, since we expect the electromagnetic interaction to play a secondary role in the formation of hadrons. This time, we take $\alpha$ to be the expansion parameter of the perturbation series. For a perturbative treatment of $H$, we need to know the eigenvalues and -spaces of the unperturbed part, $H_\text{QCD}$, first. The determination of these properties was the concern of the previous sections in this chapter. We have found that we can group hadrons, i.e., eigenstates of $H_\text{QCD}$, into \text{SU}(3)-multiplets $\sigma$ such that every hadron in $\sigma$ is fully determined by three quantum numbers $Y$, $I$, and $I_3$. Usually, we have to mind the degeneracy of the eigenvalues of the unperturbed operator when applying perturbation theory. If $H_\text{QCD}$ was completely flavor symmetric, i.e., invariant under \text{SU}(3)-flavor transformations, all hadrons in a multiplet $\sigma$ would be degenerate. However, the symmetry breaking induced by $\varepsilon_8$ lifts the degeneracy between different isospin multiplets\footnote{Isospin multiplets are the irreducible representations of the isospin symmetry group $\text{SU}(2)\times\text{U}(1)$. Isospin multiplets within a \text{SU}(3)-multiplet $\sigma$ are formed by hadrons which have the same hypercharge $Y$ and total isospin $I$ (cf. \autoref{sec:rel_within_multiplets}).}, i.e., only hadrons with the same $Y$ and $I$ would be degenerate, if the term proportional to $\varepsilon_8$ was the only term breaking the flavor symmetry (cf. \autoref{sec:GMO_formula}). Additionally, the symmetry breaking induced by $\varepsilon_3$ lifts the remaining degeneracy within those isospin multiplets (cf. first part of \autoref{sec:add_con}). Thus, we can take the degeneracy of hadrons in $\sigma$ to be already lifted by the symmetry breaking terms that are included in $H_\text{QCD}$\footnote{On a deeper level, this is not true. We can deduce from experimental data that the contributions arising from the electromagnetic interaction and from $\varepsilon_3$ are in the same order of magnitude, thus, it is not sensible to consider the electromagnetic interaction to be a perturbation of the $\varepsilon_3$-term included in $H_\text{QCD}$. However, one can use the residual $\text{U}(1)\times\text{U}(1)$-flavor symmetry of the electromagnetic interaction to show that the electromagnetic interaction as a perturbation of {${H^0_\text{QCD} + \varepsilon_8\cdot H^8_\text{QCD;8}}$} is diagonal in the basis $\Ket{Y,I,I_3}$ when restricted to the eigenspaces of {${H^0_\text{QCD} + \varepsilon_8\cdot H^8_\text{QCD;8}}$}. This way, we obtain the same results as if we considered the electromagnetic interaction to be a perturbation of $H_\text{QCD}$ including the $\varepsilon_3$-term.}. If we denote the mass (including electromagnetic corrections) of a hadron in $\sigma$ with hypercharge $Y$, total isospin $I$, and third component $I_3$ of the isospin by $m_\alpha(Y,I,I_3)$, we find to first order in perturbation theory:
\begin{align*}
m_\alpha(Y,I,I_3) = m(Y,I,I_3) + \alpha\cdot\Braket{Y,I,I_3|\Delta H|Y,I,I_3} + \mathcal{O}\left(\alpha^2\right) + \mathcal{O}\left(\alpha\varepsilon_3\right) + \mathcal{O}\left(\alpha\varepsilon_8\right),
\end{align*}
where $m(Y,I,I_3)$ is the mass of the hadron following from $H_\text{QCD}$, i.e., without electromagnetic corrections (cf. \autoref{eq:iso_tot_sym} and \autoref{eq:iso_octet}) and $\Ket{Y,I,I_3}$ is the eigenstate of the hadron in $\sigma$.\par
In order to calculate the first order correction of the electromagnetic interaction, we need some information about $\Delta H$. For the following arguments, we suppose that $\Delta H$ is given by:
\begin{align*}
\Delta H &=\sum_{f_1,f_2\in\{\text{u,d,s}\}}\frac{q_{f_1}q_{f_2}}{e^2}C_{f_1f_2},
\end{align*}
where $q_{f_1}$ and $q_{f_2}$ are the electric charges associated with the flavors $f_1$ and $f_2$ and the operators $C_{f_1f_2}$ mediate the electromagnetic interaction inside the hadrons at order $\alpha$. The choice of the form of $\Delta H$ is motivated by the following consideration: If one wants to introduce the electromagnetic interaction of the quarks to the Lagrangian $\mathcal{L}_\text{QCD}$, one has to add the Yang-Mills term $\mathcal{L}_\text{QED}^\text{YM}$ of the photon field and the coupling of the electromagnetic current $J_\text{QED}^\mu$ to the photon field $A_\mu$:
\begin{align*}
\mathcal{L}(x) &= \mathcal{L}_\text{QCD}(x) + \mathcal{L}_\text{QED}^\text{YM}(x) + J_\text{QED}^\mu(x) A_\mu(x)\quad\text{with}\\
J_\text{QED}^\mu(x) &\coloneqq \sum_{f\in\{\text{u,d,s}\}} q_f\bar{\psi}_f(x)\gamma^\mu\psi_f(x),
\end{align*}
where $\psi_q(x) \coloneqq q(x)$ for $q\in\{\text{u, d, s}\}$ with the notation of \autoref{sec:Trafo_QCD} and similar for $\bar{\psi}_q$. Commonly, perturbative QFT calculations involve the computation of matrix elements containing the time-ordered exponential of the interaction terms in the Lagrangian. By expanding the time-ordered exponential, we find an expansion of the matrix elements in the coupling constants, i.e., in $\alpha$ in the case of quantum electrodynamics (QED). Therefore, the terms at order $\alpha$ always involve a product of two electromagnetic currents $J^\mu_\text{QED}$:
\begin{gather*}
J^\mu_\text{QED}(x) J^\nu_\text{QED}(y) = \sum_{f_1, f_2\in\{\text{u,d,s}\}} q_{f_1}q_{f_2}\left(\bar{\psi}_{f_1}(x)\gamma^\mu\psi_{f_1}(x)\right)\left(\bar{\psi}_{f_2}(y)\gamma^\nu\psi_{f_2}(y)\right)
\end{gather*}
If we assume that the product of two electromagnetic currents also occurs in the operator $\Delta H$ describing the electromagnetic interaction inside hadrons at order $\alpha$, it seems plausible that $\Delta H$ takes the given form. Of course, the integration over the space-time variables $x$ and $y$, terms related to or arising from photon fields, and Wilson loops and lines guaranteeing the gauge invariance of $\Delta H$ also need to be considered for a complete description of $\Delta H$.\par
To proceed, we need to make use of the quark model. In the quark model, we think of baryons and mesons as composite particles and imagine the hadrons to be made out of valence and sea quarks (cf. \cite{Povh2014}). The valence quarks of a hadron dictate its properties and quantum numbers. The valence quark content of a baryon in the SM is given by three quarks, while the valence quark content of a meson consists of a quark and an antiquark. With the quark model in mind, we want to argue now that the following formula applies:
\begin{align*}
\alpha\cdot\Braket{Y,I,I_3|\Delta H|Y,I,I_3} = \Delta_\alpha\sum_{(i,j)}\frac{q_iq_j}{e^2} + \mathcal{O}(\alpha\varepsilon_3) + \mathcal{O}(\alpha\varepsilon_8),
\end{align*}
where $\Delta_\alpha$ is a quantity that is constant on $\sigma$, i.e., independent of $Y$, $I$, and $I_3$, $i$ and $j$ denote valence (anti)quarks of the hadron described by $\Ket{Y,I,I_3}$, $(i,j)$ is a pair of different valence (anti)quarks, the sum runs over every pair once, and $q_i$ and $q_j$ are the charges of the valence (anti)quarks $i$ and $j$, respectively.\par
To convince ourselves that this formula is reasonable, we have to compute:
\begin{gather*}
\alpha\cdot\Braket{Y,I,I_3|\Delta H|Y,I,I_3} = \sum_{f_1,f_2\in\{\text{u,d,s}\}}\Braket{Y,I,I_3|\alpha C_{f_1f_2}|Y,I,I_3}\frac{q_{f_1}q_{f_2}}{e^2}
\end{gather*}
For this, it is helpful to consider the electromagnetic current $J^\mu_\text{QED}$ again. $J^\mu_\text{QED}$ is a conserved Noether current. Its conserved Noether charge $Q_\text{QED}$ is just the electric charge operator, i.e, the operator counting the electric charges in a state when applied to it. In that regard, it seems plausible that the expectation value of the operator $\Delta H$ from which we expect to involve the product of two currents $J^\mu_\text{QED}$ is just the sum of some constants $\Delta_{\alpha;\, ij}(Y,I,I_3)$ times charge products $q_iq_j$ of pairs $(i,j)$ of valence and/or sea (anti)quarks in the hadron state:
\begin{align*}
\alpha\cdot\Braket{Y,I,I_3|\Delta H|Y,I,I_3} = \sum_{(i,j)}\Delta_{\alpha;\, ij}(Y,I,I_3)\frac{q_iq_j}{e^2}.
\end{align*}
We assume for the following considerations that the contribution of sea quarks to this sum is negligible. Whether this assumption is satisfied and, if not, how large the contribution of the sea quarks is, poses an interesting question for further investigations. Moreover, it is sensible to assume that the constants $\Delta_{\alpha;\, ij}(Y,I,I_3)$ are independent of $i$, $j$, $Y$, $I$, and $I_3$: 
The \text{SU}(3)-flavor transformations form an approximate symmetry of the hadrons which is only broken by $\varepsilon_3$ and $\varepsilon_8$ neglecting the electromagnetic interaction. The electromagnetic interaction only breaks this approximate symmetry, since the quarks carry different electric charge. However, the constants $\Delta_{\alpha;\, ij}(Y,I,I_3)$ do not involve the electric charge of the quarks. To this end, it seems plausible that the constants $\Delta_{\alpha;\, ij}(Y,I,I_3)$ are independent of the flavor-dependent quantities $i$, $j$, $Y$, $I$, and $I_3$ aside from corrections proportional to products of $\alpha$ and $\varepsilon_{3/8}$:
\begin{gather*}
\Delta_{\alpha;\, ij}(Y,I,I_3) = \Delta_\alpha + \mathcal{O}(\alpha\varepsilon_3) + \mathcal{O}(\alpha\varepsilon_8).
\end{gather*}
This reproduces the initially given formula. 
In total, the mass of a hadron including electromagnetic corrections of order $\alpha$ should be given by:
\begin{align}
m_\alpha(Y,I,I_3) = m(Y,I,I_3) + \Delta_\alpha\sum_{(i,j)}\frac{q_iq_j}{e^2} + \mathcal{O}\left(\alpha^2\right) + \mathcal{O}\left(\alpha\varepsilon_3\right) + \mathcal{O}\left(\alpha\varepsilon_8\right).\label{eq:mass_ele}
\end{align}

\section{Heavy Quark Symmetry}
\label{sec:heavy_quark}

So far, we have only considered light hadrons, i.e., hadrons that are formed out of the three light quarks up, down, and strange as valence quarks. The description of these hadrons led us to the classification of light hadrons into \text{SU}(3)-multiplets $\sigma$ and gave us a mass formula for the light hadrons (cf. \autoref{sec:GMO_formula} and \autoref{sec:add_con}). However, we can only use this mass formula to derive mass relations within a \text{SU}(3)-multiplet $\sigma$ (cf. \autoref{sec:mass_matrix} and \autoref{sec:rel_within_multiplets}): The mass formula is just a parametrization of the hadron masses in a \text{SU}(3)-multiplet $\sigma$ in terms of undetermined parameters like $m_0$, $m^F_3$, $m^F_8$, $m^D_3$, $m^D_8$, or $\Delta_\alpha$. Since for some \text{SU}(3)-multiplets $\sigma$ the number of undetermined parameters is smaller than the number of (degenerate) hadron masses in $\sigma$, the hadron masses in $\sigma$ have to satisfy mass relations to be in agreement with the mass formula. But since we have not found any relation between the undetermined parameters of different hadronic \text{SU}(3)-multiplets $\sigma$ yet, we are unable to formulate mass relations between different multiplets.\par
In this section, we want to incorporate the description of hadrons with heavy quarks like charm and bottom quark as part of their valence quark content into our model for the hadron masses. Furthermore, we will see that we can use a heuristic approach to hadrons containing a heavy quark to formulate mass relations between different $\text{SU}(3)$-multiplets. Let us start by considering a Lagrangian $\mathcal{L}^5_\text{QCD}$ that contains five quark flavors:
\begin{gather*}
\mathcal{L}^5_{\text{QCD}}(\bar{q},q) = \sum\limits_{q\in\{\text{u,d,s,c,b}\}}\bar{q}\left(i\slashed{D} -  m_q\right)q + \mathcal{L}_\text{YM},
\end{gather*}
where we employ the same notation as in \autoref{sec:Trafo_QCD}. In the same way as for $\mathcal{L}_\text{QCD}$ (cf. \autoref{sec:Trafo_QCD}), we can define flavor transformations of the fields $q$. The only difference is that now the group of flavor transformations is given by \text{SU}(5) instead of \text{SU}(3). The flavor transformations of the fields $q$ furnish transformations of the quark mass matrix $\mathcal{M}$ that form a representation of \text{SU}(5). Likewise, the field flavor transformations give rise to transformations of the Lagrangian $\mathcal{L}^5_\text{QCD}$ and the corresponding Hamilton operator $H^5_\text{QCD}$. Similar to the case of three flavors, we can decompose the Hamilton operator $H^5_\text{QCD}$ into a singlet of \text{SU}(5) plus terms that transform like the adjoint representation of \text{SU}(5) under flavor transformations. If all quark masses in the theory are roughly the same, i.e., if the \text{SU}(5)-flavor transformations form an approximate global symmetry of the Lagrangian and if we can treat the non-singlet terms in the Hamilton operator $H^5_\text{QCD}$, which are proportional to the quark mass differences, as a perturbation of the Hamilton operator, one can proceed similarly to \autoref{sec:GMO_formula} to find that one can group the hadrons of this 5-flavors theory into \text{SU}(5)-multiplets and derive a mass formula for the hadrons in a \text{SU}(5)-multiplet.\par
In Nature, however, we observe that the charm and bottom quark are much heavier than up, down, and strange quark: $m_\text{b},m_\text{c}\gg m_\text{u},m_\text{d},m_\text{s}$. Even though this means that the \text{SU}(5)-flavor symmetry is severely broken in Nature and we cannot apply a perturbative treatment to the Hamilton operator $H^5_\text{QCD}$, the \text{SU}(5)-flavor group and its multiplets still provide a classification for hadrons (cf. $\text{SU}(4)$-flavor group in review \textit{105. Charmed Baryons} in \cite{PDG}). To describe states in a \text{SU}(5)-multiplet, we need quantum numbers additional to $Y$, $I$, and $I_3$. We can link two of these additional quantum numbers to the charm $C$ and to the bottomness $B$. Essentially, the charm $C$ and the bottomness $B$ of a hadron specify -- together with other quantum numbers -- how many charm and bottom (anti)quarks are part of the valence quark content of the hadron.\par
The flavor transformations of only the first three quarks u, d, and s form a Lie subgroup of the \text{SU}(5)-flavor transformations which is equivalent to the Lie group \text{SU}(3). Hence, the hadronic \text{SU}(5)-multiplets decompose into a direct sum of multiplets of this \text{SU}(3)-flavor subgroup. The decomposition of a \text{SU}(5)-multiplet can be done in such a way that all hadrons contained in a given \text{SU}(3)-multiplet of this decomposition have the same number of charm and bottom valence (anti)quarks. This allows us to group hadrons into \text{SU}(3)-multiplets where each hadron in a given \text{SU}(3)-multiplet has the same number of charm and bottom valence (anti)quarks. The \text{SU}(3)-flavor group of the quarks u, d, and s can be treated as an approximate global symmetry of the Lagrangian. This allows us to apply the state formalism we introduced in the previous chapters and sections to recover the GMO mass formula including isospin symmetry breaking and electromagnetic corrections for the hadronic \text{SU}(3)-multiplets with fixed number of charm and bottom valence (anti)quarks.\par
So far, we have found that we can apply the mass formula from \autoref{sec:add_con} to hadronic \text{SU}(3)-multiplets with fixed number of charm and bottom valence (anti)quarks. The mass formula from \autoref{sec:add_con} contains undetermined parameters that can, in general, be different for different hadronic \text{SU}(3)-multiplets. In the next step, we want to link the undetermined parameters of a hadronic \text{SU}(3)-multiplet with exactly one charm valence (anti)quark and no bottom valence (anti)quark to the undetermined parameters of the hadronic \text{SU}(3)-multiplet that is obtained from the first multiplet via the exchange of charm and bottom quarks. Before we consider these hadronic \text{SU}(3)-multiplets, it is instructive and helpful to investigate an easily accessible example of a composite particle, namely the hydrogen atom. We think of the hydrogen atom as a composite particle that is formed by a proton and an electron. The mass $m_\text{H}$ of the hydrogen atom is, in a naive picture, just the sum of the proton mass $m_p$, the electron mass $m_e$, and the binding energy $E$:
\begin{gather*}
m_\text{H} = m_p + m_e + E.
\end{gather*}
In non-relativistic quantum mechanics, one can determine the binding energy $E$ of the hydrogen atom by solving the Schr\"odinger equation of the proton-electron-system. Neglecting all additional terms arising from the finestructure and the hyperfinestructure, one obtains:
\begin{gather*}
E = -\mu\frac{\alpha^2}{2n}\quad\text{with}\quad\mu\coloneqq\frac{m_em_p}{m_e+m_p},
\end{gather*}
where $\mu$ is the reduced mass and $n\in\mathbb{N}$ is the principal quantum number. As the mass of the electron is much smaller than the mass of the proton, we can expand the reduced mass $\mu$ in a Taylor series:
\begin{gather*}
\mu = m_e\sum^\infty_{k=0}\left(-\frac{m_e}{m_p}\right)^k.
\end{gather*}
Thus, we can write the mass $m_\text{H}$ of the hydrogen atom as the proton mass plus a power series in $m_e/m_p$:
\begin{gather*}
m_\text{H}(n) = m_p + m_e\left(1 - \frac{\alpha^2}{2n} + \mathcal{O}\left(\frac{m_e}{m_p}\right)\right).
\end{gather*}
Now consider an atom where we have exchanged the proton of the hydrogen atom by a particle with same electric charge, but different mass $m_d>m_e$ like, for instance, the deuteron. The mass $m_\text{D}$ of this atom is given by an expression similar to the one for the mass $m_\text{H}$ of the hydrogen atom; We simply have to replace the proton mass $m_p$ in the expression for $m_\text{H}$ with the new mass $m_d$:
\begin{gather*}
m_\text{D}(n) = m_d + m_e\left(1 - \frac{\alpha^2}{2n} + \mathcal{O}\left(\frac{m_e}{m_d}\right)\right).
\end{gather*}
We can easily see that the mass formulae for $m_\text{H}$ and $m_\text{D}$ give rise to a mass relation that is satisfied to lowest order in $m_e/m_p$ and $m_e/m_d$:
\begin{gather*}
m_\text{H}(n^\prime) - m_\text{H}(n) = m_\text{D}(n^\prime) - m_\text{D}(n) + \mathcal{O}\left(\frac{m_e}{m_p}\right) + \mathcal{O}\left(\frac{m_e}{m_d}\right),
\end{gather*}
where $n$ and $n^\prime$ are natural numbers. We can even say that the corrections to this mass relation have to be of order $m_e(1/m_p - 1/m_d)$, since this relation is trivially true, if $m_p = m_d$:
\begin{gather*}
m_\text{H}(n^\prime) - m_\text{H}(n) = m_\text{D}(n^\prime) - m_\text{D}(n) + \mathcal{O}\left(m_e\left(\frac{1}{m_p} - \frac{1}{m_d}\right)\right).
\end{gather*}
We are able to write down this mass relation, since the binding energy of a hydrogen-like atom exhibits a ``heavy nucleus symmetry'': If the nucleus of a hydrogen-like atom is much heavier than the surrounding electron, the binding energy of this hydrogen-like atom depends very little on the exact mass of the nucleus. This dependence on the mass of the nucleus is suppressed by factors of $m_e/m_N$, where $m_N$ is the mass of the nucleus. Therefore, the mass difference between two excitations of a hydrogen-like atom with a heavy nucleus is approximately symmetric under the exchange of the heavy nucleus, i.e., picks up corrections in the order of $m_e(1/m_N - 1/m_{N^\prime})$ under the exchange of a heavy nucleus $N$ with a heavy nucleus $N^\prime$.\par
The example of the hydrogen-like atoms demonstrated that a $1/m$-expansion allows us to find mass relations. To this end, we should be able to find mass relations between hadrons containing heavy valence quarks like charm or bottom quarks (and a heavy quark symmetry\footnote{For a deeper discussion of heavy quark symmetry, confer \cite{Neubert1994}.}, in general), if the hadron masses can be expanded in $1/m_Q$ where $m_Q$ is the mass of a heavy quark. Indeed, one often finds in heavy quark effective theory (HQET; cf. \cite{Neubert1994} and \cite{Jenkins1996}) that objects like the Lagrangian, fields, operators, and hadron masses are described by a $1/m_Q$-expansion, implying the existence of mass relations between hadrons containing heavy valence quarks in HQET (cf. \cite{Neubert1994}). As the parameter $1/m_Q$ is dimensionful, we need some mass scale $\Lambda$ to form a dimensionless expansion parameter $\Lambda/m_Q$. In HQET, this mass scale is typically $\Lambda_\text{QCD}$ (cf. \cite{Jenkins2008}), the mass scale of QCD obtained from dimensional transmutation.\par
Motivated by HQET, we want to incorporate a heuristic description of hadrons containing a heavy quark into our model of hadron masses. For this, consider the following matrix $S$:
\begin{gather*}
S \coloneqq \begin{pmatrix}1 & 0 & 0 & 0 & 0\\ 0 & 1 & 0 & 0 & 0\\ 0 & 0 & 1 & 0 & 0\\ 0 & 0 & 0 & 0 & 1\\ 0 & 0 & 0 & 1 & 0\end{pmatrix}.
\end{gather*}
If we label the columns from left to right and the rows from top to bottom with u, d, s, c, and b, the matrix $S$ defines a flavor transformation of the fields:
\begin{align*}
q&\xrightarrow{S\in\text{U}(5)}q^\prime \coloneqq \sum\limits_{p\in\{\text{u,d,s,c,b}\}} S_{qp}p,\\
\bar{q}&\xrightarrow{S\in\text{U}(5)}\bar{q}^\prime \coloneqq \sum\limits_{p\in\{\text{u,d,s,c,b}\}} S^\ast_{qp}\bar{p}.
\end{align*}
Note that $S$ is not a \text{SU}(5)-flavor transformation, but a \text{U}(5)-flavor transformation. The transformation given by $S$ only exchanges the charm with the bottom field, so it only exchanges charm and bottom quarks. Suppose that there exists an operator $\hat S$ such that:
\begin{gather*}
H^5_\text{QCD}(q^\prime,\bar{q}^\prime) = \hat{S}^\dagger H^5_\text{QCD}(q,\bar{q})\hat S.
\end{gather*}
We can interpret $\hat S$ as the operator that exchanges charm and bottom quarks when applied to hadronic states in the framework of the state formalism. Now consider a hadronic \text{SU}(3)-multiplet with exactly one charm and no bottom valence (anti)quarks and its bottom counterpart. Denote the hadronic states in the charm \text{SU}(3)-multiplet by $\Ket{Y,I,I_3}_\text{c}$ and the hadronic states in the bottom counterpart by $\Ket{Y,I,I_3}_\text{b}$. Motivated by HQET, we postulate that $\hat S$ links $\Ket{Y,I,I_3}_\text{c}$ and $\Ket{Y,I,I_3}_\text{b}$ in the following way:
\begin{gather*}
\hat S\Ket{Y,I,I_3}_\text{c} = \Ket{Y,I,I_3}_\text{b} + \mathcal{O}\left(\Lambda_\text{QCD}\left(\frac{1}{m_\text{c}} - \frac{1}{m_\text{b}}\right)\right).
\end{gather*}
We chose to use $\mathcal{O}\left(\Lambda_\text{QCD}\left(1/m_\text{c} - 1/m_\text{b}\right)\right)$ instead of $\mathcal{O}\left(\Lambda_\text{QCD}/m_\text{c}\right) + \mathcal{O}\left(\Lambda_\text{QCD}/m_\text{b}\right)$ here, since we want to reflect the fact that there would be an exact global \text{U}(2)-flavor symmetry between charm and bottom quark, if the charm and bottom quark masses were equal (neglecting the electroweak interaction). If this symmetry was exact, the masses of the hadrons corresponding to $\Ket{Y,I,I_3}_\text{c}$ and $\Ket{Y,I,I_3}_\text{b}$ would have to be equal. In this case, we would expect the equation above to have no corrections. However, the \text{U}(2)-flavor symmetry between charm and bottom is broken by the quark mass difference $m_\text{b} - m_\text{c}$, so corrections to the equation above have to scale with $m_\text{b} - m_\text{c}$. This is reflected in $\Lambda_\text{QCD}(1/m_\text{c} - 1/m_\text{b})$, as $\Lambda_\text{QCD}(1/m_\text{c} - 1/m_\text{b}) = \Lambda_\text{QCD}(m_\text{b}-m_\text{c})/m_\text{c}m_\text{b}$. Note that we only would expect the $\text{U}(2)$-flavor symmetry breaking terms to be proportional to $m_\text{b}-m_\text{c}$, if the $\text{U}(2)$-flavor symmetry was still an approximate symmetry, i.e., if the $\text{U}(2)$-flavor symmetry breaking was small enough such that we can treat the symmetry breaking term as a perturbation. In Nature, however, we observe that the symmetry breaking is quite large, as $m_\text{b}\gg m_\text{c}$. To this end, we have to understand the last arguments with a grain of salt. Further notice that we have completely neglected any logarithmic correction or any logarithmic scale dependence originating from quantum loops in the discussion of the corrections, even though they might be quite large.\par
Let us now decompose the Hamilton operator $H^5_\text{QCD}$ into a \text{SU}(3)-singlet plus the 3rd and 8th component of a \text{SU}(3)-octet:
\begin{gather*}
H^5_\text{QCD} = H^{5;\,0}_\text{QCD} + \varepsilon_3\cdot H^{8}_{\text{QCD};\, 3} + \varepsilon_8\cdot H^{8}_{\text{QCD};\, 8},
\end{gather*}
where $\varepsilon_3\cdot H^{8}_{\text{QCD};\, 3} + \varepsilon_8\cdot H^{8}_{\text{QCD};\, 8}$ coincides with the identically named term from \autoref{sec:EFT+H_Pert} and \autoref{sec:GMO_formula} and $H^{5;\,0}_\text{QCD}$ is the collection of the remaining terms. $H^{5;\,0}_\text{QCD}$ is a \text{SU}(3)-singlet under uds-flavor transformations, while $H^8_{\text{QCD};\, 3}$ and $H^8_{\text{QCD};\, 8}$ are the 3rd and 8th component of a \text{SU}(3)-octet under uds-flavor transformations, respectively. Let us denote the mass of the hadron corresponding to $\Ket{Y,I,I_3}_\text{c}$ by $m^\text{c}(Y,I,I_3)$ and the mass of the hadron corresponding to $\Ket{Y,I,I_3}_\text{b}$ by $m^\text{b}(Y,I,I_3)$. Following the state formalism, we find:
\begin{align*}
m^\text{b}(Y,I,I_3) &= \leftidx{_\text{b}}{\Braket{Y,I,I_3|H^{5;\,0}_\text{QCD}(q,\bar{q}) + \varepsilon_3\cdot H^{8}_{\text{QCD};\, 3}(q,\bar{q}) + \varepsilon_8\cdot H^{8}_{\text{QCD};\, 8}(q,\bar{q})|Y,I,I_3}}{_\text{b}}\\
&\ \ \ + \mathcal{O}(\varepsilon_i\varepsilon_j)\\
&= \leftidx{_\text{b}}{\Braket{Y,I,I_3|H^5_\text{QCD}(q,\bar{q})|Y,I,I_3}}{_\text{b}} + \mathcal{O}(\varepsilon_i\varepsilon_j)\\
&= \leftidx{_\text{c}}{\Braket{Y,I,I_3|\hat{S}^\dagger H^5_\text{QCD}(q,\bar{q})\hat S|Y,I,I_3}}{_\text{c}} + \mathcal{O}\left(\Lambda_\text{QCD}\left(\frac{1}{m_\text{c}} - \frac{1}{m_\text{b}}\right)\right) + \mathcal{O}(\varepsilon_i\varepsilon_j)\\
&= \leftidx{_\text{c}}{\Braket{Y,I,I_3|H^5_\text{QCD}(q^\prime,\bar{q}^\prime)|Y,I,I_3}}{_\text{c}} + \mathcal{O}\left(\Lambda_\text{QCD}\left(\frac{1}{m_\text{c}} - \frac{1}{m_\text{b}}\right)\right) + \mathcal{O}(\varepsilon_i\varepsilon_j)\\
&= \leftidx{_\text{c}}{\Braket{Y,I,I_3|H^{5;\,0}_\text{QCD}(q^\prime,\bar{q}^\prime) + \varepsilon_3\cdot H^{8}_{\text{QCD};\, 3}(q,\bar{q}) + \varepsilon_8\cdot H^{8}_{\text{QCD};\, 8}(q,\bar{q})|Y,I,I_3}}{_\text{c}}\\
&\ \ \ + \mathcal{O}\left(\Lambda_\text{QCD}\left(\frac{1}{m_\text{c}} - \frac{1}{m_\text{b}}\right)\right) + \mathcal{O}(\varepsilon_i\varepsilon_j),
\end{align*}
where $\mathcal{O}(\varepsilon_i\varepsilon_j)$ combines all higher order corrections from \autoref{eq:iso_tot_sym} and \autoref{eq:iso_octet}. The last line follows, as the flavor transformation given by $S$ only affects the charm and bottom fields, but only the \text{SU}(3)-singlet $H^{5;\, 0}_\text{QCD}$ contains charm and bottom fields. The mass $m^\text{c}(Y,I,I_3)$ of the charmed hadrons is similarly given by:
\begin{align*}
m^\text{c}(Y,I,I_3) &= \leftidx{_\text{c}}{\Braket{Y,I,I_3|H^{5;\,0}_\text{QCD}(q,\bar{q}) + \varepsilon_3\cdot H^{8}_{\text{QCD};\, 3}(q,\bar{q}) + \varepsilon_8\cdot H^{8}_{\text{QCD};\, 8}(q,\bar{q})|Y,I,I_3}}{_\text{c}}\\
&\ \ \ + \mathcal{O}(\varepsilon_i\varepsilon_j).
\end{align*}
We can see that $m^\text{c}(Y,I,I_3)$ and $m^\text{b}(Y,I,I_3)$ only differ in the \text{SU}(3)-singlet except for higher order corrections. All other terms that appear in the mass formulae are equal. This allows us to parametrize $m^\text{c}(Y,I,I_3)$ and $m^\text{b}(Y,I,I_3)$ for totally symmetric \text{SU}(3)-multiplets in the following way (cf. \autoref{eq:iso_tot_sym}):
\begin{align}
m^\text{c}(Y,I,I_3) &= m^\text{c}_0 + m^F_3\cdot I_3 + \tilde{m}^F_8\cdot Y + \mathcal{O}\left(\varepsilon_i\varepsilon_j\right) + \mathcal{O}\left(\Lambda_\text{QCD}\left(\frac{1}{m_\text{c}} - \frac{1}{m_\text{b}}\right)\right),\label{eq:mass_tot_sym_c}\\
m^\text{b}(Y,I,I_3) &= m^\text{b}_0 + m^F_3\cdot I_3 + \tilde{m}^F_8\cdot Y + \mathcal{O}\left(\varepsilon_i\varepsilon_j\right) + \mathcal{O}\left(\Lambda_\text{QCD}\left(\frac{1}{m_\text{c}} - \frac{1}{m_\text{b}}\right)\right),\label{eq:mass_tot_sym_b}
\end{align}
where $m^\text{c}_0$ and $m^\text{b}_0$ are unrelated parameters. One can write down similar parametrizations for the other \text{SU}(3)-multiplets.\par
If we want to include electromagnetic corrections in the mass parametrization, we have to consider two important points: Firstly, the charm and bottom quark carry different electric charge. Secondly, we were able to parametrize electromagnetic corrections within a hadronic \text{SU}(3)-multiplet in \autoref{sec:add_con}, as the flavors of all valence quarks within the multiplet were part of the approximate global \text{SU}(3)-flavor symmetry. This does not apply to hadronic \text{SU}(3)-multiplets containing heavy quarks. Therefore, we have to introduce two contributions to the mass formula of hadronic \text{SU}(3)-multiplets containing exactly one heavy quark: One term $\Delta^{LL}_\alpha$ that describes the electromagnetic interaction between two light quarks and another term $\Delta^{LH}_\alpha$ that describes the electromagnetic interaction between a light and a heavy quark. Note that $\Delta^{LL}_\alpha$ does not arise for mesons. With these remarks in mind, we can repeat the parametrization of electromagnetic contributions from \autoref{sec:add_con} to find for baryons:
\begin{align}
m^\text{c}_\alpha(Y,I,I_3) &= m^\text{c}(Y,I,I_3) + \Delta^{LH}_\alpha\sum_{j\in\{k,l\}}\frac{q_jq_\text{c}}{e^2} + \Delta^{LL}_\alpha\frac{q_kq_l}{e^2}\label{eq:mass_ele_c}\\
&\ \ \ + \mathcal{O}\left(\alpha\varepsilon_i\right) + \mathcal{O}\left(\Lambda_\text{QCD}\left(\frac{1}{m_\text{c}} - \frac{1}{m_\text{b}}\right)\right),\nonumber\\
m^\text{b}_\alpha(Y,I,I_3) &= m^\text{b}(Y,I,I_3) + \Delta^{LH}_\alpha\sum_{j\in\{k,l\}}\frac{q_jq_\text{b}}{e^2} + \Delta^{LL}_\alpha\frac{q_kq_l}{e^2}\label{eq:mass_ele_b}\\
&\ \ \ + \mathcal{O}\left(\alpha\varepsilon_i\right) + \mathcal{O}\left(\Lambda_\text{QCD}\left(\frac{1}{m_\text{c}} - \frac{1}{m_\text{b}}\right)\right),\nonumber
\end{align}
where we employ a notation similar to \autoref{eq:mass_ele}. $\mathcal{O}\left(\alpha\varepsilon_i\right)$ combines all higher order corrections involving $\alpha$ (cf. \autoref{eq:mass_ele}). $k$ and $l$ denote the light valence quarks of the hadron, $q_k$ and $q_l$ denote their charge. In the case of mesons, we have to slightly modify the formula: There is only one light valence quark for mesons and we have to set $\Delta^{LL}_\alpha = 0$.

\newpage
\chapter{Overview and Discussion of Mass Relations}\label{chap:mass_relations}

In \autoref{chap:hadron_masses} and \autoref{chap:GMO_formula}, we explored how $\text{SU}(3)\rightarrow\text{SU}(2)\times\text{U}(1)\rightarrow\text{U}(1)\times\text{U}(1)$ flavor symmetry breaking together with electromagnetic corrections and heavy quark symmetry enters the mass parametrization of hadrons in multiplets. We want to investigate the mass relations arising from the hadronic mass parametrizations in this chapter. In particular, we examine the mass relations of hadrons within sextets, octets, and decuplets in \autoref{sec:rel_within_multiplets} and the mass relations between charm and bottom antitriplets and sextets in \autoref{sec:rel_between_multiplets}. For each multiplet or pair of charm and bottom multiplet, we present the weight diagram(s) of the multiplet(s), the mass parametrization of the hadrons within the multiplet(s), and the mass relations following from this parametrization. To make the weight diagram(s) and mass relations more accessible, we label each weight in the weight diagram(s) with the corresponding hadron of a prominent example for that multiplet/pair of multiplets. After the presentation of the mass relations, we discuss the order of magnitude of the dominant correction for every mass relation within/between that multiplet/pair of multiplets.

\section{Mass Relations within Multiplets}\label{sec:rel_within_multiplets}

Before we dive into the discussion of the multiplets, let us first consider which multiplets we expect to be realized in Nature and make some remarks about weight diagrams and notations. As already explained, we think of hadrons as composite particles in the quark model. In this picture, we imagine the hadrons to be made out of valence and sea quarks. The valence quarks dictate the flavor structure of the hadron in this model, i.e., tell us how the hadron transforms under flavor transformations. In regard to $\text{SU}(3)$-flavor transformations, the light quarks and antiquarks up, down, and strange transform under the fundamental representations $3$ and $\bar{3}$, respectively. Baryons are considered to be made out of three valence quarks. If the baryon at hand is a light baryon, i.e., if the valence quark content of the baryon only consists of light quarks (u, d, or s), the baryon transforms under the tensor product representation of three fundamental representations $3$. Using Young tableaux, one then finds:
\begin{gather*}
3\otimes 3\otimes 3 = 1\oplus 8\oplus 8\oplus 10.
\end{gather*}
If the baryon consists of one charm or bottom and two light valence quarks, it transforms under:
\begin{gather*}
3\otimes 3 = \bar{3}\oplus 6.
\end{gather*}
From this consideration, we expect the light baryons to form singlets, octets, and decuplets and the baryons containing exactly one charm or bottom valence quark to form antitriplets and sextets. However, light baryon singlets are not yet discovered in Nature\footnote{And we do not expect to find one, as the symmetry and statistics of baryonic states forbid the existence of baryon singlets.}.\par
We can determine the multiplets formed by mesons in a similar way: The valence quark content of mesons consists of a quark and an antiquark. Therefore, light mesons, i.e., mesons only containing light valence quarks transform under:
\begin{gather*}
3\otimes \bar{3} = 1\oplus 8,
\end{gather*}
while mesons containing exactly one charm/bottom valence quark or antiquark simply transform under $\bar{3}$ or $3$, respectively. This means that light mesons form singlets and octets, while mesons containing exactly one charm/bottom valence quark or antiquark form antitriplets or triplets, respectively. Note that the hadronic (anti)triplets do not contain enough hadrons to form mass relations within the (anti)triplet.\par
When we introduced weights in \autoref{sec:GMO_formula}, we stated that the collection of all weights of one \text{SU}(3)-multiplet uniquely characterizes that multiplet. Hence, it is sufficient to consider just the weights if one wishes to describe the multiplets of $\text{SU}(3)$. Commonly, the weights of a \text{SU}(3)-multiplet are presented in graphical form, i.e., by a weight diagram. A weight diagram of a \text{SU}(3)-multiplet is a two-dimensional coordinate system where the axes correspond to the components of the weights and each weight is indicated by a dot at the appropriate position. In the following sections, we rescale the axes of the weight diagrams we display such that the axes of the weight diagram coincide with the hypercharge $Y$ and the third isospin component $I_3$. Note that the rescaling of the axes breaks the rotational symmetry of the weight diagrams (the weight diagrams of \text{SU}(3)-multiplets are invariant under rotations by $\frac{2\pi}{3}$; cf. \autoref{sec:GMO_formula} and \cite{Lichtenberg}).\par
Hadrons correspond to vectors in the multiplet, thus, each hadron in a multiplet has a weight. However, multiple hadrons in a multiplet might have the same weight. If the multiplicity of a weight in a multiplet is larger than one, i.e, if there are multiple linearly independent vectors and, thus, hadrons in the multiplet which correspond to that weight, we add a dot near that weight for every additional linearly independent vector/hadron to indicate the multiplicity in the weight diagram.\par
As {${\text{SU}(2)\times\text{U}(1)}$} is a Lie subgroup of $\text{SU}(3)$, we can group basis vectors of a $\text{SU}(3)$-multiplet and, thus, the corresponding weights in a weight diagram into {${\text{SU}(2)\times\text{U}(1)}$}-multiplets. All vectors with same hypercharge $Y$ and total isospin $I$ form a {${\text{SU}(2)\times\text{U}(1)}$-}multiplet. The trivial representation of {${\text{SU}(2)\times\text{U}(1)}$} has a hypercharge and total isospin of $0$. In the following weight diagrams, we connect weights corresponding to the same {${\text{SU}(2)\times\text{U}(1)}$-}multiplet with red lines, if the {${\text{SU}(2)\times\text{U}(1)}$-}multiplet consists of more than one weight.\par
Now consider the group of flavor transformations between the up and strange quark, denoted by {${\text{SU}(2)_\text{us}\times\text{U}(1)}$}, and the group of flavor transformations between the down and strange quark, denoted by {${\text{SU}(2)_\text{ds}\times\text{U}(1)}$}:
\begin{align*}
\text{SU}(2)_\text{us}\times\text{U}(1) &\coloneqq \left\{\begin{pmatrix}e^{i\alpha}A_{11} & 0 & e^{i\alpha}A_{12}\\ 0 & e^{-2i\alpha} & 0\\ e^{i\alpha}A_{21} & 0 & e^{i\alpha}A_{22}\end{pmatrix}\middle| \alpha\in\mathbb{R};\, \begin{pmatrix}A_{11} & A_{12}\\ A_{21} & A_{22}\end{pmatrix}\in\text{SU}(2)\right\},\\
\text{SU}(2)_\text{ds}\times\text{U}(1) &\coloneqq \left\{\begin{pmatrix}e^{-2i\alpha} & 0\\ 0 & e^{i\alpha}A\end{pmatrix}\middle| \alpha\in\mathbb{R};\, A\in\text{SU}(2)\right\}.
\end{align*}
As the groups {${\text{SU}(2)_\text{us}\times\text{U}(1)}$} and {${\text{SU}(2)_\text{ds}\times\text{U}(1)}$} are also Lie subgroups of $\text{SU}(3)$, we can likewise group the weights into {${\text{SU}(2)_\text{us}\times\text{U}(1)}$}- and {${\text{SU}(2)_\text{ds}\times\text{U}(1)}$}-multiplets, respectively. If needed, we connect weights corresponding to the same {${\text{SU}(2)_\text{us}\times\text{U}(1)}$}-multiplet with green lines and weights corresponding to the same {${\text{SU}(2)_\text{ds}\times\text{U}(1)}$}-multiplet with blue lines. A recipe for the construction of weight diagrams of $\text{SU}(3)$ and for grouping the weights of such a weight diagram into {${\text{SU}(2)\times\text{U}(1)}$-},\linebreak {${\text{SU}(2)_\text{us}\times\text{U}(1)}$-}, or {${\text{SU}(2)_\text{ds}\times\text{U}(1)}$-}multiplets can be found in \cite{Lichtenberg}.\par
$\text{U}(1)\times\text{U}(1)$ is also a Lie subgroup of $\text{SU}(3)$, so $\text{SU}(3)$-multiplets decompose into $\text{U}(1)\times\text{U}(1)$-multiplets. However, $\text{U}(1)\times\text{U}(1)$ is an Abelian Lie group, hence, all multiplets are one-dimensional. All vectors in a $\text{U}(1)\times\text{U}(1)$-multiplet are eigenvectors of $I^{(\sigma\otimes\bar{\sigma})}_3$ and $Y^{(\sigma\otimes\bar{\sigma})}$, therefore, every weight in a weight diagram corresponds to a $\text{U}(1)\times\text{U}(1)$-multiplet. The trivial representation of $\text{U}(1)\times\text{U}(1)$ has a hypercharge and third isospin component of $0$.\par
Lastly, we need to clarify some details regarding the (expansion) parameters $\varepsilon_3$, $\varepsilon_8$, $\alpha$, and $\varepsilon_\text{cb}\coloneqq\Lambda_\text{QCD}\left(1/m_\text{c} - 1/m_\text{b}\right)$ we used throughout this work before we can begin the discussion of the mass relations. $\varepsilon_3$ and $\varepsilon_8$ have mass dimension 1, while $\alpha$ and $\varepsilon_\text{cb}$ are dimensionless. A dimensionful quantity as an expansion parameter has the advantage that it allows us to at least estimate how large the first order contribution is, but it does not immediately tell us by which factor the higher order contributions are suppressed. Conversely, a dimensionless quantity as an expansion parameter does not allow us to estimate the first order contribution, but can be used as a suppression factor for the higher order contributions. For instance, we expect the first order contribution of $\varepsilon_8$ to a hadron mass to be of the order of \SI{100}{MeV}, as the current quark mass difference between strange and up/down quark is of that order (cf. comments on quark masses in \autoref{sec:Trafo_QCD} and review \textit{66. Quark Masses} in \cite{PDG}). Likewise, the first order contribution of $\varepsilon_3$ to a hadron mass should be in the order of up-down current quark mass difference, i.e., in the order of \SI{1}{MeV} (cf. comments on quark masses in \autoref{sec:Trafo_QCD} and review \textit{66. Quark Masses} in \cite{PDG}). However, $\varepsilon_3$ and $\varepsilon_8$ per se do not indicate by which factor higher order contributions are suppressed. To obtain the suppression factors of higher order terms of $\varepsilon_3$ and $\varepsilon_8$, we need to know the (unperturbed) singlet contribution $m_0$ to the mass in a hadron multiplet. As we expanded the hadron masses at $m_0$, higher order contributions are suppressed by powers of $\varepsilon_{3/8}/m_0$. In the following sections, we will use $\varepsilon_{3/8}$ both in the sense of a dimensionful and dimensionless parameter, i.e., we will use $\varepsilon_{3/8}$ to denote both $\varepsilon_{3/8}$ and $\varepsilon_{3/8}/m_0$. For $m_0\sim\SI{1}{GeV}$, the order of magnitude of the dimensionless parameters $\varepsilon_3$, $\varepsilon_8$, $\alpha$, and $\varepsilon_\text{cb}$ is given in \autoref{tab:exp_param}. The parameters are ordered from highest to lowest. The values used for this estimate can be found in \cite{PDG}.

\begin{table}[t!]
\centering
\caption{Order of magnitude of several expansion parameters for hadrons with a mass of roughly \SI{1}{GeV}.}
\begin{tabular}{|c|c|}
\hline
Expansion or suppression parameter & Order of magnitude\\\hline
$\varepsilon_\text{cb}\coloneqq\Lambda_\text{QCD}\left(\frac{1}{m_\text{c}} - \frac{1}{m_\text{b}}\right)$ & $\mathcal{O}\left(10\%\right)$\\
$\varepsilon_8$ & $\mathcal{O}\left(10\%\right)$\\
$\alpha$ & $\mathcal{O}\left(1\%\right)$\\
$\varepsilon_3$ & $\mathcal{O}\left(1\%\right)$ to $\mathcal{O}\left(0.1\%\right)$\\\hline
\end{tabular}
\label{tab:exp_param}
\end{table}
Note that the hadron masses only enter linearly in all following mass formulae and relations. Aside from a few exceptions, however, the following formulae and relations are also valid, if one replaces the hadron masses with their squares. For a deeper discussion, confer \autoref{chap:hadron_masses} and \autoref{chap:data}.

\subsection*{Sextet}

An example for a baryonic sextet is the lowest-energy charm sextet with\footnote{Note that the quantum numbers $J^P$ we present in this work are not measured directly for most hadrons, but assumed based on quark model predictions (cf. \cite{PDG}).} $J^P = 1/2^+$. Its weight diagram is shown in \autoref{fig:c_sextet}.
\begin{figure}[htpb]
\centering
\includegraphics[width=\textwidth]{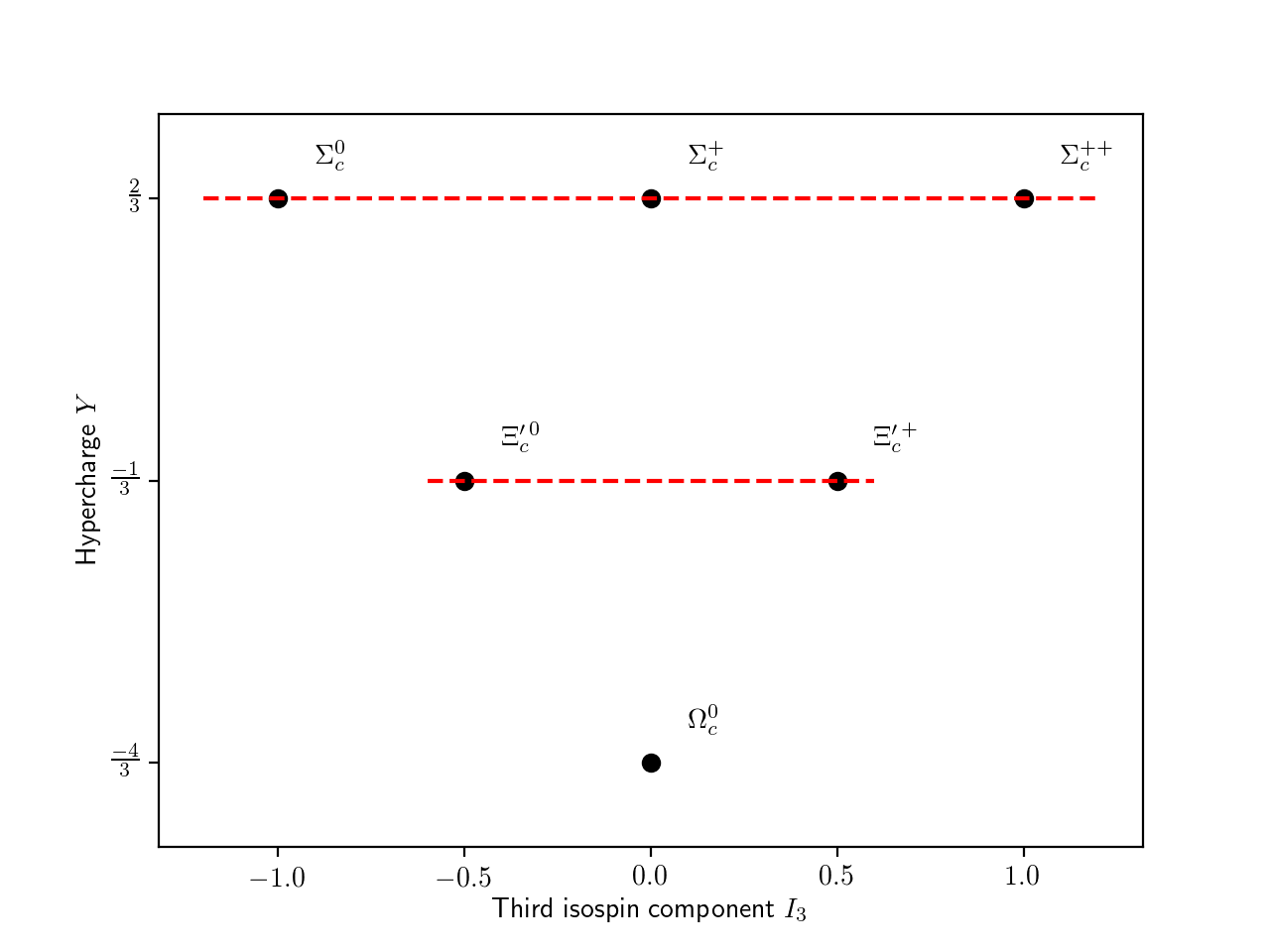}
\caption{The weight diagram of a sextet. For every weight, the name of the corresponding baryon from the charm sextet with $J^P = 1/2^+$ is included. The red lines visualize isospin multiplets of $\text{SU}(2)\times\text{U}(1)$. All weights not connected by red lines are one-dimensional isospin multiplets.}
\label{fig:c_sextet}
\end{figure}
Using \autoref{eq:mass_tot_sym_c} and \autoref{eq:mass_ele_c}, we find for the masses in the sextet:
\begin{alignat*}{4}
&m^\text{c}_\alpha\left(\frac{2}{3},1,1\right)&&\equiv m_{\Sigma^{++}_\text{c}} &&= m^\text{c}_0 + m^F_3 &&+ \frac{2}{3}\tilde{m}^F_8 + \frac{8}{9}\Delta^{LH}_\alpha + \frac{4}{9}\Delta^{LL}_\alpha,\\
&m^\text{c}_\alpha\left(\frac{2}{3},1,0\right)&&\equiv m_{\Sigma^{+}_\text{c}} &&= m^\text{c}_0 &&+ \frac{2}{3}\tilde{m}^F_8 + \frac{2}{9}\Delta^{LH}_\alpha - \frac{2}{9}\Delta^{LL}_\alpha,\\
&m^\text{c}_\alpha\left(\frac{2}{3},1,-1\right)&&\equiv m_{\Sigma^{0}_\text{c}} &&= m^\text{c}_0 - m^F_3 &&+ \frac{2}{3}\tilde{m}^F_8 - \frac{4}{9}\Delta^{LH}_\alpha + \frac{1}{9}\Delta^{LL}_\alpha,\\
&m^\text{c}_\alpha\left(-\frac{1}{3},\frac{1}{2},\frac{1}{2}\right)&&\equiv m_{\Xi^{\prime\, +}_\text{c}} &&= m^\text{c}_0 + \frac{1}{2}m^F_3 &&- \frac{1}{3}\tilde{m}^F_8 + \frac{2}{9}\Delta^{LH}_\alpha - \frac{2}{9}\Delta^{LL}_\alpha,\\
&m^\text{c}_\alpha\left(-\frac{1}{3},\frac{1}{2},-\frac{1}{2}\right)&&\equiv m_{\Xi^{\prime\, 0}_\text{c}} &&= m^\text{c}_0 - \frac{1}{2}m^F_3 &&- \frac{1}{3}\tilde{m}^F_8 - \frac{4}{9}\Delta^{LH}_\alpha + \frac{1}{9}\Delta^{LL}_\alpha,\\
&m^\text{c}_\alpha\left(-\frac{4}{3},0,0\right)&&\equiv m_{\Omega^{0}_\text{c}} &&= m^\text{c}_0 &&- \frac{4}{3}\tilde{m}^F_8 - \frac{4}{9}\Delta^{LH}_\alpha + \frac{1}{9}\Delta^{LL}_\alpha,
\end{alignat*}
where we omitted all ``$\mathcal{O}$'' for the sake of clarity. This mass parametrization allows us to find two inequivalent mass relations:
\begin{align}
m_{\Sigma^+_\text{c}} - m_{\Sigma^0_\text{c}} &= m_{\Xi^{\prime\, +}_\text{c}} - m_{\Xi^{\prime\, 0}_\text{c}} + \mathcal{O}\left(\alpha\varepsilon_8\right) + \mathcal{O}\left(\varepsilon_3\varepsilon_8\right),\label{eq:sextet_c_iso_bre}\\
m_{\Sigma^0_\text{c}} - m_{\Xi^{\prime\, 0}_\text{c}} &= m_{\Xi^{\prime\, 0}_\text{c}} - m_{\Omega^0_\text{c}} + \mathcal{O}\left(\varepsilon^2_8\right).\label{eq:sextet_c_GMO_equal_spacing}
\end{align}
This time, we have included the order of magnitude of the dominant correction(s) in the mass relations. Of course, analogous mass relations apply to all baryonic charm and bottom sextets. \autoref{eq:sextet_c_GMO_equal_spacing} is just the equal spacing rule for sextets following from the GMO mass formula (cf. \autoref{eq:GMO_mass_formula}). \autoref{eq:sextet_c_iso_bre} relates the mass splittings of different isospin multiplets in the sextet. The particular form of this mass relation is actually rather interesting. It allows us to derive \autoref{eq:sextet_c_iso_bre} and the order of magnitude of its dominant corrections without the knowledge of a mass parametrization. For this, suppose that the $\text{SU}(2)\times\text{U}(1)$-isospin symmetry was exact. Then, \autoref{eq:sextet_c_iso_bre} would have to be exact, as all baryons in one isospin multiplet would have to have the same mass. One can show this analogously to the case of $\text{SU}(3)$, for which we showed that all hadrons in one $\text{SU}(3)$-multiplet would have to have the same mass, if $\text{SU}(3)$ was an exact symmetry (cf. \autoref{sec:GMO_formula}). However, only the mass and charge difference between the up and down quark break the isospin symmetry (neglecting weak interaction), meaning that only $\varepsilon_3$ and $\alpha$ break the isospin symmetry. Hence, every correction to \autoref{eq:sextet_c_iso_bre} has to be proportional to $\varepsilon_3$ or $\alpha$. Furthermore, \autoref{eq:sextet_c_iso_bre} would also be exact, if the $\text{SU}(2)_\text{ds}\times\text{U}(1)$-flavor transformations of the down and strange quark were an exact symmetry. We can see this by rewriting \autoref{eq:sextet_c_iso_bre}:
\begin{gather*}
m_{\Sigma^+_\text{c}} - m_{\Xi^{\prime\, +}_\text{c}} = m_{\Sigma^0_\text{c}} - m_{\Xi^{\prime\, 0}_\text{c}}.
\end{gather*}
This equation is exactly satisfied for exact $\text{SU}(2)_\text{ds}\times\text{U}(1)$-symmetry, as $\Sigma^+_\text{c}$ and $\Xi^{\prime\, +}_\text{c}$ and likewise $\Sigma^0_\text{c}$ and $\Xi^{\prime\, 0}_\text{c}$ have the same mass for exact $\text{SU}(2)_\text{ds}\times\text{U}(1)$-symmetry, since they are contained in the same $\text{SU}(2)_\text{ds}\times\text{U}(1)$-multiplet (cf. \autoref{fig:sextet_ds}).
\begin{figure}[b!]
\centering
\includegraphics[width=0.7\textwidth]{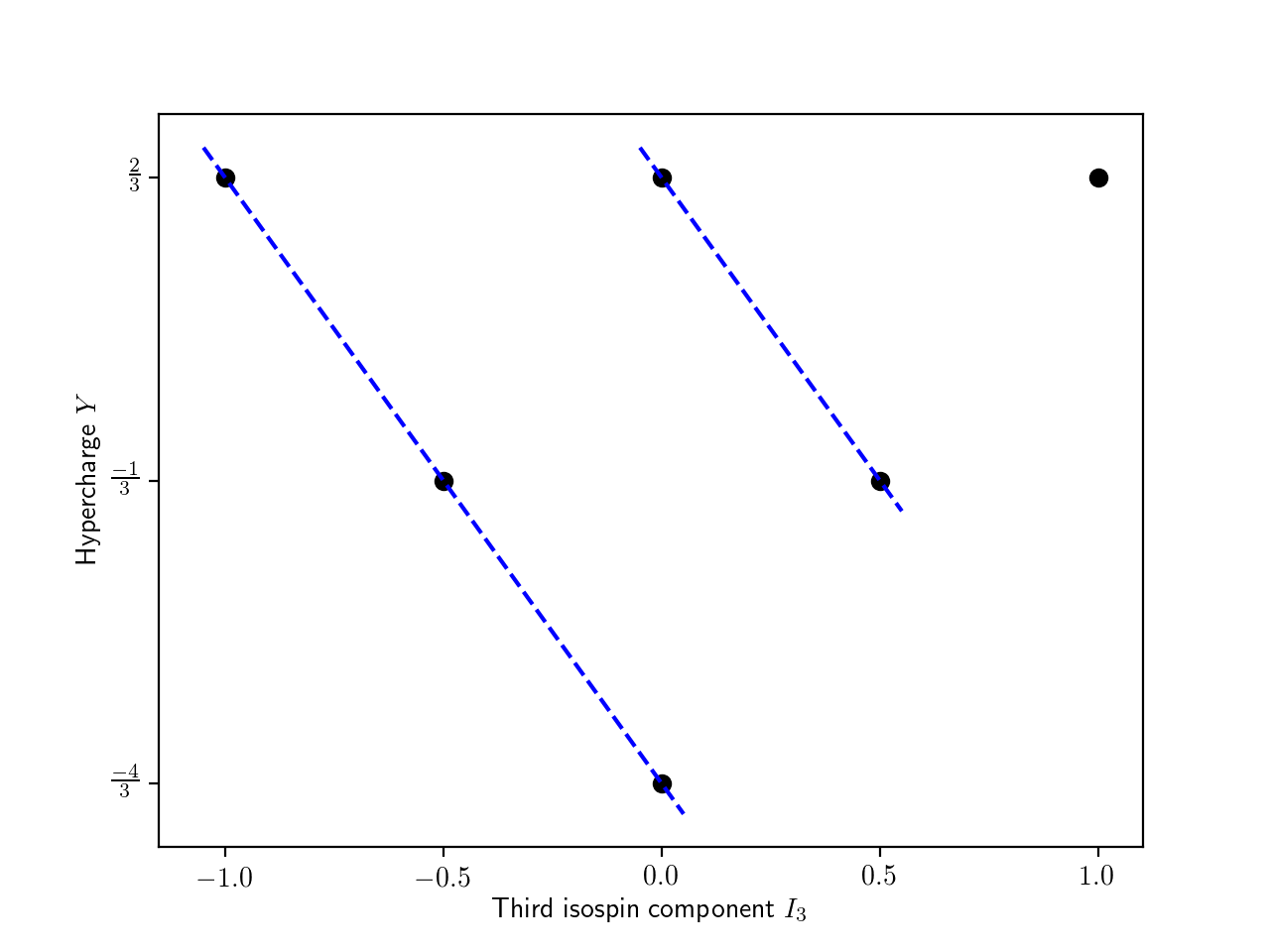}
\caption{The weight diagram of a sextet. The blue lines visualize the multiplets of $\text{SU}(2)_\text{ds}\times\text{U}(1)$. All weights not connected by blue lines are one-dimensional multiplets.}
\label{fig:sextet_ds}
\end{figure}
However, the $\text{SU}(2)_\text{ds}\times\text{U}(1)$-symmetry is only broken by the mass difference between down and strange quark, so roughly by $\varepsilon_8$. Hence, the corrections to \autoref{eq:sextet_c_iso_bre} have to be proportional to $\varepsilon_8$. In total, this means that every correction to \autoref{eq:sextet_c_iso_bre} has to be proportional to $\alpha\varepsilon_8$ or $\varepsilon_3\varepsilon_8$.\par
It is easy to see that the dominant correction to \autoref{eq:sextet_c_GMO_equal_spacing} is of order $\varepsilon^2_8$: Neglecting the weak interaction, the only parameters that give rise to corrections to \autoref{eq:mass_tot_sym_c} and \autoref{eq:mass_ele_c} are $\varepsilon_3$, $\varepsilon_8$, $\alpha$, and $\varepsilon_\text{cb}$. $\varepsilon_\text{cb}$ and $\varepsilon_8$ induce the largest corrections (cf. \autoref{tab:exp_param}). However, \autoref{eq:sextet_c_GMO_equal_spacing} does not contain any bottom hadrons, so corrections proportional to $\varepsilon_\text{cb}$ do not occur. As terms in the order of $\varepsilon_8$, $\alpha$, and $\varepsilon_3$ are respected by \autoref{eq:mass_ele_c} and, thus, by \autoref{eq:sextet_c_GMO_equal_spacing}, the dominant correction is of order $\varepsilon^2_8$.

\subsection*{Octet}

The $J^P = 1/2^+$ baryon octet is the lightest and probably one of the most known baryon multiplets. Its weight diagram is displayed in \autoref{fig:baryon_octet}.
\begin{figure}[htpb]
\centering
\includegraphics[width=\textwidth]{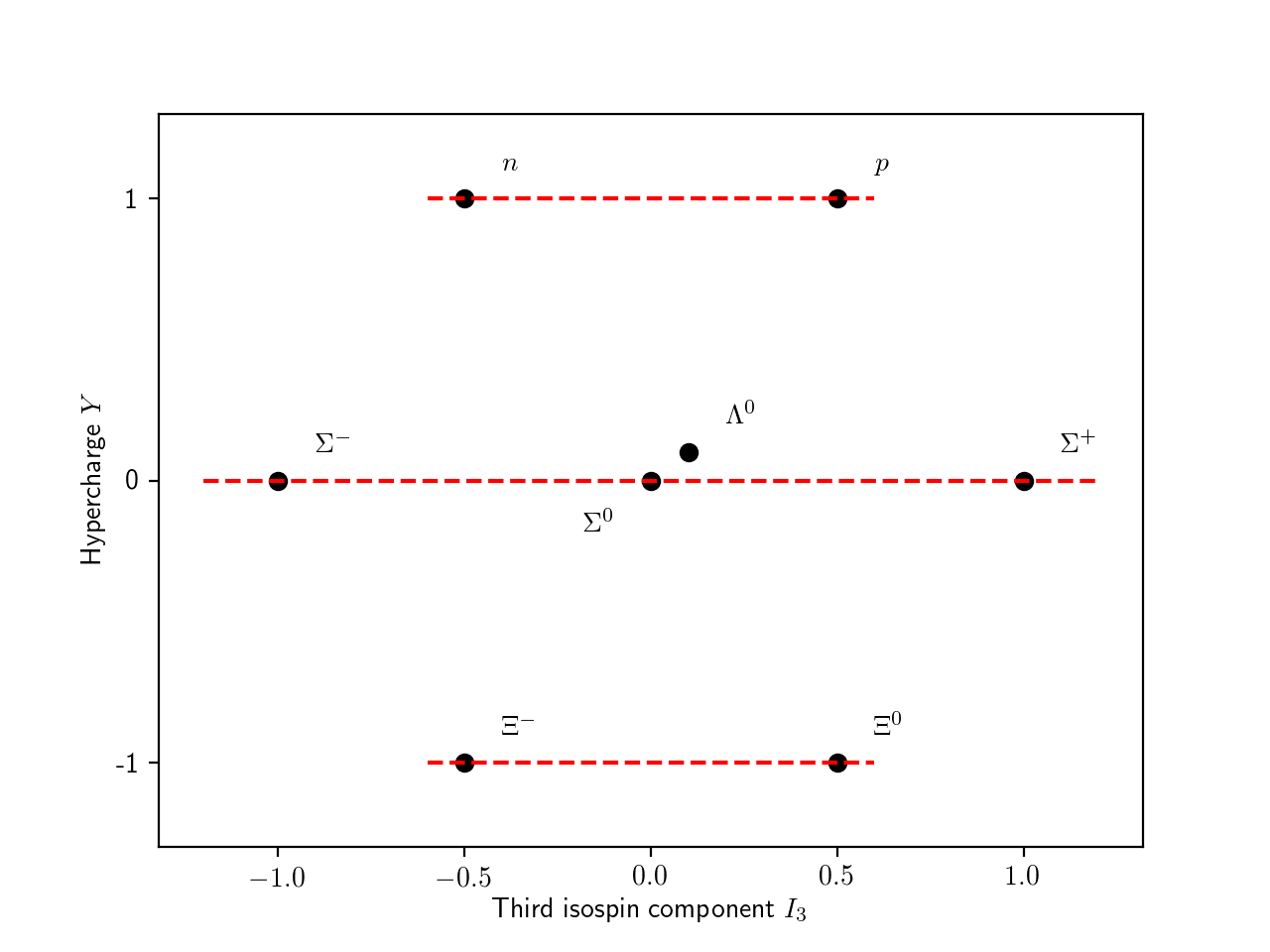}
\caption{The weight diagram of an octet. For every weight, the name of the corresponding baryon from the $J^P = 1/2^+$ baryon octet is included. The red lines visualize isospin multiplets of $\text{SU}(2)\times\text{U}(1)$. All weights not connected by red lines are isospin singlets.}
\label{fig:baryon_octet}
\end{figure}
Note that both $\Sigma^0$ as well as $\Lambda^0$ have a hypercharge and third isospin component of $0$. The difference between the two particles is that $\Lambda^0$ also has a total isospin of $0$ making it a singlet of $\text{SU}(2)\times\text{U}(1)$, while $\Sigma^0$ has a total isospin of $1$ forming together with $\Sigma^+$ and $\Sigma^-$ an isospin multiplet. We can sum this up in a very elegant way: The octet has two $\text{U}(1)\times\text{U}(1)$-singlets, $\Sigma^0$ and $\Lambda^0$, but only one $\text{SU}(2)\times\text{U}(1)$-isospin singlet, $\Lambda^0$. Using \autoref{eq:iso_octet} and \autoref{eq:mass_ele}, we find for the masses in the octet:
\begin{alignat*}{6}
&m_\alpha\left(1,\frac{1}{2},\frac{1}{2}\right)&&\equiv m_{p} &&= \tilde{m}_0 + \frac{1}{2}m^F_3 &&+ \tilde{m}^F_8 + \frac{1}{2}m^D_3 &&+ \frac{1}{2}\tilde{m}^D_8,&&\\
&m_\alpha\left(1,\frac{1}{2},-\frac{1}{2}\right)&&\equiv m_{n} &&= \tilde{m}_0 - \frac{1}{2}m^F_3 &&+ \tilde{m}^F_8 - \frac{1}{2}m^D_3 &&+ \frac{1}{2}\tilde{m}^D_8 &&- \frac{1}{3}\Delta_\alpha,\\
&m_\alpha\left(0,1,1\right)&&\equiv m_{\Sigma^{+}} &&= \tilde{m}_0 + m^F_3 && &&+ 2\tilde{m}^D_8,&&\\
&m_\alpha\left(0,1,0\right)&&\equiv m_{\Sigma^{0}} &&= \tilde{m}_0 && &&+ 2\tilde{m}^D_8 &&- \frac{1}{3}\Delta_\alpha,\\
&m_\alpha\left(0,1,-1\right)&&\equiv m_{\Sigma^{-}} &&= \tilde{m}_0 - m^F_3 && &&+ 2\tilde{m}^D_8 &&+ \frac{1}{3}\Delta_\alpha,\\
&m_\alpha\left(0,0,0\right)&&\equiv m_{\Lambda^{0}} &&= \tilde{m}_0 && && &&- \frac{1}{3}\Delta_\alpha,\\
&m_\alpha\left(-1,\frac{1}{2},\frac{1}{2}\right)&&\equiv m_{\Xi^{0}} &&= \tilde{m}_0 + \frac{1}{2}m^F_3 &&- \tilde{m}^F_8 - \frac{1}{2}m^D_3 &&+ \frac{1}{2}\tilde{m}^D_8 &&- \frac{1}{3}\Delta_\alpha,\\
&m_\alpha\left(-1,\frac{1}{2},-\frac{1}{2}\right)&&\equiv m_{\Xi^{-}} &&= \tilde{m}_0 - \frac{1}{2}m^F_3 &&- \tilde{m}^F_8 + \frac{1}{2}m^D_3 &&+ \frac{1}{2}\tilde{m}^D_8 &&+ \frac{1}{3}\Delta_\alpha,
\end{alignat*}
where we omitted all ``$\mathcal{O}$'' for the sake of clarity. This mass parametrization allows us to find two inequivalent mass relations:
\begin{align}
m_{p} - m_{n} + m_{\Xi^0} - m_{\Xi^-} &= m_{\Sigma^+} - m_{\Sigma^-} + \mathcal{O}\left(\alpha\varepsilon_8\right) + \mathcal{O}\left(\varepsilon_3\varepsilon_8\right),\label{eq:Coleman-Glashow}\\
m_{p} + m_{n} + m_{\Xi^0} + m_{\Xi^-} &= 3m_{\Lambda^0} + m_{\Sigma^+} + m_{\Sigma^-} - m_{\Sigma^0} + \mathcal{O}\left(\varepsilon^2_8\right).\label{eq:Gell-Mann--Okubo}
\end{align}
This time, we have included the order of magnitude of the dominant correction(s) in the mass relations. Of course, analogous mass relations apply to all baryonic octets. For mesonic octets, especially for the pseudoscalar meson octet, one has to mind the octet-singlet-mixing ($\eta$-$\eta^\prime$-mixing for the pseudoscalar meson octet) and the power of the meson masses.
\autoref{eq:Gell-Mann--Okubo} corresponds to the original Gell-Mann--Okubo mass relation (cf. \cite{Gell-Mann1961}). Often, exact $\text{SU}(2)\times\text{U}(1)$-isospin symmetry is assumed such that the relation can be presented in the following form:
\begin{gather*}
2(m_N + m_\Xi) = 3m_{\Lambda^0} + m_{\Sigma},
\end{gather*}
where $m_N$ is the mass of the $p$-$n$-isospin multiplet, $m_\Xi$ is the mass of the $\Xi$-isospin multiplet, and $m_{\Sigma}$ is the mass of the $\Sigma$-multiplet. In the same way as for \autoref{eq:sextet_c_GMO_equal_spacing}, we find that the dominant correction to \autoref{eq:Gell-Mann--Okubo} is of the order $\varepsilon^2_8$.\par
\autoref{eq:Coleman-Glashow} is the famous and very precise Coleman-Glashow mass relation (cf. \cite{coleman-glashow}). Like for \autoref{eq:sextet_c_iso_bre}, we can derive the Coleman-Glashow mass relation and the order of magnitude of its dominant corrections from purely group theoretical considerations: If any of the flavor transformation groups {${\text{SU}(2)\times\text{U}(1)}$,} {${\text{SU}(2)_\text{us}\times\text{U}(1)}$,} or {${\text{SU}(2)_\text{ds}\times\text{U}(1)}$} was an exact symmetry, all baryons in one {${\text{SU}(2)\times\text{U}(1)}$\text{-,}}\linebreak {${\text{SU}(2)_\text{us}\times\text{U}(1)}$\text{-,}} or $\text{SU}(2)_\text{ds}\times\text{U}(1)$-multiplet would have to have the same mass, respectively. Thus, considering the combination of masses that occur in \autoref{eq:Coleman-Glashow} and the different submultiplets in the octet (cf. \autoref{fig:baryon_octet} and \autoref{fig:baryon_octet_us_ds}), we easily find that the Coleman-Glashow mass relation would be exact, if any of the flavor transformation groups $\text{SU}(2)\times\text{U}(1)$, $\text{SU}(2)_\text{us}\times\text{U}(1)$, or $\text{SU}(2)_\text{ds}\times\text{U}(1)$ was an exact symmetry.
\begin{figure}[t!]
\hfill
\subfigure{\includegraphics[width=0.45\textwidth]{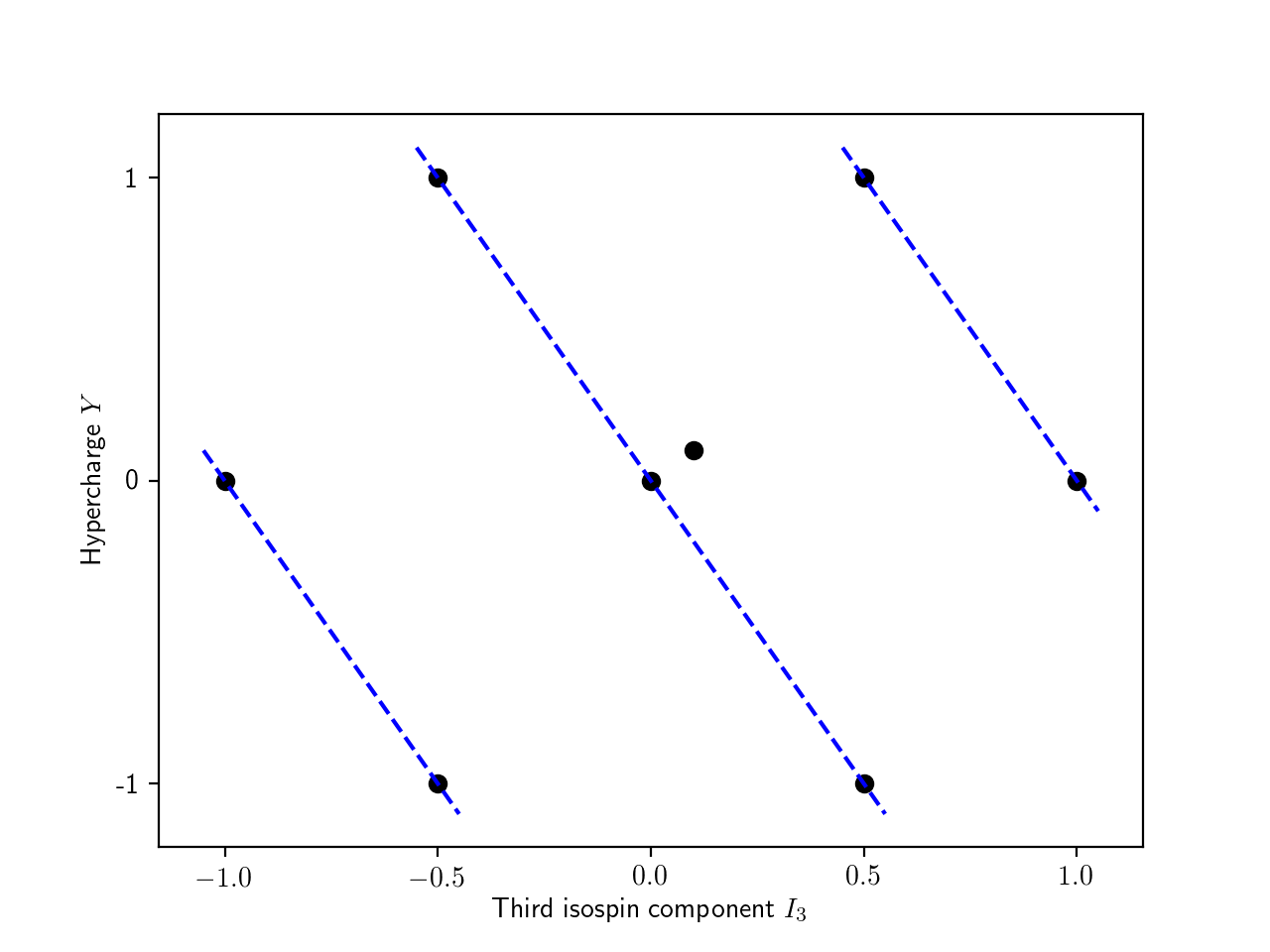}}
\hfill
\subfigure{\includegraphics[width=0.45\textwidth]{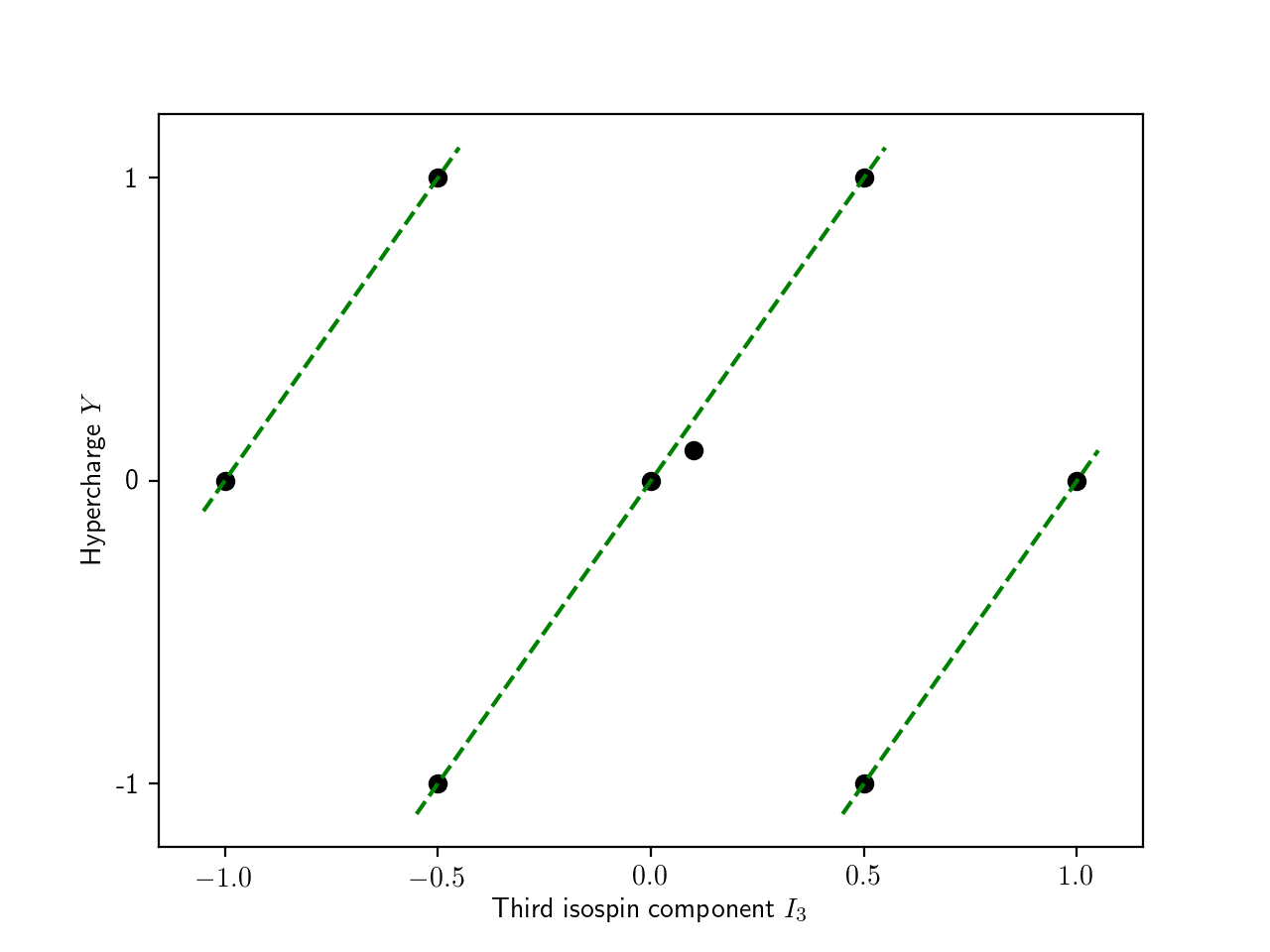}}
\hfill
\caption{Two weight diagrams of an octet. The blue lines visualize the multiplets of {${\text{SU}(2)_\text{ds}\times\text{U}(1)}$}, while the green lines visualize the multiplets of {${\text{SU}(2)_\text{us}\times\text{U}(1)}$}. All weights not connected by blue lines or green lines are singlets. Note that neither $\Sigma^0$ nor $\Lambda^0$ are part of a {${\text{SU}(2)_\text{us}\times\text{U}(1)}$-} or {${\text{SU}(2)_\text{ds}\times\text{U}(1)}$-}multiplet, but only a mixture of $\Sigma^0$ and $\Lambda^0$.}
\label{fig:baryon_octet_us_ds}
\end{figure}
However, all these symmetries are broken, but the $\text{SU}(2)\times\text{U}(1)$-symmetry is only broken by $\varepsilon_3$ and $\alpha$, while the $\text{SU}(2)_\text{ds}\times\text{U}(1)$-symmetry is only broken by $\varepsilon_8$. As for \autoref{eq:sextet_c_iso_bre}, this means that every correction to the Coleman-Glashow mass relation has to be proportional to $\alpha\varepsilon_8$ or $\varepsilon_3\varepsilon_8$.

\subsection*{Decuplet}

The $J^P = 3/2^+$ baryon decuplet is a textbook example for the application of the GMO mass formula, as Gell-Mann was able to predict the $\Omega^-$-particle contained in the $J^P = 3/2^+$ baryon decuplet and its mass using the equal spacing rules within the decuplet, before the $\Omega^-$-particle was discovered (cf. \cite{Zee2016} and \cite{Langacker2017}). Its weight diagram is displayed in \autoref{fig:baryon_decuplet}.
\begin{figure}[htpb]
\centering
\includegraphics[width=\textwidth]{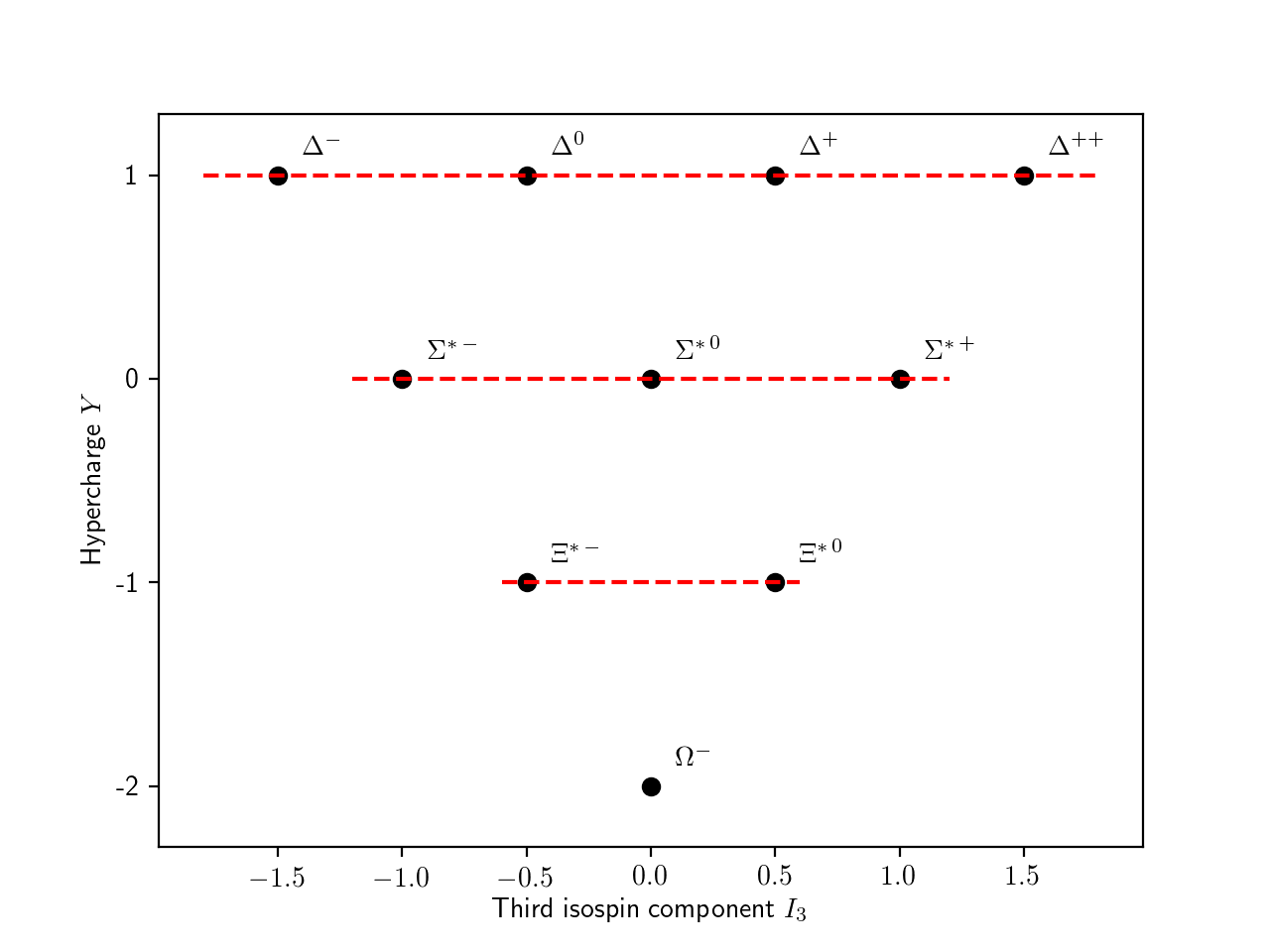}
\caption{The weight diagram of a decuplet. For every weight, the name of the corresponding baryon from the $J^P = 3/2^+$ baryon decuplet is included. The red lines visualize isospin multiplets of $\text{SU}(2)\times\text{U}(1)$. All weights not connected by red lines are one-dimensional isospin multiplets.}
\label{fig:baryon_decuplet}
\end{figure}
Using \autoref{eq:iso_tot_sym} and \autoref{eq:mass_ele}, we find for the masses in the decuplet:
\begin{alignat*}{6}
&m_\alpha\left(1,\frac{3}{2},\frac{3}{2}\right)&&\equiv m_{\Delta^{++}} &&= m_0 &&+ \frac{3}{2}m^F_3 &&+ \tilde{m}^F_8 &&+ \frac{4}{3}\Delta_\alpha,\\
&m_\alpha\left(1,\frac{3}{2},\frac{1}{2}\right)&&\equiv m_{\Delta^{+}} &&= m_0 &&+ \frac{1}{2}m^F_3 &&+ \tilde{m}^F_8,&&\\
&m_\alpha\left(1,\frac{3}{2},-\frac{1}{2}\right)&&\equiv m_{\Delta^{0}} &&= m_0 &&- \frac{1}{2}m^F_3 &&+ \tilde{m}^F_8 &&- \frac{1}{3}\Delta_\alpha,\\
&m_\alpha\left(1,\frac{3}{2},-\frac{3}{2}\right)&&\equiv m_{\Delta^{-}} &&= m_0 &&- \frac{3}{2}m^F_3 &&+ \tilde{m}^F_8 &&+ \frac{1}{3}\Delta_\alpha,\\
&m_\alpha\left(0,1,1\right)&&\equiv m_{\Sigma^{\ast\, +}} &&= m_0 &&+ m^F_3,&& &&\\
&m_\alpha\left(0,1,0\right)&&\equiv m_{\Sigma^{\ast\, 0}} &&= m_0 && && &&- \frac{1}{3}\Delta_\alpha,\\
&m_\alpha\left(0,1,-1\right)&&\equiv m_{\Sigma^{\ast\, -}} &&= m_0 &&- m^F_3 && &&+ \frac{1}{3}\Delta_\alpha,\\
&m_\alpha\left(-1,\frac{1}{2},\frac{1}{2}\right)&&\equiv m_{\Xi^{\ast\, 0}} &&= m_0 &&+ \frac{1}{2}m^F_3 &&- \tilde{m}^F_8 &&- \frac{1}{3}\Delta_\alpha,\\
&m_\alpha\left(-1,\frac{1}{2},-\frac{1}{2}\right)&&\equiv m_{\Xi^{\ast\, -}} &&= m_0 &&- \frac{1}{2}m^F_3 &&- \tilde{m}^F_8 &&+ \frac{1}{3}\Delta_\alpha,\\
&m_\alpha\left(-2,0,0\right)&&\equiv m_{\Omega^{-}} &&= m_0 && &&- 2\tilde{m}^F_8 &&+ \frac{1}{3}\Delta_\alpha,
\end{alignat*}
where we omitted all ``$\mathcal{O}$'' for the sake of clarity. This mass parametrization allows us to find seven mass relations:
\begin{gather}\small
m_{\Sigma^{\ast\, +}} - m_{\Sigma^{\ast\, -}} = m_{\Delta^{+}} - m_{\Delta^{0}} + m_{\Xi^{\ast\, 0}} - m_{\Xi^{\ast\, -}} + \mathcal{O}\left(\alpha\varepsilon_8\right) + \mathcal{O}\left(\varepsilon_3\varepsilon_8\right),\label{eq:Coleman-Glashow_decuplet}\\
m_{\Delta^{-}} = m_{\Delta^{++}} + 3\left( m_{\Delta^{0}} - m_{\Delta^{+}}\right) + \mathcal{O}\left(\alpha\varepsilon_8\right) + \mathcal{O}\left(\varepsilon_3\varepsilon_8\right),\label{eq:Delta-}\\
m_{\Delta^{++}} + m_{\Delta^{0}} - 2m_{\Delta^{+}} = m_{\Sigma^{\ast\, +}} + m_{\Sigma^{\ast\, -}} - 2m_{\Sigma^{\ast\, 0}} + \mathcal{O}\left(\alpha\varepsilon_8\right) + \mathcal{O}\left(\varepsilon_3\varepsilon_8\right),\label{eq:iso1_decuplet}\\
m_{\Sigma^{\ast\, -}} - m_{\Sigma^{\ast\, +}} = m_{\Delta^{++}} + 3m_{\Delta^{0}} - 4m_{\Delta^{+}} + \mathcal{O}\left(\alpha\varepsilon_8\right) + \mathcal{O}\left(\varepsilon_3\varepsilon_8\right),\label{eq:iso2_decuplet}\\
m_{\Delta^{+}} - m_{\Sigma^{\ast\, +}} = m_{\Sigma^{\ast\, -}} - m_{\Xi^{\ast\, -}} + \mathcal{O}\left(\varepsilon^2_8\right),\label{eq:equal_spacing1_decuplet}\\
m_{\Sigma^{\ast\, -}} - m_{\Xi^{\ast\, -}} = m_{\Xi^{\ast\, -}} - m_{\Omega^{-}} + \mathcal{O}\left(\varepsilon^2_8\right),\label{eq:equal_spacing2_decuplet}\\
4m_{\Delta^{++}} - 6\left(m_{\Delta^{+}} - m_{\Delta^{0}}\right) - 4m_{\Omega^{-}} = 6\left(m_{\Sigma^{\ast\, +}} + m_{\Sigma^{\ast\, -}}\right) - 6\left(m_{\Xi^{\ast\, 0}} + m_{\Xi^{\ast\, -}}\right)\label{eq:better_GMO_decuplet}\\
+ \mathcal{O}\left(\varepsilon^3_8\right) + \mathcal{O}\left(\alpha\varepsilon_8\right) + \mathcal{O}\left(\varepsilon_3\varepsilon_8\right).\nonumber
\end{gather}
This time, we have included the order of magnitude of the dominant correction(s) in the mass relations. Of course, analogous mass relations apply to all baryonic decuplets. In contrast to the case of the sextet and octet, only the first six mass relations are inequivalent in the sense that none of the first six mass relations follows from each other. The seventh mass relation, however, can be formed out of the first six relations.\par
\autoref{eq:Coleman-Glashow_decuplet} is the companion piece to the Coleman-Glashow mass relation in the decuplet. It and its dominant corrections can be derived in the same way as for the Coleman-Glashow mass relation of the octet (cf. the previous section ``Octet'', \autoref{fig:baryon_decuplet}, and \autoref{fig:baryon_decuplet_us_ds}).\par
\begin{figure}[htbp]
\hfill
\subfigure{\includegraphics[width=0.45\textwidth]{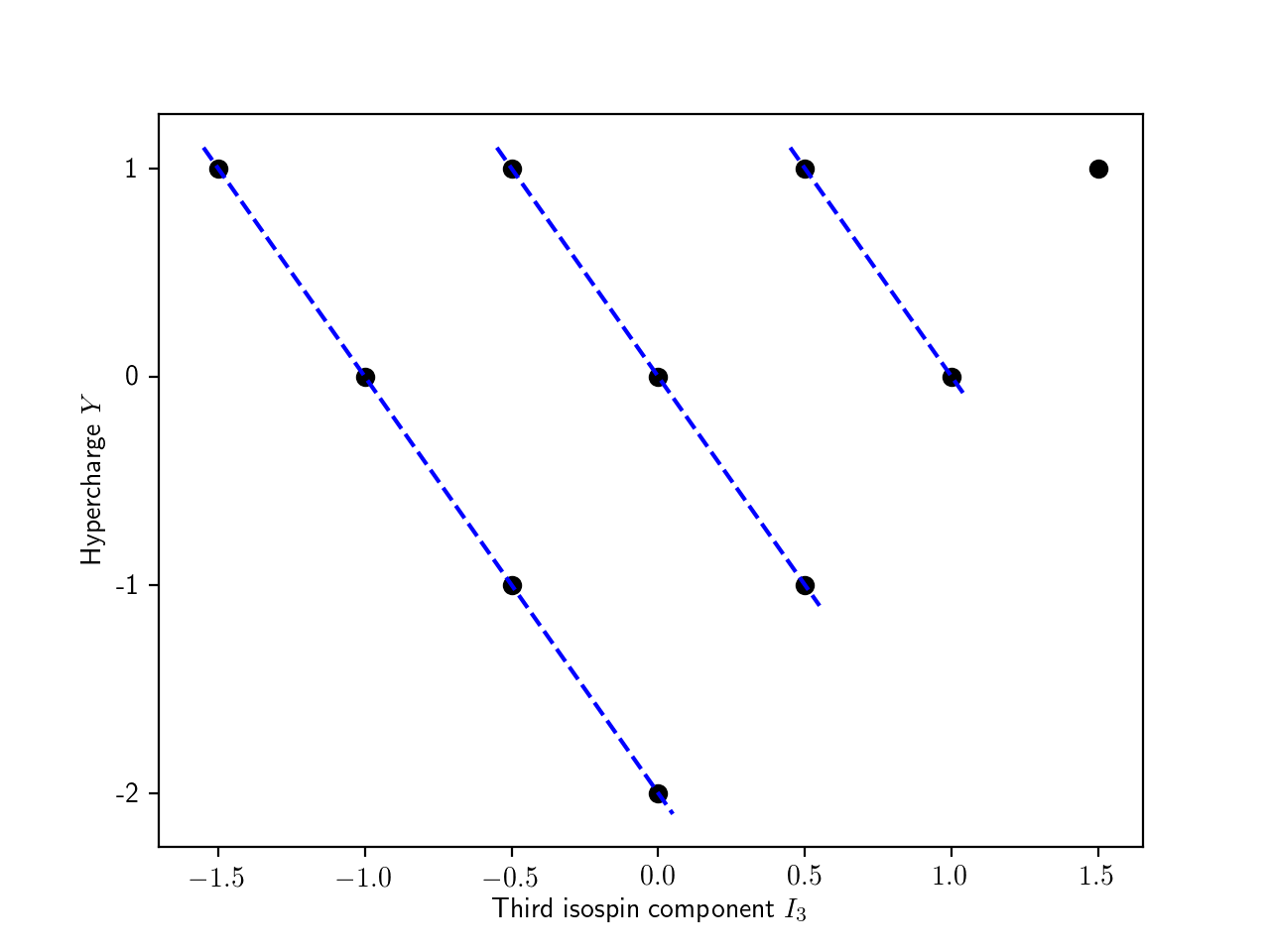}}
\hfill
\subfigure{\includegraphics[width=0.45\textwidth]{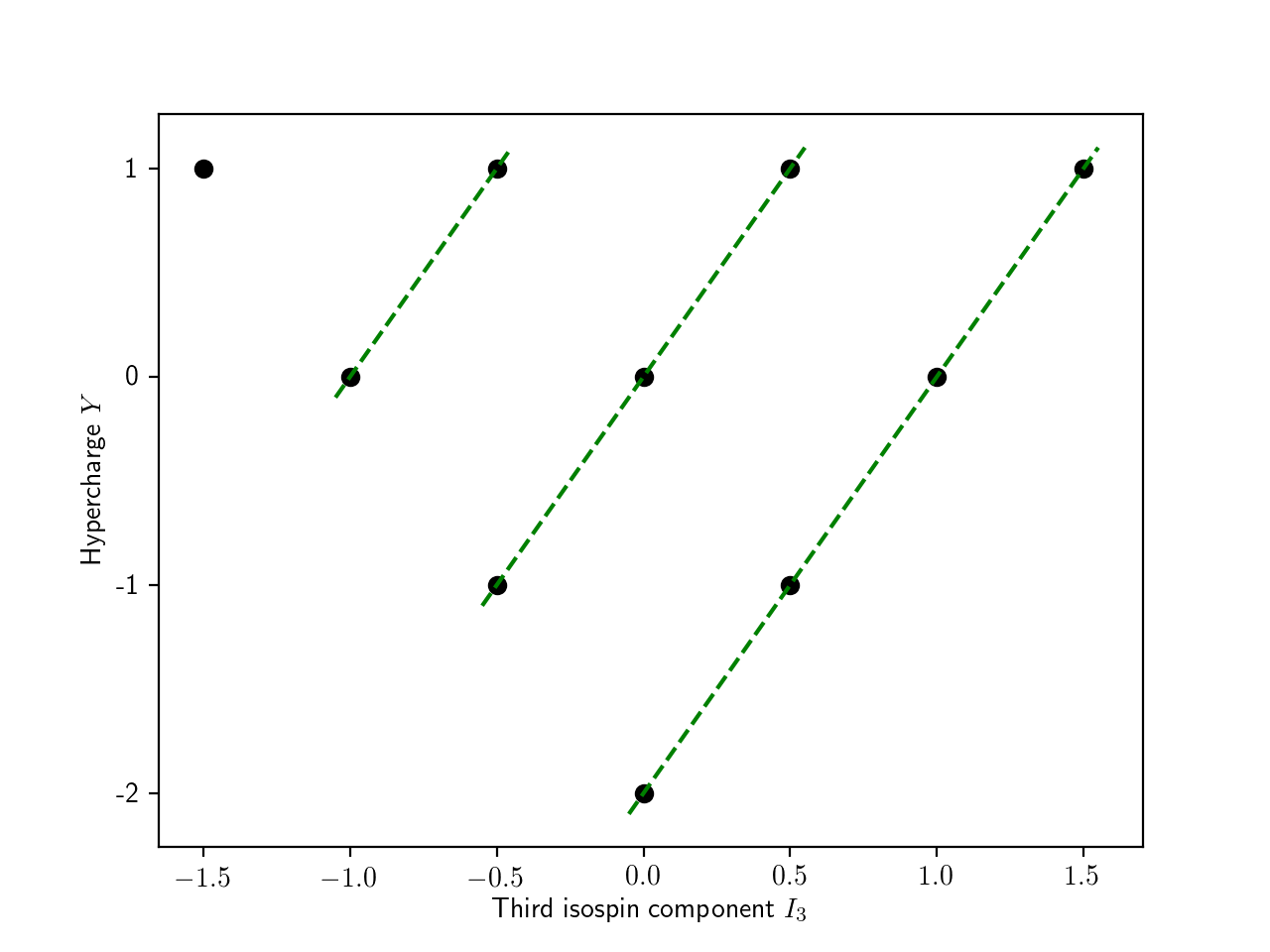}}
\hfill
\caption{Two weight diagrams of a decuplet. The blue lines visualize the multiplets of $\text{SU}(2)_\text{ds}\times\text{U}(1)$, while the green lines visualize the multiplets of $\text{SU}(2)_\text{us}\times\text{U}(1)$. All weights not connected by blue lines or green lines are one-dimensional multiplets.}
\label{fig:baryon_decuplet_us_ds}
\end{figure}
The dominant corrections to \autoref{eq:Delta-}, \autoref{eq:iso1_decuplet}, and \autoref{eq:iso2_decuplet} are in the order of $\alpha\varepsilon_8$ and $\varepsilon_3\varepsilon_8$. As all of these relations are exact in the limit of exact isospin symmetry, every correction to these relations has to be proportional to $\alpha$ or $\varepsilon_3$. However, the mass parametrization we used to derive these relations already includes the first order contributions of $\alpha$ and $\varepsilon_3$, hence, the dominant corrections are in the order of $\alpha\varepsilon_8$ and $\varepsilon_3\varepsilon_8$. At this point, it should be noted that it is also possible to derive \autoref{eq:Coleman-Glashow_decuplet}, \autoref{eq:Delta-}, and \autoref{eq:iso1_decuplet} from quark model calculations (cf. \cite{Ishida1966} and \cite{Gal1967}) and baryonic chiral perturbation theory (cf. \cite{Lebed1994}).\par
\autoref{eq:equal_spacing1_decuplet} and \autoref{eq:equal_spacing2_decuplet} are the equal spacing rules for the decuplet following from the GMO mass formula (cf. \autoref{eq:GMO_mass_formula}). Their dominant correction is in the order of $\varepsilon^2_8$. We can show this in the same way as for the equal spacing rule of the sextet.\par
\autoref{eq:better_GMO_decuplet} seems to be redundant, as this mass relation follows from the other six decuplet mass relations. However, this relation is more precise than the equal spacing rules, since it holds true to second order in $\text{SU}(3)\rightarrow\text{SU}(2)\times\text{U}(1)$-flavor symmetry breaking (for the derivation of \autoref{eq:better_GMO_decuplet} and its dominant corrections, cf. \cite{Okubo1963} and \cite{Lebed1994}), thus, the dominant corrections to \autoref{eq:better_GMO_decuplet} are in the order of $\varepsilon^3_8$, $\alpha\varepsilon_8$, and $\varepsilon_3\varepsilon_8$.

\section{Mass Relations between Multiplets}\label{sec:rel_between_multiplets}

In this section, we want to take a closer look at pairs of charm and bottom multiplets. For each pair of charm and bottom multiplets, heavy quark symmetry allows us to formulate mass relations that involve hadrons from both the charm and bottom multiplet. This is interesting for multiple reasons: First off, all mass relations we have found in \autoref{sec:rel_within_multiplets} only relate hadrons in the same multiplet, so the mass relations we introduce in this section expand our description of hadron masses. Secondly, considering several multiplets to form mass relations gives us the possibility to use hadronic multiplets that do not contain enough hadrons to give rise to a mass relation within them. This applies, in particular, to hadronic charm and bottom (anti)triplets. Lastly, the mass relations following from heavy quark symmetry provide us with a way to investigate mesonic mass relations that are not spoiled by octet-singlet-mixing in contrast to the mass relations within mesonic octets.\par
The discussions of the charm and bottom multiplets and the mass relations between them are structured very similar to \autoref{sec:rel_within_multiplets}. All remarks at the beginning of \autoref{sec:rel_within_multiplets} apply in very similar fashion to the following discussions.

\subsection*{Baryonic Charm and Bottom Antitriplets}

The lightest charm and bottom baryons form antitriplets with $J^P = 1/2^+$. Their weight diagrams are displayed in \autoref{fig:c_b_baryon_triplets}.
\begin{figure}
\centering
\subfigure{\includegraphics[width=0.85\textwidth]{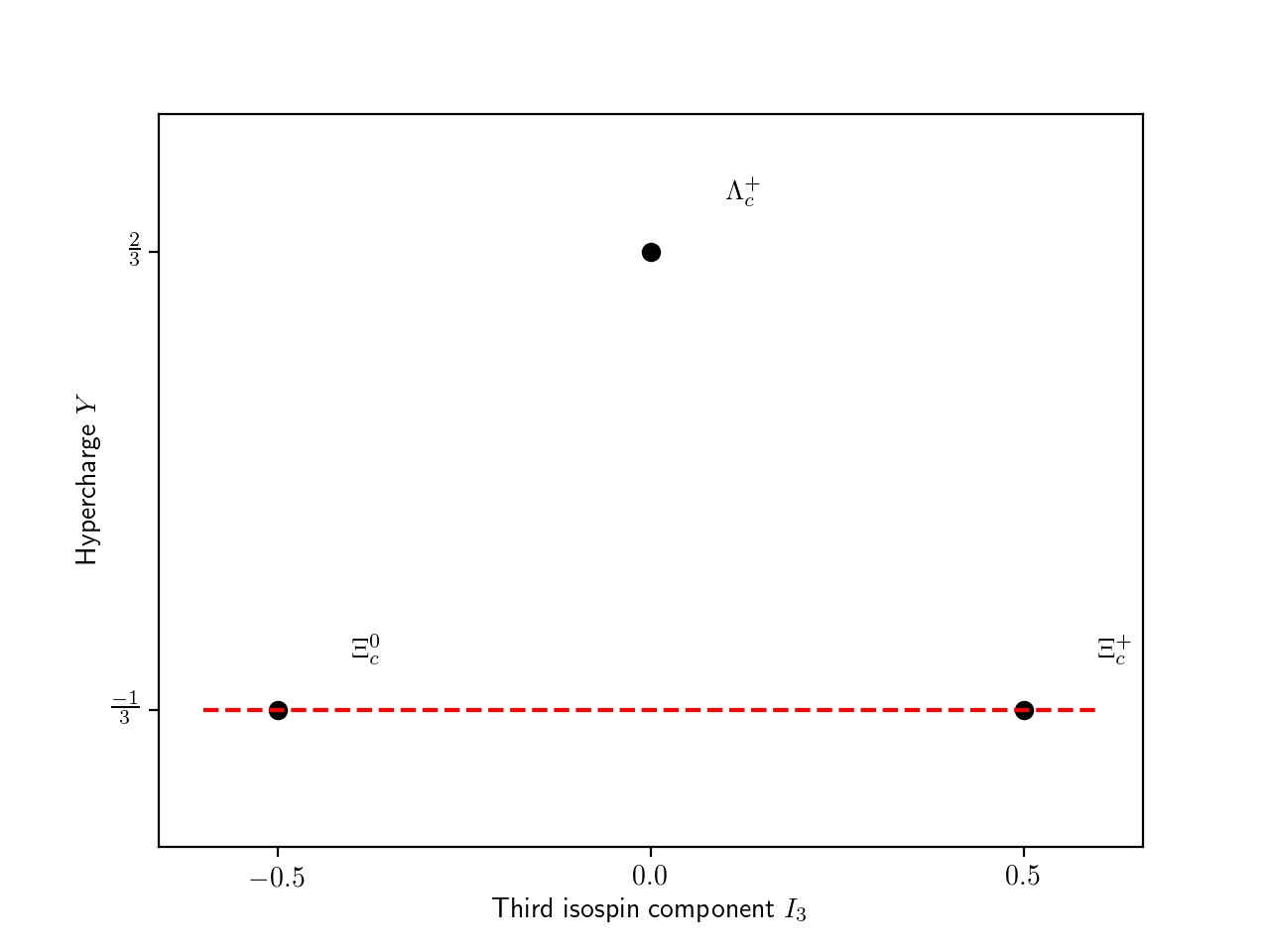}}
\vspace{0.1cm}
\subfigure{\includegraphics[width=0.85\textwidth]{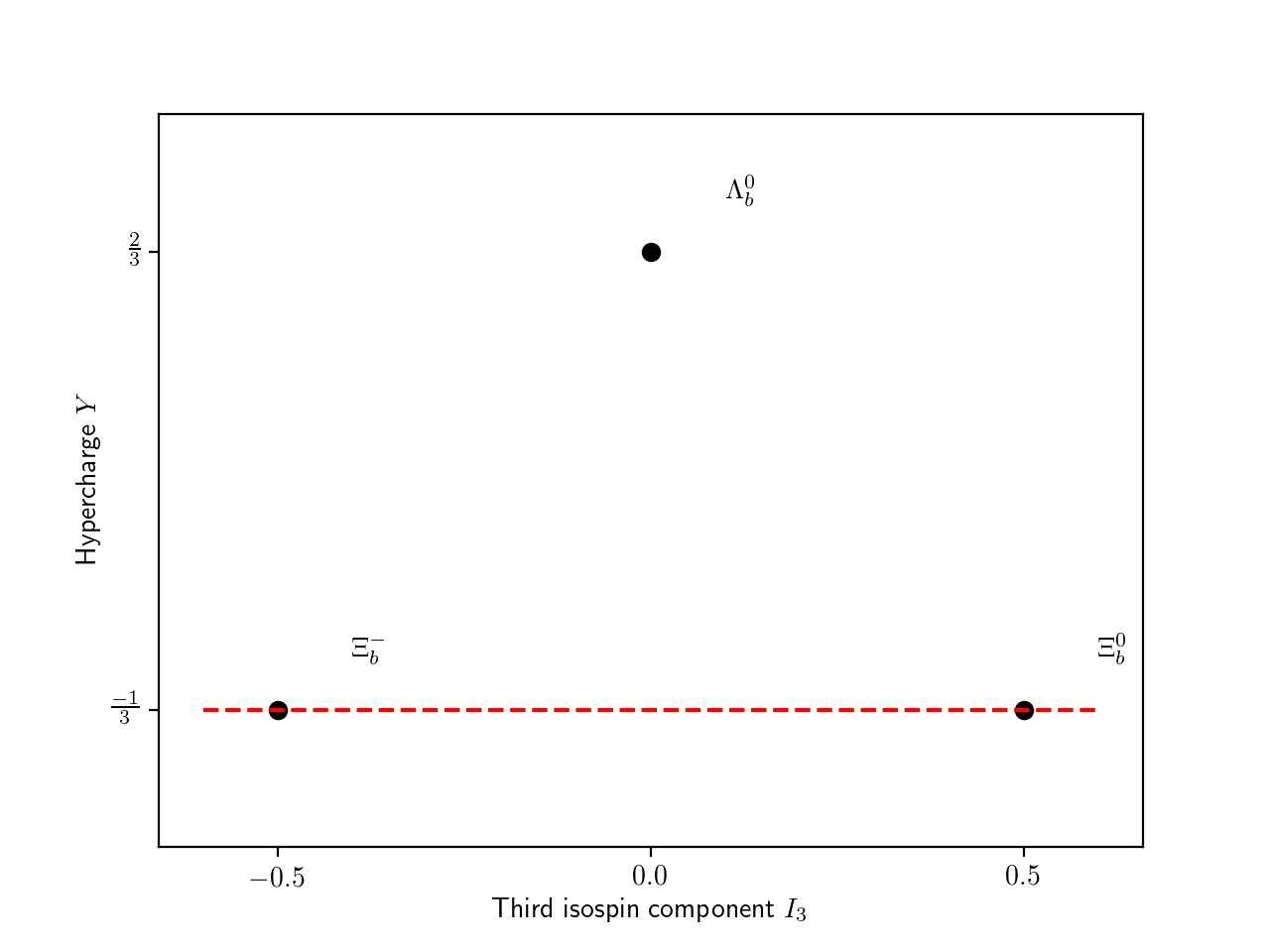}}
\caption{Two weight diagrams of an antitriplet. For every weight, the name of the corresponding baryon from the charm or bottom antitriplet with $J^P = 1/2^+$ is included in the upper or lower weight diagram, respectively. The red lines visualize isospin multiplets of $\text{SU}(2)\times\text{U}(1)$. All weights not connected by red lines are one-dimensional isospin multiplets.}
\label{fig:c_b_baryon_triplets}
\end{figure}
Using \autoref{eq:mass_tot_sym_c}, \autoref{eq:mass_tot_sym_b}, \autoref{eq:mass_ele_c}, and \autoref{eq:mass_ele_b}, we find for the masses in the charm and bottom antitriplets:
\begin{alignat*}{7}
&m^\text{c}_\alpha\left(\frac{2}{3}, 0, 0\right) &&\equiv m_{\Lambda^{+}_\text{c}} &&= m^\text{c}_0 && &&+\frac{2}{3}\tilde{m}^F_8 &&+ \frac{2}{9}\Delta^{LH}_\alpha &&- \frac{2}{9}\Delta^{LL}_\alpha,\\
&m^\text{c}_\alpha\left(-\frac{1}{3}, \frac{1}{2}, \frac{1}{2}\right) &&\equiv m_{\Xi^{+}_\text{c}} &&= m^\text{c}_0 &&+\frac{1}{2}m^F_3 &&-\frac{1}{3}\tilde{m}^F_8 &&+ \frac{2}{9}\Delta^{LH}_\alpha &&- \frac{2}{9}\Delta^{LL}_\alpha,\\
&m^\text{c}_\alpha\left(-\frac{1}{3}, \frac{1}{2}, -\frac{1}{2}\right) &&\equiv m_{\Xi^{0}_\text{c}} &&= m^\text{c}_0 &&-\frac{1}{2}m^F_3 &&-\frac{1}{3}\tilde{m}^F_8 &&- \frac{4}{9}\Delta^{LH}_\alpha &&+ \frac{1}{9}\Delta^{LL}_\alpha,\\
&m^\text{b}_\alpha\left(\frac{2}{3}, 0, 0\right) &&\equiv m_{\Lambda^{0}_\text{b}} &&= m^\text{b}_0 && &&+\frac{2}{3}\tilde{m}^F_8 &&- \frac{1}{9}\Delta^{LH}_\alpha &&- \frac{2}{9}\Delta^{LL}_\alpha,\\
&m^\text{b}_\alpha\left(-\frac{1}{3}, \frac{1}{2}, \frac{1}{2}\right) &&\equiv m_{\Xi^{0}_\text{b}} &&= m^\text{b}_0 &&+\frac{1}{2}m^F_3 &&-\frac{1}{3}\tilde{m}^F_8 &&- \frac{1}{9}\Delta^{LH}_\alpha &&- \frac{2}{9}\Delta^{LL}_\alpha,\\
&m^\text{b}_\alpha\left(-\frac{1}{3}, \frac{1}{2}, -\frac{1}{2}\right) &&\equiv m_{\Xi^{-}_\text{b}} &&= m^\text{b}_0 &&-\frac{1}{2}m^F_3 &&-\frac{1}{3}\tilde{m}^F_8 &&+ \frac{2}{9}\Delta^{LH}_\alpha &&+ \frac{1}{9}\Delta^{LL}_\alpha,
\end{alignat*}
where we omitted all ``$\mathcal{O}$'' for the sake of clarity. This mass parametrization allows us to formulate one mass relation:
\begin{gather}
m_{\Lambda^{+}_\text{c}} - m_{\Xi^{+}_\text{c}} = m_{\Lambda^{0}_\text{b}} - m_{\Xi^{0}_\text{b}} + \mathcal{O}\left(\varepsilon_\text{cb}\varepsilon_8\right).\label{eq:bary_c_b_tri}
\end{gather}
This time, we have included the order of magnitude of the dominant correction in the mass relation. Of course, an analogous mass relation applies to all baryonic charm and bottom antitriplets. \autoref{eq:bary_c_b_tri} states that the mass splittings between the isospin multiplets in the charm and bottom antitriplets are equal. The dominant correction to \autoref{eq:bary_c_b_tri} is in the order of $\varepsilon_\text{cb}\varepsilon_8$. This can be shown by the following consideration: $\Lambda^+_\text{c}$ and $\Xi^+_\text{c}$ as well as $\Lambda^0_\text{b}$ and $\Xi^0_\text{b}$ form a multiplet of $\text{SU}(2)_\text{ds}\times\text{U}(1)$. Therefore, their masses would be equal and \autoref{eq:bary_c_b_tri} would be exact, if $\text{SU}(2)_\text{ds}\times\text{U}(1)$ was an exact symmetry. Neglecting the weak interaction, this symmetry is only broken by the mass difference of the down and strange quark, so roughly by $\varepsilon_8$. Thus, every correction to \autoref{eq:bary_c_b_tri} has to be proportional to $\varepsilon_8$. Following the line of reasoning from \autoref{sec:heavy_quark}, \autoref{eq:bary_c_b_tri} only picks up corrections proportional to $\varepsilon_\text{cb}$ or to a product of at least two of the parameters $\varepsilon_3$, $\varepsilon_8$, and $\alpha$. As every correction has to be proportional to $\varepsilon_8$, corrections proportional to $\varepsilon_\text{cb}$ have to be proportional to $\varepsilon_\text{cb}\varepsilon_8$ as well. This leaves us with $\varepsilon_\text{cb}\varepsilon_8$ and $\varepsilon^2_8$ as the largest corrections to \autoref{eq:bary_c_b_tri}. However, we have an exact $\text{SU}(2)_\text{cb}\times\text{U}(1)$-flavor symmetry between charm and bottom quark in the limit of $m_\text{c} = m_\text{b}$ and $\alpha = 0$. In this limit, the corresponding charm and bottom baryons of the charm and bottom antitriplet have to have the same mass. This implies that \autoref{eq:bary_c_b_tri} is exact in this limit. Nevertheless, corrections proportional to $\varepsilon^2_8$ do not vanish in this limit which means that they cannot be present in the first place. Therefore, the dominant correction to \autoref{eq:bary_c_b_tri} is in the order of $\varepsilon_\text{cb}\varepsilon_8$.

\subsection*{Mesonic Charm and Bottom Antitriplets}

Mesonic multiplets include both triplets and antitriplets. However, we can restrict ourselves to considering just triplets or antitriplets, since for every mesonic triplet there is a mesonic antitriplet that contains the antiparticles of the mesons in the triplet. As CPT-invariance dictates that particles and their corresponding antiparticles have to have the same mass, the mass relations of triplets also apply to antitriplets and vice versa. In reference to the baryonic case, we restrict ourselves to antitriplets.\par
The lightest charm and bottom mesons form (anti)triplets with $J^P = 0^-$. Their weight diagrams are displayed in \autoref{fig:c_b_meson_triplets}.
\begin{figure}
\centering
\subfigure{\includegraphics[width=0.85\textwidth]{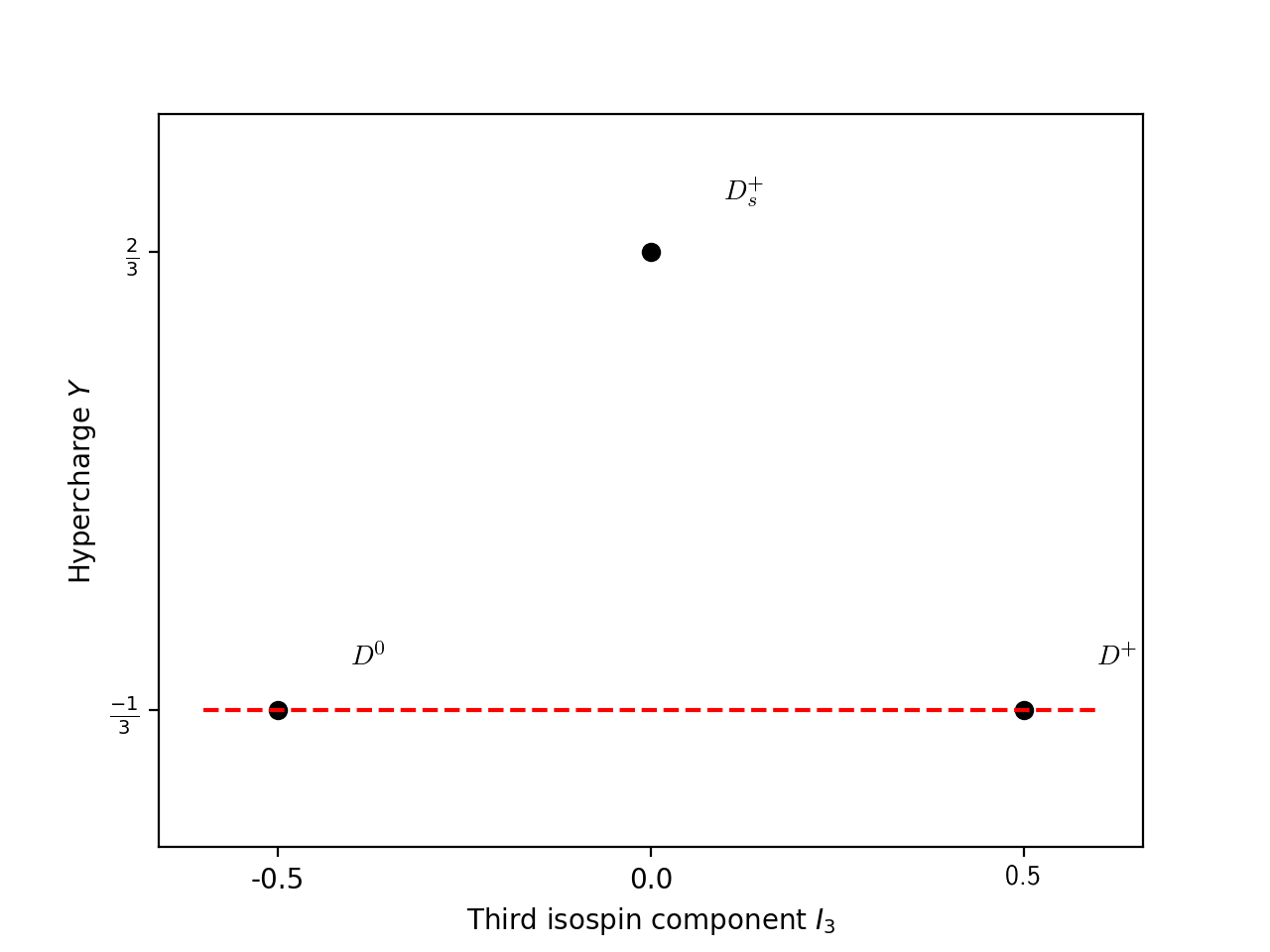}}
\vspace{0.1cm}
\subfigure{\includegraphics[width=0.85\textwidth]{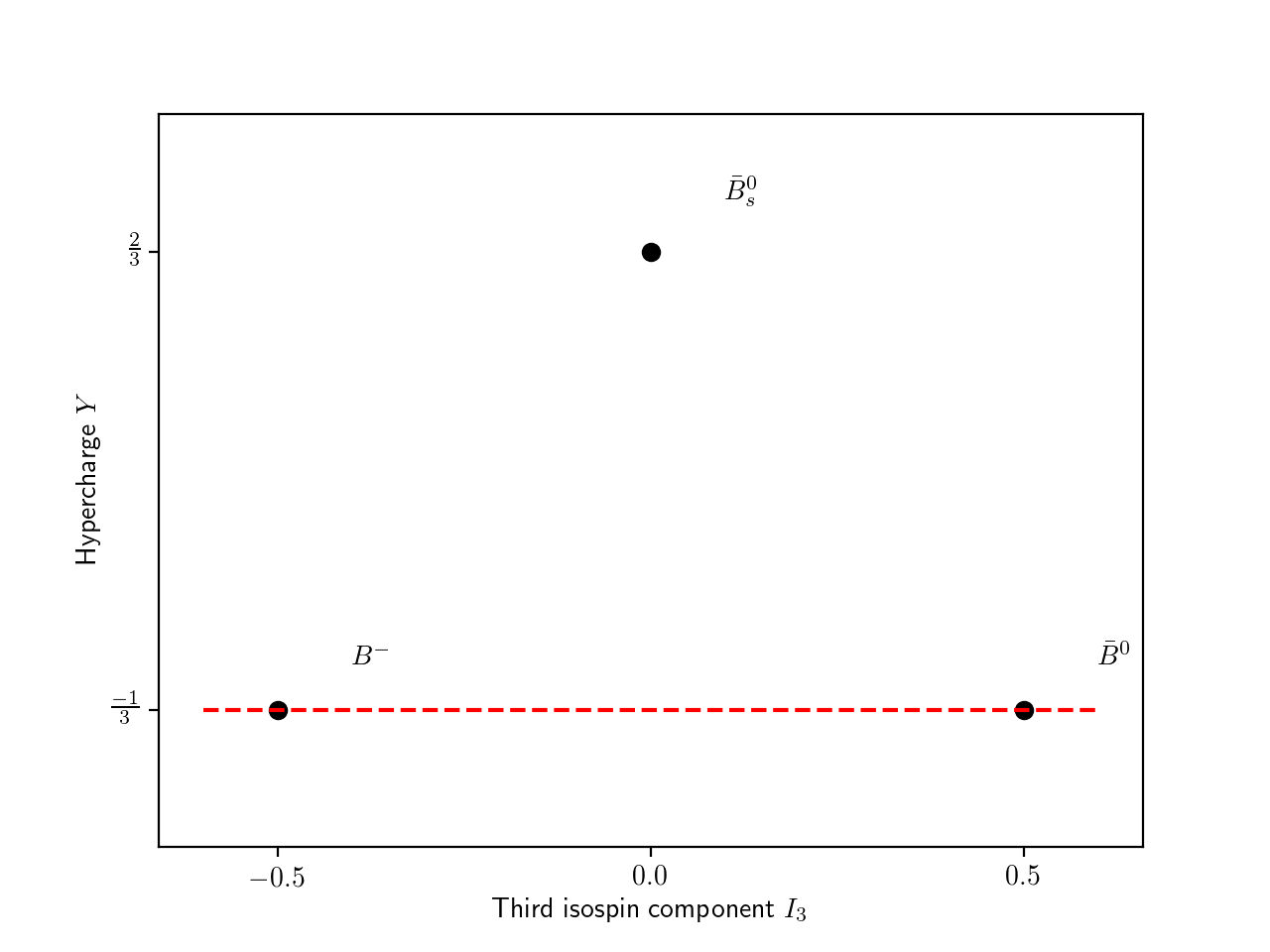}}
\caption{Two weight diagrams of an antitriplet. For every weight, the name of the corresponding meson from the charm or bottom antitriplet with $J^P = 0^-$ is included in the upper or lower weight diagram, respectively. The red lines visualize isospin multiplets of $\text{SU}(2)\times\text{U}(1)$. All weights not connected by red lines are one-dimensional isospin multiplets.}
\label{fig:c_b_meson_triplets}
\end{figure}
Using \autoref{eq:mass_tot_sym_c}, \autoref{eq:mass_tot_sym_b}, \autoref{eq:mass_ele_c}, and \autoref{eq:mass_ele_b}, we find for the masses in the charm and bottom antitriplets:
\begin{alignat*}{6}
&m^\text{c}_\alpha\left(\frac{2}{3}, 0, 0\right) &&\equiv m_{D^{+}_\text{s}} &&= m^\text{c}_0 && &&+\frac{2}{3}\tilde{m}^F_8 &&+ \frac{2}{9}\Delta^{LH}_\alpha,\\
&m^\text{c}_\alpha\left(-\frac{1}{3}, \frac{1}{2}, \frac{1}{2}\right) &&\equiv m_{D^{+}} &&= m^\text{c}_0 &&+\frac{1}{2}m^F_3 &&-\frac{1}{3}\tilde{m}^F_8 &&+ \frac{2}{9}\Delta^{LH}_\alpha,\\
&m^\text{c}_\alpha\left(-\frac{1}{3}, \frac{1}{2}, -\frac{1}{2}\right) &&\equiv m_{D^{0}} &&= m^\text{c}_0 &&-\frac{1}{2}m^F_3 &&-\frac{1}{3}\tilde{m}^F_8 &&- \frac{4}{9}\Delta^{LH}_\alpha,\\
&m^\text{b}_\alpha\left(\frac{2}{3}, 0, 0\right) &&\equiv m_{\bar{B}^{0}_\text{s}} &&= m^\text{b}_0 && &&+\frac{2}{3}\tilde{m}^F_8 &&- \frac{1}{9}\Delta^{LH}_\alpha,\\
&m^\text{b}_\alpha\left(-\frac{1}{3}, \frac{1}{2}, \frac{1}{2}\right) &&\equiv m_{\bar{B}^{0}} &&= m^\text{b}_0 &&+\frac{1}{2}m^F_3 &&-\frac{1}{3}\tilde{m}^F_8 &&- \frac{1}{9}\Delta^{LH}_\alpha,\\
&m^\text{b}_\alpha\left(-\frac{1}{3}, \frac{1}{2}, -\frac{1}{2}\right) &&\equiv m_{B^{-}} &&= m^\text{b}_0 &&-\frac{1}{2}m^F_3 &&-\frac{1}{3}\tilde{m}^F_8 &&+ \frac{2}{9}\Delta^{LH}_\alpha,
\end{alignat*}
where we omitted all ``$\mathcal{O}$'' for the sake of clarity. This mass parametrization allows us to formulate one mass relation:
\begin{gather}
m_{D^{+}_\text{s}} - m_{D^{+}} = m_{\bar{B}^{0}_\text{s}} - m_{\bar{B}^{0}} + \mathcal{O}\left(\varepsilon_\text{cb}\varepsilon_8\right).\label{eq:meso_c_b_tri}
\end{gather}
This time, we have included the order of magnitude of the dominant corrections in the mass relation. Of course, an analogous mass relation applies to all mesonic charm and bottom (anti)triplets. \autoref{eq:meso_c_b_tri} is the mesonic companion piece to \autoref{eq:bary_c_b_tri}. Indeed, its dominant corrections can be obtained in the same way. Relations like \autoref{eq:meso_c_b_tri} are a typical result of heavy quark symmetry (cf. \cite{Neubert1994}).

\subsection*{Charm and Bottom Sextets}

The lightest charm and bottom sextets are characterized by $J^P = 1/2^+$. Their weight diagrams are displayed in \autoref{fig:c_b_sextets}.
\begin{figure}
\centering
\subfigure{\includegraphics[width=0.85\textwidth]{c_sextet.png}}
\vspace{0.1cm}
\subfigure{\includegraphics[width=0.85\textwidth]{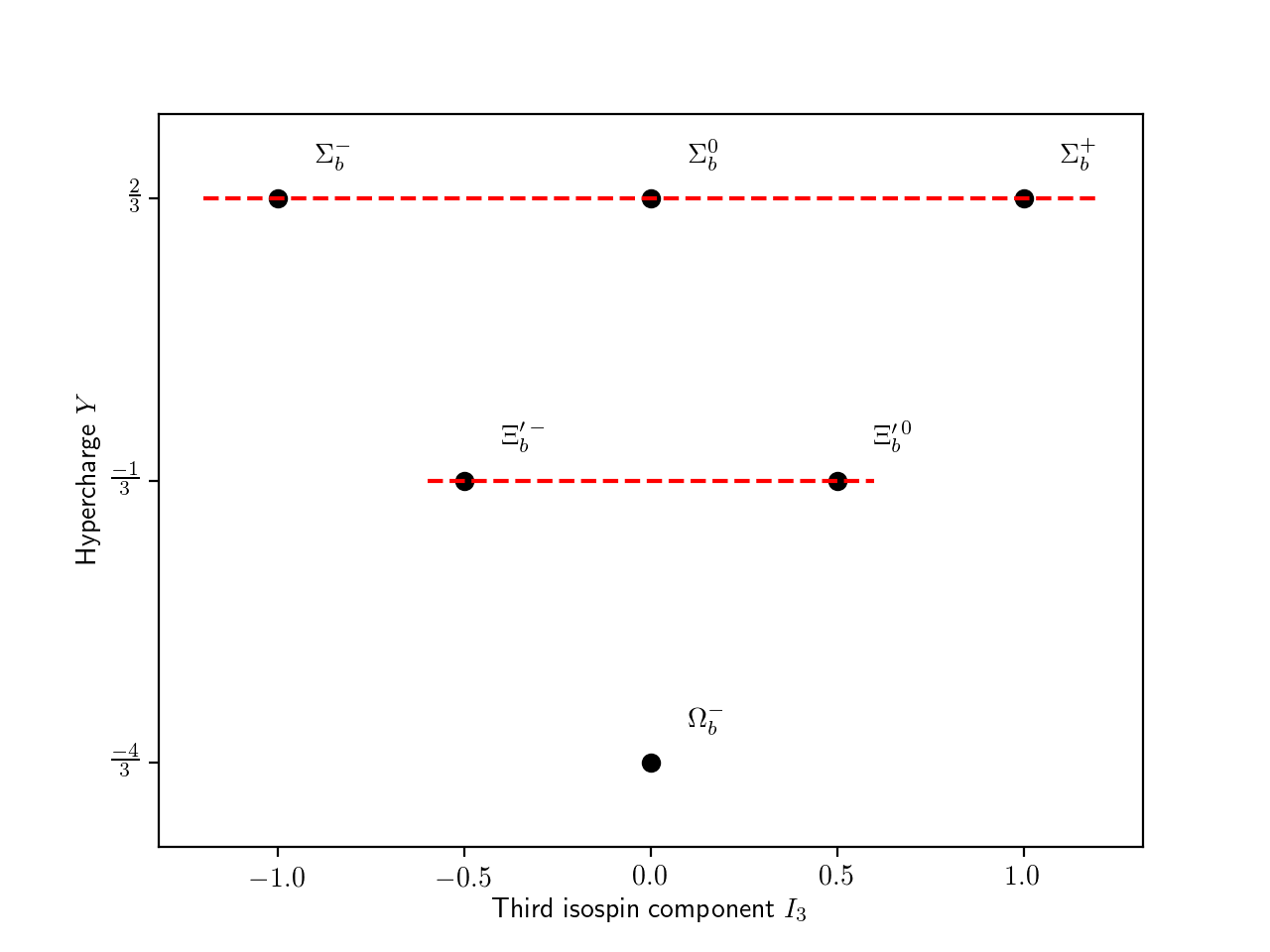}}
\caption{Two weight diagrams of a sextet. For every weight, the name of the corresponding baryon from the charm or bottom sextet with $J^P = 1/2^+$ is included in the upper or lower weight diagram, respectively. The red lines visualize isospin multiplets of $\text{SU}(2)\times\text{U}(1)$. All weights not connected by red lines are one-dimensional isospin multiplets.}
\label{fig:c_b_sextets}
\end{figure}
Using \autoref{eq:mass_tot_sym_c}, \autoref{eq:mass_tot_sym_b}, \autoref{eq:mass_ele_c}, and \autoref{eq:mass_ele_b}, we find for the masses in the charm and bottom sextets:
\begin{alignat*}{4}
&m^\text{c}_\alpha\left(\frac{2}{3},1,1\right)&&\equiv m_{\Sigma^{++}_\text{c}} &&= m^\text{c}_0 + m^F_3 &&+ \frac{2}{3}\tilde{m}^F_8 + \frac{8}{9}\Delta^{LH}_\alpha + \frac{4}{9}\Delta^{LL}_\alpha,\\
&m^\text{c}_\alpha\left(\frac{2}{3},1,0\right)&&\equiv m_{\Sigma^{+}_\text{c}} &&= m^\text{c}_0 &&+ \frac{2}{3}\tilde{m}^F_8 + \frac{2}{9}\Delta^{LH}_\alpha - \frac{2}{9}\Delta^{LL}_\alpha,\\
&m^\text{c}_\alpha\left(\frac{2}{3},1,-1\right)&&\equiv m_{\Sigma^{0}_\text{c}} &&= m^\text{c}_0 - m^F_3 &&+ \frac{2}{3}\tilde{m}^F_8 - \frac{4}{9}\Delta^{LH}_\alpha + \frac{1}{9}\Delta^{LL}_\alpha,\\
&m^\text{c}_\alpha\left(-\frac{1}{3},\frac{1}{2},\frac{1}{2}\right)&&\equiv m_{\Xi^{\prime\, +}_\text{c}} &&= m^\text{c}_0 + \frac{1}{2}m^F_3 &&- \frac{1}{3}\tilde{m}^F_8 + \frac{2}{9}\Delta^{LH}_\alpha - \frac{2}{9}\Delta^{LL}_\alpha,\\
&m^\text{c}_\alpha\left(-\frac{1}{3},\frac{1}{2},-\frac{1}{2}\right)&&\equiv m_{\Xi^{\prime\, 0}_\text{c}} &&= m^\text{c}_0 - \frac{1}{2}m^F_3 &&- \frac{1}{3}\tilde{m}^F_8 - \frac{4}{9}\Delta^{LH}_\alpha + \frac{1}{9}\Delta^{LL}_\alpha,\\
&m^\text{c}_\alpha\left(-\frac{4}{3},0,0\right)&&\equiv m_{\Omega^{0}_\text{c}} &&= m^\text{c}_0 &&- \frac{4}{3}\tilde{m}^F_8 - \frac{4}{9}\Delta^{LH}_\alpha + \frac{1}{9}\Delta^{LL}_\alpha,\\
&m^\text{b}_\alpha\left(\frac{2}{3},1,1\right)&&\equiv m_{\Sigma^{+}_\text{b}} &&= m^\text{b}_0 + m^F_3 &&+ \frac{2}{3}\tilde{m}^F_8 - \frac{4}{9}\Delta^{LH}_\alpha + \frac{4}{9}\Delta^{LL}_\alpha,\\
&m^\text{b}_\alpha\left(\frac{2}{3},1,0\right)&&\equiv m_{\Sigma^{0}_\text{b}} &&= m^\text{b}_0 &&+ \frac{2}{3}\tilde{m}^F_8 - \frac{1}{9}\Delta^{LH}_\alpha - \frac{2}{9}\Delta^{LL}_\alpha,\\
&m^\text{b}_\alpha\left(\frac{2}{3},1,-1\right)&&\equiv m_{\Sigma^{-}_\text{b}} &&= m^\text{b}_0 - m^F_3 &&+ \frac{2}{3}\tilde{m}^F_8 + \frac{2}{9}\Delta^{LH}_\alpha + \frac{1}{9}\Delta^{LL}_\alpha,\\
&m^\text{b}_\alpha\left(-\frac{1}{3},\frac{1}{2},\frac{1}{2}\right)&&\equiv m_{\Xi^{\prime\, 0}_\text{b}} &&= m^\text{b}_0 + \frac{1}{2}m^F_3 &&- \frac{1}{3}\tilde{m}^F_8 - \frac{1}{9}\Delta^{LH}_\alpha - \frac{2}{9}\Delta^{LL}_\alpha,\\
&m^\text{b}_\alpha\left(-\frac{1}{3},\frac{1}{2},-\frac{1}{2}\right)&&\equiv m_{\Xi^{\prime\, -}_\text{b}} &&= m^\text{b}_0 - \frac{1}{2}m^F_3 &&- \frac{1}{3}\tilde{m}^F_8 + \frac{2}{9}\Delta^{LH}_\alpha + \frac{1}{9}\Delta^{LL}_\alpha,\\
&m^\text{b}_\alpha\left(-\frac{4}{3},0,0\right)&&\equiv m_{\Omega^{-}_\text{b}} &&= m^\text{b}_0 &&- \frac{4}{3}\tilde{m}^F_8 + \frac{2}{9}\Delta^{LH}_\alpha + \frac{1}{9}\Delta^{LL}_\alpha,
\end{alignat*}
where we omitted all ``$\mathcal{O}$'' for the sake of clarity. This mass parametrization allows us to find ten mass relations:
\begin{gather}
m_{\Sigma^+_\text{c}} - m_{\Sigma^0_\text{c}} = m_{\Xi^{\prime\, +}_\text{c}} - m_{\Xi^{\prime\, 0}_\text{c}} + \mathcal{O}\left(\alpha\varepsilon_8\right) + \mathcal{O}\left(\varepsilon_3\varepsilon_8\right),\label{eq:sextet_c_iso_bre2}\\
m_{\Sigma^0_\text{b}} - m_{\Sigma^-_\text{b}} = m_{\Xi^{\prime\, 0}_\text{b}} - m_{\Xi^{\prime\, -}_\text{b}} + \mathcal{O}\left(\alpha\varepsilon_8\right) + \mathcal{O}\left(\varepsilon_3\varepsilon_8\right),\label{eq:sextet_b_iso_bre}\\
m_{\Sigma^0_\text{c}} - m_{\Xi^{\prime\, 0}_\text{c}} = m_{\Xi^{\prime\, 0}_\text{c}} - m_{\Omega^0_\text{c}} + \mathcal{O}\left(\varepsilon^2_8\right),\label{eq:sextet_c_GMO_equal_spacing2}\\
m_{\Sigma^-_\text{b}} - m_{\Xi^{\prime\, -}_\text{b}} = m_{\Xi^{\prime\, -}_\text{b}} - m_{\Omega^-_\text{b}} + \mathcal{O}\left(\varepsilon^2_8\right),\label{eq:sextet_b_GMO_equal_spacing}\\
m_{\Sigma^{++}_\text{c}} + m_{\Sigma^{0}_\text{c}} - 2m_{\Sigma^{+}_\text{c}} = m_{\Sigma^{+}_\text{b}} + m_{\Sigma^{-}_\text{b}} - 2m_{\Sigma^{0}_\text{b}} + \mathcal{O}\left(\alpha\varepsilon_\text{cb}\right) + \mathcal{O}\left(\alpha\varepsilon_8\right),\label{eq:sextet_c_b_sigma}\\
m_{\Sigma^0_\text{c}} - m_{\Xi^{\prime\, 0}_\text{c}} = m_{\Sigma^-_\text{b}} - m_{\Xi^{\prime\, -}_\text{b}} + \mathcal{O}\left(\varepsilon_\text{cb}\varepsilon_8\right),\label{eq:sextet_c_b_spacing}\\
m_{\Sigma^+_\text{c}} - m_{\Sigma^0_\text{c}} + m_{\Xi^{\prime\, 0}_\text{c}} - m_{\Xi^{\prime\, +}_\text{c}} = m_{\Sigma^0_\text{b}} - m_{\Sigma^-_\text{b}} + m_{\Xi^{\prime\, -}_\text{b}} - m_{\Xi^{\prime\, 0}_\text{b}} + \mathcal{O}\left(\alpha\varepsilon_8\right),\label{eq:sextet_c_b_very_precise}\\
2m_{\Xi^{\prime\, 0}_\text{c}} - m_{\Sigma^0_\text{c}} - m_{\Omega^0_\text{c}} = 2m_{\Xi^{\prime\, -}_\text{b}} - m_{\Sigma^-_\text{b}} - m_{\Omega^-_\text{b}} + \mathcal{O}\left(\varepsilon_\text{cb}\varepsilon^2_8\right),\label{eq:sextet_c_b_precise}\\
m_{\Sigma^{0}_\text{b}} = m_{\Xi^{\prime\, +}_\text{c}} - m_{\Xi^{\prime\, 0}_\text{c}} + \frac{1}{2}\left(m_{\Sigma^{+}_\text{b}} + m_{\Sigma^{-}_\text{b}} + m_{\Sigma^{0}_\text{c}} - m_{\Sigma^{++}_\text{c}}\right)\label{eq:sextet_c_b_sigma_b}\\
+ \mathcal{O}\left(\alpha\varepsilon_\text{cb}\right) + \mathcal{O}\left(\varepsilon_3\varepsilon_\text{cb}\right) + \mathcal{O}\left(\alpha\varepsilon_8\right) + \mathcal{O}\left(\varepsilon_3\varepsilon_8\right),\nonumber\\
m_{\Xi^{\prime\, 0}_\text{b}} = m_{\Xi^{\prime\, -}_\text{b}} + m_{\Xi^{\prime\, +}_\text{c}} - m_{\Xi^{\prime\, 0}_\text{c}} + \frac{1}{2}\left(m_{\Sigma^{+}_\text{b}} - m_{\Sigma^{-}_\text{b}} + m_{\Sigma^{0}_\text{c}} - m_{\Sigma^{++}_\text{c}}\right)\label{eq:sextet_c_b_xi}\\
+ \mathcal{O}\left(\alpha\varepsilon_\text{cb}\right) + \mathcal{O}\left(\varepsilon_3\varepsilon_\text{cb}\right) + \mathcal{O}\left(\alpha\varepsilon_8\right).\nonumber
\end{gather}
This time, we have included the order of magnitude of the dominant correction(s) in the mass relations. Of course, analogous mass relations apply to all pairs of baryonic charm and bottom sextets. Only the first six relations are inequivalent in the sense that none of the first six mass relations follows from each other. The seventh and eighth mass relation (\autoref{eq:sextet_c_b_very_precise} and \autoref{eq:sextet_c_b_precise}) are more precise variants of the first two mass relations (\autoref{eq:sextet_c_iso_bre2} and \autoref{eq:sextet_b_iso_bre}) and of the following two mass relations (\autoref{eq:sextet_c_GMO_equal_spacing2} and \autoref{eq:sextet_b_GMO_equal_spacing}), respectively. Experimentally, some baryon masses in some sextets are not measured yet, therefore, we will use the last two mass relations (\autoref{eq:sextet_c_b_sigma_b} and \autoref{eq:sextet_c_b_xi}) to calculate the masses of the missing baryons in \autoref{sec:mass_predictions}.\par
The first mass relation (\autoref{eq:sextet_c_iso_bre2}) is just \autoref{eq:sextet_c_iso_bre} from \autoref{sec:rel_within_multiplets} and listed here for completeness' sake. How to obtain the dominant corrections to this mass relation is explained in \autoref{sec:rel_within_multiplets}. The second mass relation (\autoref{eq:sextet_b_iso_bre}) is just the bottom companion piece to the first mass relation. Similar to the first two mass relations, the third mass relation (\autoref{eq:sextet_c_GMO_equal_spacing2}) is just \autoref{eq:sextet_c_GMO_equal_spacing} from \autoref{sec:rel_within_multiplets} and the fourth mass relation (\autoref{eq:sextet_b_GMO_equal_spacing}) is the bottom companion piece to the third mass relation.\par
The fifth and sixth mass relation (\autoref{eq:sextet_c_b_sigma} and \autoref{eq:sextet_c_b_spacing}) involve masses of both the charm and bottom sextet. To determine the dominant corrections to these mass relations, we note that both mass relations are exact in the limit of $m_\text{c} = m_\text{b}$ and $\alpha = 0$, since, in this limit, the mass of a baryon in the charm sextet is equal to the mass of the corresponding baryon in the bottom sextet. This means that every correction to \autoref{eq:sextet_c_b_sigma} and \autoref{eq:sextet_c_b_spacing} has to be proportional to $\varepsilon_\text{cb}$ or $\alpha$. Moreover, we find that \autoref{eq:sextet_c_b_sigma} would also be exact, if the $\text{SU}(2)\times\text{U}(1)$-isospin symmetry was exact. As this symmetry is only broken by $\varepsilon_3$ and $\alpha$, every correction to \autoref{eq:sextet_c_b_sigma} has to be proportional to $\varepsilon_3$ or $\alpha$ as well. Corrections in the order of $\alpha$ cannot occur, as the given mass parametrization already takes such correction into account. This implies that the dominant corrections to \autoref{eq:sextet_c_b_sigma} have to be in the order of $\alpha\varepsilon_\text{cb}$, $\alpha\varepsilon_8$, or $\varepsilon_3\varepsilon_\text{cb}$. However, a correction in the order of $\varepsilon_3\varepsilon_\text{cb}$ cannot occur, because neither the combination of all charm baryons nor the combination of all bottom baryons in \autoref{eq:sextet_c_b_sigma}, i.e., neither the left-hand side nor the right-hand side of \autoref{eq:sextet_c_b_sigma} contains any terms proportional to a single power of $\varepsilon_3$ (cf. \autoref{eq:iso_tot_sym}):
\begin{gather*}
m_{\Sigma^{++}_\text{c}} + m_{\Sigma^{0}_\text{c}} - 2m_{\Sigma^{+}_\text{c}} = 0 + \mathcal{O}\left(\varepsilon^2_3\right) + \mathcal{O}\left(\alpha\right),\\
m_{\Sigma^{+}_\text{b}} + m_{\Sigma^{-}_\text{b}} - 2m_{\Sigma^{0}_\text{b}} = 0 + \mathcal{O}\left(\varepsilon^2_3\right) + \mathcal{O}\left(\alpha\right).
\end{gather*}
Thus, the dominant corrections to \autoref{eq:sextet_c_b_sigma} are in the order of\linebreak {${\alpha\varepsilon_\text{cb}}$} and {${\alpha\varepsilon_8}$}.\par
The sixth mass relation (\autoref{eq:sextet_c_b_spacing}) is not only exact in the limit of $m_\text{c} = m_\text{b}$ and $\alpha = 0$, but also in the limit of exact $\text{SU}(2)_\text{ds}\times\text{U}(1)$-flavor symmetry (cf. \autoref{fig:sextet_ds}). This means that every correction to \autoref{eq:sextet_c_b_spacing} has to be proportional to $\varepsilon_8$ and proportional to $\varepsilon_\text{cb}$ or $\alpha$. Hence, the dominant correction to \autoref{eq:sextet_c_b_spacing} is in the order of $\varepsilon_\text{cb}\varepsilon_8$.\par
We obtain the seventh and eighth mass relation (\autoref{eq:sextet_c_b_very_precise} and \autoref{eq:sextet_c_b_precise}) by subtracting \autoref{eq:sextet_b_iso_bre} from \autoref{eq:sextet_c_iso_bre2} and by subtracting \autoref{eq:sextet_b_GMO_equal_spacing} from \autoref{eq:sextet_c_GMO_equal_spacing2}, respectively, and rewriting the results. This implies that both \autoref{eq:sextet_c_b_very_precise} and \autoref{eq:sextet_c_b_precise} are exact in the limit of exact heavy quark symmetry, i.e., in the limit of $m_\text{c}=m_\text{b}$ and $\alpha=0$. Thus, every correction to \autoref{eq:sextet_c_b_very_precise} and \autoref{eq:sextet_c_b_precise} has to be proportional to $\varepsilon_\text{cb}$ or $\alpha$. Furthermore, if the $\text{SU}(2)\times\text{U}(1)$-isospin symmetry or the $\text{SU}(2)_\text{ds}\times\text{U}(1)$-flavor symmetry was exact, \autoref{eq:sextet_c_b_very_precise} would also be exact. We can deduce from this that the dominant correction to \autoref{eq:sextet_c_b_very_precise} is in the order of $\alpha\varepsilon_8$. For \autoref{eq:sextet_c_b_precise}, we note that it is exact in the limit of exact $\text{SU}(2)_\text{ds}\times\text{U}(1)$-flavor symmetry and that both the left-hand and right-hand side of \autoref{eq:sextet_c_b_precise} only pick up corrections proportional to $\varepsilon^2_8$ (cf. \autoref{eq:sextet_c_GMO_equal_spacing2} and \autoref{eq:sextet_b_GMO_equal_spacing}):
\begin{gather*}
2m_{\Xi^{\prime\, 0}_\text{c}} - m_{\Sigma^{0}_\text{c}} - m_{\Omega^{0}_\text{c}} = 0 + \mathcal{O}\left(\varepsilon^2_8\right),\\
2m_{\Xi^{\prime\, -}_\text{b}} - m_{\Sigma^{-}_\text{b}} - m_{\Omega^{-}_\text{b}} = 0 + \mathcal{O}\left(\varepsilon^2_8\right).
\end{gather*}
This implies that the dominant correction to \autoref{eq:sextet_c_b_precise} is in the order of $\varepsilon_\text{cb}\varepsilon^2_8$.\par
The last two mass relations (\autoref{eq:sextet_c_b_sigma_b} and \autoref{eq:sextet_c_b_xi}) follow from all previous sextet mass relations. Both relations are exact in the limit of exact $\text{SU}(2)\times\text{U}(1)$-isospin symmetry, so corrections to both relations have to be proportional to $\varepsilon_3$ or $\alpha$. The contributions in the order of $\varepsilon_3$ and $\alpha$ are already taken into account by the mass parametrization, so the dominant corrections are given by products of $\varepsilon_3$ or $\alpha$ with the remaining parameters\footnote{\autoref{eq:sextet_c_b_xi} is exact for $m_\text{c} = m_\text{b}$ and $\alpha=0$, thus the correction $\varepsilon_3\varepsilon_8$ cannot occur.}.

\newpage
\chapter{Testing Mass Relations on Experimental Data}\label{chap:data}

In the previous chapters, we discussed the symmetry properties of hadron masses regarding global flavor transformations in different approaches to find mass parametrizations and relations. In the course of this process, especially in \autoref{sec:EFT+H_Pert}, the question arose whether these mass relations apply to the hadron masses or their squares. To answer this question, we referred to the experimentally determined values for the hadron masses. However, we have refrained from showing any numbers or results so far. Therefore, we want to test the hadronic mass relations from \autoref{chap:mass_relations} on experimental data in this chapter.\par
To compare the mass relations from \autoref{chap:mass_relations} with experimental values, we first have to say how hadron masses can be accessed experimentally. For this, we will motivate a general definition of particle masses in the framework of S-matrix theory, the pole mass. In the course of this process, we will see that we will be faced with several difficulties regarding the definition and determination of hadron masses. The discussion of mass determination and pole mass is the concern of \autoref{sec:polemass}.\par
In \autoref{sec:mass_testing}, we will discuss to which extent the mass relations from \autoref{chap:mass_relations} are satisfied for both linear and quadratic hadron masses. The aim of this section is two-fold: On the one hand, we want to convince ourselves that the mass relations from \autoref{chap:mass_relations} actually apply up to their dominant correction(s) and that the assumptions which guided us to these mass relations are at least somewhat justified, as they lead to reasonable results. On the other hand, we want to back up our claims concerning the relation between linear and quadratic mass relations we made throughout this work, in particular at the end of \autoref{sec:EFT+H_Pert}.\par
Lastly, we want to use the mass relations from \autoref{chap:mass_relations} and empirical observations to group yet unassigned hadrons into multiplets and to predict the mass of hadrons missing within almost complete multiplets in \autoref{sec:mass_predictions}.

\section{Pole Mass and the Experimental Accessibility of Particle Masses}\label{sec:polemass}

The mass of a point-like particle in a non-quantized, relativistic theory is an intrinsic property of the particle, i.e., independent of the choice of the inertial frame and invariant under all Poincar\'{e} transformations. Following the laws of special relativity, the mass $m$ of a point-like particle with four-momentum $p^\mu$ can be obtained by using the Minkowski metric: $p^2 = m^2$. Thus, it is sufficient to measure the four-momentum of a particle to determine its mass in the framework of such a theory. If we try to apply this method of mass determination to hadrons, we are faced with two problems: On the one hand, most hadrons we are interested in are very short-lived, so it is often not viable in experiments to create a hadron and measure its four-momentum. On the other hand, such an approach neglects the very nature of subatomic particles, namely that they are subject to the laws of the quantum regime.\par
Therefore, a more sophisticated ansatz is needed to describe the mass of particles. This ansatz should factor in both special relativity as well as the laws of the quantum regime. To extent special relativity to quantum theories, we have to incorporate Poincar\'{e} transformations into a quantized theory. We can do this by requiring that the Poincar\'{e} transformations act via a representation on a set of states. We then identify particle states with irreducible representations of the Poincar\'{e} group (cf. \cite{Weinberg1995}). Irreducible representations of the Poincar\'{e} group are characterized by eigenvalues of Casimir operators which are constant on irreducible representations. One Casimir operator of the Poincar\'{e} group is the quadratic Casimir operator $\hat P^\mu\hat P_\mu$ of the translation operators $\hat P^\mu$. Its eigenvalue, denoted by $m^2$, defines the squared mass of the particle.\par
Even though this approach is a more sophisticated, it neither gives us an experimental prescription to determine the mass of a particle nor is it clear whether unstable, composite particles are described by the same formalism. Therefore, we are interested in a more general definition of particle masses. To guide our intuition, consider the following Lagrangian of a free quantized scalar field $\Phi$:
\begin{gather*}
\mathcal{L} = (\partial_\mu\Phi)^\dagger(\partial^\mu\Phi) - m^2\Phi^\dagger\Phi,
\end{gather*}
where $m$ is a real positive parameter. If one requires the fields $\Phi$ and $\Phi^\dagger$ to satisfy the Euler-Lagrange equations following from $\mathcal{L}$ and canonical commutation relations, we can interpret the fields $\Phi$ and $\Phi^\dagger$ -- their Fourier modes to be precise -- as operators annihilating and creating states which we can identify as particle states with mass $m$. Again, we understand particle states as vectors ``living'' in irreducible representations of the Poincar\'{e} group and the mass as the eigenvalue of $\hat P^\mu\hat P_\mu$. Thus, the mass $m$ of a particle in a free theory is directly connected to the coefficients of quadratic field terms in the Lagrangian. Now, consider the propagator $G(p)$ of $\Phi$:
\begin{gather*}
G(p) = \frac{i}{p^2-m^2}.
\end{gather*}
We see that $G$ has a pole at $p^2 = m^2$. This illustrates that particle masses in a QFT are in some way related to poles.\par
Still, we have refrained from giving an experimental prescription for measuring hadron masses. To account for this, consider now a simplistic model for an interacting theory:
\begin{gather*}
\mathcal{L} = (\partial_\mu\Phi)^\dagger(\partial^\mu\Phi) - m^2\Phi^\dagger\Phi + \overline{\Psi}\left(i\slashed{\partial} - M\right)\Psi + g\Phi\overline{\Psi}\Psi,
\end{gather*}
where $\Phi$ is a scalar field, $\Psi$ is a fermionic field, and $m$, $M$, and $g$ are real parameters. Now imagine we make a scattering experiment $\bar{f}f\to\bar{f}f$ within this theory, where we denote the (anti)fermions of this theory by $(\bar{f})f$. Experimentally, we would be able to determine the cross section of this reaction. Theoretically, the cross section for the reaction $\bar{f}f\to\bar{f}f$ is given by a phase space integral of the modulus squared of the transition amplitude $\mathcal{M}(\bar{f}f\to\bar{f}f)$. At tree level, $\mathcal{M}(\bar{f}f\to\bar{f}f)$ is given by the following Feynman diagram:\\
\begin{minipage}{0.49\textwidth}
\centering
\begin{tikzpicture}
  \begin{feynman}
    \vertex (a);
    \vertex [below left=of a] (b) {\(f\)};
    \vertex [above left=of a] (c) {\(\bar{f}\)};
    \vertex [right=of a] (d);
    \vertex [above right=of d] (e) {\(\bar{f}\)};
    \vertex [below right=of d] (f) {\(f\)};

    \diagram* {
      (b) -- [fermion, momentum' = \(p_1\)] (a) -- [fermion, rmomentum'=\(p_2\)] (c),
      (a) -- [scalar] (d),
      (d) -- [anti fermion] (e),
      (d) -- [fermion] (f),
    };
  \end{feynman}
\end{tikzpicture}
\end{minipage}
\begin{minipage}{0.49\textwidth}
\begin{flalign*}
\propto\frac{i}{(p_1+p_2)^2 - m^2}.
\end{flalign*}
\end{minipage}\\
In terms of Mandelstam variables, $(p_1+p_2)^2$ is the square $s$ of the center-of-mass energy. Thus, the cross section of the reaction $\bar{f}f\to\bar{f}f$ has -- at least at tree level -- a pole in $s$ at $s = m^2$. We would like to deduce from this result that cross sections have poles at particle masses -- as we will formulate later on --, since we have seen that the parameter $m$ corresponds to the mass of the particle associated with $\Phi$ in the case of a free field. However, we have to be cautious: The parameter $m$ is a bare parameter in the case of an interacting theory. This means that the parameter $m$ itself is not well defined and renormalized parameters need to be introduced to describe physical quantities. The definition of these parameters depends heavily on the chosen renormalization scheme. This means that the renormalized parameter $m$ might not bear much semblance, if any, to any pole of the cross section. In this regard, we can only understand the previous considerations as a motivation.\par
Nevertheless,
in the S-matrix formalism, particles are commonly related to s-poles of the scattering amplitude and particle masses are defined via the corresponding pole positions (cf. review \textit{48. Resonances} in \cite{PDG}). The notion of mass following this approach is known as \textit{pole mass}. There are two conventions for parametrizing a pole $s_R$ of the S-matrix:
\begin{align*}
s_R &= M_s^2 - iM_s\Gamma_s,\\
w_R :&= \sqrt{s_R} = M_w - i\frac{\Gamma_w}{2},
\end{align*}
whereby $M_{s/w}$ and $\Gamma_{s/w}$ define the pole mass and width of a particle, respectively, in the $s/w$-plane with $w\coloneqq\sqrt{s}$. Usually, the second parametrization is used and referred to as pole mass and decay width (cf. review \textit{48. Resonances} in \cite{PDG}). In both cases, the imaginary part is chosen to be negative which is always possible since $s_R^*$ is a pole, if $s_R$ is a pole (cf. review \textit{48. Resonances} in \cite{PDG}). The parametrizations can be converted into each other by:
\begin{align}
M_s &= \sqrt{M_w^2 - \left(\frac{\Gamma_w}{2}\right)^2},\label{eq:s-w-relation}\\
\Gamma_s &= \frac{M_w}{M_s}\Gamma_w.\nonumber
\end{align}
Stable particles do not decay, thus, $\Gamma_{s/w}$ equals zero and $s_R$ is real. In this case, both definitions of mass and width coincide. For a lot of hadrons, a similar statement applies: By expanding \autoref{eq:s-w-relation}, we see that the difference of $M_s$ and $M_w$ is proportional to $(\Gamma_w/M_w)^2$. This means that the difference between $M_s$ and $M_w$ is very small in comparison to $M_s$ and $M_w$ for $\Gamma_w\ll M_w$. 
If the experimental uncertainties of $M_s$ and $M_w$ are larger than their difference, it does not really matter which convention for the pole mass one uses. This applies to most hadrons we consider in the following sections. One should note that under certain circumstances, usually for very broad resonances, the relation between the imaginary part of the pole and the decay width breaks down (cf. review \textit{48. Resonances} in \cite{PDG}). In these cases, the decay width is related to the residue of the S-matrix pole (cf. review \textit{48. Resonances} in \cite{PDG}).\par
Most hadrons are only experimentally accessible as resonances in formation experiments or as subchannel resonances together with spectator particles in scattering experiments (cf. review \textit{48. Resonances} in \cite{PDG}). Thus, most data on hadron masses come from those experiments. To extract information about hadron masses from the scattering cross section, one needs to describe the scattering amplitude $\mathcal{M}$ in the vicinity of the resonance. For simplicity, a decomposition of $\mathcal{M}$ into a part containing the pole and a background amplitude is often performed (cf. review \textit{48. Resonances} in \cite{PDG}):
\begin{gather*}
\mathcal{M} = \mathcal{M}_{\text{pole}} + \mathcal{M}_\text{B}.
\end{gather*}
Obviously, this decomposition is not unique without further restrictions. The common choice for the background amplitude $\mathcal{M}_\text{B}$ is either to take it to be constant near the pole or to omit $\mathcal{M}_\text{B}$ entirely. For instance, $\mathcal{M}$ for a single scalar resonance with small decay width neglecting the background $\mathcal{M}_\text{B}$ can be parametrized near the pole by (cf. review \textit{48. Resonances} in \cite{PDG}):
\begin{gather*}
\mathcal{M} = \frac{\mathcal{M}_0}{s - M^2_\text{BW} + i\sqrt{s}\Gamma_\text{BW}},
\end{gather*}
where $M_\text{BW}$ and $\Gamma_\text{BW}$ are the Breit-Wigner (BW) parameters of mass and width. Assuming this parametrization is exact, the Breit-Wigner parameters are directly linked to the pole position (for instance, in the $w$-plane) via:
\begin{align*}
M_w &= \sqrt{M_\text{BW}^2 - \left(\frac{\Gamma_\text{BW}}{2}\right)^2},\\
\Gamma_w &= \Gamma_\text{BW}.
\end{align*}
If $\Gamma_\text{BW}$ is much smaller than $M_\text{BW}$, $\sqrt{s}$ can be replaced by $M_\text{BW}$:
\begin{gather*}
\mathcal{M} = \frac{\mathcal{M}_0}{s - M^2_\text{BW} + iM_\text{BW}\Gamma_\text{BW}}.
\end{gather*}
In this case, the Breit-Wigner parameters coincide with the $M_s$-$\Gamma_s$-convention.\par
In the context of scattering experiments, the expression ``small decay width'' which was used laxly in the discussion of the Breit-Wigner parametrization refers to the width in relation to other resonances and thresholds nearby. A decay width is small if the width is much smaller than the difference between pole position and other resonances/thresholds, meaning e.g. ${M_s\Gamma_s\ll|s_\text{Res./Thres.} - M^2_s|}$ (cf. review \textit{48. Resonances} in \cite{PDG}). For some strongly decaying resonances, this relation is not satisfied. The $\Delta^+(1232)$-resonance with a width of about $\SI{130}{MeV}$ (taken from \cite{PDG}), for instance, in the photoreaction $p+\gamma$ is located closely to the $N+\pi$ threshold around approximately $\SI{1080}{MeV}$.\par
In general, the Breit-Wigner parametrization is a rather crude approximation and analysis method for resonances. The interaction of spins and angular momenta, the vicinity of other resonances and thresholds, the non-negligible size of decay widths and dominant background effects, among others, call for more sophisticated analysis procedures to obtain the physical pole positions (cf. review \textit{48. Resonances} in \cite{PDG}). This task which is partially the concern of current research is quite troublesome.
As stated in review \textit{98. $N$ and $\Delta$ Resonances} in \textit{The Review of Particle Physics} of the \textit{Particle Data Group} (cf. \cite{PDG}), ``the accurate determination of pole parameters from the analysis of data on the real energy axis is not necessarily simple, or even straightforward. It requires the implementation of the correct analytic structure of the relevant (often coupled) channels.'' Some methods for the experimental determination of pole parameters are faced with ``almost unsurmountable difficulties.'' (Quotes are taken from review \textit{98. $N$ and $\Delta$ Resonances} in \cite{PDG}.)\par
Next to these more technical problems, we are also faced with a theoretical problem on a deeper level: We still have to explain how the notion of pole mass is related to the notion of hadron mass we employed in the previous chapters. For the state formalism, this means to formulate a relation between the eigenvalues of the Hamilton operator and the poles of the scattering matrix. As it stands now, I am unable to give a complete and rigorous answer to this question. However, one might formulate the wrong question, if one asks how to relate the eigenvalues of the Hamilton operator to the poles of the scattering matrix: Historically, the \text{SU}(3)-structure present in the hadronic sector was first recognized for the positions of resonances related to hadrons, i.e., for the notion of pole mass. In this sense, we should start by assuming approximate \text{SU}(3)-symmetry for the pole masses instead of the Lagrangian, if we want to follow empirical observations very closely. We do not need to consider the relation between the eigenvalues of the Hamilton operator and the poles of the scattering matrix in this case. Obviously, we are faced with different complications in such an approach like, for instance, the description of electromagnetic contributions and heavy quark symmetry. Nevertheless, we operate on the assumption 
that the mass formulae and relations from \autoref{chap:GMO_formula} and \autoref{chap:mass_relations} apply to pole masses for the time being.\par
The difficulties to define and determine the mass of a particle should be kept in mind when discussing the results of mass relations on experimental data and their predictions.

\section{Verification of Mass Relations}\label{sec:mass_testing}

We want to test the mass relations from \autoref{chap:mass_relations} for both linear and quadratic hadron masses in this section. Our goal is to verify that the mass relations given in \autoref{chap:mass_relations} apply within their range of validity, i.e., are correct up to the given corrections. Moreover, we want to back up our claims from \autoref{sec:EFT+H_Pert}: In this section, we stated that, in general, both linear and quadratic mass relations apply to both baryons and mesons. If they do not, we can argue that the distinction we have to make does not originate from the distinction in baryons and mesons.\par
There are multiple ways to check the mass formulae and relations from \autoref{chap:GMO_formula} and \autoref{chap:mass_relations} and to decide whether and in which cases linear mass formulae and relations are preferable to their quadratic counterparts or vice versa. One method is to determine the free parameters in the linear and quadratic mass formulae by fitting the formulae to the measured hadron masses using the $\chi^2$-method. One could then compare $\chi^2$ of the linear and quadratic mass formulae to find which mass formula is favored by the observed masses. Performing such a fit would also give an estimation for the free parameters in the mass formulae. However, there are a lot of drawbacks to such an approach. First off, the mass formulae from \autoref{chap:GMO_formula} pick up different corrections of different scale. However, a fit is not sensitive to every correction and dominated by the largest one. Moreover, the number of hadron masses in a multiplet or in a pair of multiplets exceeds the number of undetermined parameters in the corresponding mass formulae by at most 6. Considering that some multiplets are incomplete, i.e., that some experimental values for the hadron masses are missing, we lack a sufficient number of data points for the fit-method to be meaningful. For the same reasons, a statistical analysis is, in general, not advisable for testing the implications of the theory of global flavor symmetry breaking and, in particular, the state formalism.\par
Another more refined method
is to determine how ``good'' or ``bad'' the mass relations from \autoref{chap:mass_relations} are satisfied for linear and quadratic hadron masses. In contrast to the previous method, different mass relations are sensitive to different corrections. Additionally, the method of checking mass relations give us more control over the particles involved in the analysis: If the mass of a hadron in a multiplet is not measured, we might be able to find mass relations that do not involve this hadron. But in order to apply this method, we have to elaborate on what we mean by ``determine how `good' or `bad' the mass relations from \autoref{chap:mass_relations} are satisfied for linear and quadratic hadron masses''. Let us start this discussion by clarifying what we mean by ``linear'' and ``quadratic'' mass relations: We can write all mass relations from \autoref{chap:mass_relations} as:
\begin{gather*}
\sum_x k_x m_x = 0,
\end{gather*}
where $x$ denotes a hadron, the sum runs over all hadrons in a multiplet or in a pair of charm and bottom multiplets, $k_x$ is a rational number, and $m_x$ is the mass of the hadron $x$. This form can be achieved by, for instance, subtracting the right-hand side from the left-hand side of a mass relation from the previous chapter. Since such a mass relation is a linear combination of hadron masses, we call it linear mass relation. For every linear mass relation, there is a quadratic version of this mass relation which is obtained by replacing the hadron masses with the squared hadron masses:
\begin{gather*}
\sum_x k_x m^2_x = 0.
\end{gather*}
This means that we can describe linear and quadratic mass relations by:
\begin{gather*}
\sum_x k_x m^n_x = 0,
\end{gather*}
where the exponent $n\in\{1,2\}$ distinguishes between the linear and the quadratic version.\par
Next, we have to figure out how we can check the linear and quadratic mass relations. One of the easiest ways to do this is to simply insert the experimentally determined values for the hadron masses $m_x$ into the left-hand side of the linear and quadratic mass relations given in the form above and to see how much the calculated values deviate from zero. However, this procedure is not helpful in any way, as it suffers from two major problems: On the one hand, the values for the linear and quadratic mass relations have different units: The value obtained from the linear mass relation has a mass dimension of 1, while the value of the quadratic version has mass dimension 2. This makes the values incomparable. On the other hand, the values given by this method are not invariant under rescaling the mass relation, i.e., under multiplying the mass relation with a constant. This is especially troublesome, since the mass relations we found in \autoref{chap:mass_relations} are to some extent arbitrary. For instance, we could have used 1/4 times \autoref{eq:Gell-Mann--Okubo} instead of \autoref{eq:Gell-Mann--Okubo} which is done sometimes (cf. \cite{Gell-Mann1961}). Of course, this implies that the previously suggested procedure itself is arbitrary.\par
We can avoid these issues by employing a more sophisticated method. For this, consider a mass relation given in the following form:
\begin{gather*}
\sum_x k_x m^n_x = 0.
\end{gather*}
If this mass relation is not trivial, there is at least one hadron $y$ for which the coefficient $k_y$ does not vanish. Solving this mass relation for $m_y$ yields:
\begin{gather}
m_y = \sqrt[n]{\sum_{x\neq y}\frac{-k_x}{k_y}m^n_x}.\label{eq:prediction}
\end{gather}
Let us now denote the experimentally determined value for $m_y$ by $m^\text{exp}$ and the value for $m_y$ calculated from the remaining hadrons using \autoref{eq:prediction} by $m^\text{cal}$. We can then compute the following expression:
\begin{gather}
m^\text{cal} - m^\text{exp} = \sqrt[n]{\sum_{x\neq y}\frac{-k_x}{k_y}m^n_x} - m_y.\label{eq:cal-exp}
\end{gather}
If the mass relation was exact, the expression would be zero. Therefore, the deviation of this expression from zero tells us how strongly the mass relation is broken. Clearly, \autoref{eq:cal-exp} is invariant under rescaling the initial mass relation: If we were to multiply the entire mass relation with a constant, $k_x$ and $k_y$ would be multiplied with that constant as well. As only their ratio enters \autoref{eq:cal-exp}, $m^\text{cal}-m^\text{exp}$ remains unchanged. Furthermore, the value obtained from \autoref{eq:cal-exp} has mass dimension 1 for both linear and quadratic mass relations allowing us to compare both values. It becomes clear that the comparison of the linear ($n=1$) values with their quadratic ($n=2$) counterpart is meaningful, if we try to find a physical interpretation for the comparison. The physical relevance of \autoref{eq:cal-exp} is quite obvious: $m^\text{cal}$ is just the mass of the hadron $y$ predicted by the mass relation and the other hadron masses, thus, $m^\text{cal} - m^\text{exp}$ simply states how accurate the prediction for the mass of the hadron $y$ is. Comparing the linear and quadratic value for $m^\text{cal} - m^\text{exp}$ then tells us which mass relation is more accurate making $m^\text{cal} - m^\text{exp}$ a useful measure for comparing linear with quadratic mass relations.\par
To employ this method, we need to say which hadron $y$ we choose for computing $m^\text{cal}-m^\text{exp}$ for an arbitrary mass relation. In principle, there is no reason to favor one hadron over another (except for mesonic octets because of octet-singlet-mixing). We only have to be cautious not to employ systematics such that the selection of the hadron $y$ itself is biased and favors linear or quadratic mass relations. This would be the case, for instance, if we chose the hadron $y$ for every mass relation such that $m^\text{cal}-m^\text{exp}$ attains the smallest-possible value for the linear mass relation. For the following calculations of $m^\text{cal}-m^\text{exp}$ (except for the mesonic octets), we choose to solve the mass relations for the highest mass appearing in the relation, i.e, choose the hadron $y$ to be the hadron with the highest mass (and $k_y\neq 0$) in the multiplet or pair of multiplets.\par
Before we dive into the comparison of linear and quadratic mass relations, we should first tend to the problem the chosen comparison method faces: Even though $m^\text{cal}-m^\text{exp}$ allows us to compare a linear mass relation with its quadratic counterpart, one might struggle to use $m^\text{cal} - m^\text{exp}$ to compare the linear (or quadratic) versions of different mass relations with each other. The reason for this is illustrated by the following example: Consider the lowest-energy baryon octet and charm sextet with $J^P = 1/2^+$. Let us compute the values of $m^\text{cal}-m^\text{exp}$ corresponding to the linear versions of \autoref{eq:Gell-Mann--Okubo} and \autoref{eq:sextet_c_GMO_equal_spacing} for these multiplets, where we solve for the highest mass in each multiplet, i.e., for $y = \Xi^-$ and $y = \Omega^0_\text{c}$ in the case of the octet and sextet, respectively. We find for the values of $m^\text{cal}-m^\text{exp}$ (cf. \autoref{eq:cal-exp}):
\begin{gather*}
3m_{\Lambda^0} + m_{\Sigma^+} + m_{\Sigma^-} - m_{\Sigma^0} - m_{p} - m_{n} - m_{\Xi^0} - m_{\Xi^-} = (26.8\pm 0.2)~\si{MeV},\\
2m_{\Xi^{\prime\, 0}_\text{c}} - m_{\Sigma^0_\text{c}} - m_{\Omega^0_\text{c}} = (9.4\pm 2.0)~\si{MeV}.
\end{gather*}
The given uncertainties follow from the uncertainties of the hadron masses via Gaussian error propagation. The values for the hadron masses and their uncertainties are taken from \cite{PDG}. We can see that in this case the value $m^\text{cal}-m^\text{exp}$ of the octet is three times larger than the corresponding value of the charm sextet. However, the observation we made depends heavily on the hadrons $y$ we solve for, since $m^\text{cal}-m^\text{exp}$ depends heavily on the hadron $y$. If we use $y = \Lambda^0$ and $y = \Omega^0_\text{c}$ for the determination of $m^\text{cal} - m^\text{exp}$ of the octet and sextet, respectively, we obtain:
\begin{gather*}
\frac{1}{3}\left(m_{\Sigma^0} + m_{p} + m_{n} + m_{\Xi^0} + m_{\Xi^-} - m_{\Sigma^+} - m_{\Sigma^-}\right) - m_{\Lambda^0} = (-8.9\pm 0.1)~\si{MeV},\\
2m_{\Xi^{\prime\, 0}_\text{c}} - m_{\Sigma^0_\text{c}} - m_{\Omega^0_\text{c}} = (9.4\pm 2.0)~\si{MeV}.
\end{gather*}
The absolute values of both $m^\text{cal} - m^\text{exp}$ now agree within their range of uncertainty and no significant difference can be observed.\par
As we have just seen, the dependence of $m^\text{cal} - m^\text{exp}$ on the hadron $y$ we solve for makes it difficult to compare $m^\text{cal} - m^\text{exp}$ for different mass relations. Indeed, other quantities were proposed to circumvent this problem. In \cite{Jenkins1995}, for instance, the \textit{experimental accuracy}\footnote{This expression is used in the paper \cite{Jenkins1995} itself.} $q$ is used as a quality measure for mass relations. The experimental accuracy $q$ of a mass relation
\begin{gather*}
\sum_x k_x m_x = 0
\end{gather*}
is in this paper defined to be:
\begin{gather*}
q\coloneqq \frac{\left|\sum_x k_x m_x\right|}{\frac{1}{2}\sum_x |k_x|m_x}.
\end{gather*}
Similar to $m^\text{cal} - m^\text{exp}$, $q$ is invariant under rescaling the mass relation, but in contrast to $m^\text{cal} - m^\text{exp}$ we do not have to single out a special hadron like $y$. One might also find appealing that $q$ is a dimensionless quantity, while $m^\text{cal} - m^\text{exp}$ has mass dimension 1. Nevertheless, we are able to relate $q$ and $m^\text{cal}-m^\text{exp}$. For $n=1$, we have:
\begin{gather*}
m^\text{cal} - m^\text{exp} = \sum_{x\neq y}\frac{-k_x}{k_y}m_x - m_y.
\end{gather*}
If the mass relation at hand only contains hadron masses from one $\text{SU}(3)$-multiplet which is not a mesonic octet, the hadron masses $m_x$ deviate very little from the average $M$ of the masses in the multiplet. Using this, we find:
\begin{gather*}
q = \frac{\left|\sum_x k_x m_x\right|}{\frac{1}{2}\sum_x |k_x|m_x} = \frac{\left|\sum_{x\neq y} \frac{-k_x}{k_y} m_x - m_y\right|}{\frac{1}{2}\left(\sum_{x\neq y} |\frac{k_x}{k_y}|m_x + m_y\right)}\approx \frac{|m^\text{cal}-m^\text{exp}|}{N\frac{M}{2}},
\end{gather*}
where $N$ is a normalization factor:
\begin{gather}
N\coloneqq \sum_{x\neq y}\left|\frac{k_x}{k_y}\right| + 1 = \sum_{x}\left|\frac{k_x}{k_y}\right|.\label{eq:normalization}
\end{gather}
The reason why we favor $m^\text{cal}-m^\text{exp}$ over $q$ in this work becomes clear when we try to extend $q$ to quadratic mass relations. The natural generalization of $q$ to both linear and quadratic mass relations is to replace the hadron masses in the original definition by the $n$-th power of the hadron masses:
\begin{gather*}
q\coloneqq \frac{\left|\sum_x k_x m^n_x\right|}{\frac{1}{2}\sum_x |k_x|m^n_x}.
\end{gather*}
This quantity $q$ is likewise invariant under rescaling of the mass relation and dimensionless. Let us try to relate $q$ to $m^\text{cal}-m^\text{exp}$ again. $m^\text{cal}-m^\text{exp}$ is for $n=2$ given by:
\begin{gather*}
m^\text{cal}-m^\text{exp} = \sqrt{\sum_{x\neq y}\frac{-k_x}{k_y}m^2_x} - m_y
\end{gather*}
Under the same assumptions as before, we find:
\begin{align*}
q &= \frac{\left|\sum_{x\neq y} \frac{-k_x}{k_y} m^2_x - m^2_y\right|}{\frac{1}{2}\left(\sum_{x\neq y} |\frac{k_x}{k_y}|m^2_x + m^2_y\right)} \approx \frac{\left|\sqrt{\sum_{x\neq y} \frac{-k_x}{k_y} m^2_x} - m_y\right|\cdot\left|\sqrt{\sum_{x\neq y} \frac{-k_x}{k_y} m^2_x} + m_y\right|}{N\frac{M^2}{2}}\\
&\approx \frac{|m^\text{cal} - m^\text{exp}|\cdot 2m_y}{N\frac{M^2}{2}}\approx 2\frac{|m^\text{cal} - m^\text{exp}|}{N\frac{M}{2}}.
\end{align*}
This allows us to sum up the relation between $q$ and $m^\text{cal}-m^\text{exp}$ for $n\in\{1,2\}$ in the following way:
\begin{gather*}
q\approx n\frac{|m^\text{cal} - m^\text{exp}|}{N\frac{M}{2}}.
\end{gather*}
Now one can see why we favor $m^\text{cal}-m^\text{exp}$ over $q$ for our purposes. If we chose $q$ as a quality measure for comparing linear with quadratic mass relations, we would find undesirable results: Since the relation between $q$ and $m^\text{cal}-m^\text{exp}$ is linear in $n$, it may occur that $q$ favors a linear mass relation over a quadratic one, even though the mass prediction of the quadratic relation is more accurate than the one of the linear relation. To prevent such unwanted results, we use $m^\text{cal}-m^\text{exp}$ instead of $q$ as a quality measure for the comparison of linear and quadratic mass relations.\par
To be able to compare the linear or quadratic versions of different mass relations with each other, we also provide the quantity $(m^\text{cal}-m^\text{exp})/N$ in the following sections. This quantity is also (mostly) independent of the choice of the hadron $y$, as $q$ is independent of that choice and $|m^\text{cal}-m^\text{exp}|/N$ and $q$ only differ by a factor of roughly $\frac{M}{2n}$. The normalization factors $N$ used for our analysis are given in \autoref{tab:normalization_factors}. If one wishes to retrieve $q$ from $(m^\text{cal}-m^\text{exp})/N$, one simply has to take the absolute value of $(m^\text{cal}-m^\text{exp})/N$ and divide by $\frac{M}{2n}$ (cf. \autoref{tab:mass_scales}).\par
\begin{table}[!htb]
\small
\centering
\caption{Normalization factors $N$ for hadronic mass relations. All mass relations aside from the ones of the mesonic octets are solved for the hadron with the highest mass, i.e., hadron $y$ is chosen to be the hadron with the highest mass. The value in brackets is only used for the mesonic octets.}
\begin{tabular}{|c|c|}
\hline
Mass relation & Normalization factor $N$\\
\hline\hline
\autoref{eq:Coleman-Glashow} & 6\\
\hline
\autoref{eq:Gell-Mann--Okubo} & 10 (10/3)\\
\hline
\autoref{eq:Coleman-Glashow_decuplet} & 6\\
\hline
\autoref{eq:Delta-} & 8\\
\hline
\autoref{eq:iso1_decuplet} & 8\\
\hline
\autoref{eq:iso2_decuplet} & 10\\
\hline
\autoref{eq:equal_spacing1_decuplet} & 4\\
\hline
\autoref{eq:equal_spacing2_decuplet} & 4\\
\hline
\autoref{eq:better_GMO_decuplet} & 11\\
\hline
\autoref{eq:bary_c_b_tri} & 4\\
\hline
\autoref{eq:meso_c_b_tri} & 4\\
\hline
\autoref{eq:sextet_c_iso_bre2} & 4\\
\hline
\autoref{eq:sextet_b_iso_bre} & 4\\
\hline
\autoref{eq:sextet_c_GMO_equal_spacing2} & 4\\
\hline
\autoref{eq:sextet_b_GMO_equal_spacing} & 4\\
\hline
\autoref{eq:sextet_c_b_sigma} & 8\\
\hline
\autoref{eq:sextet_c_b_spacing} & 4\\
\hline
\autoref{eq:sextet_c_b_very_precise} & 8\\
\hline
\autoref{eq:sextet_c_b_precise} & 8\\
\hline
\autoref{eq:sextet_c_b_sigma_b} & 5\\
\hline
\autoref{eq:sextet_c_b_xi} & 6\\
\hline
\end{tabular}
\label{tab:normalization_factors}
\end{table}
Let us now turn to the actual analysis of the mass relations. We begin with an overview over the hadronic multiplets we evaluate in this section. For our analysis, we choose to only use hadronic multiplets which satisfy three criteria:
\begin{itemize}
\item There has to be evidence for every multiplet that the assignment of hadrons to the multiplet is correct.
\item Enough hadron masses in each multiplet have to be measured to evaluate at least one mass relation.
\item The uncertainties of the hadron masses in each multiplet have to be small enough to make meaningful statements.
\end{itemize}
\newpage
\noindent These criteria leave us with 18 hadronic multiplets we use for our analysis:\par
The lowest-energy pseudoscalar and vector meson octet ($J^P = 0^-$ and $J^P = 1^-$), the lowest-energy baron octet ($J^P = 1/2^+$), the lowest-energy baryon decuplet ($J^P = 3/2^+$), the lowest-energy pair of charm and bottom baryon antitriplets ($J^P = 1/2^+$), four pairs of charm and bottom meson (anti)triplets ($J^P = 0^-$, $J^P = 1^-$, $J^P = 1^+$, and $J^P = 2^+$), and two pairs of charm and bottom baryon sextets ($J^P = 1/2^+$ and $J^P = 3/2^+$).\par
A tabular summary of this list together with the mass scales of the multiplets is given by \autoref{tab:mass_scales}.
\begin{table}[t!]
\small
\centering
\caption{Number of multiplets used for the analysis of mass relations together with a crude estimation for their mass average $M$. The last three entries in the table denote mesonic multiplets, the rest is baryonic. The charm and bottom multiplets come and are analyzed in pairs.}
\begin{tabular}{|c|c|c|}
\hline
\# & Hadronic multiplet & Mass average $M$ (roughly)\\
\hline\hline
1 & Baryon octet & \SI{1}{GeV}\\
\hline
1 & Baryon decuplet & \SI{1.4}{GeV}\\
\hline
1 & Charm baryon antitriplet & \SI{2.4}{GeV}\\
1 & Bottom baryon antitriplet & \SI{5.7}{GeV}\\
\hline
2 & Charm sextet & \SI{2.5}{GeV}\\
2 & Bottom sextet & \SI{6}{GeV}\\
\hline\hline
2 & Meson octet & \SI{0.4}{GeV} and \SI{0.8}{GeV}\\
\hline
4 & Charm meson (anti)triplet & \SI{2}{GeV} to \SI{2.5}{GeV}\\
4 & Bottom meson (anti)triplet & \SI{5.3}{GeV} to \SI{5.8}{GeV}\\
\hline
\end{tabular}
\label{tab:mass_scales}
\end{table}
The assignment of hadrons to these multiplets used for our analysis is based on the assignments and quantum numbers provided and favored by the \textit{Particle Data Group} (cf. \cite{PDG}), but is also based on the results from \cite{Faustov_HM} and \cite{Faustov_HB}. The experimentally determined values for the hadron masses and their uncertainties are taken from \cite{PDG} and from other references\footnote{\cite{Aaij2012}, \cite{Aaij2013}, \cite{Aaij2014}, \cite{Aaij2015}, \cite{Aaij2017}, \cite{Aaij2018}, \cite{Aaij2019}, \cite{Abe04}, \cite{Aubert2008}, \cite{Bernicha1995}, \cite{Ganenko1979}, \cite{Gridnev2006}, \cite{Hanstein1996}, \cite{Kato2016}, \cite{Li2018}, \cite{Lichtenberg1974}, \cite{Mohr2016}}\addtocounter{footnote}{-1}\addtocounter{Hfootnote}{-1}. If the experimental value for the mass of a hadron is given in \cite{PDG} and in one of the other references\footnotemark, the value given by \cite{PDG} is used. If \cite{PDG} gives two values for the mass of a hadron, a mass fit and a mass average (cf. \cite{PDG} for the meaning of ``mass fit'' and ``mass average''), the mass fit is used in instead of the mass average.\par
We present the results of our analysis in the form of tables (cf., for instance, \autoref{table:mass_relations_octet_decuplet} and \autoref{table:mass_relations_octet_decuplet_rescaled_2}). Each table has four columns: The first column specifies the mass relation by referencing the corresponding mass relation from \autoref{chap:mass_relations} and the hadronic multiplet or pair of charm and bottom multiplets the mass relation is applied to. The second column presents the value for $m^\text{cal}-m^\text{exp}$ or $(m^\text{cal}-m^\text{exp})/N$ together with its experimental uncertainty. This uncertainty is obtained from the experimental uncertainties\footnote{If the experimental uncertainty of the hadron mass is asymmetric, the larger uncertainty is used for the Gaussian error propagation. If the uncertainty is split up into a systematic and statistical uncertainty, the square root of the sum of the squared uncertainties is used.} of the hadron masses via Gaussian error propagation. The third column indicates whether the linear or quadratic version of the mass relation was used for the calculation of the second column by presenting the exponent $n$. In the last column, the order of the dominant correction(s) to the mass relation at hand is listed, where ``$\mathcal{O}$'' is dropped and $\varepsilon_\text{cb}\coloneqq \Lambda_\text{QCD}(1/m_\text{c} - 1/m_\text{b})$ is used for the sake of clarity.\par
First, consider the results for the baryon octet and decuplet ($J^P = 1/2^+$ and $J^P = 3/2^+$), shown in \autoref{table:mass_relations_octet_decuplet} and \autoref{table:mass_relations_octet_decuplet_rescaled_2}.
\begin{table}[t!]
\small
\centering
\caption{Linear vs. quadratic mass relations in the lowest-energy baryon octet ($J^P = 1/2^+$) and decuplet ($J^P = 3/2^+$).}
\begin{tabular}{|c||c|c|c|}
\hline
Mass relation & $m^{\text{cal}}-m^{\text{exp}}$ in \si{MeV} & Exponent & Correction(s)\\
\hline\hline
\autoref{eq:Coleman-Glashow} & $-0.06 \pm 0.23$ & 1 & $\alpha\varepsilon_8;\ \varepsilon_3\varepsilon_8$\\
baryon occtet ($J^P = 1/2^+$) & $-0.46 \pm 0.22$ & 2 & $\alpha\varepsilon_8;\ \varepsilon_3\varepsilon_8$\\
\hline
\autoref{eq:Gell-Mann--Okubo} & $26.8 \pm 0.2$ & 1 & $\varepsilon^2_8$\\
baryon occtet ($J^P = 1/2^+$) & $-30.1 \pm 0.2$ & 2 & $\varepsilon^2_8$\\
\hline
\autoref{eq:Coleman-Glashow_decuplet} & $2.7 \pm 1.7$ & 1 & $\alpha\varepsilon_8;\ \varepsilon_3\varepsilon_8$\\
decuplet ($J^P = 3/2^+$) & $2.0 \pm 1.4$ & 2 & $\alpha\varepsilon_8;\ \varepsilon_3\varepsilon_8$\\
\hline
\autoref{eq:Coleman-Glashow_decuplet} & $-1.0 \pm 2.2$ & 1 & $\alpha\varepsilon_8;\ \varepsilon_3\varepsilon_8$\\
decuplet ($J^P = 3/2^+$; pole) & $-0.9 \pm 2.0$ & 2 & $\alpha\varepsilon_8;\ \varepsilon_3\varepsilon_8$\\
\hline
\autoref{eq:iso1_decuplet} & $-8.5 \pm 3.5$ & 1 & $\alpha\varepsilon_8;\ \varepsilon_3\varepsilon_8$\\
decuplet ($J^P = 3/2^+$) & $-7.8 \pm 3.3$ & 2 & $\alpha\varepsilon_8;\ \varepsilon_3\varepsilon_8$\\
\hline
\autoref{eq:iso2_decuplet} & $-13.2 \pm 5.7$ & 1 & $\alpha\varepsilon_8;\ \varepsilon_3\varepsilon_8$\\
decuplet ($J^P = 3/2^+$) & $-12.3 \pm 5.1$ & 2 & $\alpha\varepsilon_8;\ \varepsilon_3\varepsilon_8$\\
\hline
\autoref{eq:iso2_decuplet} & $4.1 \pm 4.7$ & 1 & $\alpha\varepsilon_8;\ \varepsilon_3\varepsilon_8$\\
decuplet ($J^P = 3/2^+$; pole) & $3.1 \pm 4.2$ & 2 & $\alpha\varepsilon_8;\ \varepsilon_3\varepsilon_8$\\
\hline
\autoref{eq:equal_spacing1_decuplet} & $0.1 \pm 1.6$ & 1 & $\varepsilon^2_8$\\
decuplet ($J^P = 3/2^+$) & $-14.6 \pm 1.4$ & 2 & $\varepsilon^2_8$\\
\hline
\autoref{eq:equal_spacing1_decuplet} & $16.6 \pm 2.1$ & 1 & $\varepsilon^2_8$\\
decuplet ($J^P = 3/2^+$; pole) & $-2.1 \pm 1.9$ & 2 & $\varepsilon^2_8$\\
\hline
\autoref{eq:equal_spacing2_decuplet} & $10.3 \pm 1.3$ & 1 & $\varepsilon^2_8$\\
decuplet ($J^P = 3/2^+$) & $-2.7 \pm 1.2$ & 2 & $\varepsilon^2_8$\\
\hline
\autoref{eq:equal_spacing2_decuplet} & $13.4 \pm 2.4$ & 1 & $\varepsilon^2_8$\\
decuplet ($J^P = 3/2^+$; pole) & $-0.3 \pm 2.2$ & 2 & $\varepsilon^2_8$\\
\hline
\autoref{eq:better_GMO_decuplet} & $1.0 \pm 2.6$ & 1 & $\varepsilon^3_8;\ \alpha\varepsilon_8;\ \varepsilon_3\varepsilon_8$\\
decuplet ($J^P = 3/2^+$) & $3.2 \pm 2.0$ & 2 & $\varepsilon^3_8;\ \alpha\varepsilon_8;\ \varepsilon_3\varepsilon_8$\\
\hline
\autoref{eq:better_GMO_decuplet} & $-0.6 \pm 3.3$ & 1 & $\varepsilon^3_8;\ \alpha\varepsilon_8;\ \varepsilon_3\varepsilon_8$\\
decuplet ($J^P = 3/2^+$; pole) & $2.9 \pm 2.7$ & 2 & $\varepsilon^3_8;\ \alpha\varepsilon_8;\ \varepsilon_3\varepsilon_8$\\
\hline
\end{tabular}
\label{table:mass_relations_octet_decuplet}
\end{table}
We can directly see that both the linear and the quadratic version of the Coleman-Glashow mass relation (\autoref{eq:Coleman-Glashow}) in the baryon octet are outstandingly well satisfied. Both versions are violated by a value of $m^\text{cal}-m^\text{exp}$ below \SI{1}{MeV}. With a mass scale of roughly \SI{1}{GeV} for the baryon octet, the precision of their predictions is better than 0.1\%. Even though the value of $m^\text{cal}-m^\text{exp}$ for the quadratic version is about 8 times larger than the corresponding linear value, the difference between both values is not significant in consideration of their uncertainties, as the values agree within two $2\sigma$. Judging from \autoref{tab:exp_param}, the dominant corrections to the Coleman-Glashow relation should be in the order of \SI{1}{MeV}. This seems plausible, although it is difficult to make a definite statement, since the uncertainties are also of that order.\par
The GMO mass relation in the baryon octet (\autoref{eq:Gell-Mann--Okubo}) exhibits a different behavior. As expected, both the linear as well as the quadratic GMO mass relation are much stronger violated than the Coleman-Glashow relation. Using \autoref{tab:exp_param} again, we can estimate that the dominant correction to the GMO mass relation should be in the order of \SI{10}{MeV}. With values between \SI{25}{MeV} and \SI{30}{MeV} for $m^\text{cal}-m^\text{exp}$, this seems to be reasonable. This time, however, the difference of \SI{3}{MeV} to \SI{4}{MeV} between the linear and quadratic version of the mass relation is significant. One might be tempted to deduce from this that only the linear GMO mass relation applies to the baryon octet. Nonetheless, we only expect the GMO mass relation to be satisfied up to a correction that is in the order of \SI{10}{MeV}. In consideration of such a correction, a difference of \SI{3}{MeV} to \SI{4}{MeV} does not make the quadratic GMO relation inapplicable, even though the prediction of the linear mass relation is more precise.\par
To understand the results of the baryon decuplet, we have to note some important remarks. First off, the uncertainties of the baryon masses in the decuplet are much larger than the ones of the octet. This makes it more difficult to obtain meaningful results from the available data. Secondly, the mass of the $\Delta^-$-baryon is not measured yet which means that we cannot check \autoref{eq:Delta-}. Nevertheless, this mass relation can be used for the determination of the $\Delta^-$-mass (cf. \autoref{sec:mass_predictions}). Lastly, we have to note that the $\Delta$- and $\Sigma^\ast$-resonances are not only rather broad compared to the other baryons in the decuplet and to the octet baryons, but also in the sense of \autoref{sec:polemass} (cf. ``small decay width''). This is problematic for two reasons: On the one hand, there are two conventions $M_s$ and $M_w$ for defining the pole mass (cf. \autoref{sec:polemass}). We have seen that the mass conventions differ by terms in the order of $(\Gamma_{w}/M_{w})^2$. As we are not able to say to which definition of the pole mass the mass relations from \autoref{chap:mass_relations} apply, we have to expect to pick up corrections of the order of $(\Gamma_{w}/M_{w})^2$. This means that the relations involving $\Delta$-masses face additional corrections in the order of 1\% of the $\Delta$-masses or roughly \SI{10}{MeV}, while the additional corrections to $\Sigma^\ast$-mass relations are in the order of \SI{1}{MeV}. On the other hand, the accurate determination of the pole positions is more challenging for broader resonances (cf. \autoref{sec:polemass}). In particular, the parameters of a Breit-Wigner fit to a resonance do not represent the pole position accurately. However, there are still ways to retrieve the pole mass from the resonance (cf. \cite{Bernicha1995}, \cite{Hanstein1996}, and \cite{Lichtenberg1974}; the pole masses for the decuplet are taken from these papers as well). Nonetheless, the pole mass is not known for every baryon in the decuplet. The pole masses for $\Delta^-$, $\Sigma^{\ast\, 0}$, and $\Omega^-$ are missing, even though we can take the Breit-Wigner mass of the $\Omega^-$-baryon to be its pole mass, as the $\Omega^-$-resonance is rather narrow. In order to use the available data to full extent, we present $m^\text{cal}-m^\text{exp}$ and $(m^\text{cal}-m^\text{exp})/N$ for both Breit-Wigner and pole masses in the decuplet. To distinguish between relations using Breit-Wigner masses and relations using pole masses in \autoref{table:mass_relations_octet_decuplet} and \autoref{table:mass_relations_octet_decuplet_rescaled_2}, we indicate the pole mass relations with ``pole''.\par
Considering the values of $m^\text{cal}-m^\text{exp}$ for the baryon decuplet, we observe that the pole mass relations are -- aside from two exceptions, namely the linear versions of \autoref{eq:equal_spacing1_decuplet} and \autoref{eq:equal_spacing2_decuplet} -- better satisfied than their Breit-Wigner counterpart. Since we operate on the assumption that the mass relations apply to the pole masses (cf. \autoref{sec:polemass}), this matches our expectation. Moreover, we can see that every quadratic mass prediction -- aside from \autoref{eq:equal_spacing1_decuplet} involving Breit-Wigner masses -- is at least as precise as its linear counterpart within the range of uncertainty. We even find that the equal spacing rules (\autoref{eq:equal_spacing1_decuplet} and \autoref{eq:equal_spacing2_decuplet}) are -- aside from one exception\footnote{\autoref{eq:equal_spacing1_decuplet} involving Breit-Wigner masses. The observation that \autoref{eq:equal_spacing1_decuplet} seems to be an exception may be contributed to the fact that it is the only equal spacing rule involving the problematic $\Delta$-resonance.} -- significantly more precise than their linear counterparts. Furthermore, we observe that the linear and quadratic values of $m^\text{cal}-m^\text{exp}$ for the decuplet mass relations whose dominant correction is in the order of $\varepsilon^2_8$ fluctuate between roughly \SI{1}{MeV} and \SI{15}{MeV}. As corrections associated with $\varepsilon^2_8$ should be in the order of \SI{10}{MeV}, it seems plausible that the dominant correction to these mass relations is indeed\footnote{One might argue that the dominant correction to these mass relations appears to be smaller than $\varepsilon^2_8$, judging from the experimental data, but the experimental data lacks the precision to make a definite statement.} $\varepsilon^2_8$. However, we have to be cautious with our observations for the decuplet because of the problems explained above.\par
\begin{table}[!htb]
\small
\centering
\caption{Linear vs. quadratic mass relations (normalized) in the lowest-energy baryon octet ($J^P = 1/2^+$) and decuplet ($J^P = 3/2^+$).}
\begin{tabular}{|c||c|c|c|}
\hline
Mass relation & $(m^{\text{cal}}-m^{\text{exp}})/N$ in \si{MeV} & Exponent & Correction(s)\\
\hline\hline
\autoref{eq:Coleman-Glashow} & $-0.01 \pm 0.04$ & 1 & $\alpha\varepsilon_8;\ \varepsilon_3\varepsilon_8$\\
baryon occtet ($J^P = 1/2^+$) & $-0.08 \pm 0.04$ & 2 & $\alpha\varepsilon_8;\ \varepsilon_3\varepsilon_8$\\
\hline
\autoref{eq:Gell-Mann--Okubo} & $2.68 \pm 0.02$ & 1 & $\varepsilon^2_8$\\
baryon occtet ($J^P = 1/2^+$) & $-3.01 \pm 0.02$ & 2 & $\varepsilon^2_8$\\
\hline
\autoref{eq:Coleman-Glashow_decuplet} & $0.5 \pm 0.3$ & 1 & $\alpha\varepsilon_8;\ \varepsilon_3\varepsilon_8$\\
decuplet ($J^P = 3/2^+$) & $0.3 \pm 0.2$ & 2 & $\alpha\varepsilon_8;\ \varepsilon_3\varepsilon_8$\\
\hline
\autoref{eq:Coleman-Glashow_decuplet} & $-0.2 \pm 0.4$ & 1 & $\alpha\varepsilon_8;\ \varepsilon_3\varepsilon_8$\\
decuplet ($J^P = 3/2^+$; pole) & $-0.2 \pm 0.3$ & 2 & $\alpha\varepsilon_8;\ \varepsilon_3\varepsilon_8$\\
\hline
\autoref{eq:iso1_decuplet} & $-1.1 \pm 0.4$ & 1 & $\alpha\varepsilon_8;\ \varepsilon_3\varepsilon_8$\\
decuplet ($J^P = 3/2^+$) & $-1.0 \pm 0.4$ & 2 & $\alpha\varepsilon_8;\ \varepsilon_3\varepsilon_8$\\
\hline
\autoref{eq:iso2_decuplet} & $-1.3 \pm 0.6$ & 1 & $\alpha\varepsilon_8;\ \varepsilon_3\varepsilon_8$\\
decuplet ($J^P = 3/2^+$) & $-1.2 \pm 0.5$ & 2 & $\alpha\varepsilon_8;\ \varepsilon_3\varepsilon_8$\\
\hline
\autoref{eq:iso2_decuplet} & $0.4 \pm 0.5$ & 1 & $\alpha\varepsilon_8;\ \varepsilon_3\varepsilon_8$\\
decuplet ($J^P = 3/2^+$; pole) & $0.3 \pm 0.4$ & 2 & $\alpha\varepsilon_8;\ \varepsilon_3\varepsilon_8$\\
\hline
\autoref{eq:equal_spacing1_decuplet} & $0.02 \pm 0.41$ & 1 & $\varepsilon^2_8$\\
decuplet ($J^P = 3/2^+$) & $-3.66 \pm 0.35$ & 2 & $\varepsilon^2_8$\\
\hline
\autoref{eq:equal_spacing1_decuplet} & $4.15 \pm 0.51$ & 1 & $\varepsilon^2_8$\\
decuplet ($J^P = 3/2^+$; pole) & $-0.54 \pm 0.46$ & 2 & $\varepsilon^2_8$\\
\hline
\autoref{eq:equal_spacing2_decuplet} & $2.59 \pm 0.33$ & 1 & $\varepsilon^2_8$\\
decuplet ($J^P = 3/2^+$) & $-0.67 \pm 0.3$ & 2 & $\varepsilon^2_8$\\
\hline
\autoref{eq:equal_spacing2_decuplet} & $3.34 \pm 0.61$ & 1 & $\varepsilon^2_8$\\
decuplet ($J^P = 3/2^+$; pole) & $-0.08 \pm 0.55$ & 2 & $\varepsilon^2_8$\\
\hline
\autoref{eq:better_GMO_decuplet} & $0.1 \pm 0.23$ & 1 & $\varepsilon^3_8;\ \alpha\varepsilon_8;\ \varepsilon_3\varepsilon_8$\\
decuplet ($J^P = 3/2^+$) & $0.3 \pm 0.18$ & 2 & $\varepsilon^3_8;\ \alpha\varepsilon_8;\ \varepsilon_3\varepsilon_8$\\
\hline
\autoref{eq:better_GMO_decuplet} & $-0.06 \pm 0.3$ & 1 & $\varepsilon^3_8;\ \alpha\varepsilon_8;\ \varepsilon_3\varepsilon_8$\\
decuplet ($J^P = 3/2^+$; pole) & $0.27 \pm 0.25$ & 2 & $\varepsilon^3_8;\ \alpha\varepsilon_8;\ \varepsilon_3\varepsilon_8$\\
\hline
\end{tabular}
\label{table:mass_relations_octet_decuplet_rescaled_2}
\end{table}
In principle, \autoref{table:mass_relations_octet_decuplet_rescaled_2} shows the same features as \autoref{table:mass_relations_octet_decuplet}, just scaled down. The only additional noteworthy aspect is that the values of $(m^\text{cal}-m^\text{exp})/N$ for mass relations whose dominant corrections are in the same order of magnitude seem to be more compatible with each other than the corresponding values of $m^\text{cal}-m^\text{exp}$, but, again, we have to be cautious with our observations for the decuplet.\par
Next, consider the charm and bottom sextets ($J^P = 1/2^+$ and $J^P = 3/2^+$). Let us start by evaluating the results for the mass relations that only involve baryons from one sextet (cf. \autoref{table:mass_relations_within_heavy_quarks} and \autoref{table:mass_relations_within_heavy_quarks_rescaled_2}). As we can see, every quadratic value of $m^\text{cal}-m^\text{exp}$ and $(m^\text{cal}-m^\text{exp})/N$ is smaller than its linear counterpart. In the case of the GMO equal spacing rules (\autoref{eq:sextet_c_GMO_equal_spacing2} and \autoref{eq:sextet_b_GMO_equal_spacing}), the difference is actually significant. This observation is particularly interesting in consideration of Feynman's distinction\footnote{With this expression, we denote the statement that baryons should be subject to linear mass relations, while mesons should satisfy quadratic mass relations.} of mass relations: Feynman's distinction predicts that sextet mass relations as baryonic mass relations have to be linear. As the experimental data clearly favor quadratic over linear mass relations, Feynman's distinction is in disagreement with the experimental data for sextets. Even though there is strong evidence that quadratic relations are more precise than linear ones for sextets, this does not necessarily mean that linear mass relations do not apply within their range of validity, i.e., do not apply up to their dominant correction. Indeed, the given results seem to indicate that both linear and quadratic mass relations apply within their range of validity, although we have to understand this statement with a grain of salt, since we can only estimate the size of the dominant corrections.\par
\begin{table}[!htb]
\small
\centering
\caption{Linear vs. quadratic mass relations within baryonic charm and bottom sextets.}
\begin{tabular}{|c||c|c|c|}
\hline
Mass relation & $m^{\text{cal}}-m^{\text{exp}}$ in \si{MeV} & Exponent & Correction(s)\\
\hline\hline
\autoref{eq:sextet_c_iso_bre2} & $0.05 \pm 0.82$ & 1 & $\alpha\varepsilon_8;\ \varepsilon_3\varepsilon_8$\\
charm sextet ($J^P = 1/2^+$) & $0.009 \pm 0.814$ & 2 & $\alpha\varepsilon_8;\ \varepsilon_3\varepsilon_8$\\
\hline
\autoref{eq:sextet_c_iso_bre2} & $0.17 \pm 2.33$ & 1 & $\alpha\varepsilon_8;\ \varepsilon_3\varepsilon_8$\\
charm sextet ($J^P = 3/2^+$) & $0.12 \pm 2.22$ & 2 & $\alpha\varepsilon_8;\ \varepsilon_3\varepsilon_8$\\
\hline
\autoref{eq:sextet_c_GMO_equal_spacing2} & $9.4 \pm 2.0$ & 1 & $\varepsilon^2_8$\\
charm sextet ($J^P = 1/2^+$) & $3.6 \pm 2.0$ & 2 & $\varepsilon^2_8$\\
\hline
\autoref{eq:sextet_c_GMO_equal_spacing2} & $8.4 \pm 2.1$ & 1 & $\varepsilon^2_8$\\
charm sextet ($J^P = 3/2^+$) & $2.5 \pm 2.0$ & 2 & $\varepsilon^2_8$\\
\hline
\autoref{eq:sextet_b_GMO_equal_spacing} & $8.3 \pm 1.7$ & 1 & $\varepsilon^2_8$\\
bottom sextet ($J^P = 1/2^+$) & $5.9 \pm 1.7$ & 2 & $\varepsilon^2_8$\\
\hline
\end{tabular}
\label{table:mass_relations_within_heavy_quarks}
\end{table}
\begin{table}[!htb]
\small
\centering
\caption{Linear vs. quadratic mass relations (normalized) within baryonic charm and bottom sextets.}
\begin{tabular}{|c||c|c|c|}
\hline
Mass relation & $(m^{\text{cal}}-m^{\text{exp}})/N$ in \si{MeV} & Exponent & Correction(s)\\
\hline\hline
\autoref{eq:sextet_c_iso_bre2} & $0.01 \pm 0.21$ & 1 & $\alpha\varepsilon_8;\ \varepsilon_3\varepsilon_8$\\
charm sextet ($J^P = 1/2^+$) & $0.002 \pm 0.203$ & 2 & $\alpha\varepsilon_8;\ \varepsilon_3\varepsilon_8$\\
\hline
\autoref{eq:sextet_c_iso_bre2} & $0.04 \pm 0.58$ & 1 & $\alpha\varepsilon_8;\ \varepsilon_3\varepsilon_8$\\
charm sextet ($J^P = 3/2^+$) & $0.03 \pm 0.56$ & 2 & $\alpha\varepsilon_8;\ \varepsilon_3\varepsilon_8$\\
\hline
\autoref{eq:sextet_c_GMO_equal_spacing2} & $2.4 \pm 0.5$ & 1 & $\varepsilon^2_8$\\
charm sextet ($J^P = 1/2^+$) & $0.9 \pm 0.5$ & 2 & $\varepsilon^2_8$\\
\hline
\autoref{eq:sextet_c_GMO_equal_spacing2} & $2.1 \pm 0.5$ & 1 & $\varepsilon^2_8$\\
charm sextet ($J^P = 3/2^+$) & $0.6 \pm 0.5$ & 2 & $\varepsilon^2_8$\\
\hline
\autoref{eq:sextet_b_GMO_equal_spacing} & $2.1 \pm 0.4$ & 1 & $\varepsilon^2_8$\\
bottom sextet ($J^P = 1/2^+$) & $1.5 \pm 0.4$ & 2 & $\varepsilon^2_8$\\
\hline
\end{tabular}
\label{table:mass_relations_within_heavy_quarks_rescaled_2}
\end{table}
It is also noteworthy that the mass relations whose dominant corrections are in the order of $\alpha\varepsilon_8$ and $\varepsilon_3\varepsilon_8$ are significantly more precise than the mass relations whose dominant correction is of the order of $\varepsilon^2_8$, matching our expectation.\par
As for the baryon octet and decuplet, \autoref{table:mass_relations_within_heavy_quarks_rescaled_2} shows very similar features to \autoref{table:mass_relations_within_heavy_quarks}. The only noteworthy difference is that the values of $(m^\text{cal}-m^\text{exp})/N$ for mass relations with the same dominant correction(s) (cf. \autoref{table:mass_relations_octet_decuplet_rescaled_2} and \autoref{table:mass_relations_within_heavy_quarks_rescaled_2}) seem to be more compatible with each other than the corresponding values of $m^\text{cal}-m^\text{exp}$ (cf. \autoref{table:mass_relations_octet_decuplet} and \autoref{table:mass_relations_within_heavy_quarks}). Before we move on, we should note that we can only check one mass relation for the bottom sextets ($J^P = 1/2^+$ and $J^P = 3/2^+$). The reason for this is that the bottom sextets are incomplete: In the bottom sextet with $J^P = 1/2^+$, the masses of $\Sigma^0_\text{b}$ and $\Xi^0_\text{b}$ are not measured yet, while the masses of $\Sigma^{\ast\, 0}_\text{b}$ and the counterpart to $\Omega^-_\text{b}$ (we will call this particle $\Omega^-_\text{b}(6070)$ in \autoref{sec:mass_predictions}) are missing for the bottom sextet with $J^P = 3/2^+$. We will use the mass relations from \autoref{chap:mass_relations} to predict the masses of these baryons in \autoref{sec:mass_predictions}.\par
So far, we have only considered mass relations that involve hadrons from only one multiplet. Thus, we now want to check mass relations that involve hadrons from a pair of charm and bottom multiplets. The results for these mass relations are displayed in \autoref{table:mass_relations_between_heavy_quarks} and \autoref{table:mass_relations_between_heavy_quarks_rescaled_2}.\par
\begin{table}[!htb]
\small
\centering
\caption{Linear vs. quadratic mass relations between pairs of hadronic charm and bottom triplets and sextets. Assignments marked with ($\ast$) are based on \cite{Faustov_HM}.}
\hspace{-0.6cm}\begin{tabular}{|c||c|c|c|}
\hline
Mass relation & $m^{\text{cal}}-m^{\text{exp}}$ in \si{MeV} & Exponent & Correction(s)\\
\hline\hline
\autoref{eq:sextet_c_b_xi} & $-0.05 \pm 1.0$ & 1 & $\alpha\varepsilon_{\text{cb}};\ \varepsilon_3\varepsilon_{\text{cb}};\ \alpha\varepsilon_8$\\
c and b sextet ($J^P = 3/2^+$) & $-0.52 \pm 0.94$ & 2 & $\alpha;\ \varepsilon_3$\\
\hline
\autoref{eq:sextet_c_b_precise} & $-1.1 \pm 2.6$ & 1 & $\varepsilon_{\text{cb}}\varepsilon^2_8$\\
c and b sextet ($J^P = 1/2^+$) & $4.3 \pm 1.9$ & 2 & $\varepsilon^2_8$\\
\hline
\autoref{eq:sextet_c_b_spacing} & $6.1 \pm 0.6$ & 1 & $\varepsilon_{\text{cb}}\varepsilon_8$\\
c and b sextet ($J^P = 1/2^+$) & $-65.3 \pm 0.4$ & 2 & $\varepsilon_8$\\
\hline
\autoref{eq:sextet_c_b_spacing} & $7.3 \pm 0.4$ & 1 & $\varepsilon_{\text{cb}}\varepsilon_8$\\
c and b sextet ($J^P = 3/2^+$) & $-64.3 \pm 0.3$ & 2 & $\varepsilon_8$\\
\hline
\autoref{eq:bary_c_b_tri} & $9.2 \pm 0.6$ & 1 & $\varepsilon_{\text{cb}}\varepsilon_8$\\
baryonic c and b triplet ($J^P = 1/2^+$) & $-96.1 \pm 0.5$ & 2 & $\varepsilon_8$\\
\hline
\autoref{eq:meso_c_b_tri} & $11.4 \pm 0.2$ & 1 & $\varepsilon_{\text{cb}}\varepsilon_8$\\
mesonic c and b triplet ($J^P = 0^-$) & $-51.5 \pm 0.2$ & 2 & $\varepsilon_8$\\
\hline
\autoref{eq:meso_c_b_tri} & $11.2 \pm 1.9$ & 1 & $\varepsilon_{\text{cb}}\varepsilon_8$\\
mesonic c and b triplet ($J^P = 1^-$) ($\ast$) & $-51.4 \pm 1.8$ & 2 & $\varepsilon_8$\\
\hline
\autoref{eq:meso_c_b_tri} & $9.3 \pm 2.7$ & 1 & $\varepsilon_{\text{cb}}\varepsilon_8$\\
mesonic c and b triplet ($J^P = 1^+$) ($\ast$) & $-54.4 \pm 1.6$ & 2 & $\varepsilon_8$\\
\hline
\autoref{eq:meso_c_b_tri} & $3.3 \pm 1.7$ & 1 & $\varepsilon_{\text{cb}}\varepsilon_8$\\
mesonic c and b triplet ($J^P = 2^+$) & $-55.0 \pm 1.0$ & 2 & $\varepsilon_8$\\
\hline
\end{tabular}
\label{table:mass_relations_between_heavy_quarks}
\end{table}
\begin{table}[!htb]
\small
\centering
\caption{Linear vs. quadratic mass relations (normalized) between pairs of hadronic charm and bottom triplets and sextets. Assignments marked with ($\ast$) are based on \cite{Faustov_HM}.}
\hspace{-1cm}\begin{tabular}{|c||c|c|c|}
\hline
Mass relation & $(m^{\text{cal}}-m^{\text{exp}})/N$ in \si{MeV} & Exponent & Correction(s)\\
\hline\hline
\autoref{eq:sextet_c_b_xi} & $-0.01 \pm 0.17$ & 1 & $\alpha\varepsilon_{\text{cb}};\ \varepsilon_3\varepsilon_{\text{cb}};\ \alpha\varepsilon_8$\\
c and b sextet ($J^P = 3/2^+$) & $-0.09 \pm 0.16$ & 2 & $\alpha;\ \varepsilon_3$\\
\hline
\autoref{eq:sextet_c_b_precise} & $-0.1 \pm 0.3$ & 1 & $\varepsilon_{\text{cb}}\varepsilon^2_8$\\
c and b sextet ($J^P = 1/2^+$) & $0.5 \pm 0.2$ & 2 & $\varepsilon^2_8$\\
\hline
\autoref{eq:sextet_c_b_spacing} & $1.5 \pm 0.1$ & 1 & $\varepsilon_{\text{cb}}\varepsilon_8$\\
c and b sextet ($J^P = 1/2^+$) & $-16.3 \pm 0.1$ & 2 & $\varepsilon_8$\\
\hline
\autoref{eq:sextet_c_b_spacing} & $1.8 \pm 0.1$ & 1 & $\varepsilon_{\text{cb}}\varepsilon_8$\\
c and b sextet ($J^P = 3/2^+$) & $-16.1 \pm 0.1$ & 2 & $\varepsilon_8$\\
\hline
\autoref{eq:bary_c_b_tri} & $2.3 \pm 0.1$ & 1 & $\varepsilon_{\text{cb}}\varepsilon_8$\\
baryonic c and b triplet ($J^P = 1/2^+$) & $-24.0 \pm 0.1$ & 2 & $\varepsilon_8$\\
\hline
\autoref{eq:meso_c_b_tri} & $2.9 \pm 0.1$ & 1 & $\varepsilon_{\text{cb}}\varepsilon_8$\\
mesonic c and b triplet ($J^P = 0^-$) & $-12.9 \pm 0.1$ & 2 & $\varepsilon_8$\\
\hline
\autoref{eq:meso_c_b_tri} & $2.8 \pm 0.5$ & 1 & $\varepsilon_{\text{cb}}\varepsilon_8$\\
mesonic c and b triplet ($J^P = 1^-$) ($\ast$) & $-12.8 \pm 0.5$ & 2 & $\varepsilon_8$\\
\hline
\autoref{eq:meso_c_b_tri} & $2.3 \pm 0.7$ & 1 & $\varepsilon_{\text{cb}}\varepsilon_8$\\
mesonic c and b triplet ($J^P = 1^+$) ($\ast$) & $-13.6 \pm 0.4$ & 2 & $\varepsilon_8$\\
\hline
\autoref{eq:meso_c_b_tri} & $0.8 \pm 0.4$ & 1 & $\varepsilon_{\text{cb}}\varepsilon_8$\\
mesonic c and b triplet ($J^P = 2^+$) & $-13.8 \pm 0.2$ & 2 & $\varepsilon_8$\\
\hline
\end{tabular}
\label{table:mass_relations_between_heavy_quarks_rescaled_2}
\end{table}
Immediately, we observe that the values of $m^\text{cal}-m^\text{exp}$ for the quadratic mass relations are all about 5 to 10 times larger than their linear counterparts. Naively, one might think that this behavior contradicts the predictions of the state formalism, however, it actually confirms the state formalism. To see this, we have to remind ourselves how we derived the mass relations connecting charm and bottom multiplets. In \autoref{sec:heavy_quark}, we identified the hadron masses with the eigenvalues of the Hamiltonian $H^5_\text{QCD}$ and determined them in a perturbative expansion:
\begin{gather*}
H^5_\text{QCD} = H^{5;\,0}_\text{QCD} + \varepsilon_3\cdot H^{8}_{\text{QCD};\, 3} + \varepsilon_8\cdot H^{8}_{\text{QCD};\, 8}.
\end{gather*}
The operators $H^{8}_{\text{QCD};\, 3}$ and $H^{8}_{\text{QCD};\, 8}$ which we treat as a perturbation are independent of the charm and bottom quark fields and masses. Therefore, only $H^{5;\, 0}_\text{QCD}$ changes under the exchange of charm and bottom quarks. We deduced from this that the first order contributions of the flavor symmetry breaking to the hadron masses are the same for pairs of charm and bottom multiplets, while the singlet mass term $m^0$ which directly corresponds to $H^{5;\, 0}_\text{QCD}$ changes. This allowed us to formulate mass relations involving hadrons from pairs of charm and bottom multiplets. However, that kind of reasoning only applies to linear mass relations. If we want to formulate quadratic mass relations, we have to consider the square of the Hamiltonian:
\begin{gather*}
\left(H^5_\text{QCD}\right)^2 = \left(H^{5;\,0}_\text{QCD}\right)^2 + \varepsilon_3\cdot H^{5;\,0}_\text{QCD} H^{8}_{\text{QCD};\, 3} + \varepsilon_8\cdot H^{5;\,0}_\text{QCD} H^{8}_{\text{QCD};\, 8} + \mathcal{O}\left(\varepsilon_i\varepsilon_j\right).
\end{gather*}
Now, the perturbation of the squared Hamilton operator varies under the exchange of charm and bottom quark. This implies that quadratic mass relations connecting charm and bottom multiplets do not apply with the same precision as linear mass relations: If the dominant corrections to a linear mass relation are suppressed by $\varepsilon_\text{cd}$ originating from heavy quark symmetry, the quadratic version of that mass relation does not have this suppression $\varepsilon_\text{cb}$, as heavy quark symmetry breaks down in the quadratic case. The dominant corrections to the quadratic mass relation are then given by the dominant corrections to the linear mass relation where we have to drop the factor $\varepsilon_\text{cb}$. This description matches the observed behavior: The values of $m^\text{cal}-m^\text{exp}$ for the quadratic mass relations are about 5 to 10 times higher than the corresponding linear values which coincides rather well with $\varepsilon^{-1}_\text{cb}$ (cf. \autoref{tab:exp_param}).\par
To this end, we can say that mass relations involving particles from different multiplets clearly favor linear over quadratic mass relations, but still agree with the state formalism. It is noteworthy that this preference of linear mass relations is reflected by baryons as well as mesons as one would expect following the state formalism.\par
Lastly, let us turn to the mesonic octets. We have pointed out multiple times throughout this thesis that the isospin singlet in a mesonic octet mixes with a meson which forms a \text{SU}(3)-singlet. While the GMO mass relation (cf. \autoref{eq:Gell-Mann--Okubo}) for mesonic octets is affected by this mixing, the Coleman-Glashow mass relation (cf. \autoref{eq:Coleman-Glashow}) is not, since it does not involve the isospin singlet. However, the Coleman-Glashow mass relation is trivially zero for mesonic octets. The reason for this is simple: For every meson in a mesonic octet, the corresponding antiparticle is contained in the same octet as well. In the weight diagram of a mesonic octet, a particle is diametrically opposed to its antiparticle. Thus, the Coleman-Glashow mass relation for mesonic octets only involves differences of particle and antiparticle masses. But since we have not observed CPT violation\footnote{The CPT invariance of most QFTs follows from the CPT theorem.} yet, we assume that particles and antiparticles have the same mass which is in agreement with all experiments so far. Hence, we find that the Coleman-Glashow mass relation is exactly zero for mesonic octets.\par
This leaves us with only the GMO mass relation for mesonic octets. As stated, this mass relation is affected by octet-singlet-mixing. Commonly, one takes the quadratic GMO mass relation for the mesonic octets to be exact in order to determine the mixing angle between the isospin singlet of the mesonic octet and the \text{SU}(3)-singlet meson. As we are aiming to compare linear and quadratic mass relations, this is not a viable approach. Nevertheless, there is still a way to analyze the mesonic octets: Prior in this section, we explained that we have to choose a hadron $y$ for whose mass we solve the mass relation at hand. If we take $y$ to be the isospin singlet for mesonic octets, the mass prediction $m^\text{cal}$ only involves particles which are not affected by mixing. Now, $m^\text{cal}$ does not predict the mass of an observable particle like $\eta$ or $\eta^\prime$, but the mass of a mixture of particles. Even though we do not know the mass of that mixture and, consequently, cannot determine the precision of the predictions $m^\text{cal}$, we still can compare the linear and quadratic versions of the mass prediction $m^\text{cal}$ with each other to determine how much the predictions $m^\text{cal}$ deviate: If the deviation of the predictions $m^\text{cal}$ is large in comparison to the meson masses, both mass relations cannot apply simultaneously, at least not with a high precision. If the predictions $m^\text{cal}$ are very similar, there is no reason to favor one mass relation over the other. In this regard, $m^\text{cal}$ is now the quantity of interest. Nevertheless, the values shown in \autoref{table:mass_relations_meson_octets} and \autoref{table:mass_relations_meson_octets_rescaled_2} are still $m^\text{cal}-m^\text{exp}$ and $(m^\text{cal}-m^\text{exp})/N$. For both linear and quadratic GMO mass relations, $m^\text{exp}$ is chosen to be the mass of $\eta$ in the case of the pseudoscalar meson octet and the mass of $\omega$ in the case of the vector meson octet. This choice corresponds to a naive picture of mesons in which we neglect mixing effects. Thus, $m^\text{cal}-m^\text{exp}$ tells us now how strongly the ``naive'' GMO mass relations are violated and, if the actual GMO mass relation is applicable, how large the mixing is. But since we have only shifted $m^\text{cal}$ by a constant, we can still take the difference of the linear and quadratic value for $m^\text{cal}-m^\text{exp}$ to obtain the corresponding difference of the values of $m^\text{cal}$.\par
\begin{table}[b!]
\small
\centering
\caption{Linear vs. quadratic mass relations within the lowest-energy pseudoscalar and vector meson octet ($J^P = 0^-$ and $J^P = 1^-$). The mass relations are solved for the mass of the isospin singlet instead of the highest mass in the octet.}
\begin{tabular}{|c||c|c|c|}
\hline
Mass relation & $m^{\text{cal}}-m^{\text{exp}}$ in \si{MeV} & Exponent & Correction\\
\hline\hline
\autoref{eq:Gell-Mann--Okubo} & $64.94 \pm 0.02$ & 1 & $\varepsilon^2_8$\\
meson occtet ($J^P = 0^-$) & $18.39 \pm 0.02$ & 2 & $\varepsilon^2_8$\\
\hline
\autoref{eq:Gell-Mann--Okubo} & $150.5 \pm 0.3$ & 1 & $\varepsilon^2_8$\\
meson occtet ($J^P = 1^-$) & $147.1 \pm 0.2$ & 2 & $\varepsilon^2_8$\\
\hline
\end{tabular}
\label{table:mass_relations_meson_octets}
\end{table}
\begin{table}[!htb]
\small
\centering
\caption{Linear vs. quadratic mass relations (normalized) within the lowest-energy pseudoscalar and vector meson octet ($J^P = 0^-$ and $J^P = 1^-$). The mass relations are solved for the mass of the isospin singlet instead of the highest mass in the octet.}
\begin{tabular}{|c||c|c|c|}
\hline
Mass relation & $(m^{\text{cal}}-m^{\text{exp}})/N$ in \si{MeV} & Exponent & Correction\\
\hline\hline
\autoref{eq:Gell-Mann--Okubo} & $19.48 \pm 0.01$ & 1 & $\varepsilon^2_8$\\
meson occtet ($J^P = 0^-$) & $5.52 \pm 0.01$ & 2 & $\varepsilon^2_8$\\
\hline
\autoref{eq:Gell-Mann--Okubo} & $45.2 \pm 0.1$ & 1 & $\varepsilon^2_8$\\
meson occtet ($J^P = 1^-$) & $44.1 \pm 0.1$ & 2 & $\varepsilon^2_8$\\
\hline
\end{tabular}
\label{table:mass_relations_meson_octets_rescaled_2}
\end{table}
\autoref{table:mass_relations_meson_octets} depicts an interesting behavior of the meson masses: All values for $m^\text{cal}-m^\text{exp}$ -- aside from the quadratic case in the pseudoscalar meson octet -- show a violation of the ``naive'' GMO mass relation about 15\% or higher with respect to their mass scale $M$ (cf. \autoref{tab:mass_scales}). Only the violation of the ``naive'' quadratic GMO mass relation for the pseudoscalar meson octet is below 5\%. Furthermore, we see that the difference between the linear and quadratic value of $m^\text{cal}-m^\text{exp}$ is about 10\% for the pseudoscalar meson octet and below 1\% for the vector meson octet (in regard to their respective mass scales $M$). We can interpret this in the following way: The vector meson octet is subject to the state formalism, as the flavor symmetry breaking in this octet is rather small and can be treated as a perturbation. Therefore, the linear as well as the quadratic GMO mass relation apply and give similar results in the vector meson octet. In the pseudoscalar meson octet, however, the state formalism breaks down, since the flavor symmetry breaking is not small anymore. Nevertheless, the quadratic GMO mass relation for the pseudoscalar meson octet can still be derived in the framework of chiral perturbation theory (cf. \cite{Scherer2011}) and, thus, applies. The linear GMO mass relation cannot be obtained as demonstrated in \autoref{sec:EFT+H_Pert} from the quadratic GMO mass relation in the pseudoscalar meson octet, because the symmetry breaking is not small. The mixing in the pseudoscalar meson octet, however, is relatively small which explains why the ``naive'' quadratic GMO mass relation applies relatively well to the pseudoscalar meson octet.\par
This interpretation is supported by the mixing angles $\theta_\text{lin}$ and $\theta_\text{quad}$ one can calculate by taking the linear or quadratic GMO mass relations to be exact, respectively (cf. Ch. 6 of \cite{Oneda1985}). For the pseudoscalar meson octet, one finds $|\theta^\text{P}_\text{lin}|\approx 23\degree$ and $|\theta^\text{P}_\text{quad}|\approx 10\degree$, while we have $|\theta^\text{V}_\text{lin}|\approx 36\degree$ and $|\theta^\text{V}_\text{quad}|\approx 39\degree$ for the vector meson octet. The mixing angles show that the mixing obtained from the linear GMO mass relation deviates a lot from the quadratic mixing in the case of the pseudoscalar meson octet, while the difference is rather small in the case of the vector meson octet. To this end, the mixing angles match our interpretation. The mixing angles of other meson octets exhibit a similar behavior as the vector meson octet (cf. Ch. 6 of \cite{Oneda1985}).\par
In conclusion, we have seen that the mass relations from \autoref{chap:mass_relations} apply within their range of validity, i.e., are correct up to their dominant correction(s). Generally, this is true for both versions of mass relations -- linear as well as quadratic. Moreover, we were able to reject Feynman's distinction of baryons and mesons into linear and quadratic mass relations for three different reasons: Firstly, every time we had to make a distinction between the two versions of mass relations we could argue that this distinction does not need to arise from the distinction of baryons and mesons into different mass relations: We had to utilize the linear version of mass relations involving both charm and bottom multiplets, as heavy quark symmetry breaks down for quadratic relations, and we had to use the quadratic GMO mass relation for the pseudoscalar meson octet, since it is predicted by chiral perturbation theory and not equivalent to its linear version because of the large symmetry breaking. Secondly, linear as well as quadratic mass relations apply within the same range of validity for a lot of multiplets. In this sense, Feynman's distinction is not necessary and, thus, artificial. Lastly, we found that the mass relations which favor the linear or quadratic version do not always match Feynman's distinction: The mass relations within baryon sextets favor the quadratic version, even though they are baryons, and the mass relations involving both charm and bottom multiplets favor the linear version for both baryons and mesons.

\section{Multiplet Assignments and Mass Predictions}\label{sec:mass_predictions}

The goal of this section is to use the mass relations from \autoref{chap:mass_relations} to make predictions for hadrons. There are two kinds of predictions we want to make: We want to assign known and measured resonances to multiplets and/or to pairs of charm and bottom multiplets and we want to make predictions for yet undetermined hadron masses.\par
Let us begin with the assignments. The idea for this is to start with an ``educated guess'' for the assignment of yet ungrouped hadrons to multiplets. In the next step, we check the assignment by examining the mass relations this assignment implies: If the mass relations are satisfied within their range of validity, the assignment is more likely to be true. If not, this is evidence for the assignment being false. To find this ``educated guess'', we use the quantum numbers provided and favored by the \textit{Particle Data Group} (cf. \cite{PDG}) and results from \cite{Faustov_HM} and \cite{Faustov_HB}. On top of that, we make use of an empirical observation: Consider the mass splittings between the isospin multiplets of hadronic (anti)triplets which are listed in \autoref{table:mass_splitting_baryon_triplet} and \autoref{table:mass_splitting_meson_triplet}. The first column of each table shows the hadrons which we assign to the same (possible) (anti)triplet, the second column displays the mass difference of the hadrons in the first column, and the third column contains the favored value for $J^P$. The assignment of the hadrons to (anti)triplets is based on results provided by \cite{PDG}, if not marked by $(\ast)$, and on \cite{Faustov_HM} and \cite{Faustov_HB}, otherwise. The hadron masses and their uncertainties are taken from the same references as in \autoref{sec:mass_testing}. Again, all uncertainties in this section are obtained via Gaussian error propagation.\par
\begin{table}[!htb]
\small
\centering
\caption{Mass splittings of baryonic charm and bottom antitriplets. Assignments marked with ($\ast$) are based on \cite{Faustov_HB}. If the mass of a particle is not measured, but the mass of a particle in the same isospin multiplet is known, this mass is used instead.}
\begin{tabular}{|c||c|c|c|}
\hline
Baryon splitting & Mass splitting in \si{MeV} & $J^P$\\
\hline\hline
$\Xi_c^+ - \Lambda_c^+$ & $181.5 \pm 0.2$ & $1/2^+$\\
\hline
$\Xi_c(2790) - \Lambda_c^+(2595)$ & $200.2 \pm 0.6$ & $1/2^-$\\
\hline
$\Xi_c(2815) - \Lambda_c^+(2625)$ & $188.6 \pm 0.3$ & $3/2^-$\\
\hline
$\Xi_c(2970) - \Lambda_c^+(2765)$ & $202.8 \pm 2.5$ & $1/2^+$ ($\ast$)\\
\hline
$\Xi_c(3055) - \Lambda_c^+(2860)$ & $199.8 \pm 6.0$ & $3/2^+$ ($\ast$)\\
\hline
$\Xi_c(3080) - \Lambda_c^+(2880)$ & $195.6 \pm 0.5$ & $5/2^+$ ($\ast$)\\
\hline
$\Xi_c(3123) - \Lambda_c^+(2940)$ & $183.3 \pm 2.0$ & $3/2^-$?\\
\hline
$\Xi_b^0 - \Lambda_b^0$ & $172.3 \pm 0.5$ & $1/2^+$\\
\hline
\end{tabular}
\label{table:mass_splitting_baryon_triplet}
\end{table}
\begin{table}[!htb]
\small
\centering
\caption{Mass splittings of mesonic charm and bottom (anti)triplets. Assignments marked with ($\ast$) are based on \cite{Faustov_HM}. If the mass of a particle is not measured, but the mass of a particle in the same isospin multiplet is known, this mass is used instead.}
\begin{tabular}{|c||c|c|c|}
\hline
Meson splitting & Mass splitting in \si{MeV} & $J^P$\\
\hline\hline
$D_s^+ - D^+$ & $98.7 \pm 0.1$ & $0^-$\\
\hline
$D_s^{\ast +} - D^{\ast +}(2010)$ & $101.9 \pm 0.4$ & $1^-$ ($\ast$)\\
\hline
$D_{s1}^+(2536) - D_1^+(2420)$ & $111.9 \pm 2.4$ & $1^+$ ($\ast$)\\
\hline
$D_{s2}^{\ast +}(2573) - D_2^{\ast +}(2460)$ & $103.7 \pm 1.5$ & $2^+$\\
\hline
$D_{s0}^{\ast +}(2317) - D_0^{\ast +}(2300)$ & $31.2 \pm 7.0$ & $0^+$\\
\hline
$D_{s1}^+(2460) - D_1^+(2430)$ & $32.5 \pm 36.0$ & $1^+$\\
\hline
$D_{s1}^{\ast +}(2700) - D_J^{\ast}(2600)$ & $85.3 \pm 12.6$ & $1^-$ ($\ast$)\\
\hline
$D_{s1}^{\ast +}(2860) - D^+(2740)$ & $122.0 \pm 29.4$ & $1^-$ ($\ast$)\\
\hline
$D_{s3}^{\ast +}(2860) - D_3^{\ast}(2750)$ & $97.0 \pm 7.8$ & $3^-$\\
\hline
$B_s^0 - B^0$ & $87.2 \pm 0.2$ & $0^-$\\
\hline
$B_s^{\ast 0} - B^{\ast}$ & $90.7 \pm 1.8$ & $1^-$\\
\hline
$B_{s1}^0(5830) - B_1^0(5721)$ & $102.6 \pm 1.3$ & $1^+$\\
\hline
$B_{s2}^{\ast 0}(5840) - B_2^{\ast 0}(5747)$ & $100.4 \pm 0.7$ & $2^+$\\
\hline
\end{tabular}
\label{table:mass_splitting_meson_triplet}
\end{table}
Considering the mass splittings listed in these tables, it seems like the mass splittings depend little on $J^P$ and are mainly determined by the type of (anti)triplet: While the mass splittings of baryonic antitriplets range from roughly \SI{170}{MeV} to \SI{200}{MeV}, the mass splittings of mesonic (anti)triplets lie mostly\footnote{The most striking exceptions are the (anti)triplets with $J^P = 0^+$ and $J^P = 1^+$. These (anti)triplets are quite odd, as they also disagree with the predictions of \cite{Faustov_HM}.} between \SI{85}{MeV} and \SI{125}{MeV}. It is also noteworthy that the mass splittings of bottom (anti)triplets are always smaller than their charm counterparts. Usually, the difference between charm and bottom mass splitting is about \SI{10}{MeV}.\par
If this behavior generalizes to all or at least many hadronic multiplets, it gives us an additional tool for assigning hadrons to multiplets. Based on that assumption, we assign the hadrons $\Lambda^+_\text{c}(2940)$ and $\Xi_\text{c}(3123)$ to the same antitriplet with $J^P = 3/2^-$\footnote{The evidence for the existence of the $\Xi_\text{c}(3123)$ particle is rather weak (cf. \cite{PDG}) which is why the value for $J^P$ of the corresponding antitriplet is marked with a question mark in \autoref{table:mass_splitting_baryon_triplet}.}. The value of $J^P$ for that multiplet is based on $J^P$ of $\Lambda^+_\text{c}(2940)$ favored by \cite{PDG}. Unfortunately, we cannot check this assignment with the help of a mass relation, since we lack the corresponding bottom baryons.\par
The mass splittings of charm and bottom sextets behave very similar to what we have observed for the hadronic (anti)triplets (cf. \autoref{table:mass_splitting_sextet}).\par
\begin{table}[t!]
\small
\centering
\caption{Mass splittings of baryonic charm and bottom sextets. Assignments marked with ($\ast$) are based on \cite{Faustov_HB}. If the mass of a particle is not measured, but the mass of a particle in the same isospin multiplet is known, this mass is used instead.}
\begin{tabular}{|c||c|c|c|}
\hline
Baryon splitting & Mass splitting in \si{MeV} & $J^P$\\
\hline\hline
$\Xi_c^{\prime 0} - \Sigma_c^0(2455)$ & $125.4 \pm 0.5$ & $1/2^+$\\
\hline
$\Xi_c^{0}(2645) - \Sigma_c^0(2520)$ & $127.9 \pm 0.3$ & $3/2^+$\\
\hline
$\Xi_c(2930) - \Sigma_c(2800)$ & $123.7 \pm 8.6$ & $1/2^-$, $3/2^-$, or $5/2^-$ ($\ast$)\\
\hline
$\Xi_b^{\prime -} - \Sigma_b^-$ & $119.4 \pm 0.3$ & $1/2^+$\\
\hline
$\Xi_b^{-}(5955) - \Sigma_b^{\ast -}$ & $120.6 \pm 0.3$ & $3/2^+$\\
\hline
$\Xi_b(6227) - \Sigma_b(6097)$ & $128.9 \pm 2.7$ & $1/2^-$, $3/2^-$, or $5/2^-$ ($\ast$)\\
\hline
$\Omega_c^{0} - \Xi_c^{\prime 0}$ & $116.0 \pm 1.8$ & $1/2^+$\\
\hline
$\Omega_c^{0}(2770) - \Xi_c^{0}(2645)$ & $119.5 \pm 2.0$ & $3/2^+$\\
\hline
$\Omega_c^{0}(3050) - \Xi_c(2930)$ & $120.5 \pm 5.0$ & $1/2^-$, $3/2^-$, or $5/2^-$ ($\ast$)\\
\hline
$\Omega_b^{-} - \Xi_b^{\prime -}$ & $111.1 \pm 1.7$ & $1/2^+$\\
\hline
\end{tabular}
\label{table:mass_splitting_sextet}
\end{table}
Based on that and \cite{Faustov_HB}, we assign the hadrons $\Sigma_\text{c}(2800)$, $\Xi_\text{c}(2930)$, $\Omega^0_\text{c}(3050)$, $\Sigma_\text{b}(6097)$, and $\Xi_\text{b}(6227)$ to the same pair of charm and bottom sextets. The values\footnote{In \cite{Faustov_HB}, three pairs of charm and bottom sextets are predicted to be very close together. This is the reason why three values for $J^P$ are given for the newly assigned pair of sextets in \autoref{table:mass_splitting_sextet}.} of $J^P$ for this pair are based on \cite{Faustov_HB}. This time, there are mass relations we can use to test the new assignment. The mass relations we can check for this assignment are satisfied to the extent we expect them to be valid (cf. \autoref{table:mass_assignments} and \autoref{table:mass_assignments_rescaled_2}), although it is hard to tell because of the large uncertainties.\par
\begin{table}[t]
\small
\centering
\caption{Linear vs. quadratic mass relations within and between the presumptive charm and bottom sextets. If the mass of a particle is not measured, but the mass of a particle in the same isospin multiplet is known, this mass is used instead.}
\begin{tabular}{|c||c|c|c|}
\hline
Mass relation & $m^{\text{cal}}-m^{\text{exp}}$ in \si{MeV} & Exponent & Correction\\
\hline\hline
\autoref{eq:sextet_c_GMO_equal_spacing2} & $3.2 \pm 12.2$ & 1 & $\varepsilon^2_8$\\
presumptive charm sextet & $-1.8 \pm 11.6$ & 2 & $\varepsilon^2_8$\\
\hline
\autoref{eq:sextet_c_b_spacing} & $-5.2 \pm 9.0$ & 1 & $\varepsilon_{\text{cb}}\varepsilon_8$\\
presumptive c and b sextet & $-71.0 \pm 4.8$ & 2 & $\varepsilon_8$\\
\hline
\end{tabular}
\label{table:mass_assignments}
\end{table}
\begin{table}[hbtp]
\small
\centering
\caption{Linear vs. quadratic mass relations (normalized) within and between the presumptive charm and bottom sextets. If the mass of a particle is not measured, but the mass of a particle in the same isospin multiplet is known, this mass is used instead.}
\begin{tabular}{|c||c|c|c|}
\hline
Mass relation & $(m^{\text{cal}}-m^{\text{exp}})/N$ in \si{MeV} & Exponent & Correction\\
\hline\hline
\autoref{eq:sextet_c_GMO_equal_spacing2} & $0.8 \pm 3.1$ & 1 & $\varepsilon^2_8$\\
presumptive charm sextet & $-0.5 \pm 2.9$ & 2 & $\varepsilon^2_8$\\
\hline
\autoref{eq:sextet_c_b_spacing} & $-1.3 \pm 2.3$ & 1 & $\varepsilon_{\text{cb}}\varepsilon_8$\\
presumptive c and b sextet & $-17.7 \pm 1.2$ & 2 & $\varepsilon_8$\\
\hline
\end{tabular}
\label{table:mass_assignments_rescaled_2}
\end{table}
\clearpage
Let us now turn to the mass predictions. How to use mass relations to make predictions for hadron masses is pretty straightforward: If we know the masses of at least a few hadrons in a multiplet, we can simply apply \autoref{eq:prediction}. This way, we can predict the masses of the hadrons which are missing to complete the multiplets from \autoref{sec:mass_testing}. Moreover, we can give an estimate for the mass of $\Omega^-_\text{b}(6350)$ which denotes the counterpart to the $\Omega^-_\text{b}$-baryon in the newly assigned pair of charm and bottom sextets. Furthermore, we can predict the masses of $\Xi^0_\text{b}(6112)$ and $\Xi^0_\text{b}(6109)$ which denote hadrons in the charm-bottom pairs $\Lambda^+_\text{c}(2595)$-$\Xi_\text{c}(2790)$-$\Lambda^0_\text{b}(5912)$ and $\Lambda^+_\text{c}(2625)$-$\Xi_\text{c}(2815)$-$\Lambda^0_\text{b}(5920)$ with $J^P = 1/2^-$ and $J^P = 3/2^-$, respectively. The predictions are displayed in \autoref{table:mass_predictions}. \autoref{table:mass_predictions} is organized similarly to the tables in \autoref{sec:mass_testing}. The uncertainties for the mass predictions $m^\text{cal}$ are obtained from the experimental uncertainties via Gaussian error propagation and do not include any theory or model error.\par
\begin{table}[!htb]
\small
\centering
\caption{Linear and quadratic mass predictions of various baryons. Quantum numbers $J^P$ marked with ($\ast$) are based on \cite{Faustov_HB} and the quantum numbers of other particles in the multiplet. If the mass of a particle is not measured, but the mass of a particle in the same isospin multiplet is known, this mass is used instead. However, this comment only applies to the calculation of $\Omega_b^-(6350)$.}
\hspace{-1.5cm}\begin{tabular}{|c||c|c|c|c|}
\hline
Hadron & $m^\text{cal}$ in \si{MeV} & Exponent & Mass relation & Correction(s)\\
\hline\hline
$\Delta^-(1232)$ & $1226.0 \pm 4.3$ & 1 & \autoref{eq:Delta-} & $\alpha\varepsilon_8;\ \varepsilon_3\varepsilon_8$\\
$J^P = 3/2^+$; BW mass & $1226.0 \pm 4.3$ & 2 & \autoref{eq:Delta-} & $\alpha\varepsilon_8;\ \varepsilon_3\varepsilon_8$\\
\hline
$\Delta^-(1232)$ & $1219.1 \pm 3.6$ & 1 & \autoref{eq:Delta-} & $\alpha\varepsilon_8;\ \varepsilon_3\varepsilon_8$\\
$J^P = 3/2^+$; pole mass & $1219.1 \pm 3.6$ & 2 & \autoref{eq:Delta-} & $\alpha\varepsilon_8;\ \varepsilon_3\varepsilon_8$\\
\hline
$\Sigma^{\ast 0}(1385)$ & $1379.2 \pm 1.3$ & 1 & \autoref{eq:iso1_decuplet} & $\alpha\varepsilon_8;\ \varepsilon_3\varepsilon_8$\\
$J^P = 3/2^+$; pole mass & $1379.4 \pm 1.2$ & 2 & \autoref{eq:iso1_decuplet} & $\alpha\varepsilon_8;\ \varepsilon_3\varepsilon_8$\\
\hline
$\Sigma_b^0$ & $5812.1 \pm 0.5$ & 1 & \autoref{eq:sextet_c_b_sigma} & $\alpha\varepsilon_\text{cb};\ \alpha\varepsilon_8$\\
$J^P = 1/2^+$ & $5812.7 \pm 0.3$ & 2 & \autoref{eq:sextet_c_b_sigma} & $\alpha$\\
\hline
$\Sigma_b^0$ & $5812.2 \pm 0.7$ & 1 & \autoref{eq:sextet_c_b_sigma_b} & $\alpha\varepsilon_\text{cb};\ \varepsilon_3\varepsilon_\text{cb};\, \alpha\varepsilon_8;\ \varepsilon_3\varepsilon_8$\\
$J^P = 1/2^+$ & $5812.7 \pm 0.4$ & 2 & \autoref{eq:sextet_c_b_sigma_b} & $\alpha;\, \varepsilon_3$\\
\hline
$\Xi_b^0$ & $5931.6 \pm 0.7$ & 1 & \autoref{eq:sextet_c_b_xi} & $\alpha\varepsilon_\text{cb};\,\varepsilon_3\varepsilon_\text{cb};\, \alpha\varepsilon_8$\\
$J^P = 1/2^+$ & $5932.1 \pm 0.4$ & 2 & \autoref{eq:sextet_c_b_xi} & $\alpha;\,\varepsilon_3$\\
\hline
$\Sigma_b^{\ast 0}$ & $5831.7 \pm 1.0$ & 1 & \autoref{eq:sextet_b_iso_bre} & $\alpha\varepsilon_8;\ \varepsilon_3\varepsilon_8$\\
$J^P = 3/2^+$ & $5831.6 \pm 1.0$ & 2 & \autoref{eq:sextet_b_iso_bre} & $\alpha\varepsilon_8;\ \varepsilon_3\varepsilon_8$\\
\hline
$\Sigma_b^{\ast 0}$ & $5831.6 \pm 2.3$ & 1 & \autoref{eq:sextet_c_b_sigma} & $\alpha\varepsilon_\text{cb};\ \alpha\varepsilon_8$\\
$J^P = 3/2^+$ & $5832.1 \pm 1.0$ & 2 & \autoref{eq:sextet_c_b_sigma} & $\alpha$\\
\hline
$\Sigma_b^{\ast 0}$ & $5831.5 \pm 2.5$ & 1 & \autoref{eq:sextet_c_b_very_precise} & $\alpha\varepsilon_8$\\
$J^P = 3/2^+$ & $5831.6 \pm 1.4$ & 2 & \autoref{eq:sextet_c_b_very_precise} & $\alpha\varepsilon_8;\ \varepsilon_3\varepsilon_8$\\
\hline
$\Sigma_b^{\ast 0}$ & $5831.8 \pm 0.4$ & 1 & \autoref{eq:sextet_c_b_sigma_b} & $\alpha\varepsilon_\text{cb};\ \varepsilon_3\varepsilon_\text{cb};\, \alpha\varepsilon_8;\ \varepsilon_3\varepsilon_8$\\
$J^P = 3/2^+$ & $5832.2 \pm 0.3$ & 2 & \autoref{eq:sextet_c_b_sigma_b} & $\alpha;\, \varepsilon_3$\\
\hline
$\Omega_b^-(6070)$ & $6075.9 \pm 0.4$ & 1 & \autoref{eq:sextet_b_GMO_equal_spacing} & $\varepsilon^2_8$\\
$J^P = 3/2^+$ & $6073.5 \pm 0.4$ & 2 & \autoref{eq:sextet_b_GMO_equal_spacing} & $\varepsilon^2_8$\\
\hline
$\Omega_b^-(6070)$ & $6067.5 \pm 2.1$ & 1 & \autoref{eq:sextet_c_b_precise} & $\varepsilon_\text{cb}\varepsilon^2_8$\\
$J^P = 3/2^+$ & $6072.4 \pm 1.0$ & 2 & \autoref{eq:sextet_c_b_precise} & $\varepsilon^2_8$\\
\hline
$\Omega_b^-(6350)$ & $6355.8 \pm 4.4$ & 1 & \autoref{eq:sextet_b_GMO_equal_spacing} & $\varepsilon^2_8$\\
$J^P = 1/2^-$, $3/2^-$, or $5/2^-$ ($\ast$) & $6353.2 \pm 4.3$ & 2 & \autoref{eq:sextet_b_GMO_equal_spacing} & $\varepsilon^2_8$\\
\hline
$\Omega_b^-(6350)$ & $6352.6 \pm 13.0$ & 1 & \autoref{eq:sextet_c_b_precise} & $\varepsilon_\text{cb}\varepsilon^2_8$\\
$J^P = 1/2^-$, $3/2^-$, or $5/2^-$ ($\ast$) & $6354.1 \pm 7.0$ & 2 & \autoref{eq:sextet_c_b_precise} & $\varepsilon^2_8$\\
\hline
$\Xi_b^0(6112)$ & $6112.4 \pm 0.6$ & 1 & \autoref{eq:bary_c_b_tri} & $\varepsilon_\text{cb}\varepsilon_8$\\
$J^P = 1/2^-$ ($\ast$) & $6002.7 \pm 0.3$ & 2 & \autoref{eq:bary_c_b_tri} & $\varepsilon_8$\\
\hline
$\Xi_b^0(6109)$ & $6108.5 \pm 0.3$ & 1 & \autoref{eq:bary_c_b_tri} & $\varepsilon_\text{cb}\varepsilon_8$\\
$J^P = 3/2^-$ ($\ast$) & $6006.0 \pm 0.2$ & 2 & \autoref{eq:bary_c_b_tri} & $\varepsilon_8$\\
\hline
\end{tabular}
\label{table:mass_predictions}
\end{table}
One might be confused that the mass of $\Xi^0_\text{b}(6112)$ is larger than the mass of $\Xi^0_\text{b}(6109)$, even though the order is reversed for the corresponding $\Lambda$-hadrons, i.e., the mass of $\Lambda^0_\text{b}(5912)$ is smaller than the mass of $\Lambda^0_\text{b}(5920)$. However, this behavior is also predicted by other authors (cf. \cite{Thakkar2017}).

\newpage
\chapter*{Summary}
\addcontentsline{toc}{chapter}{Summary}
\markboth{}{Summary}

The initial question of this work revolved around Feynman's distinction, i.e., the distinction of baryons and mesons into linear and quadratic GMO mass relations, respectively. As formulated in the introduction, we aimed to resolve the discrepancy between Feynman's distinction and the symmetry between fermion (baryon) and boson (meson) masses in a supersymmetrical world. To accomplish this task, we wanted to answer the question whether Feynman's distinction is a real physical distinction, merely artificial, or maybe even false. In the course of this thesis, we have found a clear answer to this question: While Feynman's distinction is not necessarily false, it is most certainly artificial. We have seen that this result holds true on a theoretical and experimental level.\par
From a theoretical perspective, we have considered two descriptions of hadron masses, the EFT approach and the state formalism. The EFT approach in which hadrons are identified with fields in an effective Lagrangian seemed to exhibit Feynman's distinction naturally, while Feynman's distinction did not arise in the state formalism that describes hadrons as eigenstates of the Hamilton operator. We were able to understand this difference between the EFT approach and the state formalism by considering the $\text{SU}(3)$-flavor symmetry breaking that led us to the formulation of the GMO mass relations in the first place: If the flavor symmetry breaking is small such that it can be treated as a perturbation and no heavy quark symmetry is involved, both linear and quadratic mass relations are valid to first order in flavor symmetry breaking. As the state formalism is a perturbative description of hadron masses and, thus, only applicable, if the symmetry breaking is small, every mass relation predicted by the state formalism -- omitting relations involving both charm and bottom hadrons -- is valid in both its linear and quadratic form. To this end, Feynman's distinction is artificial, because it is simply not necessary to introduce the distinction.\par
In spite of this result, we have seen that there are still multiplets where the results following from the different versions of mass relations -- linear and quadratic -- clearly differ and one version is most likely favored over the other. Examples for these multiplets include the pseudoscalar meson octet and the pairs of charm and bottom antitriplets and sextets. Nevertheless, we were able to explain the patterns exhibited by those multiplets:\par
The behavior of the pseudoscalar meson octet can be understood by considering the size of the symmetry breaking. The mesons in the pseudoscalar meson octet are not very heavy in comparison to the mass splitting induced by the flavor symmetry breaking. Hence, the symmetry breaking cannot be treated as a perturbation in this case and the different versions of the GMO mass relation in the pseudoscalar meson octet cannot be satisfied simultaneously. As the quadratic GMO mass relation of the pseudoscalar meson octet can be calculated in chiral perturbation theory (cf. \cite{Scherer2011}), the quadratic version is applicable, while the linear one is not.\par
The behavior of the pairs of charm and bottom multiplets is a direct consequence of the state formalism: Only the linear version is applicable for mass relations involving charm and bottom hadrons, because the heavy quark symmetry only holds true for linear mass relations and breaks down for quadratic ones. The interesting aspect of this result is that it applies to baryonic as well as mesonic pairs of charm and bottom multiplets.\par
We have argued in this thesis that these cases cannot be seen as a confirmation of Feynman's distinction, even though a distinction of mass relations is necessary in those cases. The found patterns simply do not match Feynman's distinction: The pseudoscalar meson octet favors the quadratic GMO mass relation over the linear one, but the heavier meson octets like the vector meson octet do not feature this preference. The baryonic pairs of charm and bottom multiplets clearly favor linear mass relations involving charm and bottom hadron masses, but so do the mesonic pairs.\par
Concerning the part of this work related to experimental data, we were able to support the presented discussion of Feynman's distinction with the most recent data on hadron masses. In the analysis of experimental data, we addressed several issues regarding the applicability of the mass relations to pole masses, the method for comparing mass relations, and the experimental data itself.\par
In addition to the discussion of Feynman's distinction, we have also incorporated the effects of isospin symmetry breaking, electromagnetic interaction, and heavy quark symmetry into the state formalism. This allowed us to obtain strongly satisfied and well known mass relations like the Coleman-Glashow relation (cf. \cite{coleman-glashow}). The validity of these mass relations could also be confirmed by the analysis of experimental data.\par
At the end of the thesis, we used the mass relations we derived to make predictions for the masses of yet undiscovered hadrons.

\newpage
\begin{appendix}
\chapter{Motivation for the State 
Formalism}\label{app:stateform}

In \autoref{sec:polemass}, the mass of a particle or resonance is defined via poles of scattering amplitudes with respect to the Mandelstam variable $s$. However, it is not clear at all what the relation of these poles and the parameters in the Lagrangian is and how to calculate the position of the poles. Therefore, in order to make any statement about hadron masses, we have to make some assumptions on the relation of the Lagrangian and the hadron masses. In \autoref{sec:EFT+H_Pert}, two approaches to this problem are given. For the first one, one assumes that the hadrons can be described as fields in an EFT and that the symmetry structure of $\mathcal{L}_\text{QCD}$ ``carries over" to the EFT Lagrangian to first order in flavor symmetry breaking.\par
The second approach, the state formalism, which is used throughout my thesis is based on three assumptions:
\begin{itemize}
 \item[1)] For every hadron $a$, there exists an eigenstate $\Ket{a}$ with $\Braket{a|a} = 1$ of the Hamilton operator $H$ from which the vacuum energy is already subtracted such that the mass $m_a$ of $a$ is given by
 \begin{gather*}
  m_a = \Bra{a}H \Ket{a}.
 \end{gather*}
 \item[2)] The subspace $V$ of the physical states which is spanned by the states $\Ket{a}$ from 1), i.e., {${V:= \overline{\text{Span}\left\{\Ket{a}\mid a\text{ hadron}\right\}}}$}, is a Hilbert space.
 \item[3)] There is a unitary representation $D^{(\rho)}:V\rightarrow V$ of \text{SU}(3) on $V$ such that the following equation holds for every $A\in\text{SU}(3)$:
 \begin{gather*}
  \Bra{a} D^{(\rho)}(A)^\dagger\circ H\left(\bar{q}_{\text{L/R}},\, q_{\text{L/R}}\right)\circ D^{(\rho)}(A) \Ket{b} = \Bra{a}H\left(\bar{q}^{\, \prime}_{\text{L/R}},\, q^\prime_{\text{L/R}}\right) \Ket{b}\ \forall\Ket{a},\Ket{b}\in V
 \end{gather*}
 where $q_{\text{L/R}}$ are the left/right-handed fields of the light quarks $(q\in\{\text{u, d, s}\})$ and $q^\prime_{\text{L/R}}\coloneqq \sum\limits_{\tilde{q}\in\{\text{u,d,s}\}}A_{q \tilde{q}}\cdot \tilde{q}_{\text{L/R}}$.
\end{itemize}
I cannot prove these assumptions in the case of $\mathcal{L}_\text{QCD}$. However, we can consider the case of a non-interacting theory to see how we can understand and motivate the assumptions 1) to 3).\par
Let us consider a theory of three free spin-$\frac{1}{2}$ particles with fields $q$ and masses $m_q$, $q\in\{\text{u, d, s}\}$:
\begin{align*}
\mathcal{L} &= \sum\limits_{q\in\{\text{u,d,s}\}}\bar{q}(i\slashed{\partial} - m_q)q\\
&= \sum\limits_{q\in\{\text{u,d,s}\}}\bar{q}_Li\slashed{\partial}q_L + \bar{q}_Ri\slashed{\partial}q_R - \bar{q}_Lm_qq_R - \bar{q}_Rm_qq_L.
\end{align*}
Following \cite{Peskin}, we can express the field $q$ as
\begin{gather*}
q(x) = \int\frac{d^3p}{(2\pi)^3}\frac{1}{\sqrt{2E_{p;q}}}\sum_{s=1,2}\left(a^s_{p;q}u_q^s(\vec p)e^{-ip\cdot x} + b^{s\dagger}_{p;q}v^s_q(\vec p)e^{ip\cdot x}\right)
\end{gather*}
with energy $p_0 \equiv E_{p;q}\coloneqq \sqrt{m_q^2 + {\vec p}^{\, 2}}$, annihilation and creation operators $a^s_{p,q}$, $a^{s\dagger}_{p;q}$, $b^s_{p,q}$, and $b^{s\dagger}_{p;q}$ of particles and antiparticles satisfying
\begin{gather*}
\left\{a^{r}_{p;q},a^{s\dagger}_{p^\prime;q^\prime}\right\} = \left\{b^{r}_{p;q},b^{s\dagger}_{p^\prime;q^\prime}\right\} = (2\pi)^3\delta^{(3)}(\vec p - {\vec p}^{\,\prime})\delta^{rs} \delta_{q q^\prime}\\
\text{+ all other anticommutators of $a$ and $b$ vanish,}
\end{gather*}
and spinors $u^s_q(\vec p)$ and $v^s_q(\vec p)$ normalized to $\bar{u}^r_q\left(\vec p\right)u^s_q\left(\vec p\right) = -\bar{v}^r_q\left(\vec p\right)v^s_q\left(\vec p\right) = 2m_q\delta^{rs}$ and satisfying the free Dirac equation.
The Hamilton operator $H$ is then given by (cf. \cite{Peskin}):
\begin{gather*}
H = \sum_{q\in\{\text{u,d,s}\}}\int d^3x\,\bar{q}\left(-i\vec\gamma\cdot\vec\nabla + m_q\right)q.
\end{gather*}
Subtracting vacuum energy, one obtains:
\begin{gather*}
H = \sum_{q\in\{\text{u,d,s}\}}\int\frac{d^3p}{(2\pi)^3}\sum_{s=1,2}E_{p;q}\left(a^{s\dagger}_{p;q}a^{s}_{p;q} + b^{s\dagger}_{p;q}b^{s}_{p;q}\right).
\end{gather*}
Let us now define the particle state $\Ket{\vec p;s,q}\coloneqq a^{s\dagger}_{p;q}\Ket{0}$ with $\Ket{0}$ being the vacuum state. The state $\Ket{\vec p;s,q}$ is an energy eigenstate with eigenvalue $E_{p;q}$. One maybe tempted to identify $\Ket{0;s,q}$ with a state from 1), however, it does not have the proper normalization as $\Braket{0;s,q|0;s,q} = (2\pi)^3\delta^{(3)}(0)$. In order to fix this, we consider wave packets now and define for $\varepsilon>0$:
\begin{gather*}
\Ket{\varepsilon;s,q}\coloneqq \frac{1}{\sqrt[4]{(2\pi)^9\varepsilon^3}}\int d^3p\ e^{-\frac{{\vec p}^{\, 2}}{4\varepsilon}}\Ket{\vec p;s,q}.
\end{gather*}
The state $\Ket{\varepsilon;s,q}$ is normalized to 1:
\begin{align*}
\Braket{\varepsilon;s,q|\varepsilon;s,q} &= \frac{1}{(2\pi)^3}(2\pi\varepsilon)^{-\frac{3}{2}}\int d^3p\int d^3p^\prime\ e^{-\frac{{\vec p}^{\, 2} + {\vec p}^{\,\prime 2}}{4\varepsilon}}\Braket{\vec p;s,q|{\vec p}^{\,\prime};s,q}\\
&= (2\pi\varepsilon)^{-\frac{3}{2}}\int d^3p\int d^3p^\prime\ e^{-\frac{{\vec p}^{\, 2} + {\vec p}^{\,\prime 2}}{4\varepsilon}}\delta^{(3)}\left(\vec p - {\vec p}^{\,\prime}\right)\\
&= (2\pi\varepsilon)^{-\frac{3}{2}}\int d^3p\ e^{-\frac{{\vec p}^{\, 2}}{2\varepsilon}}\\
&= (2\pi\varepsilon)^{-\frac{3}{2}}\left(\int\limits^{\infty}_{-\infty}dp\ e^{-\frac{p^2}{2\varepsilon}}\right)^3\\
&= (2\pi\varepsilon)^{-\frac{3}{2}}(2\pi\varepsilon)^{\frac{3}{2}} = 1.
\end{align*}
Let us calculate the energy of this state:
\begin{align*}
\Braket{\varepsilon;s,q|H|\varepsilon;s,q} &= \sum_{\tilde{q}\in\{\text{u,d,s}\}}\int\frac{d^3p}{(2\pi)^3}\sum_{r=1,2}E_{p;\tilde{q}}\Braket{\varepsilon;s,q|a^{r\dagger}_{p;\tilde{q}}a^{r}_{p;\tilde{q}}|\varepsilon;s,q}\\
&= \frac{1}{(2\pi)^6}(2\pi\varepsilon)^{-\frac{3}{2}}\int d^3p\int d^3k \int d^3k^\prime\sum_{\tilde{q}\in\{\text{u,d,s}\}}\sum_{r=1,2}E_{p;\tilde{q}}\, e^{-\frac{{\vec k}^{\, 2} + {\vec k}^{\,\prime 2}}{4\varepsilon}}\\
&\qquad\times\Braket{\vec k;s,q|a^{r\dagger}_{p;\tilde{q}}a^{r}_{p;\tilde{q}}|{\vec k}^{\,\prime};s,q}\\
&= \frac{1}{(2\pi)^6}(2\pi\varepsilon)^{-\frac{3}{2}}\int d^3p\int d^3k \int d^3k^\prime\sum_{\tilde{q}\in\{\text{u,d,s}\}}\sum_{r=1,2}E_{p;\tilde{q}}\, e^{-\frac{{\vec k}^{\, 2} + {\vec k}^{\,\prime 2}}{4\varepsilon}}\\
&\qquad\times\left\{a^{s}_{k;q},a^{r\dagger}_{p;\tilde{q}}\right\}\left\{a^{r}_{p;\tilde{q}},a^{s\dagger}_{k^\prime;q}\right\}\\
&= \frac{1}{(2\pi)^6}(2\pi\varepsilon)^{-\frac{3}{2}}\int d^3p\int d^3k \int d^3k^\prime\sum_{\tilde{q}\in\{\text{u,d,s}\}}\sum_{r=1,2}E_{p;\tilde{q}}\, e^{-\frac{{\vec k}^{\, 2} + {\vec k}^{\,\prime 2}}{4\varepsilon}}\\
&\qquad\times\left(2\pi\right)^6\,\delta^{rs}\,\delta_{q\tilde{q}}\,\delta^{(3)}\left(\vec k -\vec p\right)\delta^{(3)}\left(\vec p - {\vec k}^{\,\prime}\right)\\
&= (2\pi\varepsilon)^{-\frac{3}{2}}\int d^3p\, E_{p;q}\, e^{-\frac{{\vec p}^{\, 2}}{2\varepsilon}}\\
&= \frac{4\pi}{\left(2\pi\varepsilon\right)^{\frac{3}{2}}}\int\limits^{\infty}_0 dp\, p^2\sqrt{m^2_q + p^2}\, e^{-\frac{p^2}{2\varepsilon}}\\
&= \frac{4\pi}{\left(2\pi\varepsilon\right)^{\frac{3}{2}}}\int\limits^{\infty}_0 d \left(\sqrt{2\varepsilon}y\right)\, \left(\sqrt{2\varepsilon}y\right)^2\sqrt{m^2_q + \left(\sqrt{2\varepsilon}y\right)^2}\, e^{-y^2}\\
&= \frac{4}{\sqrt{\pi}}\int\limits^{\infty}_0 dy\, y^2\sqrt{m^2_q + 2\varepsilon y^2}\, e^{-y^2}\\
&= m_q\frac{4}{\sqrt{\pi}}\int\limits^{\infty}_0 dy\, e^{-y^2}\, y^2\, \sqrt{1 + \frac{2\varepsilon}{m^2_q}y^2}.
\end{align*}
With the help of mathematics and the definition $z\coloneqq \frac{2\varepsilon}{m^2_q}$, we can rewrite this expression in terms of the modified Bessel function $K_1$ of the second kind:
\begin{gather*}
\Braket{\varepsilon;s,q|H|\varepsilon;s,q} = m_q\frac{e^{\frac{1}{2z}}K_1\left(\frac{1}{2z}\right)}{\sqrt{\pi z}}.
\end{gather*}
If we consider the limit $\varepsilon\rightarrow 0^+$, $z$ also tends to $0$ and with $\lim\limits_{z\rightarrow 0^+}\frac{e^{\frac{1}{2z}}K_1\left(\frac{1}{2z}\right)}{\sqrt{\pi z}} = 1$ we obtain:
\begin{gather*}
\lim\limits_{\varepsilon\rightarrow 0^+}\Braket{\varepsilon;s,q|H|\varepsilon;s,q} = m_q.
\end{gather*}
Furthermore, a similar calculation shows:
\begin{align*}
\Braket{\varepsilon;s,q|H^2|\varepsilon;s,q} &= \frac{4}{\sqrt{\pi}}\int\limits^{\infty}_0 dy\, y^2\left(m^2_q + 2\varepsilon y^2\right)\, e^{-y^2}\\
&= m^2_q + 3\varepsilon.
\end{align*}
Therefore, we find:
\begin{gather*}
\lim\limits_{\varepsilon\rightarrow 0^+}\left(\Delta H^2(\varepsilon;s,q)\right) = \lim\limits_{\varepsilon\rightarrow 0^+}\left(\Braket{\varepsilon;s,q|H^2|\varepsilon;s,q} - \Braket{\varepsilon;s,q|H|\varepsilon;s,q}^2\right) = 0.
\end{gather*}
In this sense, the state $\Ket{\varepsilon;s,q}$ ``tends'' for $\varepsilon\rightarrow 0^+$ to a ``state'' whose energy expectation value is the mass $m_q$ and whose energy variance $\Delta H^2$ is zero. If the variance $\Delta A^2$ of any state for a given operator $A$ is zero, the state is an eigenstate of $A$. To this end, we can say that the normalized state $\Ket{\varepsilon;s,q}$ ``tends'' for $\varepsilon\rightarrow 0^+$ to an ``eigenstate'' of $H$ with ``eigenvalue'' $m_q$. These energy ``eigenstates'' are exactly the states we are looking for to fulfill 1). Note, however, that it is neither clear whether $\Ket{\varepsilon;s,q}$ converges in a mathematical sense for $\varepsilon\rightarrow 0^+$ nor obvious that the limit of $\Ket{\varepsilon;s,q}$ exhibits the described properties, if $\Ket{\varepsilon;s,q}$ converges. We are going to ignore this for now. Similarly, we can define anti-particle states $\Ket{\varepsilon; s, \bar{q}}$ by replacing $\Ket{\vec p; s, q}$ with $\Ket{\vec p; s, \bar{q}}\coloneqq b^{s\dagger}_{p;q}\Ket{0}$ in the definition of $\Ket{\varepsilon; s, q}$ and obtain the same results for anti-particles.\par
The subspace of these energy ``eigenstates'' which we need for 2) is then given by:
\begin{gather*}
V_\varepsilon \coloneqq \text{Span}\left\{\Ket{\varepsilon;s,q}\mid s\in\{1,2\};\ q\in\{\text{u,d,s,}\bar{\text{u}},\bar{\text{d}},\bar{\text{s}}\}\right\}.
\end{gather*}
$V_\varepsilon$ is a finite-dimensional complex vector space and, therefore, trivially a Hilbert space satisfying 2).\par
In order to satisfy 3), we have to find a representation $D^{(\rho)}:V_\varepsilon\rightarrow V_\varepsilon$ of \text{SU}(3) on $V_\varepsilon$. Let us define:
\begin{align*}
D^{(\rho)}(A)\Ket{\varepsilon;s,q}&\coloneqq \sum\limits_{p\in\{\text{u,d,s}\}} A_{pq}\Ket{\varepsilon;s,p}\quad\forall A\in\text{SU}(3)\,\forall s\in\{1,2\}\,\forall q\in\{\text{u,d,s}\},\\
D^{(\rho)}(A)\Ket{\varepsilon;s,\bar{q}}&\coloneqq \sum\limits_{p\in\{\text{u,d,s}\}} A^\ast_{pq}\Ket{\varepsilon;s,\bar{p}}\quad\forall A\in\text{SU}(3)\,\forall s\in\{1,2\}\,\forall q\in\{\text{u,d,s}\}.
\end{align*}
One can easily check that $D^{(\rho)}$ defined this way is a unitary representation. With this definition of $D^{(\rho)}$, we find:
\begin{align*}
&\Braket{\varepsilon;s,p|D^{(\rho)}(A)^\dagger\circ H\circ D^{(\rho)}(A)|\varepsilon;t,l}\\
= &\Braket{\varepsilon;s,p|D^{(\rho)}(A)^\dagger\left(\sum\limits_{q\in\{\text{u,d,s}\}}\int\frac{d^3k}{(2\pi)^3}\sum\limits_{r=1,2}E_{k;q}a^{r\dagger}_{k;q}a^r_{k;q}\right)D^{(\rho)}(A)|\varepsilon;t,l}\\
= &\left(\sum\limits_{\tilde{p}\in\{\text{u,d,s}\}} A_{\tilde{p}p} \Ket{\varepsilon;s,\tilde{p}}\right)^\dagger \left(\sum\limits_{q\in\{\text{u,d,s}\}}\int\frac{d^3k}{(2\pi)^3}\sum\limits_{r=1,2}E_{k;q}a^{r\dagger}_{k;q}a^r_{k;q}\right) \left(\sum\limits_{\tilde{l}\in\{\text{u,d,s}\}} A_{\tilde{l}l} \Ket{\varepsilon;t,\tilde{l}}\right)\\
= &\sum\limits_{\tilde{l},\tilde{p},q\in\{\text{u,d,s}\}} A^\ast_{\tilde{p}p} A_{\tilde{l}l} \int\frac{d^3k}{(2\pi)^3}\sum\limits_{r=1,2} E_{k;q}\Braket{\varepsilon;s,\tilde{p}|a^{r\dagger}_{k;q}a^r_{k;q}|\varepsilon;t,\tilde{l}}\\
= &\sum\limits_{\tilde{l},\tilde{p},q\in\{\text{u,d,s}\}} A^\ast_{\tilde{p}p} A_{\tilde{l}l} \int\frac{d^3k}{(2\pi)^3}\sum\limits_{r=1,2} E_{k;q}\ (2\pi)^3\ (2\pi\varepsilon)^{-\frac{3}{2}}\ \delta^{rs}\,\delta^{rt}\,\delta_{q\tilde{p}}\,\delta_{q\tilde{l}}\ e^{-\frac{{\vec k}^{\, 2}}{2\varepsilon}}\\
= &\sum\limits_{q\in\{\text{u,d,s}\}} A^\ast_{qp}A_{ql} \int d^3k\ \frac{e^{-\frac{{\vec k}^{\, 2}}{2\varepsilon}}}{(2\pi\varepsilon)^{\frac{3}{2}}}\ E_{k;q}\,\delta^{st},
\end{align*}
where we used the following relation:
\begin{align*}
&\Braket{\varepsilon;s,\tilde{p}|a^{r\dagger}_{k;q}a^r_{k;q}|\varepsilon;t,\tilde{l}}\\
= &\, (2\pi)^{-3}\, (2\pi\varepsilon)^{-\frac{3}{2}}\int d^3k^\prime\int d^3k^{\prime\prime} e^{-\frac{{\vec k}^{\, \prime 2} + {\vec k}^{\, \prime\prime 2}}{4\varepsilon}} \Braket{{\vec k}^{\, \prime}; s, \tilde{p}|a^{r\dagger}_{k;q}a^{r}_{k;q}|{\vec k}^{\, \prime\prime}; t, \tilde{l}}\\
= &\, (2\pi)^{-3}\, (2\pi\varepsilon)^{-\frac{3}{2}}\int d^3k^\prime\int d^3k^{\prime\prime} e^{-\frac{{\vec k}^{\, \prime 2} + {\vec k}^{\, \prime\prime 2}}{4\varepsilon}} \left\{a^{s}_{k^\prime; \tilde{p}}, a^{r\dagger}_{k;q}\right\} \left\{a^{r}_{k; q}, a^{t\dagger}_{k^{\prime\prime};\tilde{l}}\right\}\\
= &\, (2\pi)^{-3}\, (2\pi\varepsilon)^{-\frac{3}{2}}\int d^3k^\prime\int d^3k^{\prime\prime} e^{-\frac{{\vec k}^{\, \prime 2} + {\vec k}^{\, \prime\prime 2}}{4\varepsilon}} (2\pi)^6\delta^{sr}\delta^{rt}\delta_{\tilde{p}q}\delta_{q\tilde{l}}\delta^{(3)}\left({\vec k}^{\, \prime} - \vec k\right)\delta^{(3)}\left(\vec k - {\vec k}^{\, \prime\prime}\right)\\
= &\, (2\pi)^3\, (2\pi\varepsilon)^{-\frac{3}{2}}\ \delta^{rs}\,\delta^{rt}\,\delta_{q\tilde{p}}\,\delta_{q\tilde{l}}\ e^{-\frac{{\vec k}^{\, 2}}{2\varepsilon}}.
\end{align*}
This result corresponds to the left-hand side of the equation form 3). In order to calculate the right-hand side, we have to calculate $H(\bar{q}^{\, \prime}, q^\prime)$ where\linebreak $q^\prime = \sum\limits_{\tilde{l}\in\{\text{u,d,s}\}} A_{q\tilde{l}}\,  \tilde{l}$. $H(\bar{q},q)$ is the normal ordered Hamilton operator:
\begin{align*}
H(\bar{q},q) &= \mathopen{:} \sum_{q\in\{\text{u,d,s}\}}\int d^3x\,\bar{q}\left(-i\vec\gamma\cdot\vec\nabla + m_q\right)q \mathclose{:}\\
&= \sum_{q\in\{\text{u,d,s}\}}\int\frac{d^3p}{(2\pi)^3}\sum_{s=1,2}E_{p;q}\left(a^{s\dagger}_{p;q}a^{s}_{p;q} + b^{s\dagger}_{p;q}b^{s}_{p;q}\right).
\end{align*}
$H(\bar{q}^{\, \prime}, q^\prime)$ is, therefore, given by:
\begin{align*}
H(\bar{q}^{\, \prime}, q^\prime) &= \mathopen{:} \sum_{q\in\{\text{u,d,s}\}}\int d^3x\,\bar{q}^{\, \prime}\left(-i\vec\gamma\cdot\vec\nabla + m_q\right)q^\prime \mathclose{:}\\
&= \mathopen{:} \sum_{q\in\{\text{u,d,s}\}}\int d^3x\,\left(\sum\limits_{\tilde{p}\in\{\text{u,d,s}\}} A^\ast_{q\tilde{p}}\,\bar{\tilde{p}}\right)\left(-i\vec\gamma\cdot\vec\nabla + m_q\right)\left(\sum\limits_{\tilde{l}\in\{\text{u,d,s}\}} A_{q\tilde{l}}\, \tilde{l}\right) \mathclose{:}\\
&= \sum\limits_{\tilde{l},\tilde{p},q\in\{\text{u,d,s}\}} A^\ast_{q\tilde{p}} A_{q\tilde{l}}\int d^3x\ \mathopen{:}\bar{\tilde{p}}\left(-i\vec\gamma\cdot\vec\nabla + m_q\right)\tilde{l}\mathclose{:}.
\end{align*}
The Fourier decomposition of $\tilde{l}(x)$ and $\bar{\tilde{p}}(x)$ then gives:
\begin{align*}
&\left(-i\vec\gamma\cdot\vec\nabla + m_q\right)\tilde{l}(x)\\
=& \int\frac{d^3k^\prime}{(2\pi)^3}\, \frac{1}{\sqrt{2 E_{k^\prime;\tilde{l}}}}\\
&\times\sum\limits_{s=1,2}\left[a^s_{k^\prime;\tilde{l}}\left({\vec k}^{\, \prime}\cdot\vec\gamma + m_q\right)u^s_{\tilde{l}}({\vec k}^{\, \prime})e^{-ik^\prime\cdot x} + b^{s\dagger}_{k^\prime;\tilde{l}}\left(-{\vec k}^{\, \prime}\cdot\vec\gamma + m_q\right)v^s_{\tilde{l}}({\vec k}^{\, \prime})e^{ik^\prime\cdot x}\right],\\
\bar{\tilde{p}}(x) =& \int\frac{d^3k}{(2\pi)^3}\frac{1}{\sqrt{2E_{k;\tilde{p}}}}\sum_{r=1,2}\left(a^{r\dagger}_{k;\tilde{p}}\bar{u}_{\tilde{p}}^r(\vec k)e^{ik\cdot x} + b^{r}_{k;\tilde{p}}\bar{v}^r_{\tilde{p}}(\vec k)e^{-ik\cdot x}\right).
\end{align*}
Inserting this in $\mathopen{:}\bar{\tilde{p}}\left(-i\vec\gamma\cdot\vec\nabla + m_q\right)\tilde{l}\mathclose{:}$ from $H(\bar{q}^{\, \prime}, q^\prime)$ yields:
\begin{align*}
&\mathopen{:}\bar{\tilde{p}}\left(-i\vec\gamma\cdot\vec\nabla + m_q\right)\tilde{l}\mathclose{:}\\
=& \int\frac{d^3k}{(2\pi)^3}\int\frac{d^3k^\prime}{(2\pi)^3}\, \frac{1}{2\sqrt{E_{k;\tilde{p}} E_{k^\prime;\tilde{l}}}}\\
&\times\sum\limits_{r,s = 1,2}\left[a^{r\dagger}_{k;\tilde{p}}a^s_{k^\prime;\tilde{l}}\bar{u}^r_{\tilde{p}}(\vec k)\left({\vec k}^{\, \prime}\cdot\vec\gamma + m_q\right)u^s_{\tilde{l}}({\vec k}^{\, \prime}) e^{i(k-k^\prime)\cdot x}\ +\ \text{terms with }b\right].
\end{align*}
Using $\int d^3x\ e^{i(k-k^\prime)\cdot x} = (2\pi)^3 \delta^{(3)}({\vec k}^{\, \prime} - \vec k)e^{i(E_{k;\tilde{p}} - E_{k^\prime ;\tilde{l}})t}$, we obtain:
\begin{align*}
&\int d^3x\ \mathopen{:}\bar{\tilde{p}}\left(-i\vec\gamma\cdot\vec\nabla + m_q\right)\tilde{l}\mathclose{:}\\
=& \int\frac{d^3k}{(2\pi)^3}\, \frac{1}{2\sqrt{E_{k;\tilde{p}} E_{k;\tilde{l}}}}\\
&\times\sum\limits_{r,s = 1,2}\left[a^{r\dagger}_{k;\tilde{p}}a^s_{k;\tilde{l}}\bar{u}^r_{\tilde{p}}(\vec k)\left(\vec k\cdot\vec\gamma + m_q\right)u^s_{\tilde{l}}(\vec k) e^{i(E_{k;\tilde{p}} - E_{k;\tilde{l}})t}\ +\ \text{terms with }b\right].
\end{align*}
This implies that the operator $H(\bar{q}^{\, \prime}, q^\prime)$ has, in general, a non-trivial time dependence, while the left-hand side of the equation from 3) is time-independent. However, we only need assumption 3) to be true at one point in time. Therefore, we can choose $t=0$ and find for the right-hand side of the equation from 3):
\begin{align*}
& \Braket{\varepsilon;s,p|H(\bar{q}^{\, \prime}, q^\prime)|\varepsilon;t,l}\\
=& \sum\limits_{\tilde{l},\tilde{p},q\in\{\text{u,d,s}\}} A^\ast_{q\tilde{p}} A_{q\tilde{l}}\int\frac{d^3k}{(2\pi)^3}\, \frac{1}{2\sqrt{E_{k;\tilde{p}} E_{k;\tilde{l}}}}\\
&\times\sum\limits_{r,\tilde{r} = 1,2} \Braket{\varepsilon;s,p|a^{r\dagger}_{k;\tilde{p}}a^{\tilde{r}}_{k;\tilde{l}}|\varepsilon;t,l}\bar{u}^r_{\tilde{p}}(\vec k)\left(\vec k\cdot\vec\gamma + m_q\right)u^{\tilde{r}}_{\tilde{l}}(\vec k)\\
=& \sum\limits_{\tilde{l},\tilde{p},q\in\{\text{u,d,s}\}} A^\ast_{q\tilde{p}} A_{q\tilde{l}}\int\frac{d^3k}{(2\pi)^3}\, \frac{1}{2\sqrt{E_{k;\tilde{p}} E_{k;\tilde{l}}}}\\
&\times\sum\limits_{r,\tilde{r} = 1,2} (2\pi)^3\, (2\pi\varepsilon)^{-\frac{3}{2}}\ \delta^{rs}\,\delta^{\tilde{r}t}\,\delta_{p\tilde{p}}\,\delta_{l\tilde{l}}\ e^{-\frac{{\vec k}^{\, 2}}{2\varepsilon}}\bar{u}^r_{\tilde{p}}(\vec k)\left(\vec k\cdot\vec\gamma + m_q\right)u^{\tilde{r}}_{\tilde{l}}(\vec k)\\
=& \sum\limits_{q\in\{\text{u,d,s}\}} A^\ast_{qp} A_{ql}\int d^3k\, \frac{e^{-\frac{{\vec k}^{\, 2}}{2\varepsilon}}}{(2\pi\varepsilon)^{\frac{3}{2}}}\, \frac{1}{2\sqrt{E_{k;p} E_{k;l}}} \bar{u}^s_{p}(\vec k)\left(\vec k\cdot\vec\gamma + m_q\right)u^{t}_{l}(\vec k)\\
\eqqcolon& \sum\limits_{q\in\{\text{u,d,s}\}} A^\ast_{qp} A_{ql}\int d^3k\, \frac{e^{-\frac{{\vec k}^{\, 2}}{2\varepsilon}}}{(2\pi\varepsilon)^{\frac{3}{2}}}\, \tilde{E}(\vec k;q,p,l,s,t).
\end{align*}
%
%
With this, we have written the left-hand and right-hand side of the equation from 3) in a very similar form. They only differ in the integrand of the 3-momentum integration. As we found that assumption 1) is fulfilled in the case of $\varepsilon\rightarrow 0^+$, we expect assumption 3) to be valid in the same limit. To show this, we note that the $\varepsilon$-dependence of both integrands is given by multiplication with:
\begin{gather*}
\frac{e^{-\frac{{\vec k}^{\, 2}}{2\varepsilon}}}{(2\pi\varepsilon)^{\frac{3}{2}}} = \prod\limits^3_{i=1} \frac{e^{-\frac{k^2_i}{2\varepsilon}}}{\sqrt{2\pi\varepsilon}}.
\end{gather*}
This is a nascent delta function known as the heat kernel. In the limit of $\varepsilon\rightarrow 0^+$, the heat kernel can be replaced by a delta distribution of $\vec k$. With this, one obtains:
\begin{align*}
&\lim\limits_{\varepsilon\rightarrow 0^+}\left(\Braket{\varepsilon;s,p|D^{(\rho)}(A)^\dagger\circ H\circ D^{(\rho)}(A)|\varepsilon;t,l}\right)\\
=& \sum\limits_{q\in\{\text{u,d,s}\}} A^\ast_{qp}A_{ql} \int d^3k\, \delta^{(3)}(\vec k)\, E_{k;q}\,\delta^{st}\\
=& \sum\limits_{q\in\{\text{u,d,s}\}} A^\ast_{qp}A_{ql}\, m_q\,\delta^{st},
\end{align*}
\begin{align*}
&\lim\limits_{\varepsilon\rightarrow 0^+}\left(\Braket{\varepsilon;s,p|H(\bar{q}^{\, \prime}, q^\prime)|\varepsilon;t,l}\right)\\
=& \sum\limits_{q\in\{\text{u,d,s}\}} A^\ast_{qp} A_{ql}\int d^3k\, \delta^{(3)}(\vec k)\, \tilde{E}(\vec k;q,p,l,s,t)\\
=& \sum\limits_{q\in\{\text{u,d,s}\}} A^\ast_{qp} A_{ql}\, \tilde{E}(0;q,p,l,s,t)\\
=& \sum\limits_{q\in\{\text{u,d,s}\}} A^\ast_{qp} A_{ql}\, \frac{1}{2\sqrt{m_p m_l}}m_q\, \bar{u}^s_{p}(0)u^{t}_{l}(0)\\
=& \sum\limits_{q\in\{\text{u,d,s}\}} A^\ast_{qp} A_{ql}\, \frac{1}{2\sqrt{m_p m_l}}m_q\, 2\sqrt{m_p m_l}\, \delta^{st}\\
=& \sum\limits_{q\in\{\text{u,d,s}\}} A^\ast_{qp}A_{ql}\, m_q\,\delta^{st}\\
=& \lim\limits_{\varepsilon\rightarrow 0^+}\left(\Braket{\varepsilon;s,p|D^{(\rho)}(A)^\dagger\circ H\circ D^{(\rho)}(A)|\varepsilon;t,l}\right),
\end{align*}
where used that the spinors are given by (cf. \cite{Peskin})
\begin{gather*}
\bar{u}^s_p(0) = \sqrt{m_p}(\xi^{s\dagger}\ \xi^{s\dagger})\text{ and }u^{t}_{l}(0) = \sqrt{m_l}\begin{pmatrix}\xi^t\\\xi^t\end{pmatrix} \text{ with } \xi^{s\dagger}\xi^t = \delta^{st}
\end{gather*}
in the Weyl representation. In the sense of the limit $\varepsilon\rightarrow 0^+$, we have shown assumption 3) for the particle states $\Ket{\varepsilon;s,q}$. For the other states in $V_\varepsilon$, we can proceed in similar fashion to show assumption 3).\par
At this point, it is helpful to recapitulate what we have done. For $\varepsilon>0$, we have defined normalized states $\Ket{\varepsilon;s,q}$ and $\Ket{\varepsilon;s,\bar{q}}$. As $\varepsilon$ tends to $0$, these states exhibit the very same properties as states satisfying assumption 1), 2), and 3). However, as stated earlier, it is neither clear that the states converge for $\varepsilon\rightarrow 0^+$ nor obvious that the limit, if it exists, has these desired properties. To this end, we have to understand the assumptions 1), 2), and 3) in a different way. For instance, we have to understand 1) as:
\begin{itemize}
 \item[1)] For every hadron $a$, there exists a parameter $\varepsilon\geq0$ and a state $\Ket{\varepsilon;a}$ with $\Braket{\varepsilon;a|\varepsilon;a} = 1$ such that the mass $m_a$ of $a$ is given by
 \begin{align*}
 m_a + \mathcal{O}(\varepsilon) &= \Braket{\varepsilon;a|H|\varepsilon;a}\text{ with}\\
 m^2_a + \mathcal{O}(\varepsilon) &= \Braket{\varepsilon;a|H^2|\varepsilon;a}.
 \end{align*}
\end{itemize}
If $\varepsilon$ is 0, we recover the original formulation of 1). In this sense, we may view $\varepsilon$ as a parameter that states how ``strongly'' the assumptions 1), 2), and 3) are violated. In the case of the free fields, we have seen that we can choose $\varepsilon$ to be arbitrarily small. However, it is not clear that this is the case for all theories and, in particular, QCD to which we apply the assumptions 1), 2), and 3) in this work. Nevertheless, the assumptions 1), 2), and 3) are used throughout this thesis as they were presented in the beginning of this section, since these assumptions are only used to calculate relations of the hadron masses in a perturbative approach. If $\varepsilon$ is indeed greater than 0, we might pick up corrections of order $\varepsilon$ in the mass relations.\par
This is a rather interesting point, especially if we try to interpret $\varepsilon$. The $\varepsilon$ we introduced for free fields can be considered to be the momentum fluctuation as:
\begin{gather*}
\Braket{\varepsilon;s,q|\Delta P_i^2|\varepsilon;s,q} = \Braket{\varepsilon;s,q|P_i^2|\varepsilon;s,q} = \varepsilon.
\end{gather*}
Therefore, assuming 1), 2), and 3) in their original form might only be a valid approximation if the hadron at hand can be seen as a non-relativistic system.

\chapter{Clebsch-Gordan Series of $D(p,q)\otimes D(1,1)$}\label{app:Young}

As we have seen in \autoref{sec:GMO_formula}, the Clebsch-Gordan series of $\sigma\otimes 8$ for finite-dimensional multiplets $\sigma$ can be used to calculate the number of octets in $\sigma\otimes\bar{\sigma}$. Therefore, these Clebsch-Gordan series are of great importance for the derivation of the GMO mass relations. We want to compute these series now by using Young tableaux. Young tableaux were introduced in \autoref{sec:GMO_formula}. The rules on how to use Young tableaux to compute the Clebsch-Gordan series of tensor product representations are not presented in this thesis, but can commonly be found in literature, for instance in \cite{Lichtenberg} and in review \textit{46. $\text{SU}(n)$ Multiplets and Young Diagrams} in \cite{PDG}. For these calculations, we denote $8$ with $D(1,1)$ and $\sigma$ with $D(p,q)$. This notation was also introduced in \autoref{sec:GMO_formula}.\par
Before we start with the actual computations, we can derive one relation that simplifies the following calculations: Suppose that we have found the Clebsch-Gordan series of $D(p,q)\otimes D(1,1)$:
\begin{gather*}
D(p,q)\otimes D(1,1) = \bigoplus_{P,Q} n_{D(P,Q)}(D(p,q)\otimes D(1,1))\ D(P,Q),
\end{gather*}
where $n_{D(P,Q)}(D(p,q)\otimes D(1,1))$ denotes the multiplicity of $D(P,Q)$ in {${D(p,q)\otimes D(1,1)}$}. We then have:
\begin{align*}
D(q,p)\otimes D(1,1) &= \overline{D(p,q)}\otimes\overline{D(1,1)} = \overline{D(p,q)\otimes D(1,1)}\\
&= \overline{\bigoplus_{P,Q} n_{D(P,Q)}(D(p,q)\otimes D(1,1))\ D(P,Q)}\\
&= \bigoplus_{P,Q} n_{D(P,Q)}(D(p,q)\otimes D(1,1))\  \overline{D(P,Q)}\\
&= \bigoplus_{P,Q} n_{D(P,Q)}(D(p,q)\otimes D(1,1))\ D(Q,P).
\end{align*}
This means that the Clebsch-Gordan series of {${D(p,q)\otimes D(1,1)}$} directly implies the Clebsch-Gordan series of {${D(q,p)\otimes D(1,1)}$}. Let us now turn to the computations:\\\\
\textbf{Singlet (}$p=0$\textbf{ and }$q=0$\textbf{)}
\begin{align*}\ytableausetup{boxsize = 1.3em}
D(0,0)\otimes D(1,1) = \emptyset\otimes\ydiagram{2,1} = \ydiagram{2,1} = D(1,1)
\end{align*}
\newpage
\noindent\textbf{Totally symmetric multiplets (}$p\geq 1$\textbf{ and }$q=0$\textbf{ or }$p=0$\textbf{ and }$q\geq 1$\textbf{)}
\begin{align*}
D(1,0)\otimes D(1,1) &= \ydiagram{1}\otimes\ydiagram{2,1} = \ydiagram{1}\otimes\ytableaushort{aa,b} = \left(\ytableaushort{\, a}\oplus\ytableaushort{\,,a}\right)\otimes\ytableaushort{a,b}\\
&= \left(\ytableaushort{\, a,a}\oplus\ytableaushort{\, aa}\right)\otimes\ytableaushort{b}\\
&= \ytableaushort{\, a,a,b}\oplus\ytableaushort{\, a,ab}\oplus\ytableaushort{\, aa,b}\\
&= \ydiagram{1}\oplus\ydiagram{2,2}\oplus\ydiagram{3,1}\\
&= D(1,0)\oplus D(0,2)\oplus D(2,1)\\\\
D(p,0)\otimes D(1,1) &= \begin{ytableau}\, & \none[\dots] & \,\end{ytableau}\otimes\ydiagram{2,1}= \begin{ytableau}\, & \none[\dots] & \,\end{ytableau}\otimes\ytableaushort{aa,b}\\
&= \left(\begin{ytableau}\, & \none[\dots] & \, & a\end{ytableau}\oplus\begin{ytableau}\, & \, & \none[\dots] & \,\\ a\end{ytableau}\right)\otimes\ytableaushort{a,b}\\
&= \left(\begin{ytableau}\, & \none[\dots] & \, & a & a\end{ytableau}\oplus\begin{ytableau}\, & \, & \none[\dots] & \, & a\\ a\end{ytableau}\oplus\begin{ytableau}\, & \, & \, & \none[\dots] & \,\\ a & a\end{ytableau}\right)\otimes\ytableaushort{b}\\
&= \begin{ytableau}\, & \, & \none[\dots] & \, & a & a\\ b\end{ytableau}\oplus\begin{ytableau}\, & \, & \, & \none[\dots] & \, & a\\ a & b\end{ytableau}\oplus\begin{ytableau}\, & \, & \none[\dots] & \, & a\\ a\\ b\end{ytableau}\\ &\ \ \ \oplus\begin{ytableau}\, & \, & \, & \none[\dots] & \,\\ a & a \\ b\end{ytableau}\\
&= \begin{ytableau}\, & \, & \none[\dots] & \, & \, & \,\\ \,\end{ytableau}\oplus\begin{ytableau}\, & \, & \, & \none[\dots] & \, & \,\\ \, & \,\end{ytableau}\oplus\begin{ytableau}\, & \, & \none[\dots] & \, & \,\\ \,\\ \,\end{ytableau}\\ &\ \ \ \oplus\begin{ytableau}\, & \, & \, & \none[\dots] & \,\\ \, & \, \\ \,\end{ytableau}\\
&= \begin{ytableau}\, & \, & \none[\dots] & \, & \, & \,\\ \,\end{ytableau}\oplus\begin{ytableau}\, & \, & \, & \none[\dots] & \, & \,\\ \, & \,\end{ytableau}\oplus\begin{ytableau}\, & \none[\dots] & \, & \,\end{ytableau}\\ &\ \ \ \oplus\begin{ytableau}\, & \, & \none[\dots] & \,\\ \,\end{ytableau}\\
&= D(p+1,1)\oplus D(p-1,2)\oplus D(p,0)\oplus D(p-2,1),\quad p\geq 2
\end{align*}
The Clebsch-Gordan series of $D(0,q)\otimes D(1,1)$ follows directly from the Clebsch-Gordan series of $D(p,0)\otimes D(1,1)$:
\begin{align*}
D(0,1)\otimes D(1,1) &= D(0,1)\oplus D(2,0)\oplus D(1,2)\\
D(0,q)\otimes D(1,1) &= D(1,q+1)\oplus D(2,q-1)\oplus D(0,q)\oplus D(1,q-2),\quad q\geq 2
\end{align*}\\
\textbf{Remaining multiplets (}$p \geq 1$\textbf{ and }$q \geq 1$\textbf{)}
\begin{align*}
D(1,1)\otimes D(1,1) &= \ydiagram{2,1}\otimes\ydiagram{2,1} = \ydiagram{2,1}\otimes\ytableaushort{aa,b}\\
&= \left(\ytableaushort{\,\,a,\,}\oplus\ytableaushort{\,\,,\,a}\oplus\ytableaushort{\,\,,\,,a}\right)\otimes\ytableaushort{a,b}\\
&= \left(\ytableaushort{\,\,aa,\,}\oplus\ytableaushort{\,\,a,\,a}\oplus\ytableaushort{\,\,a,\,,a}\oplus\ytableaushort{\,\,,\,a,a}\right)\otimes\ytableaushort{b}\\
&= \ytableaushort{\,\,aa,\,b}\oplus\ytableaushort{\,\,aa,\,,b}\oplus\ytableaushort{\,\,a,\,ab}\oplus\ytableaushort{\,\,a,\,a,b}\oplus\ytableaushort{\,\,a,\,b,a}\\ &\ \ \ \oplus\ytableaushort{\,\,,\,a,ab}\\
&= \ydiagram{4,2}\oplus\ydiagram{4,1,1}\oplus\ydiagram{3,3}\oplus\ydiagram{3,2,1}\oplus\ydiagram{3,2,1}\\ &\ \ \ \oplus\ydiagram{2,2,2}\\
&= \emptyset\oplus\ydiagram{2,1}\oplus\ydiagram{2,1}\oplus\ydiagram{3}\oplus\ydiagram{3,3}\oplus\ydiagram{4,2}\\
&= D(0,0)\oplus D(1,1)\oplus D(1,1)\oplus D(3,0)\oplus D(0,3)\oplus D(2,2)
\end{align*}
\begin{align*}
D(p,1)\otimes D(1,1) &= \begin{ytableau}\, & \, & \none[\dots] & \,\\ \,\end{ytableau}\otimes\ydiagram{2,1} = \begin{ytableau}\, & \, & \none[\dots] & \,\\ \,\end{ytableau}\otimes\ytableaushort{aa,b}\\
&= \left(\begin{ytableau}\, & \, & \none[\dots] & \, & a\\ \,\end{ytableau}\oplus\begin{ytableau}\, & \, & \, & \none[\dots] & \,\\ \, & a\end{ytableau}\oplus\begin{ytableau}\, & \, & \none[\dots] & \,\\ \,\\ a\end{ytableau}\right)\otimes\ytableaushort{a,b}\\
&= \left(\begin{ytableau}\, & \, & \none[\dots] & \, & a & a\\ \,\end{ytableau}\oplus\begin{ytableau}\, & \, & \, & \none[\dots] & \, & a\\ \, & a\end{ytableau}\oplus\begin{ytableau}\, & \, & \none[\dots] & \, & a\\ \,\\ a\end{ytableau}\right.\\ &\ \ \ \left.\oplus\begin{ytableau}\, & \, & \, & \, & \none[\dots] & \,\\ \, & a & a\end{ytableau}\oplus\begin{ytableau}\, & \, & \, & \none[\dots] & \,\\ \, & a\\ a\end{ytableau}\right)\otimes\ytableaushort{b}\\
&= \begin{ytableau}\, & \, & \, & \none[\dots] & \, & a & a\\ \, & b\end{ytableau}\oplus\begin{ytableau}\, & \, & \none[\dots] & \, & a & a\\ \,\\ b\end{ytableau}\oplus\begin{ytableau}\, & \, & \, & \, & \none[\dots] & \, & a\\ \, & a & b\end{ytableau}\\ &\ \ \ \oplus\begin{ytableau}\, & \, & \, & \none[\dots] & \, & a\\ \, & a\\ b\end{ytableau}\oplus\begin{ytableau}\, & \, & \, & \none[\dots] & \, & a\\ \, & b\\ a\end{ytableau}\oplus\begin{ytableau}\, & \, & \, & \, & \none[\dots] & \,\\ \, & a & a\\ b\end{ytableau}\\ &\ \ \ \oplus\begin{ytableau}\, & \, & \, & \none[\dots] & \,\\ \, & a\\ a & b\end{ytableau}\\
&= \begin{ytableau}\, & \, & \, & \none[\dots] & \, & \, & \,\\ \, & \,\end{ytableau}\oplus\begin{ytableau}\, & \, & \none[\dots] & \, & \, & \,\\ \,\\ \,\end{ytableau}\oplus\begin{ytableau}\, & \, & \, & \, & \none[\dots] & \, & \,\\ \, & \, & \,\end{ytableau}\\ &\ \ \ \oplus\begin{ytableau}\, & \, & \, & \none[\dots] & \, & \,\\ \, & \,\\ \,\end{ytableau}\oplus\begin{ytableau}\, & \, & \, & \none[\dots] & \, & \,\\ \, & \,\\ \,\end{ytableau}\oplus\begin{ytableau}\, & \, & \, & \, & \none[\dots] & \,\\ \, & \, & \,\\ \,\end{ytableau}\\ &\ \ \ \oplus\begin{ytableau}\, & \, & \, & \none[\dots] & \,\\ \, & \,\\ \, & \,\end{ytableau}\\
&= \begin{ytableau}\, & \, & \, & \none[\dots] & \, & \, & \,\\ \, & \,\end{ytableau}\oplus\begin{ytableau}\, & \none[\dots] & \, & \, & \,\end{ytableau}\oplus\begin{ytableau}\, & \, & \, & \, & \none[\dots] & \, & \,\\ \, & \, & \,\end{ytableau}\\ &\ \ \ \oplus\begin{ytableau}\, & \, & \none[\dots] & \, & \,\\ \,\end{ytableau}\oplus\begin{ytableau}\, & \, & \none[\dots] & \, & \,\\ \,\end{ytableau}\oplus\begin{ytableau}\, & \, & \, & \none[\dots] & \,\\ \, & \,\end{ytableau}\oplus\begin{ytableau}\, & \none[\dots] & \,\end{ytableau}\\
&= D(p+1,2)\oplus D(p+2,0)\oplus D(p-1,3)\\ &\ \ \ \oplus D(p,1)\oplus D(p,1)\oplus D(p-2,2)\oplus D(p-1,0),\quad p\geq 2\\
\Rightarrow D(1,q)\otimes D(1,1) &= D(2,q+1)\oplus D(0,q+2)\oplus D(3,q-1)\\ &\ \ \ \oplus D(1,q)\oplus D(1,q)\oplus D(2,q-2)\oplus D(0,q-1),\quad q\geq 2
\end{align*}
\begin{align*}
\hspace{-1.5cm} D(p,q)\otimes D(1,1) &= \begin{ytableau}\, & \none[\dots] & \, & \, & \none[\dots] & \,\\ \, & \none[\dots] & \,\end{ytableau}\otimes\ydiagram{2,1} = \begin{ytableau}\, & \none[\dots] & \, & \, & \none[\dots] & \,\\ \, & \none[\dots] & \,\end{ytableau}\otimes\ytableaushort{aa,b}\\
&= \left(\begin{ytableau}\, & \none[\dots] & \, & \, & \none[\dots] & \, & a\\ \, & \none[\dots] & \,\end{ytableau}\oplus\begin{ytableau}\, & \none[\dots] & \, & \, & \, & \none[\dots] & \,\\ \, & \none[\dots] & \, & a\end{ytableau}\oplus\begin{ytableau}\, & \, & \none[\dots] & \, & \, & \none[\dots] & \,\\ \, & \, & \none[\dots] & \,\\ a\end{ytableau}\right)\\ &\ \ \ \otimes\ytableaushort{a,b}\\
&= \left(\begin{ytableau}\, & \none[\dots] & \, & \, & \none[\dots] & \, & a & a\\ \, & \none[\dots] & \,\end{ytableau}\oplus\begin{ytableau}\, & \none[\dots] & \, & \, & \, & \none[\dots] & \, & a\\ \, & \none[\dots] & \, & a\end{ytableau}\oplus\begin{ytableau}\, & \, & \none[\dots] & \, & \, & \none[\dots] & \, & a\\ \, & \, & \none[\dots] & \,\\ a\end{ytableau}\right.\\ &\ \ \ \left.\oplus\begin{ytableau}\, & \none[\dots] & \, & \, & \, & \, & \none[\dots] & \,\\ \, & \none[\dots] & \, & a & a\end{ytableau}\oplus\begin{ytableau}\, & \, & \none[\dots] & \, & \, & \, & \none[\dots] & \,\\ \, & \, & \none[\dots] & \, & a\\ a\end{ytableau}\oplus\begin{ytableau}\, & \, & \, & \none[\dots] & \, & \, & \none[\dots] & \,\\ \, & \, & \, & \none[\dots] & \,\\ a & a\end{ytableau}\right)\\ &\ \ \ \otimes\ytableaushort{b}\\
&= \begin{ytableau}\, & \none[\dots] & \, & \, & \, & \none[\dots] & \, & a & a\\ \, & \none[\dots] & \, & b\end{ytableau}\oplus\begin{ytableau}\, & \, & \none[\dots] & \, & \, & \none[\dots] & \, & a & a\\ \, & \, & \none[\dots] & \,\\ b\end{ytableau}\\ &\ \ \ \oplus\begin{ytableau}\, & \none[\dots] & \, & \, & \, & \, & \none[\dots] & \, & a\\ \, & \none[\dots] & \, & a & b\end{ytableau}\oplus\begin{ytableau}\, & \, & \none[\dots] & \, & \, & \, & \none[\dots] & \, & a\\ \, & \, & \none[\dots] & \, & a\\ b\end{ytableau}\\ &\ \ \ \oplus\begin{ytableau}\, & \, & \none[\dots] & \, & \, & \, & \none[\dots] & \, & a\\ \, & \, & \none[\dots] & \, & b\\ a\end{ytableau}\oplus\begin{ytableau}\, & \, & \, & \none[\dots] & \, & \, & \none[\dots] & \, & a\\ \, & \, & \, & \none[\dots] & \,\\ a & b\end{ytableau}\\ &\ \ \ \oplus\begin{ytableau}\, & \, & \none[\dots] & \, & \, & \, & \, & \none[\dots] & \,\\ \, & \, & \none[\dots] & \, & a & a\\ b\end{ytableau}\oplus\begin{ytableau}\, & \, & \, & \none[\dots] & \, & \, & \, & \none[\dots] & \,\\ \, & \, & \, & \none[\dots] & \, & a\\ a & b\end{ytableau}
\end{align*}
\begin{align*}
\Rightarrow D(p,q)\otimes D(1,1) &= \begin{ytableau}\, & \none[\dots] & \, & \, & \, & \none[\dots] & \, & \, & \,\\ \, & \none[\dots] & \, & \,\end{ytableau}\oplus\begin{ytableau}\, & \, & \none[\dots] & \, & \, & \none[\dots] & \, & \, & \,\\ \, & \, & \none[\dots] & \,\\ \,\end{ytableau}\\ &\ \ \ \oplus\begin{ytableau}\, & \none[\dots] & \, & \, & \, & \, & \none[\dots] & \, & \,\\ \, & \none[\dots] & \, & \, & \,\end{ytableau}\oplus\begin{ytableau}\, & \, & \none[\dots] & \, & \, & \, & \none[\dots] & \, & \,\\ \, & \, & \none[\dots] & \, & \,\\ \,\end{ytableau}\\ &\ \ \ \oplus\begin{ytableau}\, & \, & \none[\dots] & \, & \, & \, & \none[\dots] & \, & \,\\ \, & \, & \none[\dots] & \, & \,\\ \,\end{ytableau}\oplus\begin{ytableau}\, & \, & \, & \none[\dots] & \, & \, & \none[\dots] & \, & \,\\ \, & \, & \, & \none[\dots] & \,\\ \, & \,\end{ytableau}\\ &\ \ \ \oplus\begin{ytableau}\, & \, & \none[\dots] & \, & \, & \, & \, & \none[\dots] & \,\\ \, & \, & \none[\dots] & \, & \, & \,\\ \,\end{ytableau}\oplus\begin{ytableau}\, & \, & \, & \none[\dots] & \, & \, & \, & \none[\dots] & \,\\ \, & \, & \, & \none[\dots] & \, & \,\\ \, & \,\end{ytableau}\\
&= \begin{ytableau}\, & \none[\dots] & \, & \, & \, & \none[\dots] & \, & \, & \,\\ \, & \none[\dots] & \, & \,\end{ytableau}\oplus\begin{ytableau}\, & \none[\dots] & \, & \, & \none[\dots] & \, & \, & \,\\ \, & \none[\dots] & \,\end{ytableau}\\ &\ \ \ \oplus\begin{ytableau}\, & \none[\dots] & \, & \, & \, & \, & \none[\dots] & \, & \,\\ \, & \none[\dots] & \, & \, & \,\end{ytableau}\oplus\begin{ytableau}\, & \none[\dots] & \, & \, & \, & \none[\dots] & \, & \,\\ \, & \none[\dots] & \, & \,\end{ytableau}\\ &\ \ \ \oplus\begin{ytableau}\, & \none[\dots] & \, & \, & \, & \none[\dots] & \, & \,\\ \, & \none[\dots] & \, & \,\end{ytableau}\oplus\begin{ytableau}\, & \none[\dots] & \, & \, & \none[\dots] & \, & \,\\ \, & \none[\dots] & \,\end{ytableau}\\ &\ \ \ \oplus\begin{ytableau}\, & \none[\dots] & \, & \, & \, & \, & \none[\dots] & \,\\ \, & \none[\dots] & \, & \, & \,\end{ytableau}\oplus\begin{ytableau}\, & \none[\dots] & \, & \, & \, & \none[\dots] & \,\\ \, & \none[\dots] & \, & \,\end{ytableau}\\
&= D(p+1,q+1)\oplus D(p+2,q-1)\oplus D(p-1,q+2)\oplus D(p,q)\oplus D(p,q)\\ &\ \ \ \oplus D(p+1,q-2)\oplus D(p-2,q+1)\oplus D(p-1,q-1),\quad p,q\geq 2
\end{align*}

\chapter{Properties of $F^{(\sigma\otimes\bar{\sigma})}_k$ and $D^{(\sigma\otimes\bar{\sigma})}_k$}\label{app:F-and_D-symbols}

In \autoref{sec:GMO_formula}, we found that there are at most two octets in the representation $\sigma\otimes\bar{\sigma}$ of \text{SU}(3) for multiplets $\sigma$ of \text{SU}(3). We parametrized these octets by introducing the matrices $F^{(\sigma\otimes\bar{\sigma})}_k$ and $D^{(\sigma\otimes\bar{\sigma})}_k$ as elements of the vector space the representation $\sigma\otimes\bar{\sigma}$ acts on. However, we skipped some technical and rather lengthy calculations in the process of showing that these matrices actually form the octets in $\sigma\otimes\bar{\sigma}$. In this part of the appendix, we want to perform these computations. Some of the calculations can be found or follow computations in \cite{Lichtenberg}. Throughout this part, $\sigma$ always denotes a non-trivial, complex, finite-dimensional, and unitary multiplet of \text{SU}(3). Before we begin, let us quickly summarize the most important results and formulae given in \autoref{sec:GMO_formula}:
\begin{gather}
D^{(\sigma\otimes\bar{\sigma})}(A)\circ DD^{(\sigma)}\vert_\mathbb{1} = DD^{(\sigma)}\vert_\mathbb{1}\circ D^{(8)}(A)\quad\forall A\in\text{SU}(3),\label{eq:sigma8}\\
F^{(\sigma\otimes\bar{\sigma})}_k\coloneqq DD^{(\sigma)}\vert_\mathbb{1}(F_k)\quad\forall k\in\{1,\ldots, 8\}\label{eq:Fdef},\\
D^{(\sigma\otimes\bar{\sigma})}_k\coloneqq \frac{2}{3}\sum^{8}_{l,m=1}d_{klm} F^{(\sigma\otimes\bar{\sigma})}_l F^{(\sigma\otimes\bar{\sigma})}_m\quad\forall k\in\{1,\ldots, 8\},\\
[F_k,F_l] = \sum^{8}_{m=1}if_{klm} F_m\quad\forall k,l\in\{1,\ldots, 8\},\label{eq:com_F}\\
[F^{(\sigma\otimes\bar{\sigma})}_k,F^{(\sigma\otimes\bar{\sigma})}_l] = \sum^{8}_{m=1}if_{klm} F^{(\sigma\otimes\bar{\sigma})}_m\quad\forall k,l\in\{1,\ldots, 8\},\label{eq:com_F_sigma}\\
\text{Tr}(F_kF_l) = \frac{\delta_{kl}}{2}\quad\forall k,l\in\{1,\ldots, 8\},\label{eq:trace}\\
f_{klm} = -2i\text{Tr}\left([F_k,F_l]F_m\right)\quad\forall k,l,m\in\{1,\ldots, 8\},\label{eq:fklm}\\
d_{klm} = 2\text{Tr}\left(\{F_k,F_l\}F_m\right)\quad\forall k,l,m\in\{1,\ldots, 8\}.\label{eq:dklm}
\end{gather}
For definitions and notations, confer \autoref{sec:GMO_formula}.\par
Let us begin with the quadratic Casimir operator $C^{(\sigma\otimes\bar{\sigma})}$:
\begin{gather*}
C^{(\sigma\otimes\bar{\sigma})}\coloneqq \sum^8_{k,l=1}\delta_{kl}F^{(\sigma\otimes\bar{\sigma})}_kF^{(\sigma\otimes\bar{\sigma})}_l.
\end{gather*}
In \autoref{sec:GMO_formula}, we stated that the quadratic Casimir operator commutes with $F^{(\sigma\otimes\bar{\sigma})}_k$ for every $k\in\{1,\ldots, 8\}$. This can be shown with a straightforward calculation using \autoref{eq:com_F_sigma}:
\begin{align*}
\left[C^{(\sigma\otimes\bar{\sigma})}, F^{(\sigma\otimes\bar{\sigma})}_k\right] &= \sum^8_{l,m=1}\delta_{lm}\left[F^{(\sigma\otimes\bar{\sigma})}_lF^{(\sigma\otimes\bar{\sigma})}_m, F^{(\sigma\otimes\bar{\sigma})}_k\right]\\
&= \sum^8_{l,m=1}\delta_{lm}\left(F^{(\sigma\otimes\bar{\sigma})}_l\left[F^{(\sigma\otimes\bar{\sigma})}_m,F^{(\sigma\otimes\bar{\sigma})}_k\right] + \left[F^{(\sigma\otimes\bar{\sigma})}_l,F^{(\sigma\otimes\bar{\sigma})}_k\right]F^{(\sigma\otimes\bar{\sigma})}_m\right)\\
&= \sum^8_{l,m,n=1}\delta_{lm}\left(if_{mkn}F^{(\sigma\otimes\bar{\sigma})}_lF^{(\sigma\otimes\bar{\sigma})}_n + if_{lkn}F^{(\sigma\otimes\bar{\sigma})}_nF^{(\sigma\otimes\bar{\sigma})}_m\right)\\
&= \sum^8_{l,n=1}if_{lkn}\left\{F^{(\sigma\otimes\bar{\sigma})}_l,F^{(\sigma\otimes\bar{\sigma})}_n\right\}\\
&= 0\quad\forall k\in\{1,\ldots,8\}
\end{align*}
The last line follows from the fact that $f_{lkn}$ is antisymmetric in $l$ and $n$ -- this follows directly from \autoref{eq:fklm} -- and $\left\{F^{(\sigma\otimes\bar{\sigma})}_l,F^{(\sigma\otimes\bar{\sigma})}_n\right\}$ is symmetric in $l$ and $n$. A contraction of an antisymmetric with a symmetric tensor always yields zero.\par
In \autoref{sec:GMO_formula}, we also stated that the quadratic Casimir operator $C^{(\sigma\otimes\bar{\sigma})}$ is a constant times the identity. The property that an operator commutes with all generators of a connected, simple, and compact Lie group is sufficient to show that the operator has to be constant for multiplets, nevertheless, we will show explicitly that $C^{(\sigma\otimes\bar{\sigma})} = c^{(\sigma\otimes\bar{\sigma})}\cdot\mathbb{1}$ for some $c^{(\sigma\otimes\bar{\sigma})}\in\mathbb{R}\backslash\{0\}$. In order to do this, we will show that $C^{(\sigma\otimes\bar{\sigma})}$ is a singlet under $\text{SU}(3)$. First, consider the transformation of $F_k\equiv F^{(3\otimes\bar{3})}_k$ under $A\in\text{SU}(3)$ (cf. \autoref{sec:GMO_formula}):
\begin{gather*}
\sum^{8}_{l=1}\left(D^{(8)}(A)\right)_{kl}F_l \coloneqq D^{(8)}(A)(F_k) = D^{(3\otimes\bar{3})}(A)(F_k) = AF_kA^\dagger\quad\forall k\in\{1,\ldots,8\},
\end{gather*}
where $\left(D^{(8)}(A)\right)_{kl}$ are the coefficients mediating the transformation of $F_k$. As the matrices $F_k$ are Hermitian, the coefficients $\left(D^{(8)}(A)\right)_{kl}$ are real:
\begin{align*}
\sum^{8}_{l=1}\left(D^{(8)}(A)\right)^\ast_{kl}F_l &= \left(\sum^{8}_{l=1}\left(D^{(8)}(A)\right)_{kl}F_l\right)^\dagger = \left(AF_kA^\dagger\right)^\dagger = AF_kA^\dagger\\
&= \sum^{8}_{l=1}\left(D^{(8)}(A)\right)_{kl}F_l\quad\forall k\in\{1,\ldots,8\}\\
\Rightarrow \left(D^{(8)}(A)\right)^\ast_{kl} &= \left(D^{(8)}(A)\right)_{kl}\quad\forall k,l\in\{1,\ldots, 8\}.
\end{align*}
Using the cyclicity of the trace and \autoref{eq:trace}, we also find that $D^{(8)}(A)$ is unitary:
\begin{align*}
\delta_{kl} &= 2\text{Tr}\left(F_kF_l\right) = 2\text{Tr}\left(F_kA^\dagger AF_lA^\dagger A\right) = 2\text{Tr}\left(AF_kA^\dagger AF_lA^\dagger\right)\\
&= \sum^8_{m,n=1}\left(D^{(8)}(A)\right)_{km}\left(D^{(8)}(A)\right)_{ln}2\text{Tr}\left(F_mF_n\right)\\
&= \sum^8_{m,n=1}\left(D^{(8)}(A)\right)_{km}\left(D^{(8)}(A)\right)_{ln}\delta_{mn}\\
&= \sum^8_{m=1}\left(D^{(8)}(A)\right)_{km}\left(D^{(8)}(A)^\text{T}\right)_{ml}\quad\forall k,l\in\{1,\ldots,8\}\\
\Rightarrow& \left(D^{(8)}(A)^{-1}\right)_{kl} = \left(D^{(8)}(A)^\text{T}\right)_{kl} = \left(D^{(8)}(A)^\dagger\right)_{kl}\quad\forall k,l\in\{1,\ldots,8\}.
\end{align*}
The fact that the coefficients $\left(D^{(8)}(A)\right)_{kl}$ are real was used in the last line.\par
Using \autoref{eq:sigma8} and \autoref{eq:Fdef}, it is easy to see that the coefficients $\left(D^{(8)}(A)\right)_{kl}$ also mediate the transformation of $F^{(\sigma\otimes\bar{\sigma})}_k$ under $A\in\text{SU}(3)$:
\begin{align*}
D^{(\sigma\otimes\bar{\sigma})}(A)\left(F^{(\sigma\otimes\bar{\sigma})}_k\right) &= \left(DD^{(\sigma)}\vert_\mathbb{1}\circ D^{(8)}(A)\right)\left(F_k\right)\\
&= \sum^{8}_{l=1} \left(D^{(8)}(A)\right)_{kl} F^{(\sigma\otimes\bar{\sigma})}_l\quad\forall k\in\{1,\ldots,8\}.
\end{align*}
Now consider the transformation of the quadratic Casimir $C^{(\sigma\otimes\bar{\sigma})}$ under $A\in\text{SU}(3)$:
\begin{align*}
D^{(\sigma\otimes\bar{\sigma})}(A)\left(C^{(\sigma\otimes\bar{\sigma})}\right) &=D^{(\sigma)}(A)\left(\sum^8_{k,l=1}\delta_{kl}F^{(\sigma\otimes\bar{\sigma})}_kF^{(\sigma\otimes\bar{\sigma})}_l\right)D^{(\sigma)}(A)^\dagger\\
&=\sum^8_{k,l=1}\delta_{kl} D^{(\sigma)}(A)F^{(\sigma\otimes\bar{\sigma})}_k D^{(\sigma)}(A)^\dagger D^{(\sigma)}(A)F^{(\sigma\otimes\bar{\sigma})}_l D^{(\sigma)}(A)^\dagger\\
&= \sum^8_{k,l=1}\delta_{kl} D^{(\sigma\otimes\bar{\sigma})}(A)\left(F^{(\sigma\otimes\bar{\sigma})}_k\right) D^{(\sigma\otimes\bar{\sigma})}(A)\left(F^{(\sigma\otimes\bar{\sigma})}_l\right)\\
&= \sum^8_{k,l,m,n=1}\delta_{kl}\left(D^{(8)}(A)\right)_{km}\left(D^{(8)}(A)\right)_{ln} F^{(\sigma\otimes\bar{\sigma})}_m F^{(\sigma\otimes\bar{\sigma})}_n\\
&= \sum^8_{k,l,m,n=1}\delta_{kl}\left(D^{(8)}(A)^\text{T}\right)_{mk}\left(D^{(8)}(A)^\text{T}\right)_{nl} F^{(\sigma\otimes\bar{\sigma})}_m F^{(\sigma\otimes\bar{\sigma})}_n\\
&= \sum^8_{k,l,m,n=1}\delta_{kl}\left(D^{(8)}(A^{-1})\right)_{mk}\left(D^{(8)}(A^{-1})\right)_{nl} F^{(\sigma\otimes\bar{\sigma})}_m F^{(\sigma\otimes\bar{\sigma})}_n\\
&= \sum^8_{m,n=1}\delta_{mn}F^{(\sigma\otimes\bar{\sigma})}_mF^{(\sigma\otimes\bar{\sigma})}_n\\
&= C^{(\sigma\otimes\bar{\sigma})}.
\end{align*}
This shows that $C^{(\sigma\otimes\bar{\sigma})}$ is a singlet in $\sigma\otimes\bar{\sigma}$. In \autoref{sec:GMO_formula}, we showed that there is exactly one singlet in $\sigma\otimes\bar{\sigma}$, namely the multiples of the identity. This means that $C^{(\sigma\otimes\bar{\sigma})} = c^{(\sigma\otimes\bar{\sigma})}\cdot\mathbb{1}$ for some $c^{(\sigma\otimes\bar{\sigma})}\in\mathbb{C}$. Furthermore, $C^{(\sigma\otimes\bar{\sigma})}$ is an Hermitian matrix, as it is the sum of squares of Hermitian matrices. This enforces $c^{(\sigma\otimes\bar{\sigma})}$ to be real. Moreover, $c^{(\sigma\otimes\bar{\sigma})}$ cannot be zero, since if it was zero, we would find:
\begin{gather*}
0 = v^\dagger C^{(\sigma\otimes\bar{\sigma})}(v) = \sum^8_{k=1} \left|F^{(\sigma\otimes\bar{\sigma})}_k(v)\right|^2\quad\forall v\in V,
\end{gather*}
where $V$ is the vector space $\sigma$ acts on. However, this equation can only be satisfied, if all $F^{(\sigma\otimes\bar{\sigma})}_k$ are zero. This contradicts the statement that the $F^{(\sigma\otimes\bar{\sigma})}_k$ form an octet. Therefore, we find $C^{(\sigma\otimes\bar{\sigma})} = c^{(\sigma\otimes\bar{\sigma})}\cdot\mathbb{1}$ for some $c^{(\sigma\otimes\bar{\sigma})}\in\mathbb{R}\backslash\{0\}$.\par
Let us now turn to the operators $\tilde{F}^{(\sigma\otimes\bar{\sigma})}_k$:
\begin{gather*}
\tilde{F}^{(\sigma\otimes\bar{\sigma})}_k\coloneqq \sum^{8}_{l,m=1}f_{klm}F^{(\sigma\otimes\bar{\sigma})}_lF^{(\sigma\otimes\bar{\sigma})}_m\quad\forall k\in\{1,\ldots,8\}.
\end{gather*}
In \autoref{sec:GMO_formula}, we found that these operators are elements of the octet spanned by $F^{(\sigma\otimes\bar{\sigma})}_k$. We also claimed that these operators are linearly independent. We can convince ourselves that this is true by showing that the matrices $\tilde{F}^{(\sigma\otimes\bar{\sigma})}_k$ form an invariant subspace of the octet. For this, consider the transformation of the structure constants $f_{klm}$ under $A\in\text{SU}(3)$. With \autoref{eq:fklm}, we obtain:
\begin{align*}
f_{klm} &= -2i\text{Tr}\left(\left[F_k,F_l\right]F_m\right) = -2i\text{Tr}\left(\left[F_kA^\dagger A,F_lA^\dagger A\right]F_mA^\dagger A\right)\\
&= -2i\text{Tr}\left(\left[AF_kA^\dagger,AF_lA^\dagger\right]AF_mA^\dagger\right)\\
&= -2i\sum^8_{n,r,s=1}\left(D^{(8)}(A)\right)_{kn}\left(D^{(8)}(A)\right)_{lr}\left(D^{(8)}(A)\right)_{ms}\text{Tr}\left(\left[F_n,F_r\right]F_s\right)\\
&= \sum^8_{n,r,s=1}\left(D^{(8)}(A)\right)_{kn}\left(D^{(8)}(A)\right)_{lr}\left(D^{(8)}(A)\right)_{ms}f_{nrs}\quad\forall k,l,m\in\{1,\ldots,8\}.
\end{align*}
Using the fact that $D^{(8)}(A)$ is unitary and real, we can rewrite this as:
\begin{gather*}
\sum^8_{k=1}\left(D^{(8)}(A)\right)_{kn}f_{klm} = \sum^8_{r,s=1}\left(D^{(8)}(A)\right)_{lr}\left(D^{(8)}(A)\right)_{ms}f_{nrs}\quad\forall l,m,n\in\{1,\ldots,8\}.
\end{gather*}
Now consider the transformation of $\tilde{F}^{(\sigma\otimes\bar{\sigma})}_k$ under $A\in\text{SU}(3)$:
\begin{align*}
D^{(\sigma\otimes\bar{\sigma})}(A)\left(\tilde{F}^{(\sigma\otimes\bar{\sigma})}_k\right) &= \sum^{8}_{l,m=1}f_{klm}D^{(\sigma\otimes\bar{\sigma})}(A)\left(F^{(\sigma\otimes\bar{\sigma})}_l\right) D^{(\sigma\otimes\bar{\sigma})}(A)\left(F^{(\sigma\otimes\bar{\sigma})}_m\right)\\
&= \sum^{8}_{l,m,n,r=1}f_{klm}\left(D^{(8)}(A)\right)_{ln}\left(D^{(8)}(A)\right)_{mr}F^{(\sigma\otimes\bar{\sigma})}_nF^{(\sigma\otimes\bar{\sigma})}_r\\
&= \sum^{8}_{l,m,n,r=1}\left(D^{(8)}(A^{-1})\right)_{nl}\left(D^{(8)}(A^{-1})\right)_{rm}f_{klm}F^{(\sigma\otimes\bar{\sigma})}_nF^{(\sigma\otimes\bar{\sigma})}_r\\
&= \sum^8_{n,r,s=1}\left(D^{(8)}(A^{-1})\right)_{sk}f_{snr}F^{(\sigma\otimes\bar{\sigma})}_nF^{(\sigma\otimes\bar{\sigma})}_r\\
&= \sum^8_{s=1}\left(D^{(8)}(A)\right)_{ks}\tilde{F}^{(\sigma\otimes\bar{\sigma})}_s.
\end{align*}
As the matrices $\tilde{F}^{(\sigma\otimes\bar{\sigma})}_k$ are only mapped into each other under $A\in\text{SU}(3)$, they span an invariant subspace of the octet spanned by $F_k$. However, the octet is a multiplet, thus, the space spanned by $\tilde{F}^{(\sigma\otimes\bar{\sigma})}_k$ is either $\{0\}$ or the entire octet. The invariant space cannot be $\{0\}$, as if it was $\{0\}$, all $\tilde{F}^{(\sigma\otimes\bar{\sigma})}_k$ would have to be zero. Nevertheless, we find for $\tilde{F}^{(\sigma\otimes\bar{\sigma})}_1$ with the help of the structure constants $f_{klm}$, which can be found, for instance, in \cite{Lichtenberg}:
\begin{gather*}
\tilde{F}^{(\sigma\otimes\bar{\sigma})}_1 = \frac{3i}{2}F^{(\sigma\otimes\bar{\sigma})}_1.
\end{gather*}
So if the span of the matrices $\tilde{F}^{(\sigma\otimes\bar{\sigma})}_k$ was $\{0\}$, $F^{(\sigma\otimes\bar{\sigma})}_1$ would be zero which is a contradiction, as the matrices $F^{(\sigma\otimes\bar{\sigma})}_k$ provide a basis for an octet in $\sigma\otimes\bar{\sigma}$. This implies that the matrices $\tilde{F}^{(\sigma\otimes\bar{\sigma})}_k$ span the entire octet spanned by $F^{(\sigma\otimes\bar{\sigma})}_k$. Since there are 8 matrices $\tilde{F}^{(\sigma\otimes\bar{\sigma})}_k$, they have to be linearly independent concluding the proof.\par
In \autoref{sec:GMO_formula}, we aimed to compute the commutator of $F^{(\sigma\otimes\bar{\sigma})}_k$ and $D^{(\sigma\otimes\bar{\sigma})}_l$. In order to do this, we introduced the following formula for $f_{klm}$ and $d_{klm}$:
\begin{gather*}
\sum^8_{m=1}f_{klm}d_{nrm} = -\sum^8_{m=1}\left(f_{nlm}d_{krm} +f_{rlm}d_{nkm}\right)\quad\forall k,l,n,r\in\{1,\ldots,8\}.
\end{gather*}
However, we skipped the derivation of this formula. Thus, we now want to proof it. The proof is actually just a straightforward calculation using \autoref{eq:com_F}, \autoref{eq:dklm}, and the cyclicity of the trace. Consider the following expression:
\begin{align*}
\frac{i}{2}\sum^8_{m=1}f_{klm}d_{nrm} &= i\sum^8_{m=1}f_{klm}\text{Tr}\left(\{F_n,F_r\}F_m\right) = \text{Tr}\left(\{F_n,F_r\}\sum^8_{m=1}if_{klm}F_m\right)\\
&= \text{Tr}\left(\{F_n,F_r\}\left[F_k,F_l\right]\right)\\
&= \text{Tr}\left(F_nF_rF_kF_l + F_rF_nF_kF_l - F_nF_rF_lF_k - F_rF_nF_lF_k\right)\\
&= \text{Tr}\left(F_rF_kF_lF_n + F_nF_kF_lF_r - F_kF_nF_rF_l - F_kF_rF_nF_l\right)\\
&= \text{Tr}\left(F_rF_kF_lF_n + F_nF_kF_lF_r - F_kF_nF_rF_l - F_kF_rF_nF_l\right.\\
&\ \ \ \ \ \,\left.+ F_kF_rF_lF_n - F_nF_kF_rF_l + F_kF_nF_lF_r - F_rF_kF_nF_l\right)\\
&= \text{Tr}\left(F_kF_rF_lF_n + F_rF_kF_lF_n + F_nF_kF_lF_r + F_kF_nF_lF_r\right.\\
&\ \ \ \ \ \,\left. - F_kF_rF_nF_l - F_rF_kF_nF_l - F_nF_kF_rF_l - F_kF_nF_rF_l\right)\\
&= \text{Tr}\left(\{F_k,F_r\}[F_l,F_n] + \{F_n,F_k\}[F_l,F_r]\right)\\
&= -\text{Tr}\left(\{F_k,F_r\}[F_n,F_l] + \{F_n,F_k\}[F_r,F_l]\right)\\
&= -i\sum^8_{m=1}\left(f_{nlm}\text{Tr}\left(\{F_k,F_r\}F_m\right) + f_{rlm}\text{Tr}\left(\{F_n,F_k\}F_m\right)\right)\\
&= \frac{-i}{2}\sum^8_{m=1}\left(f_{nlm}d_{krm} +f_{rlm}d_{nkm}\right)\quad\forall k,l,n,r\in\{1,\ldots,8\}.
\end{align*}
We also postponed the computation of the commutator of $F^{(\sigma\otimes\bar{\sigma})}_k$ and $D^{(\sigma\otimes\bar{\sigma})}_l$ itself to the appendix. For this, apply \autoref{eq:com_F_sigma} and the previous formula:
\begin{align*}
\left[F^{(\sigma\otimes\bar{\sigma})}_k,D^{(\sigma\otimes\bar{\sigma})}_l\right] &= \frac{2}{3}\sum^8_{m,n=1}d_{lmn}\left[F^{(\sigma\otimes\bar{\sigma})}_k,F^{(\sigma\otimes\bar{\sigma})}_mF^{(\sigma\otimes\bar{\sigma})}_n\right]\\
&= \frac{2}{3}\sum^8_{m,n=1}d_{lmn}\left(\left[F^{(\sigma\otimes\bar{\sigma})}_k,F^{(\sigma\otimes\bar{\sigma})}_m\right]F^{(\sigma\otimes\bar{\sigma})}_n + F^{(\sigma\otimes\bar{\sigma})}_m\left[F^{(\sigma\otimes\bar{\sigma})}_k,F^{(\sigma\otimes\bar{\sigma})}_n\right]\right)\\
&= \frac{2}{3}\sum^8_{m,n,r=1}d_{lmn}\left(if_{kmr}F^{(\sigma\otimes\bar{\sigma})}_rF^{(\sigma\otimes\bar{\sigma})}_n + if_{knr}F^{(\sigma\otimes\bar{\sigma})}_mF^{(\sigma\otimes\bar{\sigma})}_r\right)\\
&= \frac{2i}{3}\sum^8_{m,n,r=1}\left(d_{lmn}f_{kmr}F^{(\sigma\otimes\bar{\sigma})}_rF^{(\sigma\otimes\bar{\sigma})}_n + d_{lrm}f_{kmn}F^{(\sigma\otimes\bar{\sigma})}_rF^{(\sigma\otimes\bar{\sigma})}_n\right)\\
&= \frac{2i}{3}\sum^8_{m,n,r=1}\left(d_{lmn}f_{kmr} + d_{lrm}f_{kmn}\right)F^{(\sigma\otimes\bar{\sigma})}_rF^{(\sigma\otimes\bar{\sigma})}_n\\
&= \frac{2i}{3}\sum^8_{m,n,r=1}\left(f_{rkm}d_{lnm} + f_{nkm}d_{rlm}\right)F^{(\sigma\otimes\bar{\sigma})}_rF^{(\sigma\otimes\bar{\sigma})}_n\\
&= -\frac{2i}{3}\sum^8_{m,n,r=1}f_{lkm}d_{rnm}F^{(\sigma\otimes\bar{\sigma})}_rF^{(\sigma\otimes\bar{\sigma})}_n\\
&= \sum^8_{m=1}if_{klm}\left(\frac{2}{3}\sum^8_{n,r=1} d_{mrn}F^{(\sigma\otimes\bar{\sigma})}_rF^{(\sigma\otimes\bar{\sigma})}_n\right)\\
&= \sum^8_{m=1}if_{klm}D^{(\sigma\otimes\bar{\sigma})}_m\quad\forall k,l\in\{1,\ldots,8\}
\end{align*}
Lastly, we want to calculate the transformation of $D^{(\sigma\otimes\bar{\sigma})}_k$ under $A\in\text{SU}(3)$. This computation can be performed analogously to the calculation for the transformation behavior of the matrices $\tilde{F}^{(\sigma\otimes\bar{\sigma})}_k$. First, consider the transformation of $d_{klm}$ under $A\in\text{SU}(3)$. With \autoref{eq:dklm}, we obtain:
\begin{align*}
d_{klm} &= 2\text{Tr}\left(\left\{F_k,F_l\right\}F_m\right) = 2\text{Tr}\left(\left\{F_kA^\dagger A,F_lA^\dagger A\right\}F_mA^\dagger A\right)\\
&= 2\text{Tr}\left(\left\{AF_kA^\dagger,AF_lA^\dagger\right\}AF_mA^\dagger\right)\\
&= \sum^8_{n,r,s=1}\left(D^{(8)}(A)\right)_{kn}\left(D^{(8)}(A)\right)_{lr}\left(D^{(8)}(A)\right)_{ms}2\text{Tr}\left(\left\{F_n,F_r\right\}F_s\right)\\
&= \sum^8_{n,r,s=1}\left(D^{(8)}(A)\right)_{kn}\left(D^{(8)}(A)\right)_{lr}\left(D^{(8)}(A)\right)_{ms}d_{nrs}\quad\forall k,l,m\in\{1,\ldots,8\}.
\end{align*}
Using the fact that $D^{(8)}(A)$ is unitary and real, we can rewrite this as:
\begin{gather*}
\sum^8_{k=1}\left(D^{(8)}(A)\right)_{kn}d_{klm} = \sum^8_{r,s=1}\left(D^{(8)}(A)\right)_{lr}\left(D^{(8)}(A)\right)_{ms}d_{nrs}\quad\forall l,m,n\in\{1,\ldots,8\}.
\end{gather*}
Now consider the transformation of $D^{(\sigma\otimes\bar{\sigma})}_k$ under $A\in\text{SU}(3)$:
\begin{align*}
D^{(\sigma\otimes\bar{\sigma})}(A)\left(D^{(\sigma\otimes\bar{\sigma})}_k\right) &= \frac{2}{3}\sum^{8}_{l,m=1}d_{klm}D^{(\sigma\otimes\bar{\sigma})}(A)\left(F^{(\sigma\otimes\bar{\sigma})}_l\right) D^{(\sigma\otimes\bar{\sigma})}(A)\left(F^{(\sigma\otimes\bar{\sigma})}_m\right)\\
&= \frac{2}{3}\sum^{8}_{l,m,n,r=1}d_{klm}\left(D^{(8)}(A)\right)_{ln}\left(D^{(8)}(A)\right)_{mr}F^{(\sigma\otimes\bar{\sigma})}_nF^{(\sigma\otimes\bar{\sigma})}_r\\
&= \frac{2}{3}\sum^{8}_{l,m,n,r=1}\left(D^{(8)}(A^{-1})\right)_{nl}\left(D^{(8)}(A^{-1})\right)_{rm}d_{klm}F^{(\sigma\otimes\bar{\sigma})}_nF^{(\sigma\otimes\bar{\sigma})}_r\\
&= \frac{2}{3}\sum^8_{n,r,s=1}\left(D^{(8)}(A^{-1})\right)_{sk}d_{snr}F^{(\sigma\otimes\bar{\sigma})}_nF^{(\sigma\otimes\bar{\sigma})}_r\\
&= \sum^8_{s=1}\left(D^{(8)}(A)\right)_{ks}D^{(\sigma\otimes\bar{\sigma})}_s.
\end{align*}

\end{appendix}

\newpage
\pagenumbering{Roman}

\providecommand{\href}[2]{#2}\begingroup\raggedright\endgroup

\newpage
\chapter*{Acknowledgments}
\addcontentsline{toc}{chapter}{Acknowledgments}
\markboth{}{Acknowledgments}

I thank Professor Hubert Spiesberger and Professor Jens Erler for supervision of my master's thesis, insightful conversations, and helpful comments. Furthermore, I would like to thank the members of THEP and, in particular, the members of the ``AG Spiesberger'' for their kind welcome, Professor a. D. J\"urgen K\"orner for his helpful literature recommendations, Christian P. M. Schneider for (often) enlightening conversations and proofreading my thesis, my former roommates Florian Stuhlmann and Alexander Segner for proofreading and technical support, and my family and friends for moral support.

\end{document}